\documentclass[prx, amsfonts, amssymb, amsmath, nofootinbib, showkeys, singlecolumn,superscriptaddress, notitlepage, preprintnumbers]{revtex4-2}

\usepackage{epsfig}
\usepackage{bm}
\usepackage{amssymb}
\usepackage{amsmath}
\usepackage{color}
\usepackage{dcolumn}
\usepackage{tensor}
\usepackage[utf8]{inputenc}
\usepackage{amsmath}
\usepackage{amssymb}
\usepackage{amsthm} 
\usepackage{amscd, amsfonts, mathrsfs}
\usepackage{cases}
\usepackage{verbatim} 
\usepackage{epsf}
\usepackage{xcolor}
\usepackage[bookmarksnumbered,pdfpagelabels=true,plainpages=false,colorlinks=true,
   linkcolor=black,citecolor=black,urlcolor=black]{hyperref}
\usepackage{enumitem}
\usepackage{empheq}

\theoremstyle{plain}
\newtheorem{theorem}{Theorem}

\newtheorem{corollary}[theorem]{Corollary}

\theoremstyle{definition}
\newtheorem{remark}[theorem]{Remark}

\newtheorem{definition}{Definition}

\allowdisplaybreaks

\newcommand{\al}{\alpha}

\newcommand{\Umb}{\textrm{\textbf{U}}}

\newcommand{\norm}[1]{\left\Vert#1\right\Vert}

\newcommand{\be}{\begin{equation}}
\newcommand{\ee}{\end{equation}}
\newcommand{\bea}{\begin{eqnarray}}
\newcommand{\eea}{\end{eqnarray}}

\newcommand{\bml}{\begin{subequations}}
\newcommand{\eml}{\end{subequations}}

\newcommand{\bbm}{\begin{bmatrix}}
\newcommand{\ebm}{\end{bmatrix}}
\newcommand{\bvm}{\begin{vmatrix}}
\newcommand{\evm}{\end{vmatrix}}

\usepackage{soul}

\newcommand{\ISReconstruction}{\mathcal{I}}

\begin{document}
\title{Nonlinear Causality and Strong Hyperbolicity of Einstein-Israel-Stewart Theories of Transient Relativistic Fluid Dynamics}

\date{\today}

\author{Ian Cordeiro}
\email{icordei2@jh.edu}
\affiliation{Hopkins Extreme Materials Institute, Johns Hopkins University, Baltimore, Maryland 21218, USA}

\author{Enrico Speranza}
\email{enrico.speranza@unifi.it}
\affiliation{Department of Physics and Astronomy, University of Florence, 50019 Sesto Fiorentino, Italy}
\affiliation{INFN Sezione di Firenze, 50019 Sesto Fiorentino, Italy}

\author{F\'abio S.\ Bemfica}
\email{fabio.bemfica@ufrn.br}
\affiliation{Department of Mathematics, Vanderbilt University, Nashville, TN, USA}
\affiliation{Escola de Ci\^encias e Tecnologia, Universidade Federal do Rio Grande do Norte, RN, 59072-970, Natal, Brazil}

\author{Marcelo M. Disconzi}
\email{marcelo.disconzi@vanderbilt.edu}
\affiliation{Department of Mathematics, Vanderbilt University, Nashville, TN, USA}

\author{Jorge Noronha}
\email{jn0508@illinois.edu}
\affiliation{Illinois Center for Advanced Studies of the Universe\\ Department of Physics, 
University of Illinois Urbana-Champaign, Urbana, IL 61801, USA}

\begin{abstract}
We present the first complete analysis of nonlinear causality and local well-posedness for a very general class of bulk and shear viscous theories of relativistic transient fluid dynamics, which encompasses (i) the original Israel-Stewart theory derived from entropy-current arguments, (ii) approaches derived from kinetic theory, and (iii) resummed gradient-expansion based formulations as particular subcases. Our work establishes, for the first time, simultaneously necessary and sufficient algebraic conditions for causality, alongside sufficient conditions guaranteeing strong hyperbolicity, in the full nonlinear regime. These results are rigorously proven for both systems coupled to Einstein's equations featuring a dynamic metric and on a fixed background, with or without a cosmological constant, and include baryon conservation  (in the absence of heat/diffusion currents). The conditions are purely algebraic, require no simplifying spacetime symmetry assumptions or a specific equation of state, and allow all transport coefficients to depend on the dissipative currents. We also demonstrate that the normalization, orthogonality, symmetry, and tracelessness physical constraints on the dynamical variables are properly propagated during the lifetime of the solutions. Our results provide a readily usable toolset with which one can investigate the domain of applicability of relativistic viscous fluid dynamics in numerical and phenomenological studies in heavy-ion collisions and astrophysics.
\end{abstract}

\maketitle

\tableofcontents


\section{Introduction}\label{sec:introduction}

The application of viscous fluid dynamics in a myriad of problems, ranging from sub-nuclear (heavy-ion collisions and the quark-gluon plasma) \cite{Heinz:2013th, Gale:2013da, Romatschke:2017ejr} to astrophysical length scales (neutron stars and black hole accretion flows) \cite{TheLIGOScientific:2017qsa, GBM:2017lvd, Monitor:2017mdv,Akiyama_2019, Event_Horizon_Telescope_Collaboration_2022, gravitycollab} reflects the ubiquitous presence of fluid-like behavior across an equally diverse range of out-of-equilibrium systems. Given that these applications often involve strong gravitational fields and relativistic flows, the underlying theory of \emph{relativistic fluid dynamics} must itself adhere to the core principles of relativity.  However, ensuring that relativistic hydrodynamics remains consistent with fundamental principles such as causality (i.e., that information cannot propagate faster than the speed of light) \cite{ChoquetBruhatGRBook} and the covariant stability of the equilibrium state \cite{Hiscock_Lindblom_stability_1983, Gavassino:2021cli,Gavassino:2021kjm,Gavassino:2021owo} is not automatically guaranteed by a simple covariant generalization of the non-relativistic Navier–Stokes equations \cite{LandauLifshitzFluids}. Searching for ways to formulate hydrodynamics consistent with causality and stability is a longstanding challenge, driving intense research efforts for decades, see \cite{Rocha:2023ilf}.

Early formulations of relativistic viscous fluid dynamics, such as those developed by Eckart \cite{EckartViscous} and Landau and Lifshitz \cite{LandauLifshitzFluids}, expressed the dissipative contributions to the energy-momentum tensor and other conserved currents as a truncated derivative expansion constructed using first-order derivatives of the standard hydrodynamic variables of ideal fluid dynamics: temperature, fluid velocity, and chemical potential. The specific form of the constitutive relations that define the dissipative currents in such theories can be derived by imposing that the second law of thermodynamics holds exactly (even for arbitrarily large variations from equilibrium) \cite{LandauLifshitzFluids}. However, these theories were later shown to display undesired features, such as superluminal signal propagation (acausality) \cite{PichonViscous} and exponentially growing instabilities \cite{Hiscock_Lindblom_instability_1985} around the equilibrium state. 

These issues were long thought to be an inherent flaw of the first-order approach. This perception was changed with the advent of the Bemfica-Disconzi-Noronha-Kovtun (BDNK) formalism \cite{Bemfica:2017wps,Kovtun:2019hdm,Bemfica:2019knx,Hoult:2020eho,Bemfica:2020zjp,Abboud:2023hos}, which clarified that the issues in the Eckart and Landau-Lifshitz theories stemmed not from the first-order formulation itself, but from particular choices when defining the hydrodynamic variables out of equilibrium ---a choice referred to as the \emph{hydrodynamic frame}.\footnote{This term should not be confused with the well-known concept of ``frame of reference'' in relativity, see \cite{Kovtun:2012rj}.} Causality and stability of first-order theories are intrinsically tied to the choice of hydrodynamic frame \cite{Bemfica:2017wps}, and the point of the BDNK formalism is to exploit this freedom created by the first-order truncation to employ hydrodynamic frames where the issues present in Eckart and Landau-Lifshitz theories do not appear. In addition to being linearly stable and causal in both linear and nonlinear regimes, BDNK theories have also been shown in \cite{DisconziBemficaRodriguezShaoSobolevConformal,Bemfica:2020zjp} to be strongly hyperbolic \cite{ReulaStrongHyperbolic,Shao:2023psr}, a mathematically rigorous property that guarantees local well-posedness---ensuring the existence and uniqueness of solutions over a finite timescale. 

However, in both heavy-ion and astrophysical contexts, the current standard is still predominantly Israel-Stewart (IS) theory \cite{ISRAEL1976310, MIS-6}---or, more precisely, its modern formulations such as the Denicol-Niemi-Molnár-Rischke (DNMR) equations~\cite{Denicol:2012cn}\footnote{It is this general class of kinetic-theory-derived equations that we consider in this work, and for which we will derive the first simultaneously necessary and sufficient conditions for nonlinear causality.}. In particular, IS-like theories address the acausal and unstable nature of Eckart and Landau-Lifshitz first-order theories by introducing a fundamentally different approach: they promote the viscous deviations to independent dynamic degrees of freedom. These dissipative fields obey their own nonlinear evolution equations, which are of a relaxation type. This means they dynamically relax from their initial non-equilibrium state towards the asymptotic values predicted by the Navier-Stokes (first-order) theory. This formalism, based on expanding the entropy current to second order in deviations from equilibrium, successfully patches the linear-order pathologies \cite{Hiscock_Lindblom_stability_1983,Olson:1989ey} while preserving the physical correspondence to the Eckart or Landau frames. In the linear regime near equilibrium, IS, DNMR, and other formulations, such as resummed Baier-Romatschke-Son-Starinets-Stephanov (BRSSS) theory \cite{Baier:2007ix} are equivalent, with important differences only arising from nonlinear terms.

Unlike BDNK theory, whose nonlinear causality and strong hyperbolicity are now well-established \cite{Bemfica:2017wps,Bemfica:2019knx,Bemfica:2020zjp}, the nonlinear properties of Israel-Stewart-type theories remain largely unknown. This fundamental disparity stems from the intrinsic mathematical structure of each theory. In BDNK, the principal part---the matrix governing the highest-order derivatives and thus the propagation of signals---remains unchanged between the linear and nonlinear regimes\footnote{Following standard practice in the community, by the linear regime we mean not a generic linearization of the system, but rather a linearization about thermodynamic equilibrium, i.e., linearization about states where all viscous contributions vanish.} \cite{Hoult:2023clg,Abboud:2023hos}. This property means that conditions which ensure causality at linear order naturally extend to the full nonlinear theory. In stark contrast, the principal part of IS-type theories undergoes a drastic change in the nonlinear regime. Fields that do not contribute to the principal part in a linear analysis actively participate nonlinearly, introducing new, complicated terms. This vastly increases the complexity of a full nonlinear analysis, making it difficult to establish rigorous guarantees of causality and hyperbolicity beyond linear perturbations.

Furthermore, IS-theory is traditionally applied in a specific hydrodynamic frame, analogous to the Landau or Eckart frames known from first-order approaches. This choice of frame has profound consequences for analyzing nonlinear causality. To illustrate, even for a simple system with a single non-scalar viscous current (like heat diffusion) in the Eckart frame, establishing causality requires understanding the roots of a 5th-order polynomial. Since no general analytic solution exists for quintic polynomials, this makes a complete causal characterization extremely difficult. Conversely, an analogous calculation in the Landau frame does yield a complete, analytic description of its causal region \cite{cordeiro2025diffusion}. This demonstrates that the problem of nonlinear causality in Israel-Stewart theories is not only a major open challenge but also one whose complexity is highly sensitive to the technical aspects concerning the choice of hydrodynamic frame.\footnote{For a formulation of Israel-Stewart theory in a general hydrodynamic frame, see \cite{Noronha:2021syv}.}

To date, general statements characterizing the full region of nonlinear causality are limited to specific, simplified cases of Israel-Stewart-type theories. These include radially expanding fluids \cite{Floerchinger_2018}, fluids where only particle number diffusion is taken into account \cite{cordeiro2025diffusion}, only bulk viscosity\footnote{In \cite{Bemfica:2019cop}, we note that the transport coefficients are allowed to depend on the bulk stress $\Pi$ directly (e.g., $\zeta = \zeta(\varepsilon, P, \Pi)$ where $\varepsilon$ and $P$ are the total energy density and equilibrium pressure, respectively), making it more general than traditional IS theory.} \cite{Bemfica:2019cop}, and only shear viscosity when paired with a magnetic field in the large-field limit \cite{Cordeiro:2023ljz}. For more general systems, progress has been made by deriving separate sets of necessary and sufficient conditions for causality. A significant advance was made for a broad class of equations---encompassing phenomenological Israel-Stewart \cite{MIS-6,Betz:2008me} and the more general DNMR theory \cite{Denicol:2012cn}---describing a single fluid with bulk and shear viscosity (without a baryon current) \cite{Bemfica:2020xym}. The analysis in \cite{Bemfica:2020xym} considered a class of equations notably more general than that in \cite{MIS-6,Betz:2008me}, as it allowed transport coefficients to depend directly on the viscous fluxes $\Pi$ and $\pi^{\mu\nu}$. These DNMR-derived equations include multiple terms that cannot be obtained from a standard entropy current approach, highlighting the complexity of establishing causality in realistic, high-order hydrodynamic theories.

The necessity of a nonlinear analysis is twofold. First, there is a fundamental mathematical difference from the linear case. In a linear causality analysis, the viscous fluxes (e.g., $\pi^{\mu\nu}$, $\Pi$) are treated as first-order perturbations around a background equilibrium state where these fluxes are zero. Consequently, they do not contribute to the principal part of the system of equations, and their nonlinear self-interactions are neglected. In the full nonlinear regime, however, these fluxes are fully dynamical variables. Their values are not small, and as mentioned above they contribute directly to the principal part, leading to causality conditions that explicitly bound their magnitudes. Second, from a physical perspective, many important instabilities---such as the firehose \cite{RosenbluthFirehose,ParkerFirehose} and mirror \cite{Southwood_Kivelson_Mirror} instabilities---are inherently nonlinear phenomena that depend directly on the amplitude of these viscous fluxes. Such phenomena are of significant interest in fields like black hole accretion and plasma physics. Therefore, a nonlinear causal analysis is essential not only for mathematical consistency but also for capturing or ruling out these potential physical phenomena in a hydrodynamical description \cite{Cordeiro:2023ljz}.

In particular, some models of the initial state of heavy-ion collisions can be so far from equilibrium that causality can be violated in up to 75\% of fluid cells at early times in state-of-the-art viscous hydrodynamical simulations \cite{Plumberg:2021bme}. This result, demonstrated in simulations using the DNMR equations, highlights a critical issue when  applying hydrodynamics at very early times in heavy-ion collisions \cite{Plumberg:2021bme,Chiu:2021muk,ExTrEMe:2023nhy}. 
As the system evolves, a comparable fraction of cells continues to exhibit causality violations. The measured extent of this acausal behavior is currently an underestimate, resulting from the use of conditions obtained in Ref.~\cite{Bemfica:2020xym} that are necessary albeit not sufficient to define the exact causal region. Thus, determining where hydrodynamic theories are consistent with causality and stability is both an outstanding theoretical question regarding how relativistic hydrodynamics should be formulated and an imperative phenomenological tool for extracting physically consistent initial conditions for hydrodynamics while removing or modifying those that are not. For a detailed study and discussion of the effects of acausality-driven instabilities in this context, see \cite{Gavassino:2025bsn}. 

Moreover, the persistence of acausal solutions in simulations is a profound physical concern. A foundational result by Gavassino \cite{Gavassino:2021owo} demonstrates that for any Lorentz covariant dissipative theory, acausality necessarily implies instability in some Lorentz frame. This conclusion  stems from a general physical mechanism: acausal propagation allows the thermodynamic arrow of time to be reversed via a Lorentz boost, converting dissipation into growth. Consequently, a theory admitting widespread nonlinearly acausal solutions can be physically pathological. The consequences of acausality for numerical simulations have been explored in \cite{Gavassino:2025bsn}. This underscores the critical need for the simultaneously necessary and sufficient nonlinear causality conditions we provide for the first time in this work, which are essential for determining the physical validity of hydrodynamic simulations.\

This paper resolves several major open questions in the field regarding the nonlinear properties and validity of the generalized Israel-Stewart class of theories that was addressed in \cite{Bemfica:2020xym}. First, our work extends the analysis to a dynamic spacetime metric and a non-zero baryon current, achieving two primary results: (i) we derive simultaneously necessary and sufficient algebraic inequalities that define the exact region of nonlinear causality, and (ii) we rigorously prove strongly hyperbolicity and local well-posedness under the additional algebraic hypotheses of Theorem~\ref{thm:nodiffusionstronghyperbolicity}.

These results have significant physical ramifications. By incorporating a baryon current, our conditions become applicable to a wide range of baryon-rich environments like neutron star mergers and low-energy heavy-ion collisions, moving beyond the zero net baryon limit relevant for ultrarelativistic heavy-ion collisions (such as at the LHC). Furthermore, coupling to Einstein's equations self-consistently captures the dynamic interplay between the viscous fluid and spacetime, which is essential for modeling systems where the fluid's gravity is non-negligible, such as in black hole accretion flows and neutron star mergers. 

The derived causality conditions are algebraic inequalities that are completely general: they make no assumptions about the spacetime metric $g_{\mu\nu}$ (even allowing for a dynamical coupling to Einstein's equations, the intrinsic symmetries of the system) or an equation of state. Crucially, these conditions are simultaneously necessary and sufficient, thereby resolving contentious regions of unknown causality/acausality identified in \cite{Plumberg:2021bme,ExTrEMe:2023nhy}. \emph{Because such conditions are algebraic, they can in principle be readily checked at each time step in numerical simulations, providing a definitive set of criteria to ensure causality of simulations.}

We also demonstrate that for a broad subclass of these theories---which can be viewed as a second-order extension of Maxwell-Cattaneo hydrodynamics---the necessary and sufficient causality conditions simplify significantly under reasonable physical conditions. This simplification makes the application of these conditions far less computationally intensive, providing a practical tool for future phenomenological studies. Finally, we prove for the first time that all the defining properties of the dynamic variables (e.g., the fact that the shear stress tensor is orthogonal to the flow velocity) are preserved for the lifetime of the solutions of the nonlinear equations.

This paper is organized as follows. In Section~\ref{sec:EOM}, we present the relevant equations of motion that define the class of relativistic viscous fluids we consider, and express them in matrix form. In Section~\ref{sec:causality}, we formally define what causality means for a system of partial differential equations (PDEs), and connect this definition to energy conditions and historical definitions of causality. We then apply this definition to the method of characteristics for PDEs and use it to derive a system of inequalities prescribing the exact region where the theory is nonlinearly causal. We also show how these conditions can generate significantly simpler necessary conditions using the conformal fluid case, though this process works in general. Next, in Section~\ref{sec:MC}, we consider a subclass of theories that maintain the same overall structure of Maxwell-Cattaneo fluids but include second-order transport coefficients. In this regime, we show that the complexity of the necessary and sufficient causality conditions is significantly reduced under a physically realistic assumption. In the remaining sections, we present rigorous statements characterizing the solutions of the general system, along with its compatibility when coupled dynamically to Einstein's equations. Technical material concerning the mathematical proofs presented in this work can be found in the appendices. In Section~\ref{subsec:propagate}, we provide sufficient criterion for solutions to conserve fundamental symmetry constraints such as orthogonality and normalization throughout their entire existence. In Section~\ref{sec:hyperbolicity}, we then provide a set of sufficient conditions for strong hyperbolicity. In Section~\ref{sec:Einstein}, we show that all of the previous results are compatible with Einstein's equations, in the sense that the new system of IS-type theory plus Einstein's equations with or without a (dynamic) cosmological constant satisfy the causality and strong hyperbolicity results with no added assumptions. We conclude this work and provide a few avenues for possible nonlinear analyses in other relevant cases of viscous hydrodynamics in \ref{sec:conclusions}.

\emph{Notation:} We use a mostly-plus metric signature and natural units $c=\hbar=k_B = 1$.

\section{The Equations of Motion}\label{sec:EOM}

We consider a system described by a single fluid, with baryon number density $n$ and average fluid four-velocity $u^\mu$ with normalization $u^\alpha u_\alpha = -1$, along with bulk $\Pi$ and shear $\pi^{\mu\nu}$ viscous corrections. For the sake of definiteness, this theory is formulated in the Landau hydrodynamic frame with zero particle number diffusion. The shear viscosity tensor in particular is symmetric, $\pi^{\mu\nu} = \pi^{\nu\mu}$, traceless $\tensor{\pi}{^\alpha_\alpha} = 0$, and orthogonal to the four-velocity ($u_\alpha\pi^{\alpha\mu} = 0^\mu$). The conserved matter currents describing this fluid are
\bml
\label{eq:conserved_quantities}
\bea
\label{eq:EMtensor-nodiffusion}
T^{\mu\nu} &=& \varepsilon u^\mu u^\nu + (P+\Pi)\Delta^{\mu\nu} + \pi^{\mu\nu},\\
J^\mu &=& nu^\mu,
\eea
\eml
where $T^{\mu\nu}$ is the energy-momentum tensor and $J^\mu$ is the baryon number current. Per the usual convention, $\varepsilon$ and $P$ are the energy density and equilibrium pressure, respectively, and $\Delta^{\mu\nu} = g^{\mu\nu} + u^\mu u^\nu$ is the projection tensor orthogonal to $u^\mu$. Conservation of energy $u_\alpha\nabla_\beta T^{\alpha\beta} = 0$ and momentum $\tensor{\Delta}{^\mu_\alpha}\nabla_\beta T^{\alpha\beta} = 0$, and baryon number $\nabla_\alpha J^\alpha = 0$ provide us with the respective equations of motion
\bml
\label{eq:conservationlaw}
\bea
\label{eq:energy-nodiffusion}
0 &=& u^\alpha\nabla_\alpha \varepsilon + E\nabla_\alpha u^\alpha + \pi^{\alpha\beta}\nabla_\beta u_\alpha,\\
\label{eq:momentum-nodiffusion}
0^\mu &=& Eu^\alpha\nabla_\alpha u^\mu + \Delta^{\mu\alpha}\nabla_\alpha(P+\Pi) + \nabla_\alpha\pi^{\mu\alpha} - u^\mu \pi^{\alpha\beta}\nabla_\beta u_\alpha,\\
\label{eq:number-nodiffusion}
0 &=& u^\alpha\nabla_\alpha n + n\nabla_\alpha u^\alpha,
\eea
\eml
where $E = \varepsilon + P + \Pi$ for convenience. We note that the constraints $u^\alpha u_\alpha = -1$ and $u_\alpha \pi^{\alpha\mu}=0$ have been used explicitly in deriving Eqs.\ \eqref{eq:energy-nodiffusion} and \eqref{eq:momentum-nodiffusion}. Consequently, the system \eqref{eq:conservationlaw} constitutes 5 independent dynamical equations. The supplemental DNMR equations required to complete the system (bulk $\Pi$ and 5 independent $\pi^{\mu\nu}$) are presented next. The promotion of $u^0$ and all 16 components of $\pi^{\mu\nu}$ to independent variables will be performed later when constructing the augmented system for the causality and strong hyperbolicity analysis.

\subsection{Transient Relativistic Fluid Dynamical Equations}\label{subsec:IS}

It is well-known that both the Israel-Stewart equations and the DNMR equations can be obtained from kinetic theory, for a review see \cite{Rocha:2023ilf}. A critical distinction of the DNMR formalism is that, while also derived from the Boltzmann equation, its final equations of motion are obtained via a systematic truncation procedure in powers of Knudsen and inverse Reynolds numbers \cite{Denicol:2012cn} for practical use. In this sense, DNMR can also be regarded as an effective theory consistently formulated in powers of Knudsen and inverse Reynolds numbers and, as such, their validity extends beyond the kinetic theory domain originally used in their derivation. The truncation in DNMR also means that, unlike the original phenomenological Israel-Stewart equations---which were constructed to exactly satisfy the second law of thermodynamics by design---the full, non-linear DNMR equations do not guarantee a manifestly non-negative entropy production rate (though this property is formally recovered when considering solutions sufficiently close to equilibrium). In fact, the second law of thermodynamics is satisfied in the underlying kinetic theory, but the truncated DNMR hydrodynamic equations do not preserve this property in a closed algebraic form (the additional terms in DNMR, which are absent in the original IS theory, are responsible for this difference). Nevertheless, we remark that it is possible to formally obtain IS relaxation equations using the DNMR formalism in the 14-moments approximation, see \cite{Denicol_2012}.

The relaxation equations for the dissipative currents motivated by DNMR that we use in this work can be expressed as
\bml
\label{eq:supplementalIS}
\bea
\label{eq:bulk-nodiffusion}
\tau_\Pi u^\alpha\nabla_\alpha\Pi + \Pi &=& -\zeta\nabla_\alpha u^\alpha -\delta_{\Pi\Pi}\Pi\nabla_\alpha u^\alpha - \lambda_{\Pi\pi}\pi^{\alpha\beta}\sigma_{\alpha\beta},\\
\label{eq:shear-nodiffusion}
\tau_\pi\tensor{\Delta}{^{\mu\nu}_{\beta\gamma}}u^\alpha\nabla_\alpha\pi^{\beta\gamma} + \pi^{\mu\nu} &=& -2\eta\sigma^{\mu\nu} - \delta_{\pi\pi}\pi^{\mu\nu}\nabla_\alpha u^\alpha -\tau_{\pi\pi}\pi^{\alpha\langle\mu}\tensor{\sigma}{^{\nu\rangle}_\alpha}-\lambda_{\pi\Pi}\Pi\sigma^{\mu\nu}.
\eea
\eml
Here, $\sigma^{\mu\nu} = \tensor{\Delta}{^{\mu\nu}_{\alpha\beta}}\nabla^\alpha u^\beta$ is the shear tensor in terms of the traceless symmetric projection tensor orthogonal to $u^\mu$ expressed as $\tensor{\Delta}{^{\mu\nu}_{\rho\sigma}} = \frac{1}{2}\left(\tensor{\Delta}{^\mu_\rho}\tensor{\Delta}{^\nu_\sigma} + \tensor{\Delta}{^\mu_\sigma}\tensor{\Delta}{^\nu_\rho}\right) - \frac{1}{3}\Delta^{\mu\nu}\Delta_{\rho\sigma}$. In DNMR \cite{Denicol_2012}, the transport coefficients $\{\tau_\Pi, \zeta, \tau_\pi, \eta\}$ represent contributions of first-order deviations from equilibrium, whereas $\{\delta_{\Pi\Pi}, \lambda_{\Pi\pi}, \delta_{\pi\pi}, \tau_{\pi\pi}, \lambda_{\pi\Pi}\}$ are second-order contributions. Here, we go beyond DNMR since in our analysis all of these transport coefficients can depend not only on $\varepsilon$ and $n$, but they can also depend on appropriate combinations of $\Pi$ and $\pi_{\mu\nu}$, e.g., $\zeta = \zeta(\varepsilon,n,\Pi, \pi_{\mu\nu}\pi^{\mu\nu})$. This illustrates how general the class of theories considered in this work is. Combining Eqs.~\eqref{eq:conservationlaw}--\eqref{eq:supplementalIS} provides us with the full prescription of equations of motion that define the generalized IS-theory we consider throughout the paper. \footnote{We note that in this work we have not included terms involving the kinematic vorticity $\omega_{\mu\nu}$ that are, in principle, present in the full set of DNMR equations, see \cite{Denicol:2012cn}. Such terms do appear from standard entropy current derivations \cite{MIS-6} and their phenomenological relevance to simulations is not clear at the moment. Thus, we leave their consideration to dedicated future work.}

\subsection{Covariant Formulation with Auxiliary Variables}\label{sec:augmented_system}

In the formulation of the DNMR equations above, the physical degrees of freedom are the energy density $\varepsilon$, the baryon density $n$, the three independent spatial components $u^i$ of the fluid velocity, the bulk viscous pressure $\Pi$, and the five independent components of the shear stress tensor $\pi^{\mu\nu}$. However, to facilitate a fully covariant analysis of causality and strong hyperbolicity via the characteristic determinant, it is advantageous to promote all components of $u^\mu$ and $\pi^{\mu\nu}$ to independent dynamical variables. This approach preserves manifest covariance, which is essential for the subsequent mathematical analysis. The pedagogical example below elucidates the construction of this ``augmented system'' and the circumstances under which it is equivalent to the original physical equations.

\begin{quote}
\textbf{Example: Ideal Fluid's Augmented System.}
To illustrate the core idea, let us first consider the simpler case of an ideal fluid, with an energy-momentum tensor $T^{\mu\nu}=\varepsilon u^\mu u^\nu+P\Delta^{\mu\nu}$ and a barotropic equation of state $P=P(\varepsilon)$. In the standard physical formulation, the dynamical variables are $\varepsilon$ and $u^i$, with $u^0$ determined algebraically from the normalization condition $u^\mu u_\mu = -1$. The equations of motion are derived from the conservation law $\nabla_\nu T^{\mu\nu}=0$, which yields:
\[
(\varepsilon+P)(u^\mu \nabla_\nu u^\nu+u^\nu\nabla_\nu u^\mu)+u^\mu u^\nu \nabla_\nu \varepsilon+c_s^2\Delta^{\mu\nu}\nabla_\nu \varepsilon=0,
\]
where $c_s^2=\partial P/\partial \varepsilon$ is the squared sound speed at zero baryon density. This covariant equation can be projected into 4 independent components, that is the projections $u_\mu\nabla_\nu T^{\mu\nu}=0$ and $C_i^\mu\nabla_\nu \tensor{T}{_\mu^\nu}=0$, where $\tensor{C}{^\nu_i}\equiv \tensor{g}{^\nu_i}-\tensor{g}{^\nu_0} u_i/u_0$ are three linearly independent four-vectors orthogonal to $u^\mu$. The resulting equations are:
\begin{align*}
u^\nu \nabla_\nu \varepsilon+(\varepsilon+P)C^\nu_j\nabla_\nu u^j&=0, \\
c_s^2C_i^\nu \nabla_\nu \varepsilon+(\varepsilon+P)A_{ij}u^\nu\nabla_\nu u^j&=0,
\end{align*}
where $A_{ij}\equiv \tensor{C}{^\mu_i}\Delta_{\mu\nu}\tensor{C}{^\nu_j}$ is an invertible matrix \cite{ChoquetBruhatGRBook}. While this formulation is physically equivalent to the original system, the manifest Lorentz covariance is lost, as the splitting relies on a specific choice of spatial components.

To construct a covariant ``augmented system'' for the ideal fluid, we promote $u^0$ to an independent variable, making the full set $\{\varepsilon, u^\mu\}$ dynamical. We now require five independent dynamical equations. These are obtained by splitting the conservation law into the energy equation $u_\mu \nabla_\nu T^{\mu\nu}=0$ and the four momentum equations $\Delta^\mu_{\ \nu} \nabla_\alpha T^{\alpha\nu}=0^\mu$. However, these five equations are not independent; contracting the four momentum equations with $u_\mu$ yields an identity ($0=0$), leaving an underdetermined system.

To resolve this, we explicitly use the derivative of the normalization constraint, $u_\nu \nabla_\alpha u^\nu=0$, to modify the momentum equations. This yields the following set of five independent, covariant equations for the augmented system:
\bml
\label{eq:Euler}
\begin{align}
\label{eq:Euler_a}
u^\nu \nabla_\nu \varepsilon+(\varepsilon+P)\nabla_\nu u^\nu&=0, \\
\label{eq:Euler_b}
c_s^2 \Delta^{\mu\nu} \nabla_\nu \varepsilon+(\varepsilon+P)u^\nu\nabla_\nu u^\mu&=0.
\end{align}
\eml
Crucially, if we now contract \eqref{eq:Euler_b} with $u_\mu$, we no longer find an identity but instead obtain the dynamical equation $u_\mu u^\nu\nabla_\nu u^\mu=0$, which can be recast as $u^\nu\nabla_\nu \Phi = 0$, where $\Phi \equiv u^\mu u_\mu + 1$. This equation governs the evolution of the constraint. The augmented system \eqref{eq:Euler_a}--\eqref{eq:Euler_b} thus admits a wider class of solutions than the original physical system, including those where $\Phi$ is an arbitrary constant. Therefore, it is instructive to use $\Delta^{\mu\nu}=g^{\mu\nu}-u^\mu u^\nu/u_\alpha u^\alpha$ in order to keep the orthogonality between $\Delta^{\mu\nu}$ and $u^\mu$ whenever $\Phi\ne 0$.

The equivalence between the systems is established as follows. Suppose $(\varepsilon, u^\mu)$ is a sufficiently regular solution of the augmented system \eqref{eq:Euler}. If the initial data satisfies the physical constraint $\Phi=0$, then the evolution equation $u^\nu\nabla_\nu \Phi = 0$, obtained by simply contracting \eqref{eq:Euler_b} with $u^\mu$ and which is locally well-posed, guarantees that $\Phi=0$ for the lifetime of solutions. Therefore, for such initial data, the solution of the augmented system also satisfies the original physical system with $u^\mu u_\mu = -1$.
\end{quote}
This ideal fluid example directly motivates our treatment of the full Israel-Stewart theory. For the DNMR equations, the relevant constraints on the dynamical variables are:
\begin{align}
\label{Constraints}
u^\alpha u_\alpha = -1, \quad \pi^{[\mu\nu]} = 0, \quad u_\alpha \pi^{\alpha\mu} = 0^\mu, \quad \pi^\alpha_{\ \alpha} = 0,
\end{align}
where $\pi^{[\mu\nu]} = (\pi^{\mu\nu} - \pi^{\nu\mu})/2$ denotes antisymmetrization of the respective indices. Through a similar procedure---promoting all components of $u^\mu$ and $\pi^{\mu\nu}$ to independent variables and using \eqref{Constraints} explicitly---we construct the augmented system for the DNMR equations, which is given by:
\bml
\label{eq:propsystem}
\bea
0 &=& u^\alpha\nabla_\alpha \varepsilon + E\nabla_\alpha u^\alpha + \tensor{\pi}{^\alpha_\nu}\nabla_\alpha u^\nu,\\
\label{eq:momentumprop}
0^\mu &=& Eu^\alpha\nabla_\alpha u^\mu + \Delta^{\mu\alpha}\nabla_\alpha(P+\Pi) + \nabla_\alpha\pi^{\mu\alpha} - u^\mu \pi^{\alpha\beta}\nabla_\beta u_\alpha,\\
0 &=& n\tensor{g}{^\alpha_\nu}\nabla_\alpha u^\nu + u^\alpha\nabla_\alpha n,\\
\label{eq:bulkprop}
\Pi + \tau_\Pi u^\alpha\nabla_\alpha\Pi &=& -\left[\left(\zeta + \delta_{\Pi\Pi}\Pi\right)\tensor{g}{^\alpha_\nu} + \lambda_{\Pi\pi}\tensor{\pi}{^\alpha_\nu}\right]\nabla_\alpha u^\nu ,\\
\label{eq:shearprop}
\pi^{\mu\beta} + \tau_\pi u^\alpha\tensor{g}{^\mu_\nu}\nabla_\alpha\pi^{\nu\beta}  &=& -\tensor{\mathcal{C}}{_\pi^\mu^\beta^\alpha_\nu}\nabla_\alpha u^\nu,
\eea
\eml
where 
\begin{align}
\label{eq:C_def}
\tensor{\mathcal{C}}{_\pi^\mu^\beta^\alpha_\nu}&\equiv \left(\eta + \frac{1}{2}\lambda_{\pi\Pi}\Pi\right)\left(\Delta^{\mu\alpha}\tensor{g}{^\beta_\nu}+\Delta^{\beta\alpha}\tensor{g}{^\mu_\nu}-\frac{2\Delta^{\mu\beta}\Delta^\alpha_\nu}{3}\right)- 2\tau_\pi u^\alpha \tensor{\pi}{^{(\mu}_\nu}u^{\beta)}  + \left(\delta_{\pi\pi}-\frac{\tau_{\pi\pi}}{3}\right)\pi^{\mu\beta}\tensor{g}{^\alpha_\nu}\notag\\
&\quad\hspace{1cm} + \frac{\tau_{\pi\pi}}{2}\left(\tensor{\Delta}{^\alpha^{(\beta}}\tensor{\pi}{^{\mu)}_\nu} + \tensor{\pi}{^\alpha^{(\mu}}\tensor{g}{^{\beta)}_\nu}\right)-\frac{\tau_{\pi\pi}}{3}\Delta^{\mu\beta}\tensor{\pi}{^\alpha_\nu}    
\end{align}
This new formulation necessarily differs from the original physical equations given by the conservation laws $\nabla_\nu T^{\mu\nu}=0^\mu$ and $\nabla_\mu J^\mu=0$ with the supplemental equations \eqref{eq:supplementalIS}. The new equations \eqref{eq:propsystem} are obtained by explicitly substituting the constraints into the conservation laws---already obtained in \eqref{eq:conservationlaw}---and into \eqref{eq:supplementalIS}. For instance, the term $\tensor{\Delta}{^\mu^\nu_\alpha_\beta} u^\lambda \nabla_\lambda \pi^{\alpha\beta}$ in \eqref{eq:shear-nodiffusion} simplifies to $u^\lambda \nabla_\lambda \pi^{\mu\nu}- 2\pi^{\alpha(\mu}u^{\nu)} u^\lambda\nabla_\lambda u_\nu$ only after imposing \eqref{Constraints} with the particular use of $u^\mu u_\mu=-1$ as $u_\alpha \nabla_\mu u^\alpha=0$. Similarly, the conservation law $\nabla_\nu T^{\mu\nu}=0^\mu$ is split into the energy equation $u_\mu \nabla_\nu T^{\mu\nu}=0$ and the momentum equations $\tensor{\Delta}{^\lambda_\mu} \nabla_\nu T^{\mu\nu}=0^\lambda$. Again, the latter appears to provide 4 equations, but it contains a redundancy: contracting it with $u_\lambda$ yields an identity $0=0$. To obtain a linearly independent set of 4 spatial momentum equations and to write the temporal component $u^0$ as a dynamical variable, the differential form of the constraints $u^\alpha u_\alpha = -1$ and $\pi^{\mu\nu}u_\nu=0$ has been used. The same thing happens going from \eqref{eq:shear-nodiffusion} to \eqref{eq:shearprop} after the aforementioned change in $\tensor{\Delta}{^\mu^\nu_\alpha_\beta} u^\lambda \nabla_\lambda \pi^{\alpha\beta}$, giving rise to 16 dynamical independent equations instead the original 5 independent dynamical equations in \eqref{eq:shear-nodiffusion}. Consequently, the augmented system \eqref{eq:propsystem} used for our proofs admits a wider class of solutions than the physical DNMR theory. As such, it allows for initial data where the constraints \eqref{Constraints} are \textit{not} satisfied. Therefore, it is not guaranteed \textit{a priori} that solutions of this augmented system will preserve the constraints, even if they are imposed on the initial data. This fact has not been properly appreciated before in the literature, as it can also matter to numerical simulations \cite{Shen:2014vra,Most:2021rhr}. However, as rigorously proven for the first time in Section \ref{subsec:propagate}, these constraints are dynamically preserved. Consequently, the causality and well-posedness results derived from the analysis of the augmented system are guaranteed to hold for the physical DNMR theory with no further assumptions, which is helpful for finding numerical simulations of such equations.

\subsection{The Characteristic Determinant}\label{subsec:characteristics}

For all scalars, we assume the existence of an invertible, sufficiently smooth equation of state. For example,
\be
P \equiv P(\varepsilon,n).
\ee
As such, we use the notation $P_\varepsilon = (\partial P/\partial \varepsilon)_n$ and likewise for $P_n$. However, as remarked above, we shall also allow the transport coefficients to depend not just on $\varepsilon$ and $n$, but also on the viscous fluxes $\Pi$ and $\pi^{\mu\nu}$ such that, for instance,
\be
\frac{\eta}{\tau_\pi} \equiv \frac{\eta}{\tau_\pi}(\varepsilon,n,\Pi,\pi_{\alpha\beta}\pi^{\alpha\beta}),
\ee
and similarly for all others in Eqs.~\eqref{eq:propsystem}. 

We write equations \eqref{eq:propsystem} into the matrix form
\be
\label{eq:FirstOrderQPDE}
(\mathbb{A}^\alpha\partial_\alpha + \mathbb{B})\Umb = \boldsymbol{0},
\ee
where $\Umb = (u^\nu,\varepsilon,n,\Pi,\pi^{\nu 0},\pi^{\nu 1},\pi^{\nu 2},\pi^{\nu 3})^T\in\mathbb{R}^{23}$ is the column vector of the dynamic degrees of freedom for the system and $^T$ denotes transposition. For convenience, any fixed raised indices will denote representative column vector elements, whereas lowered indices will correspond to row vectors.

Here, $\mathbb{A}^\mu \equiv \mathbb{A}^\mu(\Umb)$ are the $23\times 23$ matrices containing the \emph{coefficients} of the highest (first) order derivative terms but not any coordinate derivatives themselves nor explicit dependence on spacetime points $x$, and $\mathbb{B}$ is a $23\times 23$ matrix containing all zeroth-order terms in coordinate derivatives. In this sense, the system is a \emph{quasilinear} system of partial differential equations. The characteristic surfaces $\{\varphi(x)=0\}$ are defined relative to their normal covectors $\phi_\mu=\partial_\mu \varphi(x)$ for spacetime point $x$. Substituting this ansatz into the system's equations and requiring the existence of non-trivial solutions yields the characteristic equation, which is the foundation of the method for causality, see Chapter II, Appendix I in \cite{Courant_and_Hilbert_book_2}. In this language, $\mathbb{A}^\alpha\phi_\alpha$ is known as the principal part of the theory which has the form
\be
\label{eq:fluidprincipalpart}
\mathbb{A}^\alpha\phi_\alpha =
\begin{bmatrix}
Ez\tensor{g}{^\mu_\nu} - u^\mu w_\nu& v^\mu P_\varepsilon& v^\mu P_n& v^\mu& \tensor{g}{^\mu_\nu}\phi_0& \tensor{g}{^\mu_\nu}\phi_1& \tensor{g}{^\mu_\nu}\phi_2& \tensor{g}{^\mu_\nu}\phi_3\\
E\phi_\nu + w_\nu& z& 0 & 0 & 0_\nu& 0_\nu& 0_\nu& 0_\nu\\
n\phi_\nu& 0 & z& 0 & 0_\nu& 0_\nu& 0_\nu& 0_\nu\\
\tensor{\mathcal{C}}{_\Pi_\nu}& 0& 0& \tau_\Pi z& 0_\nu& 0_\nu& 0_\nu& 0_\nu\\
\tensor{\mathcal{C}}{_\pi^{0\mu}_\nu}& 0^{\mu}& 0^{\mu}& 0^{\mu}& \tau_\pi z\tensor{g}{^\mu_\nu}& \tensor{0}{^\mu_\nu}& \tensor{0}{^\mu_\nu}& \tensor{0}{^\mu_\nu}\\
\tensor{\mathcal{C}}{_\pi^{1\mu}_\nu}& 0^{\mu}& 0^{\mu}& 0^{\mu}& \tensor{0}{^\mu_\nu}& \tau_\pi z\tensor{g}{^\mu_\nu}& \tensor{0}{^\mu_\nu}& \tensor{0}{^\mu_\nu}\\
\tensor{\mathcal{C}}{_\pi^{2\mu}_\nu}& 0^{\mu}& 0^{\mu}& 0^{\mu}& \tensor{0}{^\mu_\nu}& \tensor{0}{^\mu_\nu}& \tau_\pi z\tensor{g}{^\mu_\nu}& \tensor{0}{^\mu_\nu}\\
\tensor{\mathcal{C}}{_\pi^{3\mu}_\nu}& 0^{\mu}& 0^{\mu}& 0^{\mu}& \tensor{0}{^\mu_\nu}& \tensor{0}{^\mu_\nu}& \tensor{0}{^\mu_\nu}& \tau_\pi z\tensor{g}{^\mu_\nu}
\end{bmatrix}
\ee
where we have defined the tensors
\bml
\label{E:IS_Shear_coefficient}
\bea
\tensor{\mathcal{C}}{_\Pi_\nu} &=& (\zeta + \delta_{\Pi\Pi}\Pi)\phi_\nu + \lambda_{\Pi\pi}w_\nu,\\
\tensor{\mathcal{C}}{_\pi^{\beta \mu}_\nu} &=& \left(\eta + \frac{1}{2}\lambda_{\pi\Pi}\Pi\right)\left[v^\mu \tensor{g}{^\beta_\nu}+ v^\beta \tensor{g}{^\mu_\nu}-\frac{2\Delta^{\mu\beta}v_\nu}{3}\right]- \tau_\pi z\left(\tensor{\pi}{^\mu_\nu}u^\beta + u^\mu\tensor{\pi}{_\nu^\beta}\right) \notag\\
&&\quad + \left(\delta_{\pi\pi}-\frac{\tau_{\pi\pi}}{3}\right)\pi^{\mu\beta}\phi_\nu + \frac{\tau_{\pi\pi}}{2}\left(v^{(\beta}\tensor{\pi}{^{\mu)}_\nu} + w^{(\mu}\tensor{g}{^{\beta)}_\nu}\right)-\frac{\tau_{\pi\pi}}{3}\Delta^{\mu\beta}w_\nu,
\eea
\eml
along with the following projections along the characteristic vector
\be
z = u^\alpha\phi_\alpha,\quad v^\mu = \Delta^{\mu\alpha}\phi_\alpha,\quad w^\mu = \pi^{\mu\alpha}\phi_\alpha.
\ee
Here, $A^{(\mu\nu)} = (A^{\mu\nu}+A^{\nu\mu})/2$ denotes the symmetric portion of $A$. The characteristic determinant can be reduced to that of a $4\times 4$ matrix by using elementary row reductions:
\be
\det(\mathbb{A}^\alpha\phi_\alpha) = \tau_\Pi\tau_\pi^{16} z^{19}\det\tensor{\mathcal{X}}{^\mu_\nu},
\ee
where the $4\times 4$ matrix can be written as
\be
\label{E:IS_Four_matrix}
\tensor{\mathcal{X}}{^\mu_\nu} \equiv Ez\tensor{g}{^\mu_\nu} - u^\mu w_\nu - \frac{P_\varepsilon v^\mu}{z}\left(E\phi_\nu + w_\nu\right) - \frac{nP_n}{z}v^\mu\phi_\nu - \frac{v^\mu}{\tau_\Pi z}\mathcal{C}_{\Pi\nu} - \frac{\phi_\alpha \tensor{\mathcal{C}}{_\pi^{\alpha\mu}_\nu}}{\tau_\pi z}.
\ee
Given that the shear stress tensor $\pi^{\mu\nu}$ is symmetric, it possesses a set of linearly independent, orthonormal eigenvectors that form a basis in which the tensor is diagonal. Note that $\tensor{\pi}{^\mu_\alpha}u^\alpha = 0^\mu$, so $u^\mu$ is an eigenvector with eigenvalue $\Lambda_0 = 0$. Let $\{e_0^\mu\equiv u^\mu\}\cup\{e_a^\mu\}_{a = 1}^3$ be the set of orthonormal eigenvectors such that $\tensor{\pi}{^\mu_\alpha}e_A^\alpha = \Lambda_Ae_A^\mu$ for $A = 0,1,2,3$. We further note that since $\pi^{\mu\nu}$ is traceless, we have $\Lambda_1 + \Lambda_2 + \Lambda_3 = 0$. This orthonormal set of vectors is complete in $\mathbb{R}^4$, such that we may define a set of dual vectors $\{e_\mu^A\}_{A = 0}^3$ with
\be
e_\mu^A=\sum\limits_{B = 0}^3\eta^{AB}(e_B)_\mu,
\ee
where $\eta^{AB} = \textrm{diag}(-1,+1,+1,+1)$. Assume that roman indices are not summed over unless otherwise stated. The determinant of $\tensor{\mathcal{X}}{^\mu_\nu}$ takes a considerably simpler form after transforming into this new basis prescribed by the transformation matrix $(e_B^\mu)$ and its inverse $(e_\nu^A)$. Due to the cylic properties of the determinant, one can recognize that $\det(\tensor{\mathcal{X}}{^\mu_\nu}) = \det(e_\mu^A\tensor{\mathcal{X}}{^\mu_\nu}e_B^\nu)$ holds. Thus, the determinant may be expressed in the factorized form
\bml
\label{eq:nodiffusiondet}
\bea
\det(\mathbb{A}^\alpha\phi_\alpha) &=& E\left[\prod_{k = 0}^3(E + \Lambda_k)\right]v^{23}\tau_\Pi\tau_\pi^{16}\hat{z}^{15}\left[\hat{z}^2 -\frac{\frac{1}{2} (2\eta + \lambda_{\pi\Pi}\Pi)+\frac{1}{4} \tau_{\pi\pi}\Lambda_{\hat{v}^2}}{\tau_\pi E}\right]\mathcal{P}_3(\hat{z}^2,\chi^2,\kappa^2),\\
\mathcal{P}_3(\hat{z}^2,\chi^2,\kappa^2) &=& \hat{z}^6 - \mathcal{A}_2(\chi^2,\kappa^2)\hat{z}^4 + \mathcal{A}_1(\chi^2,\kappa^2)\hat{z}^2 -\mathcal{A}_0(\chi^2,\kappa^2),\\
\Lambda_{\hat{v}^2} &\equiv& \sum_{k = 1}^3 \Lambda_k \hat{v}_k^2(\chi^2,\kappa^2),
\eea
\eml
where we have defined $v_A \equiv v_\alpha e_A^\alpha$, noting that $v_0 = 0$ by construction. Furthermore, we write the ``normalized" characteristics as $\hat{z}^2 \equiv z^2/v^2$ and $\hat{v}_A^2 \equiv v_A^2/v^2$. The coefficients of the cubic polynomial $\mathcal{P}_3$, given by $\mathcal{A}_0(\chi^2,\kappa^2)$, $\mathcal{A}_1(\chi^2,\kappa^2)$, $\mathcal{A}_2(\chi^2,\kappa^2)$ are highly complicated nonlinear functions of the dynamic variables in $\Umb$ and $\hat{v}_A^2$. Their exact form is suppressed here for brevity, but the reader may view them in the attached link to our Mathematica notebook \cite{CordeiroGit2026}. Note that the basis of eigenvectors is orthonormal, so there exists angles $\theta_1,\theta_2$ such that $\kappa\equiv\cos\theta_1$ and $\chi\equiv\cos\theta_2$ such that $\theta_1\in[0,\pi]$ and $\theta_2\in[0,2\pi)$ where
\be
\label{eq:normalizedcharacteristics}
\hat{v}_1^2(\chi^2,\kappa^2) = \chi^2(1 - \kappa^2),\quad \hat{v}_2^2(\chi^2,\kappa^2) = (1 - \chi^2)(1 - \kappa^2),\quad \hat{v}_3^2(\kappa^2) = \kappa^2.
\ee
The following quantity is useful to define in terms of the angles:
\be
\Lambda_{\hat{v}^2} \equiv \sum_{k = 1}^3 \Lambda_k \hat{v}_k^2.
\ee
We remark that causality and strong hyperbolicity are governed by the roots of the characteristic determinant, which are guaranteed to have solutions in radicals since they are at most order $3$ in $\hat{z}^2$.

\section{Causality}\label{sec:causality}

Causality is the fundamental notion that information cannot travel at superluminal speeds. In this section, we present a formal definition of causality and connect it to the context of systems of partial differential equations and characteristics. In Section~\ref{subsec:DEC}, we clarify that this very general definition is agnostic to physical assumptions made on the theory, such as energy conditions. 

\subsection{The Dominant Energy Condition does not imply Causality}\label{subsec:DEC}

Energy conditions are constraints imposed by physically motivated expectations about the regime of validity of relativistic theories of matter (p. 219) \cite{WaldBookGR1984}. In this section, we discuss whether or not such physical assumptions imply fundamental properties such as causality. For example, the weak energy condition (WEC) is formulated on the expectation that the energy density in the reference frame of a timelike observer with four-velocity $\xi^\mu$ should be nonnegative. That is, $T^{\alpha\beta}\xi_\alpha\xi_\beta \geq 0$ for all such $\xi^\mu$ with $\xi_\alpha\xi^\alpha < 0$. Here, we consider the so-called dominant energy condition (DEC), which assumes the WEC, along with the expectation that the flow of energy-momentum in the frame of a timelike observer $t^\mu\equiv -T^{\mu\alpha}\xi_\alpha$ should be non-spacelike. That is, energy-momentum flow in the frame of a timelike observer should not exceed that of a null trajectory. Due to the obvious link to causality, one might ask whether or not the DEC implies nonlinear causality or vice versa. Here, we clarify that the DEC is a motivated assumption on the matter sector of a theory, but fails to imply causality. Furthermore, causality does not imply the DEC. Recall that for an ideal fluid, the singular necessary and sufficient causality condition is \cite{ChoquetBruhatGRBook}
\be
0\leq c_s^2\leq 1,
\ee
where $c_s^2 = P_\varepsilon + nP_n/(\varepsilon + P)$ is the hydrodynamic speed of sound at nonzero baryon density \cite{Rezzolla_Zanotti_book,AnileBook}. We will assume zero baryon number here, for simplicity, but the argument is valid in the general case as well. Note that the energy-momentum tensor reads
\be
T_{\textrm{eq.}}^{\mu\nu} = \varepsilon u^\mu u^\nu + P\Delta^{\mu\nu}
\ee
in this case. An equivalent (e.g. necessary \emph{and} sufficient) condition for the DEC states that the energy eigenvalue should dominate the principal pressures of the energy-momentum tensor:
\be
\varepsilon \geq |P|
\ee
Suppose that $P = \varepsilon + \varepsilon_0$ where $\varepsilon_0 > 0$. Then, $T_{\textrm{eq.}}^{\mu\nu}$ does not satisfy the DEC, but $c_s^2 = P_\varepsilon = 1$, so causality is satisfied. Now, suppose that $P = -\varepsilon/3$. Then $P_\varepsilon = -1/3 < 0$, which violates causality, whereas $|P| = \varepsilon/3 < \varepsilon$, so the DEC is satisfied. Thus, neither the DEC nor causality implies the other. This notion becomes a bit more (qualitatively) clear if we consider the more general case with bulk and nonzero baryon density previously treated in \cite{Bemfica:2019cop}. Removing the dynamics of the shear-stress tensor in Eq.~\eqref{eq:shear-nodiffusion} provides the (necessary and sufficient) nonlinear causality condition
\be
0\leq P_\varepsilon + \frac{1}{E}\left[nP_n + \frac{\zeta}{\tau_\Pi}\right] \leq 1,
\ee
where we remind the reader that $E = \varepsilon + P + \Pi$. Note that the energy-momentum tensor reads
\be
T_\Pi^{\mu\nu} = \varepsilon u^\mu u^\nu + (P + \Pi)\Delta^{\mu\nu}
\ee
in this case. The DEC then states
\be
\varepsilon \geq |P + \Pi|\quad\Leftrightarrow\quad 0\leq E\leq 2\varepsilon.
\ee
We mention that $E\geq 0$ appears like a generalization of the linear stability condition $\varepsilon + P \geq 0$ provided in \cite{Olson:1989ey}, whereas $E \leq 2\varepsilon$ stipulates that the out of equilibrium pressure $P + \Pi$ should not grow above the energy density. Unlike linearized causality and stability conditions, the DEC depends explicitly on viscous fluxes like nonlinear constraints, however, its immediate drawback is that it lacks any knowledge about the transport coefficients and partial derivatives $P_\varepsilon$ and $P_n$ that are present in both the linear and nonlinear calculations. This connection is perhaps unsurprising, as the DEC only contains information about the energy-momentum tensor itself and is agnostic about the evolution of the viscous currents. Thus, there should be no reason for the two conditions to be directly related.

A relevant follow-up question would then be whether or not there exists another condition, say $X$ (for instance, linear causality) such that $X$ and the DEC implies causality. One well-known result by Hawking and Ellis which translates to a \emph{vacuum conservation theorem} states that if $T^{\mu\nu}$ vanishes in some region and satisfies the DEC, then it also must vanish along its Cauchy development \cite{HawkingEllisBook}. Physically, this theorem states that the vacuum is preserved along non-spacelike curves, meaning that this interpretation of causality states that ``nothing can be created from nothingness."

We stress, however, that this vacuum conservation result is distinct than the standard precise form of formulating causality for system described by evolutionary PDEs, which is as follows, see \cite[Appendix B.2]{Disconzi-2024}.
Consider a globally hyperbolic spacetime\footnote{We recall that spacetimes that arise as a solution to the initial-value problem for Einstein's equations (with or without matter coupling) as well as Minkowski space are globally hyperbolic. Global hyperbolicity precludes the existence of closed timelike curves, although it must be stressed that this by itself is not enough to ensure causality, since there are many examples, such as Eckart's theory, that are defined in Minkowski space but violate causality.} and let $\psi$ be a field defined on spacetime. We say that $\psi$, propagates causally, if for any spacetime point $x$, $\psi(x)$ depends only on the causal past of $x$, noting that the causal past is well-defined in view of the global hyperbolicity assumption. In other words, any physical (e.g. causal) quantity can only be influenced by its causally connected past, and thus, it is completely oblivious to information from other points in spacetime that lie outside the lightcone. We highlight our definition of causality using Fig.~\ref{fig:causality}, which provides a rough sketch of the Cauchy development of a causal function. The key for this definition is to make precise what it means to ``depend only on the causal past of $x$," which is formalized in Definition \ref{def:causality}. 
\begin{figure}
    \centering
  \includegraphics[width=0.75\linewidth]{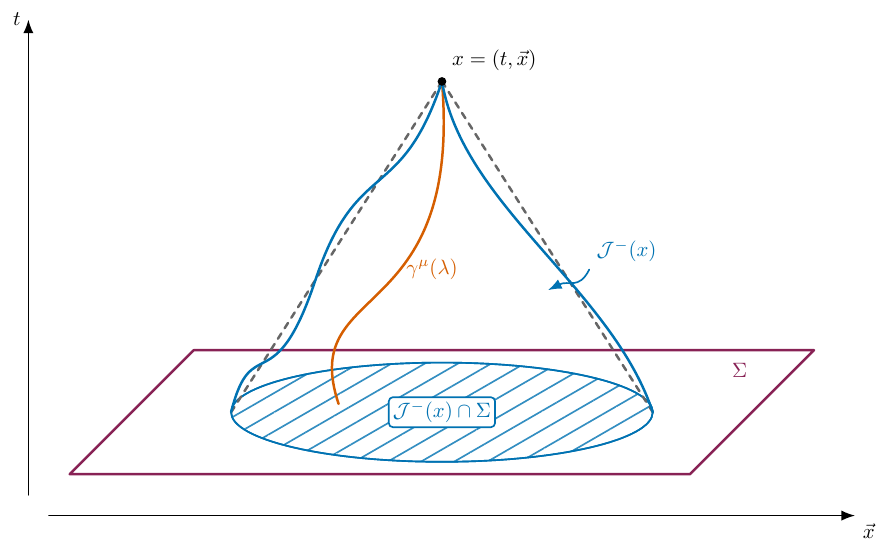}
    \caption{Sketch of the domain of influence. The Cauchy surface $\Sigma$ usually corresponds to the initial data, with $\mathcal{J}^-(x)\cap\Sigma$ being those points on $\Sigma$ that lie in the causal past of point $x$. The timelike curve $\gamma$ parameterized by some affine parameter $\lambda$ represents an instance of causal information transport between $\Sigma$ and $x$. A key facet of this definition is that any curves outside of $\mathcal{J}_-(x)$ do not influence the solution at $x$.}
    \label{fig:causality}
\end{figure}
\begin{definition}[Causality]\label{def:causality}
Let $(\mathcal{M}, g)$ be a globally hyperbolic spacetime. We are concerned with the Cauchy problem for a system of differential equations
\be
\label{E:PDE_causality_def}
\mathcal{D}\textbf{\textrm{U}} = \boldsymbol{0},
\ee
where $\mathcal{D}$ be a differential operator with unknown(s) $\textbf{\textrm{U}}$. 
Let $\textbf{\textrm{U}}_1$ be a solution to \eqref{E:PDE_causality_def} with initial data given along a Cauchy surface $\Sigma\subset\mathcal{M}$. Let $x \in \mathcal{M}$ be on the future of $\Sigma$. We say that $\textbf{\textrm{U}}_1(x)$ depends only on the causal past of $x$, denoted $\mathcal{J}^-(x)$, if the following happens. Let
$\textbf{\textrm{U}}_2$ be a solution to \eqref{E:PDE_causality_def} with initial data given along the same Cauchy surface $\Sigma$. Assume that the initial data for $\textbf{\textrm{U}}_1(x)$ and $\textbf{\textrm{U}}_2(x)$ agree on $\mathcal{J}^-(x) \cap \Sigma$. Then, $\textbf{\textrm{U}}_1(y) = \textbf{\textrm{U}}_2(y)$ for every $y \in \mathcal{J}^-(x)$. We say that \eqref{E:PDE_causality_def} is causal if solutions to \eqref{E:PDE_causality_def} depend only on the causal past of $x$ for every $x \in \mathcal{M}$ in the future of $\Sigma$.
\end{definition}

The physical interpretation of Definition \ref{def:causality} is simple. Since no causal curve can connect $\Sigma \backslash (\mathcal{J}^-(x) \cap \Sigma)$ to $x$ (the complement of $\mathcal{J}^-(x) \cap \Sigma$ inside $\Sigma$), the value of a solution at $x$ cannot depend on the data in the region  $\Sigma \backslash (\mathcal{J}^-(x) \cap \Sigma)$. Thus, two solutions whose data agree on $\mathcal{J}^-(x) \cap \Sigma$ (but might otherwise be different in $\Sigma \backslash (\mathcal{J}^-(x) \cap \Sigma)$) cannot be distinguished within $\mathcal{J}^-(x)$, which is the region composed by all causal curves that start on $\Sigma$ and reach $x$.

We observe that causality is a \emph{conditional statement:} given solutions to \eqref{E:PDE_causality_def}, it states a property of such solutions, namely, that the solution fields depend only on the causal past of $x$. It makes no claim about whether such solutions in fact exist. Thus, causality can be studied independently of the problem of existence of solutions.

In order to determine whether a PDE system is causal, the key ingredient is to show that the characteristic surfaces $\{ \varphi(x) = 0 \}$ of the PDE emanating from $x$ lie inside the lightcone with vertex at $x$ (for every $x$ in $\mathcal{M}$). This, in turn, boils down to show that the roots  of the characteristic determinant are always real and spacelike \cite[Appendix B.2]{Disconzi-2024}. In practice, such condition can be taken as a statement of causality in that it allows us to investigate it separately from the issue of existence of solutions. We note, however, that the necessary and sufficient conditions we derive in this paper are not showed to be enough to determine existence. This leaves the door open for the possibility that minimal conditions for existence are stronger than those needed to establish causality. See Section \ref{S:IS_domain_of_dependence}.
Thus, the bulk of the work in our causality analysis is the investigation of the roots of the characteristic determinant.

We may now distinguish this definition from the vacuum conservation theorem as follows. Consider the energy-momentum tensors $T_1^{\mu\nu}$ and $T_2^{\mu\nu}$ corresponding to $\widetilde{\textbf{\textrm{U}}}_1$ and $\widetilde{\textbf{\textrm{U}}}_2$, respectively, and assume that they satisfy the dominant energy condition and vanish in some $\mathcal{J}^-(x)\cap\Sigma$. The vacuum theorem then guarantees that at some later time $t$, each of these two tensors must (independently) also vanish at $x$. However, this theorem fails to guarantee that the new energy-momentum tensor $T_{1-2}^{\mu\nu}$ of the difference of solutions $\widetilde{\textbf{\textrm{U}}}_1 - \widetilde{\textbf{\textrm{U}}}_2$ also satisfies the theorem, or even the dominant energy condition. In Section~\ref{subsec:causalityconditions}, we show that causality in the sense of Definition~\ref{def:causality} is equivalent to a set of algebraic constraints on the degrees of freedom of fluid theories of IS-type with shear and bulk viscosity. We thus reiterate that with the above discussion in mind, our causality conditions do not guarantee adherence or any relationship to the DEC or similar energy conditions themselves.

\subsection{Necessary and Sufficient conditions}\label{subsec:causalityconditions}
Here, we state the most general conditions for nonlinear causality and then connect them to relevant historical subcases. We reiterate that fully nonlinear necessary and sufficient conditions have never been derived for any subcase other than with bulk viscosity \cite{Bemfica:2019cop}, and thus, these cases either generalize previous constraints or present new bounds entirely.

As mentioned previously, the coefficients of the characteristic determinant are complicated nonlinear functions of the transport coefficients and dynamic degrees of freedom of the system. Consequently, while an analytical computation of the roots is possible, doing so is not very illuminating due to the sheer number and complexity of the terms. Although some level of complexity is unavoidable, instead of bounding the roots of the polynomials present in the characteristic determinant, here we instead constrain the \emph{coefficients of the polynomial} $\mathcal{A}_k$ in Eq.~\eqref{eq:nodiffusiondet} by an equivalent set of necessary and sufficient conditions. This method significantly reduces the computational complexity when calculating the constraints.

In the language of quasilinear PDEs, we state that a system $(\mathbb{A}^\alpha\partial_\alpha + \mathbb{B})\Umb = 0$ is \textit{causal} if, and only if (CI) the roots of the characteristic equation $\det(\mathbb{A}^\alpha\phi_\alpha) = 0$, given by $\phi_0 \equiv \phi_0(\phi_i)$ are real, where $\phi^\mu \equiv \nabla^\mu\varphi$, and $\{\varphi(x) = 0\}$ are the characteristic hypersurfaces and (CII) $\phi^\mu$ is non-timelike, i.e. $\phi^\al\phi_\al \geq 0$ \cite{ChoquetBruhatGRBook}. Put together, condition (CI) states that the characteristic speeds must be real and exist, and condition (CII) enforces that these speeds are causal. We state our primary result about causality below.

\begin{theorem}[Nonlinear Causality]\label{thm:causalitynodiffusion}
Let $\tau_\Pi,\tau_\pi\neq 0$, and for all statements below assume that
\be
E + \Lambda_A \neq 0,\quad A = 0,1,2,3.
\ee
Let $\Delta(\mathcal{P}_3;\chi^2,\kappa^2)$ be the discriminant of $\mathcal{P}_3$ defined by
\begin{align}
\label{eq:discriminant}
\Delta(\mathcal{P}_3;\chi^2,\kappa^2) \equiv& - 27\left[2\left(\frac{\mathcal{A}_2(\chi^2,\kappa^2)}{3}\right)^3 - \frac{\mathcal{A}_2(\chi^2,\kappa^2)}{3}\mathcal{A}_1(\chi^2,\kappa^2) +\mathcal{A}_0(\chi^2,\kappa^2)\right]^2\notag\\
&-4\left[\mathcal{A}_1(\chi^2,\kappa^2) - 3\left(\frac{\mathcal{A}_2(\chi^2,\kappa^2)}{3}\right)^2\right]^3.
\end{align}
Consider the inequalities:
\bml
\bea
\label{eq:C1}
0&\leq &\Delta(\mathcal{P}_3;\chi^2,\kappa^2),\\
\label{eq:C2}
0 &\leq & \frac{\mathcal{A}_2(\chi^2,\kappa^2)}{3} \leq 1,\quad 0\leq \mathcal{A}_1(\chi^2,\kappa^2),\quad 0\leq \mathcal{A}_0(\chi^2,\kappa^2)\\
\label{eq:C3}
0 &\leq & 1-\mathcal{A}_2(\chi^2,\kappa^2) + \mathcal{A}_1(\chi^2,\kappa^2) - \mathcal{A}_0(\chi^2,\kappa^2) ,\\
\label{eq:C4}
0 &\leq & 1-\frac{2}{3}\mathcal{A}_2(\chi^2,\kappa^2) + \frac{1}{3}\mathcal{A}_1(\chi^2,\kappa^2).\\
\label{eq:C5}
0 & \leq & \frac{\frac{1}{2} (2\eta + \lambda_{\pi\Pi}\Pi)+\frac{1}{4} \tau_{\pi\pi}\Lambda_{\hat{v}^2}}{\tau_\pi E} \leq 1.
\eea
\eml
Then, the system \eqref{eq:FirstOrderQPDE} (or equivalently \eqref{eq:propsystem}), is causal if and only if Eq.~\eqref{eq:C1}--\eqref{eq:C5} hold for all $\chi^2,\kappa^2 \in [0,1]$, whereas the system \eqref{eq:FirstOrderQPDE} (or equivalently \eqref{eq:propsystem})  subject to the constraints \eqref{Constraints}, is causal if, and only if, Eq.~\eqref{eq:C1}--\eqref{eq:C4} hold for all $\chi^2,\kappa^2 \in [0,1]$. 
\end{theorem}
\begin{proof}
The proof of the theorem can be divided into two major steps. The first step constrains the nontrivial characteristic wave speeds between $0$ and $1$. The second step uses fundamental results of polynomials related to Descartes sign rule which provide equivalent but significantly condensed necessary and sufficient conditions for causality than the bounds on the roots in step one. The interested reader may refer to Appendix~\ref{App:causality} for the details of the proof. 
\end{proof}
We stress that these conditions are the most general of their kind to date, and can easily reproduce previous relevant results as discussed below.
\begin{enumerate}

\item\textbf{Only Bulk Viscosity:} If one removes all transport coefficients corresponding to shear-stress (keeping only $\tau_\Pi,\zeta,\delta_{\Pi\Pi}$), Theorem~\ref{thm:causalitynodiffusion} not only reproduces the single (necessary and sufficient) causality condition in \cite{Bemfica:2019cop}, it also provides a second order correction to $\zeta/\tau_\Pi$ in the form:
\be
\label{eq:bulkcausality}
0 \leq P_\varepsilon + \frac{1}{E}\left(nP_n + \frac{\zeta + \delta_{\Pi\Pi}\Pi}{\tau_\Pi}\right) \leq 1.
\ee
Again, $E = \varepsilon + P + \Pi$ here. Setting $\delta_{\Pi\Pi} = 0$ exactly reproduces their condition.

\item\textbf{No baryon number:} Previously, the most general nonlinear causality conditions for the exact same IS-like class of equations in Eq.~\eqref{eq:supplementalIS} was done with zero baryon number $n = 0$ --- a regime relevant in the context of ultrarelativistic heavy-ion collisions at the LHC. These constraints provided two sets of conditions \cite{Bemfica:2020xym}. The first set were necessary constraints in the sense that violating them also implied causality was violated. The second set were sufficient conditions, that if satisfied, guaranteed causality. However, even if these sufficient conditions were violated, it was still possible for the particular solution to be causal, since a non-trivial amount of these data points also satisfied the necessary conditions in common heavy-ion simulations \cite{Plumberg:2021bme,ExTrEMe:2023nhy}. 

Our conditions resolve this ambiguity by combining all inequalities into a single set of necessary and sufficient conditions. Thus, if any of these statements are violated, it is guaranteed that the result is acausal, and if all are satisfied, they are causal. We further remark that our results directly generalize the zero-baryon limit case (albeit with minor notational discrepancies) in \cite{Bemfica:2020xym} by setting $nP_n\rightarrow 0$ in the statement of Theorem~\ref{thm:causalitynodiffusion}. This simple limit provides a set of necessary and sufficient conditions for the zero baryon case and does not change the structure of the characteristic determinant (save for an extra factor of $x$). In fact, the role of the baryon number is simply a scalar correction to the bulk viscosity. Due to the tensor structure of the equations of motion, baryon number appears as $nP_n + \zeta/\tau_\Pi$ and never elsewhere. The necessary conditions may be generated by certain choices of angles satisfying $\hat{v}_1^2 + \hat{v}_2^2 + \hat{v}_3^2 = 1$. It is particularly straightforward to see the connection between our work and \cite{Bemfica:2020xym} through the inequality in Eq.~\eqref{eq:C5}, which directly recovers some of the necessary conditions.
\end{enumerate}
We emphasize that these conditions are relatively resource-intensive in the numerical sense as they are parameterized by the bounded parameters $\chi^2,\kappa^2\in[0,1]$ relative to the shear stress basis. Furthermore, the most complicated bound by far is the discriminant condition, which is nonlinear in the coefficients themselves. It is therefore important to consider further special cases (when applicable) that can eliminate $\chi^2$ and $\kappa^2$ and/or remove or provide a simpler necessary and sufficient condition that guarantees the reality of the roots. In the following sections, we consider two particular cases of interest that are now accessible with our general conditions - one being a generalized Maxwell-Cattaneo set of equations, in which $\delta_{\pi\pi}=\tau_{\pi\pi} = 0$, yet all other parameters are held to be general, and also the conformal fluid case.

\subsection{A toy model illustrating the effect of a constraint}

Theorem \ref{thm:causalitynodiffusion} has a key distinction between causality of 
\eqref{eq:FirstOrderQPDE} without imposing the constraints \eqref{Constraints} (conditions \eqref{eq:C1}--\eqref{eq:C5}) and \eqref{eq:FirstOrderQPDE}  with the constraints \eqref{Constraints} satisfied (conditions \eqref{eq:C1}--\eqref{eq:C4}). The difference is that imposing the constraints might (and in our case, in fact does) reduce the characteristics. Characteristics of the unconstrained system that do not belong to the constrained one are therefore spurious, i.e., only the characteristics of the constrained system are truly physical. The proof of causality carried out in Section \ref{App:causality} (see in particular Section \ref{S:IS_constraint_effect}
) relies on this distinction. Because the proof is lengthy and the computations intricate, we illustrate this key distinction here in a simple example of 
 two scalar equations and one algebraic constraint. The
determinant \eqref{E:IS_Toy_full_determinant} has two factors, whereas
the characteristic equation \eqref{E:IS_Toy_constrained_characteristic}
contains only one when the constraint is imposed. The example also distinguishes
imposing the constraint on the field values from imposing its
differentiated form on the unknown normal derivatives.

\subsubsection{The system and its full characteristic determinant}

Let $a,b\in\mathbb{R}$ be constants, and let $u$ and $v$ be real
scalar fields on $\mathbb{R}_t\times\mathbb{R}_x$. In this example,
$u$ and $v$ denote scalar fields, not the fluid velocity or the
projection of a covector used elsewhere in the manuscript. Equip
spacetime with the Minkowski metric
\begin{align*}
 \eta&=-\mathrm{d}t^2+\mathrm{d}x^2,
\end{align*}
so that the speed of light is one. Consider the system
\begin{subequations}
\label{E:IS_Toy_system}
\begin{align}
 u_{tt}-\frac{a+b}{2}u_{xx}-\frac{a-b}{2}v_{xx}&=0,\\
 v_{tt}-\frac{a-b}{2}u_{xx}-\frac{a+b}{2}v_{xx}&=0,
\end{align}
\end{subequations}
together with the algebraic constraint
\begin{align}
\label{E:IS_Toy_constraint}
 u-v&=0.
\end{align}

First disregard \eqref{E:IS_Toy_constraint}. For a covector
\begin{align*}
 \phi&=\tau\,\mathrm{d}t+\xi\,\mathrm{d}x,
\end{align*}
the principal matrix of \eqref{E:IS_Toy_system}, obtained
by replacing $\partial_t$ by $\tau$ and $\partial_x$ by $\xi$, is
\begin{align}
\label{E:IS_Toy_symbol}
 M(\tau,\xi)
 &=\begin{pmatrix}
 \tau^2-\frac{a+b}{2}\xi^2&-\frac{a-b}{2}\xi^2\\[1mm]
 -\frac{a-b}{2}\xi^2&\tau^2-\frac{a+b}{2}\xi^2
 \end{pmatrix}.
\end{align}
Its determinant is
\begin{align}
\label{E:IS_Toy_full_determinant}
\begin{split}
 \det M(\tau,\xi)
 &=\left(\tau^2-\frac{a+b}{2}\xi^2\right)^2
      -\frac{(a-b)^2}{4}\xi^4\\
 &=(\tau^2-a\xi^2)(\tau^2-b\xi^2).
\end{split}
\end{align}
For every nonzero real $\xi$, all four temporal roots of
\eqref{E:IS_Toy_full_determinant} are real if and only if $a\geq0$
and $b\geq0$. In that case, the roots are
\begin{align*}
 \tau&=\pm\sqrt{a}\,\xi,\\
 \tau&=\pm\sqrt{b}\,\xi.
\end{align*}
The hypersurface $\tau t+\xi x=\mathrm{constant}$ has speed
$-\frac{\tau}{\xi}$. Its normal covector has squared norm
\begin{subequations}
\label{E:IS_Toy_covector_norms}
\begin{align}
 \eta^{-1}(\phi,\phi)&=(1-a)\xi^2
       \quad\text{when }\tau^2=a\xi^2,\\
 \eta^{-1}(\phi,\phi)&=(1-b)\xi^2
       \quad\text{when }\tau^2=b\xi^2.
\end{align}
\end{subequations}
Thus reality of all temporal roots and non-timelikeness of all
nonzero real characteristic covectors of \eqref{E:IS_Toy_system},
without \eqref{E:IS_Toy_constraint}, are equivalent to
\begin{subequations}
\label{E:IS_Toy_full_conditions}
\begin{align}
\label{E:IS_Toy_a_bound}
 0&\leq a\leq1,\\
\label{E:IS_Toy_b_bound}
 0&\leq b\leq1.
\end{align}
\end{subequations}
The upper bounds permit propagation at the speed of light. For
$\xi=0$, \eqref{E:IS_Toy_full_determinant} gives $\det M=\tau^4$,
so there is no additional nonzero characteristic covector.

\subsubsection{Taking the constraint into account}

For a nonzero real covector $\phi$, 
consider the homogeneous linear system
\begin{align}
\label{E:IS_Toy_amplitude_equation}
 M(\tau,\xi)r&=0
\end{align}
with unknown $r=(r_u,r_v)^T \in \mathbb{R}^2$. As in Section~\ref{S:IS_constraint_effect}, the $r$ describes
the difference between possible normal derivatives when the other
derivative data are fixed. Here the equations have order two:
completing $\varphi=\tau t+\xi x$ to linear coordinates
$(\varphi,\ell)$, the two unknown derivatives are $u_{\varphi\varphi}$
and $v_{\varphi\varphi}$, with all derivatives of order at most one
and all second derivatives containing an $\ell$ derivative fixed.
The coefficients of these two unknown derivatives are exactly
\eqref{E:IS_Toy_symbol}.

Differentiating \eqref{E:IS_Toy_constraint} twice in $\varphi$
gives $u_{\varphi\varphi}-v_{\varphi\varphi}=0$. Hence we must also require
\begin{align}
\label{E:IS_Toy_amplitude_constraint}
 r_u-r_v&=0.
\end{align}
Every vector satisfying \eqref{E:IS_Toy_amplitude_constraint} has
the form $r=(h,h)^T$, and \eqref{E:IS_Toy_symbol} gives
\begin{align}
\label{E:IS_Toy_restricted_symbol}
 M(\tau,\xi)\begin{pmatrix}h\\h\end{pmatrix}
 &=(\tau^2-a\xi^2)\begin{pmatrix}h\\h\end{pmatrix}.
\end{align}
Consequently, \eqref{E:IS_Toy_amplitude_equation} has a nonzero
solution satisfying \eqref{E:IS_Toy_amplitude_constraint} if and
only if
\begin{align}
\label{E:IS_Toy_constrained_characteristic}
 \tau^2-a\xi^2&=0.
\end{align}
Identity \eqref{E:IS_Toy_restricted_symbol} also holds for complex
$\tau$ and $h$, so \eqref{E:IS_Toy_constrained_characteristic}
identifies all temporal roots to be tested for reality. By
\eqref{E:IS_Toy_covector_norms}, the conditions for the equations
with \eqref{E:IS_Toy_constraint} reduce to
\eqref{E:IS_Toy_a_bound}, with no condition on $b$.

Simply substituting $u=v$ into the coefficients of
\eqref{E:IS_Toy_symbol} would change nothing, since those
coefficients do not depend on $u$ or $v$. The restriction
\eqref{E:IS_Toy_amplitude_constraint} on the unknown derivatives
is what changes the characteristic equation.

\subsubsection{Identifying the equation responsible for the extra factor}

Set
\begin{subequations}
\label{E:IS_Toy_change_variables}
\begin{align}
 w&=\frac{u+v}{2},\\
 q&=\frac{u-v}{2}.
\end{align}
\end{subequations}
This change is invertible, with
\begin{align*}
 u&=w+q,\\
 v&=w-q.
\end{align*}
Adding and subtracting the equations in \eqref{E:IS_Toy_system}
gives
\begin{subequations}
\label{E:IS_Toy_transformed_system}
\begin{align}
\label{E:IS_Toy_w_equation}
 w_{tt}&=\frac{u_{tt}+v_{tt}}{2}
        =\frac{a}{2}(u_{xx}+v_{xx})=a w_{xx},\\
\label{E:IS_Toy_q_equation}
 q_{tt}&=\frac{u_{tt}-v_{tt}}{2}
        =\frac{b}{2}(u_{xx}-v_{xx})=b q_{xx}.
\end{align}
\end{subequations}
In matrix form, the change \eqref{E:IS_Toy_change_variables} and
the linear combinations of equations in
\eqref{E:IS_Toy_transformed_system} give
\begin{align}
\label{E:IS_Toy_diagonal_symbol}
 \frac{1}{2}\begin{pmatrix}1&1\\1&-1\end{pmatrix}
 M(\tau,\xi)
 \begin{pmatrix}1&1\\1&-1\end{pmatrix}
 &=\begin{pmatrix}
 \tau^2-a\xi^2&0\\
 0&\tau^2-b\xi^2
 \end{pmatrix}.
\end{align}
The matrices multiplying $M$ on the left and right in
\eqref{E:IS_Toy_diagonal_symbol} are inverses.

Constraint \eqref{E:IS_Toy_constraint} is exactly $q=0$.
Therefore the solutions of \eqref{E:IS_Toy_system} satisfying
\eqref{E:IS_Toy_constraint} are precisely
\begin{align*}
 u=v&=w,
\end{align*}
where $w$ solves \eqref{E:IS_Toy_w_equation}. The factor
$\tau^2-b\xi^2$ in \eqref{E:IS_Toy_full_determinant} belongs to
\eqref{E:IS_Toy_q_equation}, whose unknown must vanish.

Because \eqref{E:IS_Toy_system} has order two, compatible initial
data must satisfy both
\begin{subequations}
\label{E:IS_Toy_initial_compatibility}
\begin{align}
 u(0,x)&=v(0,x),\\
 u_t(0,x)&=v_t(0,x).
\end{align}
\end{subequations}
For completeness, consider a given smooth solution of
\eqref{E:IS_Toy_system} on $[0,T]\times\mathbb{T}$, where
$T>0$ and $\mathbb{T}=\mathbb{R}/(2\pi\mathbb{Z})$. The Fourier
coefficients of $q$ are
\begin{align*}
 \widehat{q}_m(t)&=\frac{1}{2\pi}\int_0^{2\pi}
             q(t,x)e^{-imx}\,\mathrm{d}x,
 \qquad m\in\mathbb{Z}.
\end{align*}
Equation \eqref{E:IS_Toy_q_equation} and two integrations by parts
give
\begin{align*}
 \partial_t^2\widehat{q}_m(t)
 &=\frac{b}{2\pi}\int_0^{2\pi}q_{xx}(t,x)e^{-imx}\,\mathrm{d}x
 =-b m^2\widehat{q}_m(t).
\end{align*}
Thus \eqref{E:IS_Toy_initial_compatibility} yields
\begin{subequations}
\label{E:IS_Toy_fourier_Cauchy_problem}
\begin{align}
 \partial_t^2\widehat{q}_m+b m^2\widehat{q}_m&=0,\\
 \widehat{q}_m(0)&=0,\\
 \partial_t\widehat{q}_m(0)&=0.
\end{align}
\end{subequations}
Uniqueness for \eqref{E:IS_Toy_fourier_Cauchy_problem} gives
$\widehat{q}_m(t)=0$ for every $m\in\mathbb{Z}$ and $t\in[0,T]$.
Consequently $q=0$, for every real $b$. This is a uniqueness
statement among given smooth periodic solutions; for $b<0$ it
does not assert well-posedness for arbitrary data of
\eqref{E:IS_Toy_system} without the constraint.

\subsubsection{Two concrete choices showing what condition disappears}

Fix $a=\frac{1}{4}$. For $b=4$, the temporal roots in
\eqref{E:IS_Toy_full_determinant} are
\begin{align*}
 \tau&=\pm\frac{\xi}{2},\\
 \tau&=\pm2\xi.
\end{align*}
The roots $\tau=\pm2\xi$ have speeds of absolute value two and
violate \eqref{E:IS_Toy_b_bound}. By
\eqref{E:IS_Toy_constrained_characteristic}, they are absent
from the characteristic equation for the system with the constraint.

For $b=-1$, the temporal roots in
\eqref{E:IS_Toy_full_determinant} are
\begin{align*}
 \tau&=\pm\frac{\xi}{2},\\
 \tau&=\pm i\xi.
\end{align*}
For real $\xi\neq0$, the roots $\tau=\pm i\xi$ are nonreal.
They are again absent from
\eqref{E:IS_Toy_constrained_characteristic}. In both examples,
\eqref{E:IS_Toy_w_equation} becomes
\begin{align*}
 w_{tt}-\frac{1}{4}w_{xx}&=0,
\end{align*}
and the solutions with $u=v=w$ propagate at speed $\frac{1}{2}$.
Thus neither the lower nor the upper bound in
\eqref{E:IS_Toy_b_bound} is necessary for the system with the
constraint, although both are necessary in
\eqref{E:IS_Toy_full_conditions}.

The domain of dependence follows from the solution
formula \eqref{E:IS_Toy_wave_solution}. For $0<a\leq1$, set $c=\sqrt{a}$ and prescribe smooth
functions $w_0,w_1$ on $\mathbb{R}$. The solution of
\eqref{E:IS_Toy_w_equation} with initial value $w_0$ and initial
time derivative $w_1$ is
\begin{align}
\label{E:IS_Toy_wave_solution}
\begin{split}
 w(t,x)
 &=\frac{w_0(x-ct)+w_0(x+ct)}{2}\\
 &\quad+\frac{1}{2c}\int_{x-ct}^{x+ct}w_1(r)\,\mathrm{d}r.
\end{split}
\end{align}
Differentiation of \eqref{E:IS_Toy_wave_solution} gives
\begin{align*}
 w_{tt}-c^2w_{xx}&=0,\\
 w(0,x)&=w_0(x),\\
 w_t(0,x)&=w_1(x).
\end{align*}
To check uniqueness, let $d$ be the difference of two smooth
solutions with these data and set $p=d_t+c d_x$. Then
\begin{align*}
 p_t-cp_x&=d_{tt}-c^2d_{xx}=0,\\
 p(0,x)&=0.
\end{align*}
It follows that $p(t,x)=p(0,x+ct)=0$. Hence
$d_t+c d_x=0$ and $d(0,x)=0$, so $d(t,x)=d(0,x-ct)=0$.
Therefore \eqref{E:IS_Toy_wave_solution} represents every smooth
solution with the stated data.

For $t\geq0$, \eqref{E:IS_Toy_wave_solution} depends only on data in
\begin{align*}
 [x-ct,x+ct]&\subseteq[x-t,x+t].
\end{align*}
Thus $u=v=w$ has a domain of dependence contained in the Minkowski
causal past, independently of $b$. For $a=0$, the solution of
\eqref{E:IS_Toy_w_equation} is instead
\begin{align*}
 w(t,x)&=w_0(x)+t w_1(x),
\end{align*}
which depends only on the data at $x$ and is also independent of
$b$.

\subsubsection{\texorpdfstring{The correspondence with Section~\ref{S:IS_constraint_effect}}{The correspondence with Section A.2}}

Equation \eqref{E:IS_Toy_restricted_symbol} is the two-component
analogue of \eqref{E:IS_kernel_equivalence}. In each case the
characteristic test requires a nonzero solution of the principal
equations satisfying the differentiated constraints. The restriction
\eqref{E:IS_Toy_amplitude_constraint} replaces
\eqref{E:IS_linear_constraints}, and the scalar polynomial
\eqref{E:IS_Toy_constrained_characteristic} replaces the determinant
of $A_{11}(\phi)$.

The invertible row and column changes in
\eqref{E:IS_Toy_diagonal_symbol} give a diagonal special case of
\eqref{E:IS_second_matrix}: $w$ is an independent unknown and $q$
is the quantity required to vanish. The factor $\tau^2-b\xi^2$
plays the role of $S(\phi)$, defined in
\eqref{E:IS_quadratic_polynomial} and appearing in
\eqref{E:IS_full_physical_factorization}. It supplies no additional
characteristic after the constraint has been imposed. Consequently,
\eqref{E:IS_Toy_b_bound} drops out, just as the condition
\eqref{supeq:C5} on $y_0$ drops out when the IS characteristics are
determined by \eqref{E:IS_physical_determinant} rather than
\eqref{E:IS_full_angular_determinant}.

When $a=b$, imposing \eqref{E:IS_Toy_amplitude_constraint}
reduces the kernel at a characteristic covector from $\mathbb{R}^2$ to $\operatorname{span}\{ (1,1)^T \}$, but does not change the characteristic covectors. 
Likewise, a zero of $S(\phi)$ in
\eqref{E:IS_full_physical_factorization} remains a characteristic
of the IS equations with the constraints when it is also a zero of
\eqref{E:IS_physical_determinant}. Identity
\eqref{E:IS_Toy_restricted_symbol} uses no division by $a-b$ or by
$\tau^2-b\xi^2$, so this conclusion includes their common zeros.

\subsection{Conformal Fluid}\label{subsec:conformal}

Here, we consider the special case of a conformal fluid \cite{Baier:2007ix} to highlight how our nonlinear causality bounds can be used to generate nonlinear causality conditions. The equilibrium pressure satisfies the equation of state $P = \varepsilon/3$ and it is assumed that the only nontrivial transport coefficients are $\eta/\tau_\pi$ and $\delta_{\pi\pi} = \frac{4}{3}\tau_\pi$. Furthermore, we remove the bulk viscosity and associated transport coefficients at the level of the principal part. The shear-stress evolution equation from Eq.~\eqref{eq:shear-nodiffusion} evolves as
\be
\label{eq:shear-nodiffusion-CONF}
\tensor{\Delta}{^{\mu\nu}_{\beta\gamma}}u^\alpha\nabla_\alpha\pi^{\beta\gamma} + \frac{\pi^{\mu\nu}}{\tau_\pi} = -2\frac{\eta}{\tau_\pi}\sigma^{\mu\nu} - \frac{4}{3} \pi^{\mu\nu}\nabla_\alpha u^\alpha.
\ee
The determinant of the principal part in this case is found to be
\be
\label{eq:determinant-CONF}
\det(\mathbb{A}^\alpha\phi_\alpha) = \tau_\pi^{17}v^{17}\hat{z}^{15}\left(E\hat{z}^2 - \frac{\eta}{\tau_\pi}\right)\mathcal{P}_3(\hat{z}^2,\hat{v}_1^2,\hat{v}_2^2,\hat{v}_3^2).
\ee
We also remark that $sT = E = \varepsilon + P = 4\varepsilon/3$ in this case, where $s$ is the entropy density and $T$ the temperature. The cubic portion may be expressed analogously to Eq.~\eqref{eq:supplementalIS}, except with the coefficients as
\bml
\bea
\mathcal{A}_0 &=& \frac{4}{3}\left(\frac{\eta}{\tau_\pi}\right)^2\frac{\dfrac{5}{4}(E + \Lambda_{\hat{v}^2}) - \left(E - \dfrac{\eta}{\tau_\pi}\right)}{(E + \Lambda_1)(E + \Lambda_2)(E + \Lambda_3)} - \frac{4}{9}\frac{\eta}{\tau_\pi}\sum\limits_{\substack{j,k = 1\\ j< k}}^{3}\frac{\hat{v}_j^2\hat{v}_k^2(\Lambda_j - \Lambda_k)^2}{(E+\Lambda_1)(E + \Lambda_2)(E + \Lambda_3)},\\
\mathcal{A}_1 &=& \frac{4}{9}-\frac{4}{9}(E + \Lambda_{\hat{v}^2})\sum_{k=1}^3\frac{\hat{v}_k^2}{E+\Lambda_k}+ \left(\frac{\eta}{\tau_\pi}\right)^2\frac{12E - (E + \Lambda_{\hat{v}^2})}{3 (E+\Lambda_1)(E+\Lambda_2)(E+\Lambda_3)}\notag\\
&&\qquad -\frac{4}{3}E\frac{\eta}{\tau_\pi}\left[\frac{3E - (E + \Lambda_{\hat{v}^2})}{(E+\Lambda_1)(E+\Lambda_2)(E+\Lambda_3)}-\frac{5}{4}\sum_{k=1}^3\frac{1-\hat{v}_k^2}{E+\Lambda_k}\right],\\
\frac{\mathcal{A}_2}{3} &=& \frac{5}{9} + \frac{1}{3}\frac{\eta}{\tau_\pi}\sum_{k = 1}^3\frac{1+\hat{v}_k^2/3}{E + \Lambda_k} - \frac{4}{9}E\sum_{k = 1}^3\frac{\hat{v}_k^2}{E + \Lambda_k}.
\eea
\eml
Here, we have applied the identity
\begin{align}
\sum\limits_{\substack{j,k = 1\\ j< k}}^{3}\frac{\hat{v}_j^2\hat{v}_k^2(\Lambda_j - \Lambda_k)^2}{(E+\Lambda_j)(E + \Lambda_k)} &= \left(1-(E + \Lambda_{\hat{v}^2})\sum_{k=1}^3\frac{\hat{v}_k^2}{E+\Lambda_k}\right),
\end{align}
which can be obtained by applying the traceless condition $\Lambda_1 + \Lambda_2 + \Lambda_3 = 0$ to obtain
\begin{align}
1\cdot 3E&=\left(\sum_{j=1}^3\hat{v}_j\right)\left(\sum_{k =1}^3(E + \Lambda_k)\right) =  (E + \Lambda_{\hat{v}^2}) + \sum_{j\neq k}\hat{v}_j^2 (E + \Lambda_k)
\end{align}
The determinant in Eq.~\eqref{eq:determinant-CONF} has the exact same structure as the general case in Eq.~\eqref{eq:nodiffusiondet}, meaning that nonlinear causality is immediately granted by Theorem~\ref{thm:causalitynodiffusion} - one simply replaces the inequalities in Eq.~\eqref{eq:C1}--\eqref{eq:C4} with the conformal coefficients above. In particular, since $\lambda_{\pi\Pi} = \tau_{\pi\pi} = 0$, the condition in Eq.~\eqref{eq:C5} reduces to
\be
\label{eq:ConfC5}
0 \leq \frac{\eta}{\tau_\pi E}\leq 1,
\ee
where $\frac{\eta}{\tau_\pi E} = \frac{\eta}{s}\frac{1}{\tau_\pi T}$. Since Theorem~\ref{thm:causalitynodiffusion} provides a set of necessary and sufficient conditions for nonlinear causality, it follows that if a single condition (e.g. for any given angle $\chi^2,\kappa^2\in[0,1]$) is violated, then so is causality. This property implies that each condition - even those for a particular angle - are all necessary. For instance, the condition $\mathcal{A}_0\geq 0$ may be rewritten as 
\be
\frac{1}{3}\sum\limits_{\substack{j,k = 1\\ j< k}}^{3}\hat{v}_j^2\hat{v}_k^2(\Lambda_j - \Lambda_k)^2\leq \frac{\eta}{\tau_\pi}\left[\dfrac{1}{4}E + \dfrac{5}{4}\Lambda_{\hat{v}^2} + \dfrac{\eta}{\tau_\pi}\right],\\
\ee
Choosing $\hat{v}_1^2 = \hat{v}_3^2 = 1/2$ and $\hat{v}_2^2 = 0$ then restricts the largest difference between shear-stress eigenvalues:
\be
|\Lambda_3 - \Lambda_1|\leq \left[3\frac{\eta}{\tau_\pi}\left(E - \frac{5}{2}\Lambda_2 + 4\dfrac{\eta}{\tau_\pi}\right)\right]^{1/2},\\
\ee
Fixing the characteristic covector along a particular principal axis, e.g. $\hat{v}_a^2 = 1$ provides us with the simple bound $\forall a = 1,2,3$:
\be
E + 5\Lambda_a + 4\frac{\eta}{\tau_\pi} \geq 0.
\ee
We also note that the set of necessary conditions for any permutation $(a,b,c)$ of $\{1,2,3\}$ may be generated from Eq.~\eqref{eq:C3} by fixing $\hat{v}_d^2 = 1$ for each $d = 1,2,3$:
\be
\left(E - \Lambda_a - 2\frac{\eta}{\tau_\pi}\right)\left(E + \Lambda_b - \frac{\eta}{\tau_\pi}\right)\left(E + \Lambda_c - \frac{\eta}{\tau_\pi}\right)\geq 0
\ee
Finally, we remark that although each of these conditions are necessary, they are not sufficient for ensuring causality alone. Therefore, all the inequalities in Theorem~\ref{thm:causalitynodiffusion} must be satisfied in order to guarantee causality for a given spacetime point. In Section~\ref{sec:MC}, we will show a particular, yet very general class of theories in which Theorem~\ref{thm:causalitynodiffusion} reduces drastically in algebraic complexity. 

\subsection{Linear Conditions}\label{subsec:linear}

Theories of hydrodynamics, including those of the Israel-Stewart class are systems of coupled, highly nonlinear PDEs whose solutions are often difficult or analytically impossible\footnote{Some notable exceptions are, for example, \cite{Marrochio:2013wla,Hatta:2014gqa,Hatta:2014gga}.} to retrieve.  One common method of gaining a rough understanding of the region of validity of these theories therefore has been to linearize the equations of motion, in the sense that one expands the degrees of freedom in Eq.~\eqref{eq:FirstOrderQPDE} up to linear fluctuations from their value at the (assumed) global equilibrium state: $U_a = U_{\textrm{eq.},a} + \delta U_a + \mathcal{O}(\delta U_a^2)$ for $a = 1,2,\dots,23$. Higher order fluctuations are then excluded. Below, we show that our nonlinear calculation reduces appropriately to the linear regime of phenomenological IS-theory \emph{exactly} \cite{Hiscock_Lindblom_stability_1983,Olson:1989ey}, since all second-order terms (including those only found in DNMR theory) vanish in this regime.

Since non-ideal effects from viscosity vanish at equilibrium, $\Pi,\Lambda_A = 0$ for all $A = 0,1,2,3$, which drastically simplifies (and in this case, decouples) the characteristic determinant in Eq.~\eqref{eq:nodiffusiondet} to
\be
\frac{\det(\mathcal{A}^\alpha\phi_\alpha)}{v^{23}} = (\varepsilon + P)^4\tau_\Pi\tau_\pi^{16}\hat{z}^{15}\left(\hat{z}^2 - \frac{\eta}{\tau_\pi (\varepsilon + P)}\right)^3\left[\hat{z}^2 - P_\varepsilon - \frac{1}{\varepsilon + P}\left(nP_n +\frac{\zeta}{\tau_\Pi} +\frac{4}{3} \frac{\eta}{\tau_\pi}\right) \right] \Bigg|_{\textrm{eq.}}+ \mathcal{O}(\delta U_a)
\ee
where the subscript `eq.' is understood to mean that all terms are evaluated at their equilibrium values. One therefore retrieves the necessary and sufficient conditions for \emph{linear} causality of the system in Eq.~\eqref{eq:FirstOrderQPDE}:
\bml
\bea
0 &\leq & \frac{\eta}{\tau_\pi (\varepsilon + P)} \leq 1,\\
0 &\leq & P_\varepsilon + \frac{1}{\varepsilon + P}\left(nP_n +\frac{\zeta}{\tau_\Pi} +\frac{4}{3} \frac{\eta}{\tau_\pi}\right) \leq 1.
\eea
\eml
These conditions are well-known, see \cite{Romatschke:2009im}.
We reiterate that although these linear conditions are significantly simpler than those in Theorem~\ref{thm:causalitynodiffusion}, they fail to constrain the bulk and shear stress $\Pi,\Lambda_a$ directly, and therefore, may fail to capture relevant out-of-equilibrium phenomena, such as plasma instabilities when coupling the theory to electromagnetic fields \cite{Cordeiro:2023ljz}. In the conformal limit (in which $E = 4\varepsilon /3$), the singular necessary and sufficient linear causality bound is
\be
0 \leq \frac{\eta}{\tau_\pi \varepsilon} \leq  \frac{2}{3},
\ee
which ends up being identical to the nonlinear case in Eq.~\eqref{eq:ConfC5} except with the out-of-equilibrium dependence of $\eta/\tau_\pi$ on $\Lambda_a$ replaced with its equilibrium values $\Lambda_a = 0$. Notice that the conditions in Eq.~\eqref{eq:C1}--\eqref{eq:C4} are not present in the linear case.

\section{Angle-Independent Generalized Maxwell-Cattaneo}\label{sec:MC}

The most general subclass of IS-like hydrodynamics we consider is that of a ``generalized" Maxwell-Cattaneo limit, in which one allows the transport coefficients to still depend on the viscous fluxes $\Pi$ and $\Lambda_a$, while also retaining second order transport coefficients that preserve the original tensor symmetries of the equations of Maxwell-Cattaneo. This subclass that we present here is significant from a phenomenological perspective, as we are able to remove the discriminant condition provided in Eq.~\eqref{eq:C1} of Theorem~\ref{thm:causalitynodiffusion} that guarantees the existence of all roots in $\mathbb{R}$, while simultaneously removing the non-trivial angular dependence of the remaining constraints in Eq.~\eqref{eq:C2}--\eqref{eq:C5} on $\chi$ and $\kappa$. This very general subclass retains most of the transport coefficients while significantly reducing the amount of constraints that need to be evaluated per fluid cell in simulations. Therefore, the subclass of theories discussed here is still sufficiently general and, thus, very well-suited for numerical simulations and phenomenological applications in heavy-ion collisions and astrophysics. 

Traditionally, one recovers the Maxwell-Cattaneo limit by setting all transport coefficients apart from $\eta,\zeta,\tau_\pi,\tau_\Pi$ to zero. Here, we will only choose $\delta_{\pi\pi} = \tau_{\pi\pi} = \lambda_{\Pi\pi} = 0$, as one can note that the other transport coefficients do not add any new non-trivial structures (e.g. tensors) to the equation. In this limit, note that Eq.~\eqref{eq:supplementalIS} takes the form
\bml
\label{eq:supplementalIS-MC}
\bea
\label{eq:bulk-nodiffusion-MC}
\tau_\Pi u^\alpha\nabla_\alpha\Pi + \Pi &=& -(\zeta + \delta_{\Pi\Pi}\Pi)\nabla_\alpha u^\alpha,\\
\label{eq:shear-nodiffusion-MC}
\tau_\pi\tensor{\Delta}{^{\mu\nu}_{\beta\gamma}}u^\alpha\nabla_\alpha\pi^{\beta\gamma} + \pi^{\mu\nu} &=& -2\left(\eta + \frac{1}{2}\lambda_{\pi\Pi}\Pi\right)\sigma^{\mu\nu},
\eea
\eml
where we remark that $\delta_{\Pi\Pi}$ and $\lambda_{\pi\Pi}$ directly act as second-order corrections to the bulk viscosity $\zeta$ and shear viscosity $\eta$, respectively, even though all transport coefficients may depend on $\Pi$ and $\Lambda_a$ in principle (thus, this class of theories is still sufficiently general). We will introduce the following notation to simplify constraints for the rest of the section:
\be
\frac{\zeta_{\textrm{eff}}}{\tau_\Pi} \equiv nP_n + \frac{\zeta + \delta_{\Pi\Pi}\Pi}{\tau_\Pi}  + \frac{1}{3}\frac{\eta_{\textrm{eff}}}{\tau_\pi};\qquad \eta_{\textrm{eff}}\equiv \eta + \frac{1}{2}\lambda_{\pi\Pi}\Pi.
\ee
This notation motivates writing
\be
c_{\textrm{eff},a}^2 \equiv P_\varepsilon + \frac{1}{E + \Lambda_a}\frac{\zeta_{\textrm{eff}}}{\tau_\Pi}
\ee
as the viscous corrections to the speed of sound (computed at constant $s/n$) squared $c_s^2 = P_\varepsilon + n P_n/(\varepsilon + P)$ along the principal axes of the shear stress tensor. An important fact to mention is that all transport coefficients and pressure gradients may be coalesced into these three effective viscous terms. The determinant of the principal part in this case is found to be
\be
\label{eq:MCdet}
\det(\mathbb{A}^\alpha\phi_\alpha) = \tau^{17}v^{17}\hat{z}^{15}\left(E\hat{z}^2 - \frac{\eta_{\textrm{eff}}}{\tau_\pi}\right)\mathcal{P}_3(\hat{z}^2,\hat{v}_1^2,\hat{v}_2^2,\hat{v}_3^2).
\ee
One can read off the nonlinear causality condition pertaining to the effective shear viscosity as
\be
0\leq \frac{\eta_{\textrm{eff}}}{\tau_\pi E}\leq 1.
\ee
The cubic portion has the same structure as Eq.~\eqref{eq:nodiffusiondet} with but instead with the reduced coefficients
\bml
\label{eq:MCcoeff}
\bea
\mathcal{A}_0(\chi^2,\kappa^2) &=& \dfrac{\eta_{\textrm{eff}}^2}{\tau_\pi^2}\frac{\dfrac{\zeta_{\textrm{eff}}}{\tau_\Pi} + \dfrac{\eta_{\textrm{eff}}}{\tau_\pi} + P_\varepsilon \left(E + \Lambda_{\hat{v}^2}\right)}{ (E + \Lambda_1) (E + \Lambda_2) (E + \Lambda_3)},\\
\mathcal{A}_1(\chi^2,\kappa^2) &=& \frac{\eta_{\textrm{eff}}}{\tau_\pi}\left[\dfrac{3\dfrac{\eta_{\textrm{eff}}^2}{\tau_\pi^2} + 2\dfrac{\zeta_{\textrm{eff}}}{\tau_\Pi}\left(E - \dfrac{1}{2}\Lambda_{\hat{v}^2}\right)}{(E + \Lambda_1)(E + \Lambda_2)(E +\Lambda_3)} + P_\varepsilon\sum_{k = 1}^3\frac{1 - \hat{v}_k^2}{E + \Lambda_k}\right],\\
\mathcal{A}_2(\chi^2,\kappa^2) &=& P_\varepsilon + \dfrac{\zeta_{\textrm{eff}}}{\tau_\Pi}\sum_{k = 1}^3\frac{\hat{v}_k^2}{E + \Lambda_k} + \dfrac{\eta_{\textrm{eff}}}{\tau_\pi}\sum_{k = 1}^3\frac{1}{E + \Lambda_k}.
\eea
\eml
Due to the analogous structure of the characteristic determinant, Theorem~\ref{thm:causalitynodiffusion} can be applied by simply setting $\lambda_{\Pi\pi} = \delta_{\pi\pi} = \tau_{\pi\pi} = 0$ and replacing the coefficient definitions with Eq.~\eqref{eq:MCcoeff}. We remark that strong hyperbolicity with these new coefficients is provided by Theorem~\ref{thm:nodiffusionstronghyperbolicity} in Section~\ref{sec:hyperbolicity}. However, we can further simplify these conditions using both multivariate optimization and properties of cubic polynomials. Recall that the shear stress tensor is traceless such that $\Lambda_1 + \Lambda_2 + \Lambda_3 = 0$. Therefore, unless $\Lambda_a = 0$ for each $a = 1,2,3$ (we ignore this trivial case), there always exists one eigenvalue that is strictly positive and one that is strictly negative. We will rearrange indices such that
\bml
\bea
\Lambda_1 &\leq& \Lambda_2\leq \Lambda_3,\\
\Lambda_1 &<&0 <\Lambda_3.
\eea
\eml
Furthermore, we shall also assume that the shear stress cannot grow large enough (relative to the energy density and pressure $E = \varepsilon + P + \Pi$) such that
\be
E + \Lambda_1 > 0.
\ee
Otherwise, the characteristic determinant would be allowed to vanish (and thus become singular) at some points, potentially causing unphysical singularities in the characteristics themselves. Naturally, this behavior would lead to a breakdown of the equations of motion. To remove the angles, we simply determine the global maxima/minima that the coefficients reach within the interval $\chi^2,\kappa^2\in [0,1]$, and check only those extrema. These extremal values provide necessary and sufficient conditions for the original bounds, while removing the angles completely. For instance, Eq.~\eqref{eq:C2} is a nonlinear expression that enforces the condition $\mathcal{A}_0(\chi^2,\kappa^2)\geq 0$ for all $\chi^2,\kappa^2\in [0,1]$. Suppose that this function reaches a minimum at $(\chi^2,\kappa^2) = (\widetilde{\chi}^2,\widetilde{\kappa}^2)\in [0,1]^2$. Then the following statements are equivalent:
\be
\forall \chi^2,\kappa^2\in [0,1];\quad \mathcal{A}_0(\chi^2,\kappa^2)\geq 0\quad\Leftrightarrow\quad \min_{\chi^2,\kappa^2\in[0,1]}\mathcal{A}_0(\chi^2,\kappa^2) \equiv \mathcal{A}_0(\widetilde{\chi}^2,\widetilde{\kappa}^2)\geq 0.
\ee
Here we are taking an angle-dependent condition and then enforcing the strictest bound possible for one particular angle (corresponding to a critical point or boundary term). This process then allows us to satisfy causality for all angles immediately. For the rest of this section, we shall assume that 
\be
\label{eq:assumptionMC}
P_\varepsilon \geq 0;\qquad \frac{\zeta_{\textrm{eff}}}{\tau_\Pi} \geq 0,
\ee
which then implies that $c_{\textrm{eff},a}^2\geq 0$. This assumption is very physical and most likely covers most realistic initial data - in fact, in the linear regime, causality enforces $\frac{\zeta_{\textrm{eff}}}{\tau_\Pi}\geq 0$, whereas $P_\varepsilon$ is always positive in the linear regime when there is no baryon current, since $c_s^2 = P_\varepsilon + \frac{nP_n}{\varepsilon + P}$ is the hydrodynamic speed of sound. The interested reader can consult Appendix~\ref{App:MC} to see how the relevant conditions change when $P_\varepsilon < 0$ for instance, but we refrain from listing them here since the former case is most likely more relevant. 
\begin{corollary}[Angle-Independent Nonlinear Causality of Generalized Maxwell-Cattaneo]
Let $\tau_\pi,\tau_\Pi\neq 0$, $E + \Lambda_1 > 0$, and assume that $\Lambda_1 \leq \Lambda_2 \leq \Lambda_3$ such that $\Lambda_1 < 0 <\Lambda_3$. Furthermore, assume that for all $a = 1,2,3$ the effective hydrodynamic sound speeds are non-negative:
\be
\normalfont
c_{\textrm{eff},a}^2\geq 0.
\ee
Assume that the following conditions are satisfied for all such $a = 1,2,3$, and that $P_\varepsilon,\zeta_{\normalfont{\textrm{eff}}}/\tau_\Pi \geq 0$:
\bml
\normalfont
\bea
\label{eq:MCC1}
0 \leq \frac{\eta_{\textrm{eff}}}{\tau_\pi E} &\leq & 1,\\
\label{eq:MCC2}
c_{\textrm{eff},1}^2 + \dfrac{1}{E + \Lambda_1}\dfrac{\eta_{\textrm{eff}}}{\tau_\pi} &\geq & 0,\\
\label{eq:MCC3}
\left(E - \dfrac{\Lambda_a}{2}\right)(E + \Lambda_a)c_{\textrm{eff},a}^2+ \dfrac{3}{2}\dfrac{\eta_{\textrm{eff}}}{\tau_\pi}E &\geq & 0,\\
\label{eq:MCC4}
c_{\textrm{eff},3}^2 + \frac{\eta_{\textrm{eff}}}{\tau_\pi}\sum_{k = 1}^3\frac{1}{E +\Lambda_k} &\geq & 0,\\
\label{eq:MCC5}
c_{\textrm{eff},1}^2 + \frac{\eta_{\textrm{eff}}}{\tau_\pi}\sum_{k = 1}^3\frac{1}{E +\Lambda_k} &\leq & 3,\\
\label{eq:MCC6}
1 -\frac{2}{3}\left[c_{\textrm{eff},a}^2 + \frac{\eta_{\textrm{eff}}}{\tau_\pi}\sum_{k = 1}^3\frac{1}{E + \Lambda_k}\right]+ \dfrac{\eta_{\textrm{eff}}}{\tau_\pi}\left[\frac{\left(E - \dfrac{\Lambda_a}{3}\right) (E + \Lambda_a)c_{\textrm{eff},a}^2 + \dfrac{8}{9}\dfrac{\eta_{\textrm{eff}}}{\tau_\pi}E}{(E + \Lambda_1)(E + \Lambda_2)(E + \Lambda_3)}\right] &\geq & 0,
\eea
\eml
along with the following condition for all permutations $(a,b,c) = (1,2,3),(1,3,2),(2,3,1)$ that
\be
\normalfont
\label{eq:MCC7}
\left(1 - \dfrac{1}{E + \Lambda_a}\dfrac{\eta_{\textrm{eff}}}{\tau_\pi}\right)\left(1 - \dfrac{1}{E + \Lambda_b}\dfrac{\eta_{\textrm{eff}}}{\tau_\pi}\right)\left(1-c_{\textrm{eff,c}}^2 - \dfrac{1}{E + \Lambda_c}\dfrac{\eta_{\textrm{eff}}}{\tau_\pi}\right)\geq 0.
\ee
Then the generalized Maxwell-Cattaneo system --- obtained by applying $\delta_{\pi\pi},\tau_{\pi\pi},\lambda_{\Pi\pi}\rightarrow 0$ to the original equations of motion in \eqref{eq:FirstOrderQPDE} (or equivalently \eqref{eq:propsystem}) --- is causal if and only if Eq.~\eqref{eq:MCC1}--\eqref{eq:MCC6},\eqref{eq:MCC7} hold, whereas the system \eqref{eq:FirstOrderQPDE} (or equivalently \eqref{eq:propsystem})  subject to the constraints \eqref{Constraints}, is causal if, and only if, Eq.~\eqref{eq:MCC2}--\eqref{eq:MCC6},\eqref{eq:MCC7} hold. Note that the above set of inequalities are independent of the continuous parameters $\chi^2,\kappa^2\in[0,1]$, and furthermore, one does not need to check the sign of the discriminant, as real roots are guaranteed under the above assumptions.

\begin{proof}
The proof begins by applying Theorem~\ref{thm:causalitynodiffusion} to the generalized Maxwell-Cattaneo subcase, which we are able to drastically simplify by calculating minima and maxima of the coefficients and linear combinations thereof with respect to the parameters $\chi^2,\kappa^2\in[0,1]$. Furthermore, we show, using elementary analysis results (e.g. intermediate value theorem, implicit function theorem \cite{krantz2013implicit}) that three real roots are always guaranteed under these assumptions, and thus, one does not need to check the discriminant. In some of the conditions, the extrema in Eq.~\eqref{eq:C2}--\eqref{eq:C5} have multiple subcases depending on the relative sizes of viscous fluxes, transport coefficients and other dynamic variables. Therefore, checking the above conditions for all given values and/or permutations of the indices $a,b,c\in\{1,2,3\}$ then condenses all of these subcases into a single statement. Appendix~\ref{App:MC} contains all details of the calculation for the interested reader.
\end{proof}
\end{corollary}

\section{Propagation of Constraints}\label{subsec:propagate}

When considering a system of nonlinear PDEs, one must ensure that the symmetries and algebraic constraints imposed on the initial data are preserved by the dynamical evolution. For the system consisting of the conservation equations for the currents in \eqref{eq:conserved_quantities} supplemented by the DNMR equations \eqref{eq:supplementalIS} for the viscous fields, these constraints are found in \eqref{Constraints}. The physical degrees of freedom are the energy density $\varepsilon$, baryon density $n$, the spatial components of the four-velocity $u^i$, the bulk viscous pressure $\Pi$, and the 5 independent components of the shear stress tensor $\pi^{\mu\nu}$. In particular, the constraints over $\pi^{\mu\nu}$ are directly obtained as algebraic solutions of \eqref{eq:shear-nodiffusion}---for example, contract this equation with $u_\nu$ to obtain that $\pi^{\mu\nu}u_\nu=0$. While one can in principle write the system only in terms of the physical degrees of freedom (e.g., in terms of $u^i$, $i=1,2,3$, using the normalization $u^\mu u_\mu = -1$ to eliminate $u^0$), such a procedure breaks the covariance of the equation.
It is often advantageous to maintain covariance, treating all 4 components of $u^\mu$ and all 16 components of $\pi^{\mu\nu}$ as independent dynamical variables within the principal part of the system. In this case, one needs to show that the constraints are recovered, or said differently, preserved  by the evolution; i.e., one needs to show that they hold for $t>0$ provided they hold initially at $t=0$.
This was illustrated in the case of an ideal fluid 
in Sec.\ \ref{sec:augmented_system}. Aside form an important structural aspect of the problem, the propagation of the constraints is used in our proof of local well-posedness in Section \ref{sec:hyperbolicity}. The main result of this Section is the following.

\begin{theorem}[Constraint propagation]
\label{T:Constraint_propagation}
    Under the same assumptions of Theorem \ref{thm:causalitynodiffusion},
the system \eqref{eq:FirstOrderQPDE} (or equivalently \eqref{eq:propsystem}) propagates the constraints \eqref{Constraints} if, for all $\chi^2,\kappa^2 \in [0,1]$, \eqref{eq:C5} holds with $0 \leq$ replaced by $0<$ (equivalently, if 
\eqref{eq:H5} holds for all $\chi^2,\kappa^2 \in [0,1]$).
\end{theorem}
\begin{proof}
    The proof is obtained by showing that the constraints satisfy a linear strongly hyperbolic homogeneous system. See Section \ref{S:Constrain_propagation_proof} for details. 
\end{proof}

We note that the causality conditions \eqref{eq:C1}--\eqref{eq:C4}, or their weaken versions \eqref{eq:H1}--\eqref{eq:H4}, are not assume in Theorem \ref{T:Constraint_propagation}. This is because Theorem \ref{T:Constraint_propagation} assumes a given solution to the nonlinear system \eqref{eq:FirstOrderQPDE} (or equivalently \eqref{eq:propsystem}). In practice, in order to show that solutions to \eqref{eq:FirstOrderQPDE} exist, we need \eqref{eq:H1}--\eqref{eq:H4} and other hypotheses in addition to \eqref{eq:H5}. See Theorem \ref{T:LWP}.

\section{Strong hyperbolicity}\label{sec:hyperbolicity}

In the previous sections, we discussed nonlinear causality conditions for solutions of the system given by Eq.~\eqref{eq:conservationlaw} and \eqref{eq:supplementalIS}. Although these equations of motion may be evolved numerically, there is no prior guarantee that the numerical initial value problem provides and faithfully reproduces, within acceptable error bounds, the actual solutions of the equations motion. It then becomes imperative to have a mathematical guarantee that such nonlinear systems of PDEs even have solutions, and that such solutions are unique. Here, we show for the first time that IS-like theories with bulk and shear viscosity have unique local solutions that exist in very general function spaces. These properties hold in the subset of the causal region satisfying the additional hypotheses of Theorem~\ref{thm:nodiffusionstronghyperbolicity}.

For a system of quasilinear PDEs, strong hyperbolicity \cite{ReulaStrongHyperbolic} is a desirable trait that can in particular be used to establish local well-posedness\footnote{When discussing existence, one must be precise about the form of such solutions. For this paper, we are concerned about solutions in Sobolev spaces. These function spaces are significantly more general than analytic and Gevrey-type functions, and allow for approximation by those in $C^\infty$ \cite{Disconzi-2024}.} and uniqueness of the Cauchy problem \cite{ChoquetBruhatGRBook}.

\begin{definition}[Strong hyperbolicity]
\label{D:Strong_hyperbolicity}
   We say a quasilinear system $(\mathbb{A}^\alpha\partial_\alpha + \mathbb{B})\Umb = \boldsymbol{0}$ is \emph{strongly hyperbolic} if, given some time-like vector $\xi^\mu$
\begin{enumerate}[label=(H\Roman*),nosep]
 \item $\det(\mathbb A^\alpha \xi_\alpha) \not= 0$, and
 \item for any space-like vector $\zeta^\mu$, the solutions of the eigenvalue equation $(\Lambda \xi_\alpha + \zeta_\alpha)\mathbb A^\alpha \mathbf r = \mathbf 0$ exist for $\Lambda\in\mathbb{R}$ and the right eigenvectors $\mathbf{r}$ span a complete basis.
\end{enumerate} 
\end{definition}
A precise pointwise formulation of strong hyperbolicity is given in
Definition~\ref{D:IS_Precise_strong_hyperbolicity}.
\begin{theorem}[Nonlinear Strong Hyperbolicity]\label{thm:nodiffusionstronghyperbolicity}
Let $\tau_\Pi,\tau_\pi\neq 0$,
\be
E + \Lambda_A \neq 0,\quad A = 0,1,2,3,
\ee
and $\Delta(\mathcal{P}_3;\chi^2,\kappa^2)$ be the discriminant of $\mathcal{P}_3$ defined by Eq.~\eqref{eq:discriminant}. Assume that the following conditions hold for all $\chi^2,\kappa^2\in[0,1]$:
\bml
\bea
\label{eq:H1}
0& < &\Delta(\mathcal{P}_3;\chi^2,\kappa^2),\\
\label{eq:H2}
0 &<& \frac{\mathcal{A}_2(\chi^2,\kappa^2)}{3} \leq 1,\quad 0< \mathcal{A}_1(\chi^2,\kappa^2),\quad 0 < \mathcal{A}_0(\chi^2,\kappa^2)\\
\label{eq:H3}
0 &\leq & 1-\mathcal{A}_2(\chi^2,\kappa^2) + \mathcal{A}_1(\chi^2,\kappa^2) - \mathcal{A}_0(\chi^2,\kappa^2),\\
\label{eq:H4}
0 &\leq & 1-\frac{2}{3}\mathcal{A}_2(\chi^2,\kappa^2) + \frac{1}{3}\mathcal{A}_1(\chi^2,\kappa^2) ,\\
\label{eq:H5}
0 &< & \frac{\frac{1}{2} (2\eta + \lambda_{\pi\Pi}\Pi)+\frac{1}{4} \tau_{\pi\pi}\Lambda_{\hat{v}^2}}{\tau_\pi E} \leq 1,
\eea
\eml
In addition, assume that at least one of the following two alternatives holds. Define
\begin{subequations}
\label{E:IS_Separation_coefficients}
\begin{align}
 a&=\frac{\eta+\frac{1}{2}\lambda_{\pi\Pi}\Pi}{\tau_\pi},\\
 b&=\frac{\tau_{\pi\pi}}{4\tau_\pi},\\
 c&=EP_\varepsilon+nP_n+
       \frac{\zeta+\delta_{\Pi\Pi}\Pi}{\tau_\Pi}+\frac{a}{3},\\
 d&=\frac{\delta_{\pi\pi}}{\tau_\pi}-\frac{b}{3},\\
 e&=P_\varepsilon+\frac{\lambda_{\Pi\pi}}{\tau_\Pi}-\frac{b}{3},\\
 m_\pi&=\min_{1\leq j\leq3}|\Lambda_j|,\\
 L_\pi&=\max_{1\leq j\leq3}|\Lambda_j|,\\
\label{E:IS_Extra_squared_speed}
 y_0(\chi^2,\kappa^2)&=\frac{a+b\Lambda_{\hat v^2}}{E}.
\end{align}
\end{subequations}
\begin{enumerate}[label=(\roman*),nosep]
\item The quadratic and cubic factors are separated:
\begin{align}
\label{E:IS_Separation_polynomial}
 \mathcal{P}_3(y_0;\chi^2,\kappa^2)
 &=y_0^3-\mathcal{A}_2y_0^2+\mathcal{A}_1y_0-\mathcal{A}_0
 \neq0
 \quad\text{for all }\chi^2,\kappa^2\in[0,1].
\end{align}
Here $y_0$ and the three coefficients are evaluated at the same state and
at the same values of $\chi^2,\kappa^2$.
\item The following inequality, which has no angular parameters, holds:
\begin{align}
\label{E:IS_Separation_transport}
 m_\pi\frac{|a-bE|-|b|L_\pi}{|E|}
 &>|c|+(|d|+|e|)L_\pi.
\end{align}
\end{enumerate}
Then, the system \eqref{eq:FirstOrderQPDE} (or equivalently \eqref{eq:propsystem}), subject to the constraints \eqref{Constraints}, is strongly hyperbolic, with or without a baryon current.
\end{theorem}
A precise statement of this Theorem is given in 
Theorem~\ref{T:IS_Precise_strong_hyperbolicity} in
Appendix~\ref{App:hyperbolicity}. The regularity and uniform bounds needed for the subsequent proof of local well-posedness are given in Section~\ref{S:IS_Linear_estimates_hypotheses}.
\begin{proof}
The coefficient inequalities give three distinct cubic roots in $(0,1]$.
Either additional hypothesis separates them from the root of the quadratic factor. See Theorem~\ref{T:IS_Precise_strong_hyperbolicity} in Appendix~\ref{App:hyperbolicity} for the details .
\end{proof}
It is important to highlight that for quasilinear systems as the one studied here, \emph{strong hyperbolicity depends on the solution variables,} since the matrices $\mathbb{A}^\alpha$ depend on $\Umb$, i.e., $\mathbb{A}^\alpha= \mathbb{A}^\alpha(\Umb)$. Thus, some attention to the argument's logic is needed when invoking strong hyperbolicity to establish existence of solutions since we need $\Umb$ to check whether $\mathbb{A}^\alpha(\Umb)$ is diagonalizable in the sense of Definition \ref{D:Strong_hyperbolicity}, but the existence of $\Umb$ is what we are trying to establish in the first place. The common practice is to establish strong hyperbolicity using the Cauchy data $\left.\Umb \right|_{t=0}$, since in this case the matrices $\left.\mathbb{A}^\alpha(\Umb)\right|_{t=0}$ are known. In such cases, under very general assumptions on the form of $\mathbb{A}^\alpha$ and the data, standard arguments can be employed to show that solutions exist (see, e.g., \cite{Shao:2023psr} for a short self-contained proof). 

In our case, however, there is an additional complication, namely, our proof of strong hyperbolicity explicitly uses the constraints \eqref{Constraints}. Even if such constraints are satisfied at $t=0$ (as they must for data for data for Israel-Stewart-like systems), there is in principle no guarantee that they will continue to hold for $t>0$ when solutions to \eqref{eq:FirstOrderQPDE} are known to exist. The standard arguments for existence that we referred above, however, apply to \emph{unconstrained systems,} and therefore cannot be immediately implemented in our setting. In order to understand where the loophole is, we need to briefly recall the structure of the argument that uses strong hyperbolicity to produce solutions, which goes as follows. Starting with $\Umb_0 \equiv \left.\Umb\right|_{t=0}$ and assuming that $\mathbb{A}^\alpha(\Umb_0)$ satisfies the conditions of Definition \ref{D:Strong_hyperbolicity}, we consider the \emph{linear} system
\begin{align}
    \begin{split}
        \label{E:Linear_system_SH}
\mathbb{A}^\alpha(\Umb_0)\partial_\alpha \Umb_1 + \mathbb{B}(\Umb_0)  &= 0,
\\
\left. \Umb_1 \right|_{t=0} &= \Umb_0
    \end{split}
\end{align}
for the unknown $\Umb_1$. By construction this is a linear strongly hyperbolic system and then techniques for such linear systems yield a solution $\Umb_1$. Among the aforementioned ``general assumptions on the form of $\mathbb{A}^\alpha$," the key one here is an \emph{openness} assumption: it is assumed that if $\mathbb{A}^\alpha(\Umb_0)$ satisfies Definition \ref{D:Strong_hyperbolicity}, then $\mathbb{A}^\alpha(\mathbf{V})$ also satisfies Definition \ref{D:Strong_hyperbolicity} for all $\mathbf{V}$ sufficiently close\footnote{Where closeness is measured in the natural topology of the problem, which in our case is in norms related to the Sobolev norm. Similarly for other topological statements here, like openness and convergence.} to $\Umb_0$. In other words, it is assumed that Definition \ref{D:Strong_hyperbolicity} holds in an open set $\mathcal{U}$ containing $\Umb_0$. Since $\left. \Umb_1 \right|_{t=0} = \Umb_0$, continuity of the solution $\Umb_1$ guarantees that $\Umb_1 \in \mathcal{U}$, at least for small time, so that $\mathbb{A}^\alpha(\Umb_1)$ also satisfies Definition \ref{D:Strong_hyperbolicity}. Proceeding in this way, we inductively construct a sequence of linear problems
\begin{align*}
    \mathbb{A}^\alpha(\Umb_n) \partial_\alpha \Umb_{n+1} + \mathbb{B}(\Umb_n) &=0,
    \\
    \left. \Umb_{n+1} \right|_{t=0} & = \Umb_0,
\end{align*}
for the unknowns $\Umb_{n+1}$, where at each step $\Umb_n$ is known and $\mathbb{A}^\alpha(\Umb_n)$ satisfies Definition \ref{D:Strong_hyperbolicity}, so that we obtain a solution $\Umb_{n+1}$ for each $n$. We can  prove that the sequence $\{\Umb_n\}_{n=0}^\infty$ converges to a limit $\Umb_\infty$. Such limit then solves the desired quasilinear problem since
\begin{align*}
    0 = \lim_{n\rightarrow\infty} \big( \mathbb{A}^\alpha(\Umb_n) \partial_\alpha \Umb_{n+1} + \mathbb{B}(\Umb_n)  \big)
    = \mathbb{A}^\alpha(\Umb_\infty) \partial_\alpha \Umb_{\infty} + \mathbb{B}(\Umb_\infty)
\end{align*}
and $ \Umb_0 =   \lim_{n\rightarrow\infty} \left. \Umb_{n+1} \right|_{t=0} = \left. \Umb_\infty \right|_{t=0}$.

Let us now see why the above argument in general fails for constrained systems. Assuming that the data satisfies the constraints, we can still construct the linear strongly hyperbolic system \eqref{E:Linear_system_SH} and produce a solution $\Umb_1$. However, if the constraints are not open (as in general constraints are not), we cannot, by the above continuity and openness argument, ensure that the subsequent iterate $\mathbb{A}^\alpha(\Umb_1)$ will satisfy Definition \ref{D:Strong_hyperbolicity}, so that the inductive argument breaks down. In our case, our proof of strong hyperbolicity uses \eqref{Constraints}. Even if the data satisfies \eqref{Constraints},
so that $\mathbb{A}^\alpha(\Umb_0)$ then satisfies Definition \ref{D:Strong_hyperbolicity}, the constraints \eqref{Constraints} need not to hold for the solution $\Umb_1$ because continuity and openness are not enough to ensure that an exact equality holds. To see this concretely in our case, consider for example the constraint $u_\alpha \pi^{\alpha \mu} = 0$. For the data, we have
$
    \left. u_\alpha  \pi^{\alpha\mu}\right|_{t=0} = 0
$,
so that for $(u_0)_\alpha$ and $(\pi_0)^{\alpha\mu}$ in $\Umb_0$ we have  $(u_0)_\alpha (\pi_0)^{\alpha\mu}=0$ and thus $\mathbb{A}^{\alpha}(\Umb_0)$ satisfies Definition \ref{D:Strong_hyperbolicity}. For $(u_1)_\alpha$ and $(\pi_1)^{\alpha\mu}$ in $\Umb_1$, however, it will not in general be true that  we have  $(u_1)_\alpha (\pi_1)^{\alpha\mu}=0$ because this does not follow from $(u_1)_\alpha$ and $(\pi_1)^{\alpha\mu}$ being close to $(u_0)_\alpha$ and $(\pi_0)^{\alpha\mu}$, i.e., for any small $\delta > 0$
\begin{align*}
    (u_0)_\alpha (\pi_0)^{\alpha\mu}=0
    \text{ and }
    | (u_0)_\alpha - (u_1)_\alpha | < \delta 
    \text{ and }
    |(\pi_0)^{\alpha\mu} - (\pi_1)^{\alpha\mu} | < \delta
    \mathrel{\rlap{\hskip .5em/}}\Longrightarrow
    (u_1)_\alpha (\pi_1)^{\alpha\mu}=0,
\end{align*}
so that $\mathbb{A}^{\alpha}(\Umb_1)$ will not in general satisfy Definition \ref{D:Strong_hyperbolicity}.

Fortunately, there is a different argument that goes back to Schauder and the earlier days of hyperbolic PDEs \cite{Schauder-1935} (see \cite[Chapter 5]{John-Book-1982} and \cite[Chapter VI]{Courant_and_Hilbert_book_2} for more recent expositions) that can be applied to a constrained system, provided the constraints are propagated by a suitable sub-class of solutions that can be obtained independently of strong hyperbolicity.  Strong hyperbolicity is still used in the argument to conclude local well-posedness in Sobolev spaces, but its use is indirect and legitimate. We discuss this in Section \ref{S:LWP}, while the required nonlinear propagation of constraints  which, is of independent interest and logically independent strong hyperbolicity and local well-posedness, is presented in Section \ref{subsec:propagate}.

\begin{remark}
    We note that if a constraint-like restriction is given by an open condition, then the above argument usually goes through. For example, for the ideal fluid, strong hyperbolicity requires $\varepsilon > 0$. We can impose this condition on the initial data $\varepsilon_0$, so that $\mathbb{A}(\Umb_0)$ satisfies Definition \ref{D:Strong_hyperbolicity}. Then, if $\varepsilon_1 \in \Umb_1$ is sufficiently close to $\varepsilon_0$, we obtain $\varepsilon_1 > 0$, and so on.
\end{remark}

\section{Local well-posedness}
\label{S:LWP}

\begin{theorem}[Local well-posedness]
    \label{T:LWP}
    The system \eqref{eq:FirstOrderQPDE} (or equivalently \eqref{eq:propsystem}),
    in a globally hyperbolic spacetime is locally well-posed in the Sobolev space $H^s$, $s > \frac{3}{2} + 2$, for data given along a Cauchy surface, provided that the data satisfies \eqref{Constraints}, the assumptions of Theorem \ref{thm:nodiffusionstronghyperbolicity}, and the transport coefficients are analytic functions of their arguments.
\end{theorem}
We refer to Theorem \ref{T:IS_Precise_Sobolev} in Appendix \ref{App:LWPSchauder} for a precise statement of Theorem \ref{T:LWP} and its proof. Here, we outline the logic of the argument, showing how we can still rely on strong hyperbolicity, albeit in an indirect fashion, despite the presence of the constraints. We consider the system \eqref{thm:nodiffusionstronghyperbolicity} written in the form \eqref{eq:FirstOrderQPDE}.

By Theorem \ref{thm:causalitynodiffusion} applied to the initial data $\Umb_0 \equiv \left. \Umb\right|_{t=0}$, the initial hypersurface, which we can take as parametrized by $t=0$, is non-characteristic. Thus, given analytic initial-data $\Umb_0^\omega$, there exists an analytic solution to \eqref{eq:FirstOrderQPDE} defined in a neighborhood of $\{t=0\}$. Given Sobolev-regular initial data $\Umb_0$, let $\Umb_{0,n}^\omega$ be a sequence of analytic initial data converging to $\Umb_0$. Without loss of generality we can assume that each $\Umb_{0,n}^\omega$ satisfies \eqref{Constraints} and the assumptions of Theorem \ref{thm:nodiffusionstronghyperbolicity}. For each $n$, there thus exists a corresponding analytic solution $\Umb_n^\omega$. It is known that we can further construct the sequence of analytic solutions in such a way that they are all defined on the same time interval $[0,T^\omega]$. Applying Theorem \ref{T:Constraint_propagation} to this sequence, we therefore obtain that the solutions $\Umb_n^\omega$ satisfy the constraints \eqref{Constraints} on $[0,T^\omega]$. Therefore, we can invoke Theorem \ref{thm:nodiffusionstronghyperbolicity} to conclude that for each $n$ the corresponding system for $\Umb_n^\omega$ is strongly hyperbolic, i.e, $\mathbb{A}^\alpha(\Umb_n^\omega)$ satisfies Definition \ref{D:Strong_hyperbolicity} on $[0,T^\omega]$. Standard energy estimates for strongly hyperbolic systems give that the sequence $\{ \Umb_n^\omega \}$ is uniformly bounded in $H^s$ and the sequence of differences $\Umb_n^\omega - \Umb_m^\omega$ is a Cauchy sequence in $H^{s-1}$ which thus converges to a limit that itself belongs to $H^s$. Passing to the limit on \eqref{eq:FirstOrderQPDE}, we see that the limit of the sequence satisfies the equation and by construction takes the correct initial data. Uniqueness follows by the same type of energy estimate for the difference of two solutions.

The key point to emphasize is that \emph{we do not obtain solutions by invoking the existence theory for strongly hyperbolic systems,} since such existence theory applies to unconstrained systems but our system is constrained, with the very proof of strong hyperbolicity relying on the use of the constraints. The only element of strong hyperbolicity that is used are the \emph{energy estimates for these systems applied to already proven to exist and constraint-preserving analytic solutions.}

\section{Coupling to Einstein's Equations}\label{sec:Einstein}

In the previous sections, we assumed a completely general, yet non-dynamical metric $g_{\mu\nu}$, in the sense that it was not evolved alongside the fluid equations of motion in Eq.~\eqref{eq:FirstOrderQPDE}. While the causality and strong hyperbolicity bounds in Theorem~\ref{thm:causalitynodiffusion} and \ref{thm:nodiffusionstronghyperbolicity} signify a significant leap in understanding the matter portion of IS-like theories, a non-dynamical metric obviously fails to account for non-trivial coupling of the matter and spacetime, which may be relevant particularly for back-reactions of matter on the spacetime in extreme astrophysical phenomena. Here, we show not only that our causality result is compatible with a dynamic spacetime --- we also show that the entire system of matter/fluid plus Einstein's equations admits a strongly hyperbolic development of the Cauchy problem under the \emph{exact same conditions} as Theorem~\ref{thm:causalitynodiffusion} and \ref{thm:nodiffusionstronghyperbolicity}. 

To see this, consider the original (full) set of hydrodynamic equations of motions given in Eq.~\eqref{eq:FirstOrderQPDE}. We shall relabel each of the components in this system with the subscript `$M$' in order to distinguish our matter variables from our gravity variables:
\be
(\mathbb{A}_M^\alpha\partial_\alpha  + \mathbb{B}_M)\textbf{\textrm{U}}_M = \boldsymbol{0}_{23}.
\ee
Here, $\boldsymbol{0}_n\in\mathbb{R}^n$ is the zero vector. In addition to this set of $23$ equations, we add $10$ degrees of freedom due to the (symmetric) metric in the form of Einstein's equations:
\be
\mathcal{R}_{\mu\nu} - \frac{1}{2}\mathcal{R}g_{\mu\nu} + \Lambda g_{\mu\nu} = 8\pi G\,T_{\mu\nu},
\ee
where $G$ is Newton's gravitational constant, and $\Lambda$ is the cosmological constant. Here, $\mathcal{R}_{\mu\nu}$ is the Ricci curvature tensor and $\mathcal{R} \equiv \tensor{\mathcal{R}}{^\alpha_\alpha}$ is the Ricci scalar. Exploiting the coordinate independence of Einstein's equations, we work in the harmonic gauge \cite{ChoquetBruhatGRBook}, defined as the wave-like coordinate $x^\mu$ solutions:
\be
\nabla^\alpha\nabla_\alpha x^\mu = 0^\mu.
\ee
This definition implies that harmonic gauge constrains the Christoffel symbols $\tensor{\Gamma}{^\rho_\mu_\nu}$ via the condition
\be
g^{\alpha\beta}\tensor{\Gamma}{^\mu_\alpha_\beta} = 0.
\ee
Using this condition and taking the trace of Einstein's equations (to rewrite the Ricci scalar in terms of matter variables) provides us with
\be
\label{eq:HarmonicEinsteins}
-\frac{1}{2}\partial^\alpha\partial_\alpha g_{\mu\nu}= \mathcal{O}(g,\textbf{\textrm{U}}_M,\partial g)
\ee
where the right side of the equation consists of lower-order derivative terms that do not contribute to the principal part. Note that, together, Eq.~\eqref{eq:HarmonicEinsteins} alongside the augmented system in \eqref{eq:propsystem} constitute a system of 33 mixed-order equations. When considering systems of mixed order, there inherently lies an ambiguity in how one assigns terms to the principal part. Historically, this is handled by assigning relative weights to each term in the equations of motion \cite{LerayOhyaNonlinear}. However, since Einstein's equations contain no coordinate derivatives of the matter variables, and vice-versa, the principal part becomes decoupled, and one can simply treat the coefficients of the first-order derivatives of matter and second-order derivatives as gravity as the sole contributors to the principal part.\footnote{In terms of Leray-Ohya indices, we assign a weight $n_j = 0$ for each of the $33$ equations, $m_{G,\frak{A}} = 2$ for each of the $\frak{A} = 1,2,\dotso,10$ Einstein's equations, and $m_{M,\mathscr{A}} = 1$ to the remaining $\mathscr{A} = 1,2,\dotso,23$ matter equations. Each variable $i$ whose derivative order satisfies $\alpha = m_i - n_j$ in equation $j$ is kept in the principal part, anything with $\alpha < m_i - n_j$ is excluded. The Leray-Ohya indices are correctly constructed if there is it precludes the case where $\alpha > m_i - n_j$ in  equation $i$.} We can then express the total system as
\be
\label{eq:mixedorderEOM}
(\hat{A}(\partial,\partial^2) + \mathbb{B})\textbf{\textrm{U}} = \boldsymbol{0}_{33}.
\ee
This system is a mixed-order quasilinear system of partial differential equations. We once again proceed by the method of characteristics, see Chapter II, Appendix I in \cite{Courant_and_Hilbert_book_2}. As per usual, $\mathbb{B}$ is the $33\times 33$ matrix containing all terms not present in the principal part, which we suppress here. The degrees of freedom are uploaded into the solution vector $\textbf{\textrm{U}} =
\textbf{\textrm{U}}_M\oplus \textbf{\textrm{U}}_G = (u^\nu,\varepsilon,n,\Pi,\pi^{0\nu},\pi^{1\nu},\pi^{2\nu},\pi^{3\nu},g_{\frak{A}})^T\in\mathbb{R}^{33}$ where $g_{\frak{A}}$ is the representative element for the $10$-dimensional vector of independent metric components:
\be
\frak{A} = 1,2,\dotso,10;\quad g_{\frak{A}}\in\{g_{\mu\nu}:\mu \leq \nu\}.
\ee
We write the principal symbol $\hat{A}$ as the matrix differential operator in block form:
\be
\hat{A}(\partial,\partial^2) = 
\begin{bmatrix}
\mathbb{A}_M^\alpha\partial_\alpha& \boldsymbol{0}_{23\times 10}\\
\boldsymbol{0}_{10\times 23}& \mathbb{A}_G^{\alpha\beta}\partial_\alpha\partial_\beta
\end{bmatrix},
\ee
We remark that since we keep all first-order derivative terms in the matter sector, $\mathbb{A}_M^\alpha\phi_\alpha$ remains the same as in Eq.~\eqref{eq:fluidprincipalpart}. Let $\boldsymbol{0}_{M\times N}$ be the null matrix in $\mathbb{R}^{M\times N}$ and $I_{N}$ the $N\times N$ identity. We then introduce the diagonal gravity coefficient matrix:
\be
\mathbb{A}_G^{\alpha\beta} = g^{\alpha\beta}I_{10}
\ee
We divided Einstein's equations by an unimportant overall factor of $-1/2$ for aesthetic purposes. Once again introducing the characteristic vector $\phi^\mu$ provides us with the principal part:
\be
\mathbb{A}(\phi) = 
\begin{bmatrix}
\mathbb{A}_M^\alpha\phi_\alpha& \boldsymbol{0}_{23\times 10}\\
\boldsymbol{0}_{10\times 23}& \mathbb{A}_G^{\alpha\beta}\phi_\alpha\phi_\beta
\end{bmatrix}
\ee
The characteristic determinant is now the product of block matrices:
\be
\det(\mathbb{A}(\phi)) = \det(\mathbb{A}_M^\alpha\phi_\alpha)\det(\mathbb{A}_G^{\alpha\beta}\phi_\alpha\phi_\beta)
\ee
We remark that since the matter portion of the principal part remains identical, so does its determinant. Thus, $ \det(\mathbb{A}_M^\alpha\phi_\alpha)$ is equivalent to Eq.~\eqref{eq:nodiffusiondet}. Since the gravity portion is diagonal, we can immediately read off the determinant. 
\be
\det(\mathbb{A}_G^{\alpha\beta}\phi_\alpha\phi_\beta) = (\phi_\alpha\phi^\alpha)^{10} = v^{20}(\hat{z}^2 - 1)^{10}
\ee
Nonlinear causality of the \emph{entire} mixed system of fluid plus Einstein's equations is immediate. We state this result as follows.
\begin{corollary}[Nonlinear Causality with a Dynamic Metric]\label{cor:causalityeinstein}
Solutins $\textbf{U}$ of the mixed system in Eq.~\eqref{eq:mixedorderEOM} are causal if and only if the matter portion of the principal part $\mathbb{A}_M^\alpha$ satisfies the assumptions and conditions prescribed by Theorem~\ref{thm:causalitynodiffusion}.
\begin{proof}
We already proved the necessary and sufficient conditions for nonlinear causality of the matter sector in Theorem~\ref{thm:causalitynodiffusion}. Thus, it remains to show that $\det(\mathbb{A}_G^{\alpha\beta}\phi_\alpha\phi_\beta) = 0$ satisfies (CI) and (CII) (or to show where it does). Note that $\hat{z}^2 = 1$ is a root with multiplicity $10$, which immediately satisfies the reality condition (CI) as this root is equivalent to $\phi_0^2 = \sum\limits_{i = 1}^3\phi_i^2$ in the local rest frame. Furthermore, (CII) holds since $\phi^\alpha\phi_\alpha = 0$ is null by definition. 
\end{proof}
\end{corollary}
We remark that this corollary implies that Theorem~\ref{thm:causalitynodiffusion} immediately provides nonlinear causality to the more general set of equations (fluid plus Einstein's) \emph{without any extra assumptions}. 

The standard definitions of strong hyperbolicity (HI) and (HII) are not directly applicable to this mixed-order system, as they are formulated for first-order quasilinear systems. However, the principal part of the coupled matter-gravity system, given in Eq. (68), is block-diagonal. This structure is key to establishing local well-posedness.

Specifically, the matter sector has been proven to be strongly hyperbolic in Theorem \ref{thm:nodiffusionstronghyperbolicity}. The gravity sector, governed by the principal symbol $g^{\mu\nu}\phi_\mu\phi_\nu I_{10}$, decouples into ten independent, identical wave equations $g^{\mu\nu}\partial_\mu\partial_\nu g_{\mathfrak{A}} = \text{(lower-order terms)}$ for each independent metric component $g_{\mathfrak{A}}$. The scalar wave operator $g^{\mu\nu}\partial_\mu\partial_\nu$ is \emph{strictly hyperbolic}. Since the right-hand side of the mater sector involves at most $\partial g_{\mu\nu}$ (coming from Christoffel symbols) and the right-hand side of the Einstein sector involves no derivatives of the matter fields (coming from the energy-momentum tensor), it is immediate to close energy estimates for the coupled system, and the same argument used in the proof of Theorem \ref{T:LWP} applies. We summarize the results below.
\begin{corollary}[Nonlinear Strong Hyperbolicity with a Dynamic Metric]\label{cor:SHeinstein}
If the assumptions and conditions of Theorem~\ref{thm:nodiffusionstronghyperbolicity} are satisfied by the matter portion of the principal part, then the mixed system in Eq.~\eqref{eq:mixedorderEOM} is strongly hyperbolic.
\end{corollary}

\begin{corollary}[Local well-posedness of the Einstein-Israel-Stewart system]
    \label{C:LWPcoupled}
    Theorem \ref{T:LWP} generalizes to the Einstein-Israel system, where it is assumed the initial metric is in $H^{s+1}$ and the initial second fundamental form is in $H^s$ (and, naturally, the Einstein constraint equations are satisfied\footnote{The existence of data satisfying the constraints can be obtained by the methods of \cite{Isenberg-Maxwell-2024}.}).
\end{corollary}

\begin{proof}
In \cite{Fischer_Marsden_Einstein_FOSH_1972}, it is shown that Einstein's equations may be cast as a first-order \emph{symmetric} hyperbolic system of equations \cite{ChoquetBruhatGRBook}. This statement paired with Theorem~\ref{thm:nodiffusionstronghyperbolicity} shows that (HI) immediately holds for the composite system. 

To prove that (HII) remains true, note that the total principal part is manifestly block-diagonal, with one block \emph{identical} to the original fluid principal part in Eq.~\eqref{eq:fluidprincipalpart}, and the other corresponding to Einstein's equations. Thus, their eigenbases are entirely independent. Since Einstein's equations were shown to be symmetric hyperbolic, and Theorem~\ref{thm:nodiffusionstronghyperbolicity} provides strong hyperbolicity of the fluid block, the two linearly-independent eigenbases form a new composite eigenbasis for the entire system, satisfying (HII). See Appendix~\ref{App:hyperbolicity} for a more explicit proof of this argument. 
\end{proof}

\section{Conclusions and Outlook}
\label{sec:conclusions}

Transient relativistic hydrodynamic theories were put forward by Israel 50 years ago \cite{MIS-2} as an attempt to provide a consistent formulation of viscous fluid dynamics in general relativity. 
This paper resolves the corresponding decades-old problem regarding the nonlinear causality and local well-posedness of strong solutions for a general class of Israel-Stewart-type theories of relativistic hydrodynamics. We have presented the first set of simultaneously necessary and sufficient conditions that govern nonlinear causality for this large class of theories, which include the equations derived from kinetic theory \cite{MIS-6,Denicol:2012cn,Rocha:2023ilf} and other approaches \cite{Baier:2007ix}, that are abundantly used in simulations of the quark-gluon plasma. Our results provide a definitive tool for numerical simulations of such relativistic hydrodynamic theories in applications in both high-energy nuclear physics and astrophysics, which ensures they operate within their causal regime, eliminating the guesswork associated with previously incomplete conditions.

We used the standard (fully nonlinear) definition of causality from general relativity, which requires that solutions are only influenced by perturbations within their past light cone. Consequently, our conditions simultaneously restrict all dynamical information in the theory---including the degrees of freedom, transport coefficients, and viscous fluxes. We further argued that this comprehensive restriction is generally unattainable by energy conditions or linear causality analyses. For instance, the dominant energy condition depends only on the energy-momentum tensor's matter content and is blind to the evolution equations (e.g., transport coefficients). Conversely, linear conditions depend solely on equilibrium properties and contain no information about the magnitude of viscous fluxes.

The generality of our results ensures their wide applicability. In fact, the generality of our results is rooted in a framework that allows all transport coefficients to depend on the viscous fluxes $\Pi$ and $\pi^{\mu\nu}$, which defines a large class of theories of transient relativistic fluid dynamics. It is within this general context that our inequalities are derived, meaning that previous formulations \cite{MIS-6,Denicol:2012cn} are particular subclasses of the theories we cover. Our results are also directly relevant to astrophysics, as they incorporate a dynamic baryon current and a dynamic spacetime metric---both essential for modeling phenomena like neutron star mergers and black hole accretion (after including the effects from electromagnetic fields). Furthermore, our conditions are agnostic to the equation of state and make no symmetry assumptions on the underlying geometry. We also demonstrate that they reproduce and generalize previous results for bulk viscosity \cite{Bemfica:2019cop} and the zero-baryon/non-dynamical metric case \cite{Bemfica:2020xym}. 

Furthermore, we prove for the first time that the causal region of these theories is also strongly hyperbolic, guaranteeing locally well-posed solutions of Sobolev type, under the additional algebraic hypotheses of Theorem~\ref{thm:nodiffusionstronghyperbolicity}. This property allows us to further prove that the solutions propagate relevant physical symmetries---for instance, the normalization of the four-velocity is preserved for the lifetime of the solution, provided it is imposed on suitable initial data. This is a critical result that is crucial for numerical simulations, as it ensures that these IS-type theories not only have locally well-posed solutions within a definite causal regime but also maintain their fundamental physical properties throughout their evolution, a behavior that is not automatically guaranteed for nonlinear systems of PDEs.

A key practical outcome of our work is that these statements simplify drastically for a generalized extension of Maxwell-Cattaneo theory. This case provides a significantly simpler set of conditions that retain most of the transport coefficients, which are still arbitrary functions of the fluid variables and dissipative currents. We expect this new generalized Maxwell-Cattaneo theory proposed here to be particularly useful for implementing robust causality conditions in current and future simulations, largely due to the removal of their angular dependence. Therefore, the generalized Maxwell-Cattaneo theory put forward in this work paves the way for the next generation of state-of-the-art causality-respecting numerical simulations of relativistic viscous fluids with shear and bulk viscosity in the baryon-rich regime not only in heavy-ion collisions, but also in astrophysics or cosmology where coupling with Einstein's equations becomes important. 

This work opens up several promising avenues for future research in both the physics and the mathematics of relativistic viscous dynamics. A primary open challenge is a complete understanding of diffusion in the sense investigated here. Currently, necessary and sufficient causality conditions for phenomenological Israel-Stewart theory are only known for restricted cases: for particle number diffusion in the Landau frame and for heat diffusion in the Eckart frame \cite{cordeiro2025diffusion} (both cases without bulk or shear viscosity). The fundamental difficulty in adding a vector current (e.g., for number or heat diffusion) is that it introduces odd-powered terms of the form $\hat{z}^{2j+1}$ into the characteristic determinant. Since the polynomial $\mathcal{P}_3$ in Eq.~\eqref{eq:nodiffusiondet} is of order 3 in $\hat{z}^2$ (i.e., equivalent to $\hat{z}^6$), the inclusion of diffusion would couple the dynamics, raising the polynomial's overall order in $\hat{z}$ beyond 5 and introducing non-trivial odd powers. This structure makes an analytical solution via root-finding impossible in general. Therefore, determining the full causal region for the complete theory \cite{Betz:2008me}---including all viscous currents---remains a significant open problem. 

An equally crucial next step is to determine the causal region upon coupling to magnetohydrodynamics (MHD), and, more ambitiously, to resistive MHD. This would be particularly impactful for modeling black hole accretion disks \cite{Chandra:2015iza,Most:2021rhr,Cordeiro:2023ljz}, with potential further applications in both heavy-ion collisions and neutron star physics. For phenomenological Israel-Stewart theory, this has been achieved thus far only in the strong magnetic field limit with shear viscosity \cite{Cordeiro:2023ljz}; a general treatment remains an open challenge. Finally, while we focused on the Landau hydrodynamic frame in this work, our results can be generalized to consider a wider class of different hydrodynamic frames (see \cite{Noronha:2021syv}). We leave this extension to future work.

The recent progress in our understanding of the physics and the mathematics of relativistic fluids in the nonlinear regime through the last decade has arguably impacted also other areas of theoretical physics, such as for instance the investigation of local well-posedness of generalizations of Einstein's theory of general relativity, see for instance \cite{Lehner:2017yes,Kovacs:2020pns,Kovacs:2020ywu,Figueras:2024bba,Figueras:2025wtx,Gavassino:2026xjw}. It would be interesting to see how the solution of the nonlinear causality and well-posedness questions in Israel-Stewart-like theories reported in this work could be useful to motivate new developments in investigations of local well-posedness and the initial value problem in extensions of general relativity.     

In summary, this paper has rigorously established the most general causality and well-posedness bounds for a large class of Israel-Stewart-type theories of viscous relativistic hydrodynamics, encompassing both traditional phenomenological IS theory and modern formulations derived from kinetic theory and other approaches. Our work thereby affirms the nonlinear validity of a broad class of second-order theories, going significantly beyond simple linear analyses \cite{Hiscock_Lindblom_stability_1983,Olson:1989ey}. We reaffirm that these conditions are instrumental not only in selecting physically realistic initial conditions but also in verifying that the evolution remains within the physical, causal regime throughout simulations. Therefore, we expect that the work presented here provides a new crucial tool for determining the regime of validity of relativistic fluid dynamics and ensuring the reliability of numerical simulations of relativistic viscous fluids across a diverse range of fields, including astrophysics, cosmology, and high-energy nuclear physics.

\acknowledgements

We thank L.~Gavassino for the helpful suggestions that simplified the nonlinear causality conditions in Section~\ref{sec:MC}.  JN and IC are partly supported by the U.S. Department of Energy, Office of Science, Office for Nuclear Physics under Award No. DE-SC0023861. E.S. was supported by the Rita Levi Montalcini Program of the Italian Ministry of University and Research. MMD is partially supported by NSF grants DMS-2406870 and NRT-2125764, and from DOE grant DE-SC0024711.
This research was partly supported by the National Science Foundation under Grant No. NSF PHY-1748958 and NSF PHYS-2316630. Any opinions, findings, and conclusions or recommendations expressed in this material are those of the author(s) and do not necessarily reflect the views of the National Science Foundation.

\subsection*{DATA AVAILABILITY}

The data that support the findings of this article are publicly available at \cite{CordeiroGit2026}.

\def\cprime{$'$}

\newpage

\appendix

\section{Causality, proof of Theorem \ref{thm:causalitynodiffusion}}\label{App:causality}

We recall the definition of causality from Definition \ref{def:causality}. This can be recast into the following algebraic conditions \cite[Appendix B]{Disconzi-2024}: $(\mathbb{A}^\alpha\partial_\alpha + \mathbb{B})\Umb = 0$ is \textit{causal} if (CI) the roots of the characteristic equation $\det(\mathbb{A}^\alpha\phi_\alpha) = 0$, given by $\phi_0 \equiv \phi_0(\phi_i)$ are real, where $\phi^\mu \equiv \nabla^\mu\varphi$, and $\{\varphi(x) = 0\}$ are the characteristic hypersurfaces and (CII) $\phi^\mu$ is non-timelike, i.e. $\phi^\al\phi_\al \geq 0$ \cite{ChoquetBruhatGRBook}. 

\subsection{Bounds from the full determinant}
\label{S:IS_full_determinant_bounds}

Recall that the determinant of the principal part of Israel-Stewart from kinetic theory with bulk, shear viscosity and a baryon current takes the form
\bml
\bea
\label{supeq:nodiffusiondet}
\frac{\det(\mathbb{A}^\alpha\phi_\alpha)}{v^{23}} &=& \left[\prod_{k = 0}^3(E + \Lambda_k)\right]\tau_\Pi\tau_\pi^{16}\hat{z}^{15}\left[\hat{z}^2 -\frac{\frac{1}{2} (2\eta + \lambda_{\pi\Pi}\Pi)+\frac{1}{4} \tau_{\pi\pi}\Lambda_{\hat{v}^2}}{\tau_\pi E}\right]\mathcal{P}_3(\hat{z}^2,\chi^2,\kappa^2),\\
\mathcal{P}_3(\hat{z}^2,\chi^2,\kappa^2) &=& \hat{z}^6 - \mathcal{A}_2(\chi^2,\kappa^2)\hat{z}^4 + \mathcal{A}_1(\chi^2,\kappa^2)\hat{z}^2 -\mathcal{A}_0(\chi^2,\kappa^2),\\
\Lambda_{\hat{v}^2} &\equiv& \sum_{k = 1}^3 \Lambda_k \hat{v}_k^2(\chi^2,\kappa^2),
\eea
\eml
Suppose that $\tau_\Pi,\tau_\pi,E+\Lambda_k\neq 0$ for all $k = 0,1,2,3$. Let $\phi^\mu$ be such that $\det(\mathbb{A}^\al\phi_\al) = 0$ (i.e. $\phi^\mu = \nabla^\mu\varphi$ for characteristic hypersurface $\varphi(x)$). Denote $\Delta(\mathcal{P}_3;\chi^2,\kappa^2)$ as the discriminant of $\mathcal{P}_3$ defined by
\begin{align}
\label{supeq:discriminant}
\Delta(\mathcal{P}_3;\chi^2,\kappa^2) \equiv& - 27\left[2\left(\frac{\mathcal{A}_2(\chi^2,\kappa^2)}{3}\right)^3 - \frac{\mathcal{A}_2(\chi^2,\kappa^2)}{3}\mathcal{A}_1(\chi^2,\kappa^2) +\mathcal{A}_0(\chi^2,\kappa^2)\right]^2\notag\\
&-4\left[\mathcal{A}_1(\chi^2,\kappa^2) - 3\left(\frac{\mathcal{A}_2(\chi^2,\kappa^2)}{3}\right)^2\right]^3.
\end{align}
We wish to show that, for all $\chi^2,\kappa^2\in[0,1]$, the following bounds are necessary and sufficient for causality:
\bml
\bea
\label{supeq:C1}
0&\leq &\Delta(\mathcal{P}_3;\chi^2,\kappa^2),\\
\label{supeq:C2}
0 &\leq & \frac{\mathcal{A}_2(\chi^2,\kappa^2)}{3} \leq 1,\quad 0\leq \mathcal{A}_1(\chi^2,\kappa^2),\quad 0\leq \mathcal{A}_0(\chi^2,\kappa^2)\\
\label{supeq:C3}
0 &\leq & 1-\mathcal{A}_2(\chi^2,\kappa^2) + \mathcal{A}_1(\chi^2,\kappa^2) - \mathcal{A}_0(\chi^2,\kappa^2) ,\\
\label{supeq:C4}
0 &\leq & 1-\frac{2}{3}\mathcal{A}_2(\chi^2,\kappa^2) + \frac{1}{3}\mathcal{A}_1(\chi^2,\kappa^2).
\eea
\eml
We also introduce the condition
\begin{align}
\label{supeq:C5}
0 \leq  y_0(\chi^2,\kappa^2)\equiv \frac{\frac{1}{2} (2\eta + \lambda_{\pi\Pi}\Pi)+\frac{1}{4} \tau_{\pi\pi}\Lambda_{\hat{v}^2}}{\tau_\pi E} \leq 1.
\end{align}
As we will see, this is a condition for causality of the full set of equations before taking into consideration the constraints \eqref{Constraints}, but once these are taken into account, condition \eqref{supeq:C5} is not part of the necessary and sufficient conditions for causality.

Define $y_k(\chi^2,\kappa^2)$, $k=1,2,3$, as the roots of $\mathcal{P}_3$. The number $y_0$ is defined by \eqref{supeq:C5}. 

\medskip
\noindent\textbf{Step 1: First set of conditions.} We first wish to show that, for all $\chi^2,\kappa^2\in[0,1]$ that the conditions prescribed by
\be
\Omega_1 \equiv \{\Delta(\mathcal{P}_3;\chi^2,\kappa^2)\geq 0\}\quad\cup\quad\{0\leq y_k(\chi^2,\kappa^2)\leq 1\}_{k = 0}^3,\\
\ee
are both necessary and sufficient for causality of the system. We proceed as follows: the conditions in $\Omega_1$ are built from Lorentz scalars and are thus invariant on a proper choice of frame. The calculation is carried out for arbitrary covectors in one orthonormal frame. Using Eq.~\eqref{supeq:nodiffusiondet}, one can identify the roots $\phi_0 = \phi_0(\phi_i)$ from $\det(\mathbb{A}^\al\phi_\al) = 0$ by rotating to the local rest frame ($g^{\mu\nu}\rightarrow \eta^{\mu\nu}$ and $u^\mu = (1,0,0,0)$ w.l.o.g.) such that $z = u^\al \phi_\al\rightarrow \phi_0$ and $v^2 = \phi_\al\phi_\beta\Delta^{\al\beta} \rightarrow \vec{\phi}^2$ where $\vec{\phi}^2 \equiv \sum\limits_{i = 1}^3\phi_i^2$. In this frame, the roots are then $\phi_0^2=y_k\vec{\phi}^{\,2}$, or $\phi_0^2 = 0$. In the latter case, the roots trivially satisfy conditions (CI) and (CII), so we will focus on the former roots.

The condition on the discriminant $\Delta \geq 0$ guarantees that the roots $\phi_0^2 = y_a\vec{\phi}^2$ for $a = 1,2,3$ are real due to the properties of the cubic discriminant. Furthermore, the condition that $y_k\geq 0$ for $k = 0,1,2,3$ is needed to guarantee that $\phi_0 = \pm\sqrt{y_k}|\vec{\phi}|\in\mathbb{R}$. Thus, we find that conditions in $\Omega_1$ satisfy condition (CI).

Finally, (CII) requires $\phi_\al\phi^\al \geq 0$ ($\phi_0^2\leq \vec{\phi}^2$ in the local rest frame), which then enforces the conditions that
\begin{equation}
y_k\leq 1,\qquad k = 0,1,2,3.
\end{equation}
Therefore, $\Omega_1$ consists of necessary and sufficient conditions for causality. When the spatial covector is zero, \eqref{E:IS_time_determinant} gives $\det(\mathbb{A}^{0}z)=c_0z^{23}$ with $c_0\neq0$ defined by \eqref{E:IS_full_coefficient}, so there is no additional nonzero characteristic covector.

\medskip
\noindent\textbf{Step 2: Second set of conditions.} The next step we wish to show is that Eq.~\eqref{supeq:C5} combined with the bounds in Eq.~\eqref{supeq:C1}--\eqref{supeq:C4} are necessary and sufficient conditions for $\Omega_1$. The condition given in Eq.~\eqref{supeq:C5} is identical to the condition $0\leq y_0\leq 1$, so we will henceforth only consider the roots $y_1,y_2$ and $y_3$ corresponding to the cubic $\mathcal{P}_3(y)$. Consider the set of conditions
\bml
\bea
\Omega_2 &\equiv & \{\Delta \geq 0,\quad \mathcal{A}_0,\mathcal{A}_1,\mathcal{A}_2\geq 0\},\\
\Omega_3 &\equiv & \left\{\Delta \geq 0,\quad 1-\mathcal{A}_2 + \mathcal{A}_1 -\mathcal{A}_0\geq 0,\quad 3 - 2\mathcal{A}_2+\mathcal{A}_1\geq 0,\quad \frac{\mathcal{A}_2}{3}\leq 1\right\}
\eea
\eml
All the conditions in $\Omega_2$ and $\Omega_3$ must hold simultaneously. We wish to show that $\Omega_2\cup\Omega_3$ for all $\chi^2,\kappa^2\in[0,1]$ provides a set of necessary and sufficient conditions for $\Omega_1$ (along with $0\leq y_0\leq 1$). We start by showing that $\Omega_2$ is necessary and sufficient for the nonnegativity of the roots, and then following by showing that $\Omega_3$ is necessary and sufficient for bounding the roots by unity. Consider writing $\mathcal{P}_3$ in two ways:
\be
(y - y_1)(y - y_2)(y - y_3) = \mathcal{P}_3(y) = y^3  -\mathcal{A}_2y^2 + \mathcal{A}_1y -\mathcal{A}_0.
\ee
Expanding the left-hand side and rewriting provides
\bml
\bea
\label{supeq:A0root}
\mathcal{A}_0 &=& y_1y_2y_3,\\
\label{supeq:A1root}
\mathcal{A}_1 &=& y_1y_2 + y_2y_3 + y_3y_1,\\
\label{supeq:A2root}
\mathcal{A}_2 &=& y_1 + y_2 + y_3.
\eea
\eml
Assume that $y_k\geq 0$ for all $k = 1,2,3$. The roots are real, so their cubic discriminant is nonnegative. Equations \eqref{supeq:A0root}, \eqref{supeq:A1root}, and \eqref{supeq:A2root} give $\mathcal{A}_0,\mathcal{A}_1,\mathcal{A}_2\geq0$. This shows that $\Omega_2$ is necessary for the nonnegativity of the roots.\\

\noindent To show that these conditions are sufficient, assume instead that $\Omega_2$ holds. The condition $\Delta  \geq 0$ again guarantees that these roots are real. As such, assume now by contradiction that there exists some $y_i<0$. Without loss of generality, we say that $i = 3$ by symmetry. The other cases hold identically since the coefficients are symmetric under exchange of indices. Since $\mathcal{A}_0 \geq 0$ and $y_3 < 0$, we find that $y_1y_2 \leq 0$. Furthermore, since $\mathcal{A}_1\geq 0$, we have
\begin{align}
(y_1 + y_2)y_3 \geq -y_1y_2 \geq 0.
\end{align}
However, since $y_3 < 0$ we have $y_1 + y_2 \leq 0$. Furthermore, note that $ \mathcal{A}_2 \geq 0$, so
\begin{align}
0 \leq  \mathcal{A}_2 = y_1 + y_2 + y_3 \leq y_3.
\end{align}
Thus, $y_3 \geq 0$. This statement contradicts our assumption, which proves that instead, it must be that $y_1,y_2,y_3 \geq 0$. Therefore, we have shown that $\Omega_2$ consists of necessary and sufficient conditions for the non-negativity of the roots.\\

\noindent Next, we must show that $\Omega_3$ consists of necessary and sufficient conditions for the roots to be bounded above by unity. Define $b = 1 - y$. Therefore, the condition that $b\geq 0$ is equivalent to $y \leq 1$. Rearranging, we know that $y = 1 - b$, meaning that our cubic can be expressed as
\begin{align}
\mathcal{P}_3(1-b) &= (1 - b)^3  -\mathcal{A}_2(1 - b)^2 + \mathcal{A}_1(1 - b) -\mathcal{A}_0,\notag\\
&= -b^3 + (-\mathcal{A}_2 + 3)b^2 - (3-2\mathcal{A}_2 + \mathcal{A}_1)b + (-\mathcal{A}_2 + \mathcal{A}_1 -\mathcal{A}_0 + 1).
\end{align}
Since $\Delta\geq 0$ guarantees the existence of the roots $y_1,y_2,y_3\in\mathbb{R}$, so too are the roots $b_1,b_2,b_3\in\mathbb{R}$ (the real numbers are a closed field, and $y_i = 1 - b_i$ for all $i = 1,2,3$) note that we also have
\be
\mathcal{P}_3(1-b) = -(b - b_1)(b - b_2)(b - b_3).
\ee
Upon equating and expanding as before, we get the Vieta's Formulae for a cubic polynomial
\bml
\bea
\label{E:IS_B_sum}
3-\mathcal{A}_2 &=& b_1 + b_2 + b_3,\\
\label{E:IS_B_pairs}
3 - 2\mathcal{A}_2 + \mathcal{A}_1  &=& b_1b_2 + b_2b_3 + b_3b_1,\\
\label{E:IS_B_product}
1-\mathcal{A}_2 + \mathcal{A}_1 -\mathcal{A}_0 &=& b_1b_2b_3.
\eea
\eml
This is the exact same list of conditions given in Eqs.~\eqref{supeq:A2root}--\eqref{supeq:A0root}, except with the substitutions
\begin{align*}
\{\mathcal{A}_2,\mathcal{A}_1,\mathcal{A}_0\}\longrightarrow\{3-\mathcal{A}_2,3-2\mathcal{A}_2 + \mathcal{A}_1, 1-\mathcal{A}_2 + \mathcal{A}_1 -\mathcal{A}_0\}
\end{align*}
and
\begin{align*}
\{y_1,y_2,y_3\}\longrightarrow\{b_1,b_2,b_3\}.
\end{align*}
Therefore, an analogous proof applies to showing that the roots $b_1,b_2,b_3$ satisfy (1) existence in $\mathbb{R}$ and (2) are nonnegative applies, identically to that of the proof that $y_1,y_2,y_3\geq 0$, except instead with the conditions provided in $\Omega_3$ (which are algebraic rearrangements of the new effective coefficients in terms of $b$ instead of $y$). Thus, $\Omega_3$ provides necessary and sufficient conditions for the condition $y_k\leq 1$. Explicitly, nonnegative $b_k$ make all sums and products in \eqref{E:IS_B_sum}--\eqref{E:IS_B_product} nonnegative. Conversely, under $\Omega_3$, if $b_3<0$, then \eqref{E:IS_B_product} gives $b_1b_2\leq0$ and \eqref{E:IS_B_pairs} gives
\begin{align*}
 (b_1+b_2)b_3&\geq-b_1b_2\geq0,\\
 b_1+b_2&\leq0,\\
 0\leq3-\mathcal{A}_2&=b_1+b_2+b_3\leq b_3<0.
\end{align*}
The contradiction, with indices permuted, excludes every negative $b_k$. Zero and repeated roots are included.\\

\noindent In summary, we find that $\forall\chi^2,\kappa^2\in[0,1]$, $\Omega_2\cup\Omega_3$, together with $0\leq y_0\leq1$, are necessary and sufficient conditions for causality of the system, as they are necessary and sufficient for $\Omega_1$. This result proves the claim.

\subsection{Effect of the constraints}
\label{S:IS_constraint_effect}

Section~\ref{S:IS_full_determinant_bounds} proves
\eqref{supeq:C1}--\eqref{supeq:C5} for the full determinant.
We now account for the IS constraints and identify the
physical characteristics within that determinant, directly from
\eqref{eq:propsystem}. Throughout the characteristic calculation, the metric
and a value of $\Umb$ satisfying \eqref{Constraints} are fixed. Set
\begin{align*}
 A(\phi)&=\mathbb{A}^{\alpha}\phi_{\alpha}.
\end{align*}

The characteristic test asks whether the IS equations and constraints
determine the normal derivatives from the field values and tangential
derivatives. Fix a point $p$ on a hypersurface $\varphi=0$ with
$\phi=\mathrm{d}\varphi\neq0$, and denote the fixed field value by
$\Umb_0$. Let $N$ be a vector with $\phi_{\alpha}N^{\alpha}=1$.
Here normal differentiation means differentiation in this transverse
direction; $N$ need not be orthogonal to the hypersurface.
Indeed, since $\phi\neq0$, there exists a vector $V$ with
$\phi_{\alpha}V^{\alpha}\neq0$. The choice
\begin{align*}
 N^{\alpha}&=\frac{V^{\alpha}}{\phi_{\beta}V^{\beta}}
\end{align*}
satisfies $\phi_{\alpha}N^{\alpha}=1$. This construction uses no
metric and applies whether $\phi$ is spacelike, null, or timelike.
Set
\begin{align*}
 T_{\alpha}&=\partial_{\alpha}-\phi_{\alpha}N,\\
 Q&=\nabla_N\Umb.
\end{align*}
The derivative $\nabla_N$ acts on each scalar and tensor field in $\Umb$,
in the component order of \eqref{eq:FirstOrderQPDE}. Since
$\phi(T_{\alpha})=0$, the vectors $T_{\alpha}$ are tangent, and
\begin{align*}
 \nabla_{\alpha}\Umb
 &=\phi_{\alpha}Q+\nabla_{T_{\alpha}}\Umb.
\end{align*}
Consequently, \eqref{eq:propsystem} is a linear system for $Q$:
\begin{align*}
 A(\phi)Q&=\mathcal{F},\\
 \mathcal{F}&=-\mathbb{A}^{\alpha}\nabla_{T_{\alpha}}\Umb
       -(0^{\mu},0,0,\Pi,\pi^{\mu\nu})^{T}.
\end{align*}
The source vector is written in the equation order of
\eqref{eq:fluidprincipalpart}. All coefficients are evaluated at $p$ and
$\Umb_0$, so $\mathcal{F}$ is fixed by the field values and tangential derivatives.

Write
\begin{align*}
 Q_u^{\mu}&=\nabla_Nu^{\mu},\\
 Q_{\pi}^{\mu\nu}&=\nabla_N\pi^{\mu\nu}.
\end{align*}
Applying the spacetime derivative $\nabla_N$ to \eqref{Constraints},
using the product rule and $\nabla_Ng=0$, gives
\begin{align*}
 2u_{\alpha}Q_u^{\alpha}
 &=\nabla_N(u^{\alpha}u_{\alpha}+1)=0,\\
 Q_{\pi}^{\mu\nu}-Q_{\pi}^{\nu\mu}
 &=\nabla_N(\pi^{\mu\nu}-\pi^{\nu\mu})=0,\\
 u_{\alpha}Q_{\pi}^{\alpha\mu}
       +\pi^{\alpha\mu}Q_{u\alpha}
 &=\nabla_N(u_{\alpha}\pi^{\alpha\mu})=0,\\
 g_{\alpha\beta}Q_{\pi}^{\alpha\beta}
 &=\nabla_N(g_{\alpha\beta}\pi^{\alpha\beta})=0.
\end{align*}
Here $Q_{u\alpha}=g_{\alpha\beta}Q_u^{\beta}$. Thus the restrictions on
$Q$ from \eqref{Constraints} are linear and homogeneous. For prescribed
field values and tangential derivatives for which $A(\phi)Q=\mathcal{F}$
and the restrictions obtained from $\nabla_N$ applied to \eqref{Constraints}
are jointly solvable,
the normal derivatives are uniquely determined exactly when the homogeneous
system has only the zero solution. Denote its unknown by
$\delta\Umb\in\mathbb{R}^{23}$. Its IS equation is
$A(\phi)\delta\Umb=0$, and its constraint equations are
\eqref{E:IS_linear_constraints}. The prefix $\delta$ labels this
homogeneous unknown; spacetime differentiation is denoted by $\nabla_N$.
A zero of the
full determinant for which every nonzero solution violates
\eqref{E:IS_linear_constraints} is therefore not a physical characteristic
covector. See also
{\hypersetup{citecolor=blue}\cite[Sec.~1, Definition~2, and Sec.~2,
Eqs.~(4)--(6)]{Pennisi-1998}} for the treatment of algebraic constraints in
the characteristic equations.

In terms of the homogeneous unknown $\delta\Umb$, the restrictions obtained by applying $\nabla_N$ to \eqref{Constraints} are
\begin{subequations}
\label{E:IS_linear_constraints}
\begin{align}
\label{E:IS_linear_normalization}
 2u_{\alpha}\delta u^{\alpha}&=0,\\
 \delta\pi^{\mu\nu}-\delta\pi^{\nu\mu}&=0,\\
\label{E:IS_linear_orthogonality}
 u_{\alpha}\delta\pi^{\alpha\mu}
       +\pi^{\alpha\mu}\delta u_{\alpha}&=0,\\
\label{E:IS_linear_trace}
 g_{\alpha\beta}\delta\pi^{\alpha\beta}&=0.
\end{align}
\end{subequations}
A nonzero real covector is a physical characteristic covector precisely when
\begin{align*}
 A(\phi)\delta\Umb&=0
\end{align*}
has a nonzero real solution satisfying \eqref{E:IS_linear_constraints}.
We recall that condition (CII) imposes $g^{\mu\nu}\phi_\mu\phi_\nu\geq0$ on
nonzero real physical characteristic covectors.

To analyze the full determinant, we now treat all $23$ entries of
$\delta\Umb$ as independent real unknowns and identify which solutions
satisfy \eqref{E:IS_linear_constraints}. The fixed field value
$\Umb_0$ itself continues to satisfy \eqref{Constraints}.

\medskip
\noindent\textbf{Step 1: The six symmetry constraints.}
Write the vector $A(\phi)\delta\Umb$ in the order of the equations as
\begin{align*}
 (F_u^{\mu},F_{\varepsilon},F_n,F_{\Pi},F_{\pi}^{\mu\nu})
 &=A(\phi)\delta\Umb.
\end{align*}
In particular, \eqref{eq:momentumprop}, \eqref{eq:shearprop}, and
\eqref{eq:C_def} give
\begin{subequations}
\begin{align}
\label{E:IS_momentum_row}
 F_u^{\mu}
 &=Ez\delta u^{\mu}
 +v^{\mu}(P_{\varepsilon}\delta\varepsilon
                   +P_n\delta n+\delta\Pi)
 +\phi_{\alpha}\delta\pi^{\mu\alpha}
 -u^{\mu}w_{\nu}\delta u^{\nu},\\
\label{E:IS_shear_row}
 F_{\pi}^{\mu\nu}
 &=\tau_{\pi}z\delta\pi^{\mu\nu}
 +\tensor{\mathcal{C}}{_\pi^\mu^\nu^\alpha_\beta}
       \phi_{\alpha}\delta u^{\beta}.
\end{align}
\end{subequations}
Here $z$, $v^{\mu}$, and $w^{\mu}$ have the definitions of
Section~\ref{subsec:characteristics}, and $\tensor{\mathcal{C}}{_\pi^\mu^\nu^\alpha_\beta}$
is the tensor in
\eqref{eq:C_def}. 
The symmetry $\pi^{[\mu\nu]}=0$ of the fixed state gives, from
\eqref{eq:C_def},
\begin{align*}
\tensor{\mathcal{C}}{_\pi^[^\mu^\nu^]^\alpha_\beta}
 &=\left(\delta_{\pi\pi}
             -\frac{\tau_{\pi\pi}}{3}\right)
       \pi^{[\mu\nu]}g^{\alpha}{}_{\beta}=0.
\end{align*}
Thus $\mathcal{C}_{\pi}^{\mu\nu\alpha}{}_{\beta}$ is symmetric in
$\mu,\nu$. Antisymmetrizing \eqref{E:IS_shear_row} eliminates its term
containing $\delta u^{\beta}$ and gives
\begin{align}
\label{E:IS_antisymmetric_rows}
 F_{\pi}^{[\mu\nu]}&=\tau_{\pi}z\delta\pi^{[\mu\nu]}.
\end{align}
This is also the principal part of \eqref{eq:pisymmetry}.
Equation \eqref{E:IS_antisymmetric_rows} holds for arbitrary
$\delta\pi^{\mu\nu}$ before the condition $\delta\pi^{\mu\nu}=\delta\pi^{\nu\mu}$
is imposed. For $\delta\pi^{[\mu\nu]}=0$, both sides vanish.
Keeping the identity for arbitrary $\delta\pi^{\mu\nu}$ identifies
the contribution of the six antisymmetric unknowns to the full determinant.

Use the ten entries $\delta\pi^{(\mu\nu)}$, $\mu\leq\nu$,
and the six entries $\delta\pi^{[\mu\nu]}$, $\mu<\nu$,
instead of the sixteen entries $\delta\pi^{\mu\nu}$.
The new column vector is
\begin{align*}
 \widehat{\delta\Umb}
 &=\begin{pmatrix}
 (\delta u^{\mu})_{\mu=0}^{3}\\
 \delta\varepsilon\\
 \delta n\\
 \delta\Pi\\
 (\delta\pi^{(\mu\nu)})_{\mu\leq\nu}\\
 (\delta\pi^{[\mu\nu]})_{\mu<\nu}
 \end{pmatrix}.
\end{align*}
The hat denotes only this change of components. Both $\delta\Umb$
and $\widehat{\delta\Umb}$ still have $23$ independent entries;
the six antisymmetric entries have not been set to zero. The hat does
not signify that \eqref{E:IS_linear_constraints} has been imposed.
The symmetric components are ordered by
$(\mu,\nu)=(0,0)$, $(0,1)$, $(0,2)$, $(0,3)$, $(1,1)$, $(1,2)$, $(1,3)$, $(2,2)$, $(2,3)$, $(3,3)$,
and the antisymmetric components by
$(\mu,\nu)=(0,1)$, $(0,2)$, $(0,3)$, $(1,2)$, $(1,3)$, $(2,3)$.
For each $\mu<\nu$, the inverse change is
\begin{align*}
 \delta\pi^{\mu\nu}
 &=\delta\pi^{(\mu\nu)}+\delta\pi^{[\mu\nu]},\\
 \delta\pi^{\nu\mu}
 &=\delta\pi^{(\mu\nu)}-\delta\pi^{[\mu\nu]}.
\end{align*}

For each fixed $\mu$, with no summation,
\begin{align*}
 \delta\pi^{(\mu\mu)}
 &=\frac{\delta\pi^{\mu\mu}+\delta\pi^{\mu\mu}}{2}
   =\delta\pi^{\mu\mu},\\
 \delta\pi^{[\mu\mu]}
 &=\frac{\delta\pi^{\mu\mu}-\delta\pi^{\mu\mu}}{2}=0.
\end{align*}
Hence the four components $\delta\pi^{\mu\mu}$ are retained unchanged.
Define $T_1\in\mathbb{R}^{23\times23}$ by
\begin{align*}
 \delta\Umb&=T_1\widehat{\delta\Umb}.
\end{align*}
It retains $\delta u^{\mu},\delta\varepsilon,\delta n,
\delta\Pi$ and the four diagonal shear components. For each
$\mu<\nu$, its action is
\begin{align*}
 \begin{pmatrix}
 \delta\pi^{\mu\nu}\\
 \delta\pi^{\nu\mu}
 \end{pmatrix}
 &=\begin{pmatrix}1&1\\1&-1\end{pmatrix}
   \begin{pmatrix}
   \delta\pi^{(\mu\nu)}\\
   \delta\pi^{[\mu\nu]}
   \end{pmatrix}.
\end{align*}
Thus $T_1$ reconstructs the original vector used in
\eqref{eq:FirstOrderQPDE} from the symmetric and antisymmetric components;
it is the inverse of the change that forms those components.
Substitution of $\delta\Umb=T_1\widehat{\delta\Umb}$ into the
characteristic equations gives
\begin{align*}
 A(\phi)T_1\widehat{\delta\Umb}&=0,
\end{align*}
so this change of unknowns multiplies $A(\phi)$ on the right.

Define $R_1\in\mathbb{R}^{23\times23}$ to retain the seven entries
$(F_u^{\mu},F_{\varepsilon},F_n,F_{\Pi})$ and the four entries
$F_{\pi}^{\mu\mu}$, and to replace each pair with $\mu<\nu$ by
\begin{align*}
 \begin{pmatrix}
 F_{\pi}^{(\mu\nu)}\\
 F_{\pi}^{[\mu\nu]}
 \end{pmatrix}
 &=\frac{1}{2}\begin{pmatrix}1&1\\1&-1\end{pmatrix}
   \begin{pmatrix}
   F_{\pi}^{\mu\nu}\\
   F_{\pi}^{\nu\mu}
   \end{pmatrix}.
\end{align*}
The output of $R_1$ is ordered as the seven entries
$(F_u^{\mu},F_{\varepsilon},F_n,F_{\Pi})$, the ten entries
$F_{\pi}^{(\mu\nu)}$, and then the six entries
$F_{\pi}^{[\mu\nu]}$, using the pair orders specified for
$\widehat{\delta\Umb}$. These linear combinations of the equations
multiply the matrix on the left, giving
\begin{align*}
 R_1A(\phi)T_1\widehat{\delta\Umb}&=0.
\end{align*}
The $2\times2$ matrices defining $R_1$ and $T_1$ are inverses because
\begin{align*}
 \frac{1}{2}\begin{pmatrix}1&1\\1&-1\end{pmatrix}
 \begin{pmatrix}1&1\\1&-1\end{pmatrix}&=I_2,
\end{align*}
where $I_2$ is the two-dimensional identity matrix. The retained components
and the specified reorderings therefore make $R_1$ and $T_1$ invertible.
Their entries are constant and independent of $\phi$.
In particular, $T_1$ preserves all $23$ components and is not a projection
onto the constraints. Imposing the symmetry constraint means separately
restricting $\widehat{\delta\Umb}$ to vectors whose six antisymmetric
entries are zero. The projection onto this subspace in the new components
is $\operatorname{diag}(I_{17},0_6)$, where $I_{17}$ is the identity
matrix on $\mathbb{R}^{17}$ and $0_6$ is the $6\times6$ zero matrix.
In the original components, that projection is
\begin{align*}
 T_1\operatorname{diag}(I_{17},0_6)T_1^{-1}.
\end{align*}
It imposes symmetry only; normalization, orthogonality, and trace must
still be imposed through \eqref{E:IS_linear_constraints}.
By \eqref{E:IS_antisymmetric_rows}, each antisymmetric equation has
coefficient $\tau_{\pi}z$ on its own unknown
$\delta\pi^{[\mu\nu]}$ and zero coefficients on every other unknown.
The seven non-shear equations and ten symmetric shear equations give
seventeen rows. Consequently,
\begin{align}
\label{E:IS_first_triangular_matrix}
 R_1A(\phi)T_1
 &=\begin{bmatrix}
       A_{17}(\phi)&B_1(\phi)\\
       0&\tau_{\pi}zI_6
    \end{bmatrix}.
\end{align}
This defines $A_{17}$ and $B_1$ as the indicated submatrices; $I_6$ is
the identity matrix on $\mathbb{R}^{6}$. The columns of $A_{17}$ describe
vectors with symmetric $\delta\pi^{\mu\nu}$, and its shear rows
are symmetric. The sizes of $A_{17}$ and $B_1$ are $17\times17$ and
$17\times6$, respectively. No division by $z$ occurs in these changes,
so \eqref{E:IS_first_triangular_matrix} also holds when $z=0$.
Taking determinants yields
\begin{align}
\label{E:IS_first_determinant}
 (\det R_1)(\det T_1)\det A(\phi)
 &=\tau_{\pi}^{6}z^{6}\det A_{17}(\phi).
\end{align}
This concludes the treatment of the symmetry constraints.

\medskip
\noindent\textbf{Step 2: Normalization, orthogonality, and trace.}
For a vector $\delta\Umb$ with symmetric $\delta\pi^{\mu\nu}$
and hence seventeen independent components,
define the six numbers

\begin{align*}
 q_{\Phi}&=2u_{\alpha}\delta u^{\alpha},\\
 q_{\Psi}^{\mu}
 &=u_{\alpha}\delta\pi^{\alpha\mu}
      +\pi^{\alpha\mu}\delta u_{\alpha},\\
 q_{\mathrm{tr}}&=g_{\alpha\beta}\delta\pi^{\alpha\beta}.
\end{align*}

For the quantities $(\Phi,\Psi^{\mu},\chi)$ of
Section~\ref{subsec:propagate}, the product rule gives
\begin{align*}
 \nabla_N\Phi&=2u_{\alpha}Q_u^{\alpha},\\
 \nabla_N\Psi^{\mu}
 &=u_{\alpha}Q_{\pi}^{\alpha\mu}
       +\pi^{\alpha\mu}Q_{u\alpha},\\
 \nabla_N\chi&=g_{\alpha\beta}Q_{\pi}^{\alpha\beta}.
\end{align*}
Replacing $Q$ by the homogeneous unknown $\delta\Umb$ in these
linear expressions for $\nabla_N\Phi$, $\nabla_N\Psi^{\mu}$, and
$\nabla_N\chi$ gives $q_{\Phi}$, $q_{\Psi}^{\mu}$, and
$q_{\mathrm{tr}}$, respectively.
The constraints \eqref{E:IS_linear_normalization},
\eqref{E:IS_linear_orthogonality}, and \eqref{E:IS_linear_trace}
are exactly $q=(q_{\Phi},q_{\Psi}^{0},\ldots,q_{\Psi}^{3},
q_{\mathrm{tr}})=0$.

The contraction of \eqref{E:IS_momentum_row} with $u_{\mu}$ is
\begin{align*}
 u_{\mu}F_u^{\mu}
 &=Ez u_{\mu}\delta u^{\mu}
   +\phi_{\alpha}u_{\mu}\delta\pi^{\mu\alpha}
   +w_{\nu}\delta u^{\nu}\\
 &=\frac{E}{2}zq_{\Phi}+\phi_{\mu}q_{\Psi}^{\mu}.
\end{align*}
To compute the shear contractions, write the tensor already used in
Section~\ref{subsec:propagate} as
\begin{align}
\label{E:IS_D_tensor}
 \mathcal{D}^{\mu\alpha}
 &=\frac{2\eta+\lambda_{\pi\Pi}\Pi}{4}
          \Delta^{\mu\alpha}
       +\frac{\tau_{\pi\pi}}{8}\pi^{\mu\alpha}.
\end{align}
At a physical state, contracting the terms of \eqref{eq:C_def} gives
\begin{subequations}
\label{E:IS_C_contractions}
\begin{align}
 u_{\nu} \tensor{\mathcal{C}}{_\pi^\nu^\mu^\alpha_\beta}
 &=\left(\eta+\frac{\lambda_{\pi\Pi}\Pi}{2}\right)
             \Delta^{\mu\alpha}u_{\beta}
      +\tau_{\pi}u^{\alpha}\pi^{\mu}{}_{\beta}
      +\frac{\tau_{\pi\pi}}{4}\pi^{\mu\alpha}u_{\beta}
 \notag\\
 &=2\mathcal{D}^{\mu\alpha}u_{\beta}
      +\tau_{\pi}u^{\alpha}\pi^{\mu}{}_{\beta},\\
 g_{\mu\nu}\tensor{\mathcal{C}}{_\pi^\mu^\nu^\alpha_\beta}
 &=\left(\eta+\frac{\lambda_{\pi\Pi}\Pi}{2}\right)
       \left(2\Delta^{\alpha}{}_{\beta}
                   -\frac{2\cdot3}{3}\Delta^{\alpha}{}_{\beta}\right)
       +\tau_{\pi\pi}\pi^{\alpha}{}_{\beta}
       -\frac{\tau_{\pi\pi}}{3}\,3\pi^{\alpha}{}_{\beta}
 =0.
\end{align}
\end{subequations}
The terms proportional to $u_{\alpha}\pi^{\alpha\mu}$ and
$g_{\alpha\beta}\pi^{\alpha\beta}$ vanish by \eqref{Constraints}.
Substitution of \eqref{E:IS_C_contractions} into
\eqref{E:IS_shear_row} gives
\begin{subequations}
\label{E:IS_constraint_row_identity}
\begin{align}
 u_{\mu}F_u^{\mu}
 &=\frac{E}{2}zq_{\Phi}+\phi_{\mu}q_{\Psi}^{\mu},\\
 u_{\nu}F_{\pi}^{\nu\mu}
 &=\mathcal{D}^{\mu\alpha}\phi_{\alpha}q_{\Phi}
       +\tau_{\pi}zq_{\Psi}^{\mu},\\
 g_{\mu\nu}F_{\pi}^{\mu\nu}
 &=\tau_{\pi}zq_{\mathrm{tr}}.
\end{align}
\end{subequations}

Equations \eqref{E:IS_constraint_row_identity} hold for every vector
$\delta\Umb$ with symmetric $\delta\pi^{\mu\nu}$, without
requiring the normalization, orthogonality, or trace relations in
\eqref{E:IS_linear_constraints}.

An orthonormal frame with $e_0=u$ makes the ranks and inverse changes
explicit. The covector $\phi$ remains arbitrary. In this frame
$u^0=1$, $u^i=0$, and $\pi^{0\mu}=0$. Use the eleven entries
\begin{align*}
 w&=(\delta\varepsilon,\delta n,\delta\Pi,
       \delta u^1,\delta u^2,\delta u^3,
       \delta\pi^{11},\delta\pi^{22},
       \delta\pi^{12},\delta\pi^{13},\delta\pi^{23})
\end{align*}

together with $q$ as the seventeen independent components of the vector.
The six entries not included in $w$ are reconstructed by
\begin{subequations}
\label{E:IS_inverse_column_change}
\begin{align}
 \delta u^0&=-\frac{q_{\Phi}}{2},\\
 \delta\pi^{00}&=-q_{\Psi}^{0},\\
 \delta\pi^{0i}=\delta\pi^{i0}
 &=\sum_{j=1}^{3}\pi^{ij}\delta u^j-q_{\Psi}^{i},
       \qquad i\in\{1,2,3\},\\
 \delta\pi^{33}
 &=q_{\mathrm{tr}}-q_{\Psi}^{0}
                   -\delta\pi^{11}-\delta\pi^{22}.
\end{align}
\end{subequations}
Spatial off-diagonal entries are reconstructed by symmetry.
Thus \eqref{E:IS_inverse_column_change} defines an invertible map $T_2$
from $(w,q)$ to the seventeen columns of $A_{17}$; in particular, the
six conditions $q=0$ are independent.

We next take invertible linear combinations of the seventeen
equations represented by $A_{17}(\phi)$. We retain the eleven entries
\begin{align*}
 f=(F_{\varepsilon},F_n,F_{\Pi},F_u^1,F_u^2,F_u^3,
          F_{\pi}^{11},F_{\pi}^{22},
          F_{\pi}^{12},F_{\pi}^{13},F_{\pi}^{23})
\end{align*}
and replace the other six entries by the contractions
\begin{align*}
 r_{\Phi}&=u_{\mu}F_u^{\mu},\\
 r_{\Psi}^{\mu}&=u_{\nu}F_{\pi}^{\nu\mu},\\
 r_{\mathrm{tr}}&=g_{\mu\nu}F_{\pi}^{\mu\nu}.
\end{align*}
This defines the row change $R_2$. Its inverse is
\begin{align*}
 F_u^0&=-r_{\Phi},\\
 F_{\pi}^{00}&=-r_{\Psi}^{0},\\
 F_{\pi}^{0i}=F_{\pi}^{i0}&=-r_{\Psi}^{i},
                  \qquad i\in\{1,2,3\},\\
 F_{\pi}^{33}&=r_{\mathrm{tr}}-r_{\Psi}^{0}
                           -F_{\pi}^{11}-F_{\pi}^{22},
\end{align*}
with the eleven entries in $f$ retained and spatial symmetry imposed.
Consequently, $R_2$ is invertible. Both $T_2$ and $R_2$ are independent
of $\phi$.
The original seventeen homogeneous equations are therefore
equivalent to requiring all eleven entries of $f$ and all six contractions
$r_{\Phi}$, $r_{\Psi}^{\mu}$, and $r_{\mathrm{tr}}$ to vanish.
The substitution defined by $T_2$ multiplies $A_{17}(\phi)$ on the
right, while the linear combinations defined by $R_2$ multiply it on
the left.

Equation \eqref{E:IS_constraint_row_identity} now gives
\begin{subequations}
\begin{align}
\label{E:IS_second_matrix}
 R_2A_{17}(\phi)T_2
 &=\begin{bmatrix}
        A_{11}(\phi)&B_2(\phi)\\
        0&M(\phi)
    \end{bmatrix},\\
\label{E:IS_constraint_matrix}
 M(\phi)
 &=\begin{bmatrix}
    \frac{E}{2}z&\phi_{\nu}&0\\
    \mathcal{D}^{\mu\alpha}\phi_{\alpha}
             &\tau_{\pi}z\delta^{\mu}{}_{\nu}&0^{\mu}\\
    0&0_{\nu}&\tau_{\pi}z
   \end{bmatrix}.
\end{align}
\end{subequations}
Here $A_{11}$ and $B_2$ are defined by \eqref{E:IS_second_matrix}.
Their entries need not be computed to take the determinant of that
identity. It yields
\begin{align}
\label{E:IS_second_determinant}
 (\det R_2)(\det T_2)\det A_{17}(\phi)
 &=\det A_{11}(\phi)\det M(\phi).
\end{align}
This concludes the treatment of normalization, orthogonality, and trace.

\medskip
\noindent\textbf{Step 3: Identification of the physical determinant.}
For $q=0$, \eqref{E:IS_inverse_column_change} reconstructs the
seventeen components from $w$, imposing the normalization,
orthogonality, and trace conditions \eqref{E:IS_linear_normalization},
\eqref{E:IS_linear_orthogonality}, and \eqref{E:IS_linear_trace}.
Setting the six antisymmetric components to zero and applying $T_1$
then reconstructs the full vector $\delta\Umb$. For each $\mu<\nu$,
the defining formulas for $T_1$ give
\begin{align*}
 \delta\pi^{\mu\nu}
 &=\delta\pi^{(\mu\nu)}=\delta\pi^{\nu\mu},
\end{align*}
so the symmetry condition is also satisfied. Thus every
$w\in\mathbb{R}^{11}$ determines a unique real vector $\delta\Umb$
satisfying \eqref{E:IS_linear_constraints}. Conversely, extracting
the eleven entries defining $w$ from any such vector recovers $w$;
\eqref{E:IS_inverse_column_change} and symmetry determine every other
entry. The map $w\mapsto T_1\big(T_2(w,0),0\big)$ is therefore a bijection
onto the real vectors satisfying \eqref{E:IS_linear_constraints}.
For such a vector,
\eqref{E:IS_first_triangular_matrix} makes the six antisymmetric equation
entries zero, and \eqref{E:IS_second_matrix} makes the six contracted
equation entries zero. Since $R_1$ and $R_2$ are invertible,
\begin{subequations}
\label{E:IS_kernel_equivalence}
\begin{align}
 \delta\Umb&=T_1\big(T_2(w,0),0\big),\\
 A_{11}(\phi)w=0
 \quad&\Longleftrightarrow\quad A(\phi)\delta\Umb=0.
\end{align}
\end{subequations}
In \eqref{E:IS_kernel_equivalence}, the zero in $T_2(w,0)$
specifies $q=0$, and the zero in the second argument of $T_1$ specifies
the six antisymmetric components.
Moreover, $w\neq0$ if and only if $\delta\Umb\neq0$.

By \eqref{E:IS_kernel_equivalence}, a nonzero real covector
$\phi$ is physical if and only if $\det A_{11}(\phi)=0$.
This equivalence also holds when either
$\tau_{\pi}^{6}z^6=0$ or $\det M(\phi)=0$, since
\eqref{E:IS_kernel_equivalence} does not divide by these factors.

Define the quadratic polynomial
\begin{align}
\label{E:IS_quadratic_polynomial}
 S(\phi)&=z^2-\frac{2\mathcal{D}^{\mu\nu}
                              \phi_{\mu}\phi_{\nu}}{E\tau_{\pi}}.
\end{align}
At $v>0$, \eqref{E:IS_D_tensor} gives
\begin{align}
\label{E:IS_quadratic_angular_form}
\begin{split}
 S(\phi)
 &=v^2\left(\hat{z}^{2}
       -\frac{\frac{1}{2}(2\eta+\lambda_{\pi\Pi}\Pi)
               +\frac{1}{4}\tau_{\pi\pi}\Lambda_{\hat{v}^2}}
                      {\tau_{\pi}E}\right)\\
 &=v^2\big(\hat{z}^2-y_0\big).
\end{split}
\end{align}
For $z\neq0$, elimination of the four rows indexed by $\mu$ in
\eqref{E:IS_constraint_matrix}, with $S$ defined in
\eqref{E:IS_quadratic_polynomial}, gives
\begin{align*}
 \det M(\phi)
 &=(\tau_{\pi}z)^5
     \left(\frac{E}{2}z
            -\frac{\mathcal{D}^{\mu\nu}\phi_{\mu}\phi_{\nu}}
                          {\tau_{\pi}z}\right)\\
 &=\frac{E\tau_{\pi}^{5}}{2}z^4S(\phi).
\end{align*}
Both sides are polynomials in the components of $\phi$ and agree on
the open set $z\neq0$. Thus
\begin{align}
\label{E:IS_constraint_determinant}
 \det M(\phi)&=\frac{E\tau_{\pi}^{5}}{2}z^4S(\phi)
\end{align}
holds for every covector.

Set
\begin{align*}
 d&=(\det R_1)(\det T_1)(\det R_2)(\det T_2)\neq0.
\end{align*}
Equations \eqref{E:IS_first_determinant}, \eqref{E:IS_second_determinant},
and \eqref{E:IS_constraint_determinant} give
\begin{align}
\label{E:IS_full_physical_factorization}
 d\det A(\phi)
 &=\frac{E\tau_{\pi}^{11}}{2}
                     z^{10}S(\phi)\det A_{11}(\phi).
\end{align}
This proves the required factorization directly for the IS equations.

For completeness, the coefficient independent of the covector in the
full determinant is
\begin{align}
\label{E:IS_full_coefficient}
 c_0&=\tau_{\Pi}\tau_{\pi}^{16}
                       \prod_{A=0}^{3}(E+\Lambda_A)\neq0.
\end{align}
The normalization in \eqref{E:IS_full_coefficient} can be checked
without dividing by spatial components.
At $\phi=(z,0,0,0)$, elimination of the energy, number, bulk, and
sixteen shear rows leaves the momentum matrix
\begin{align*}
 \begin{bmatrix}
 Ez&0\\
 0&z(E\delta_{ij}+\pi^{ij})_{i,j=1}^{3}
 \end{bmatrix}.
\end{align*}
The determinant of the eliminated rows is
$\tau_{\Pi}\tau_{\pi}^{16}z^{19}$, so
\begin{align}
\label{E:IS_time_determinant}
 \det A(z,0,0,0)
 &=\tau_{\Pi}\tau_{\pi}^{16}
             E\prod_{a=1}^{3}(E+\Lambda_a)z^{23}
 =c_0z^{23}.
\end{align}
For $v>0$, equations \eqref{supeq:nodiffusiondet} and \eqref{E:IS_quadratic_angular_form} give
\begin{align}
\label{E:IS_full_angular_determinant}
 \det A(\phi)
 &=c_0z^{15}S(\phi)
       v^6\mathcal{P}_3\left(\frac{z^2}{v^2},\chi^2,\kappa^2\right).
\end{align}
For each fixed nonzero real spatial covector, substitution of
\eqref{E:IS_full_angular_determinant} into
\eqref{E:IS_full_physical_factorization} and cancellation of the nonzero
polynomial $z^{10}S$ in the variable $z$ gives
\begin{subequations}
\label{E:IS_physical_angular_determinant}
\begin{align}
 \det A_{11}(\phi)
 &=c\,z^5 v^6\mathcal{P}_3\left(\frac{z^2}{v^2},\chi^2,\kappa^2\right),\\
 c&=\frac{2dc_0}{E\tau_{\pi}^{11}}\neq0.
\end{align}
\end{subequations}
The constants $c_0,d,c$ depend on the fixed state and bases, not on
$\phi$. By \eqref{E:IS_physical_angular_determinant}, the coefficients of
$z^j$, $j\in\{0,1,2,3,4\}$, in
$\det A_{11}(z,\phi_1,\phi_2,\phi_3)$ therefore vanish for every
nonzero real spatial covector. These coefficients are polynomials in
$(\phi_1,\phi_2,\phi_3)$, so they vanish identically. Consequently
there is a homogeneous polynomial $H_3$ of degree six such that
\begin{subequations}
\begin{align}
\label{E:IS_physical_determinant}
 \det A_{11}(\phi)&=c\,z^5H_3(\phi),\\
\label{E:IS_H_angular_expression}
 H_3(\phi)
 &=v^6\mathcal{P}_3\left(\frac{z^2}{v^2},\chi^2,\kappa^2\right),
                      \qquad v>0.
\end{align}
\end{subequations}
Equations \eqref{E:IS_full_physical_factorization} and
\eqref{E:IS_time_determinant}, with $S(z,0,0,0)=z^2$, imply
\begin{align}
\label{E:IS_H_time_axis}
 H_3(z,0,0,0)&=z^6.
\end{align}
The identities \eqref{E:IS_full_physical_factorization},
\eqref{E:IS_physical_determinant}, and \eqref{E:IS_H_time_axis}
are polynomial identities. In particular, no division at a covector where $zS=0$
has been used. A zero of $S$ is still physical if it is also a zero of
$z^5H_3$; cancellation does not remove such a covector from
\eqref{E:IS_kernel_equivalence}.

Different bases for the eleven-dimensional space of vectors satisfying
\eqref{E:IS_linear_constraints} and for the eleven
independent equation entries multiply $A_{11}$ on the right and left by
invertible matrices independent of $\phi$. Their determinants only
change $c$ in \eqref{E:IS_physical_determinant}. The physical
characteristic equation is therefore independent of these choices.
This concludes the identification of the physical determinant.

\medskip
\noindent \textbf{Step 4: Spurious characteristics.}
By \eqref{E:IS_quadratic_angular_form}, condition \eqref{supeq:C5}
is exactly the bound on the roots of $S=0$ used in
Section~\ref{S:IS_full_determinant_bounds}.
Equations \eqref{E:IS_constraint_determinant} and
\eqref{E:IS_full_physical_factorization} show that this factor occurs
in the equations for the six numbers $q$ in
\eqref{E:IS_constraint_row_identity}.
Equations \eqref{E:IS_kernel_equivalence} and
\eqref{E:IS_physical_determinant} show that it supplies no additional
physical characteristic: a zero of $S$ is physical if and only if it
also satisfies $z^5H_3=0$.

Condition (CI) is now tested on all temporal roots of the physical
determinant \eqref{E:IS_physical_determinant}. The roots contributed
only by $S$ in \eqref{E:IS_full_physical_factorization} impose no
reality requirement on the physical system.

For a nonzero real spatial covector in the frame $e_0=u$,
\eqref{E:IS_H_angular_expression} gives the physical temporal roots
$z=0$ and $z=\pm\sqrt{y_k}\,v$, $k=1,2,3$.
They are all real if and only if all three $y_k$ are real and
nonnegative. On a real branch,
\begin{align}
\label{E:IS_covector_norm}
 g^{\mu\nu}\phi_{\mu}\phi_{\nu}
 &=(1-y_k)v^2,\qquad k=1,2,3.
\end{align}
The root $z=0$ has squared norm $v^2>0$.
When the spatial covector is zero, \eqref{E:IS_H_time_axis} gives
$z^5H_3=z^{11}$, so there is no nonzero physical covector.
Consequently, reality of all physical temporal roots in this frame
and non-timelikeness of all real physical covectors are equivalent
to $0\leq y_k\leq1$ for $k=1,2,3$.
Step~2 of Section~\ref{S:IS_full_determinant_bounds} proves that
these three bounds are equivalent to
\eqref{supeq:C1}--\eqref{supeq:C4}, without any condition on $y_0$.
This proves the required necessity and sufficiency in the fluid frame,
including when \eqref{supeq:C5} fails.

\noindent\textbf{Step 5: Independence of the time direction.}
Set
\begin{align}
\label{E:IS_normalized_physical_polynomial}
 F(\phi)&=z^5H_3(\phi).
\end{align}
Let $e=-u^{\flat}$, where $(u^{\flat})_{\mu}=g_{\mu\nu}u^{\nu}$.
Equation \eqref{E:IS_H_time_axis} gives $F(e)=1$.
Under \eqref{eq:C1}--\eqref{eq:C4}, for every real covector
$\zeta$ the polynomial $t\mapsto F(\zeta+te)$ has eleven real
roots, counted with multiplicity: in the frame $e_0=u$ this just
translates $z$ while leaving the spatial components fixed.
Thus $F$ is hyperbolic with respect to $e$, in the precise polynomial
sense of this real-root property and $F(e)\neq0$.

Let $\Gamma^{*}$ be the connected component of $\{F\neq0\}$ containing
$e$. 
Corollary~2.7, Theorem~2.9, and Corollary~2.13(a) of
Harvey and Lawson~\cite{HarveyLawsonHyperbolic} imply that $\Gamma^{*}$ is
an open convex cone and that $F$ is hyperbolic with respect to every
covector in $\Gamma^{*}$.
All hypotheses hold here: $F$ is a real homogeneous polynomial of degree
eleven, $F(e)=1$, and its temporal roots in direction $e$ have just
been proved real.

Define the timelike half-cone
\begin{align}
\label{E:IS_timelike_half_cone}
 K_g^{*,+}
 &=\{\xi\in T_x^*\mathcal{M}\,|\,
           g^{\mu\nu}\xi_{\mu}\xi_{\nu}<0,
           \xi_{\mu}u^{\mu}>0\}.
\end{align}
The set \eqref{E:IS_timelike_half_cone} is connected and contains
$e$. Equation \eqref{E:IS_covector_norm} and the case $z=0$ prove that $F$ has no zero on that set. Consequently,
\begin{align}
\label{E:IS_cone_inclusion}
 K_g^{*,+}&\subseteq\Gamma^{*}.
\end{align}
Every timelike covector belongs to $K_g^{*,+}$ or its negative.
Theorem~2.9 of \cite{HarveyLawsonHyperbolic} and the substitution $t\mapsto-t$ therefore
show that $t\mapsto F(\zeta+t\xi)$ has only real roots for every
real $\zeta$ and every timelike $\xi$. Non-timelikeness follows from
\eqref{E:IS_covector_norm} and the case $z=0$, checked for arbitrary
real covectors in one frame. Conversely, reality for every timelike
time direction includes the direction $e$; together with
non-timelikeness this forces $0\leq y_k\leq1$, $k=1,2,3$.
Step~2 of Section~\ref{S:IS_full_determinant_bounds} then gives
the necessity of \eqref{eq:C1}--\eqref{eq:C4}.

\medskip
This concludes the proof of Theorem~\ref{thm:causalitynodiffusion}.

\subsection{Domain of dependence}
\label{S:IS_domain_of_dependence}

Theorem~\ref{thm:causalitynodiffusion} determines the physical
characteristics. Our existence result, Theorem \ref{T:LWP}, requires additional assumptions compare \eqref{eq:C1}--\eqref{eq:C4} with \eqref{eq:H1}--\eqref{eq:H5} and either \eqref{E:IS_Separation_polynomial} or \eqref{E:IS_Separation_transport}), but we do not know whether the additional assumptions are necessary. Causality, on the other hand, should be viewed as a conditional result, i.e., a statement about properties of solutions, regardless of whether such solutions have been shown to exist. In this Section, we close this gap by showing that any Sobolev-regular solutions that satisfy fairly natural assumptions enjoy a domain-of-dependence property with characteristics given precisely by those from Theorem~\ref{thm:causalitynodiffusion}. We present the result first\footnote{The result we present is a precise implementation of the ideas discussed in \cite[Section B.2.1]{Disconzi-2024}.}, and then discuss the meaning of its hypotheses.

Let $X$ be a common evolution region to the future of a Cauchy
hypersurface $\Sigma$ in a fixed globally hyperbolic spacetime.
Consider a class $\mathscr{S}$ of $C^1$ solutions, with four-velocity
$u$ future-directed in the fixed spacetime time orientation, of
\eqref{eq:propsystem} on $X\cup\Sigma$ that satisfy \eqref{Constraints} and the assumptions and inequalities
of Theorem~\ref{thm:causalitynodiffusion} at every point of
$X\cup\Sigma$.
For the quantitative statement, assume that the common region admits
coordinates $y=(t,x)\in[0,T]\times\mathbb{R}^3$, with
$\Sigma=\{t=0\}$, spacelike level sets of $t$, and $dt(u)>0$.
The background metric is fixed. 
Fix a smooth orthonormal frame $e_{\hat{a}}$, $a=0,1,2,3$,
independent of the solution, with $e_{\hat{0}}$ future-directed, and write
\begin{align*}
 e_{\hat{a}}&=E_{\hat{a}}{}^\mu(y)\partial_{y^\mu},\\
 g(e_{\hat{a}},e_{\hat{b}})&=\eta_{\hat{a}\hat{b}},
\end{align*}
where $(\eta_{\hat{a}\hat{b}})=\operatorname{diag}(-1,1,1,1)$.
Hatted indices denote components in this fixed frame, and Greek
indices denote coordinate components. In particular,
\begin{align*}
 u^\mu&=E_{\hat{a}}{}^\mu u^{\hat{a}},\\
 \pi^{\mu\nu}&=E_{\hat{a}}{}^\mu E_{\hat{b}}{}^\nu\pi^{\hat{a}\hat{b}}.
\end{align*}
The scalar fields $\varepsilon,n,\Pi$ and the following frame components
form eleven independent variables, all regarded as functions of $y$:
\begin{align}
\label{E:IS_independent_variables}
 Q&=(\varepsilon,n,\Pi,\pi^{\hat{1}\hat{1}},\pi^{\hat{2}\hat{2}},\pi^{\hat{1}\hat{2}},\pi^{\hat{1}\hat{3}},
             \pi^{\hat{2}\hat{3}},u^{\hat{1}},u^{\hat{2}},u^{\hat{3}}).
\end{align}
Spatial frame indices in the reconstruction formulas are raised and
lowered with $\delta_{\hat{i}\hat{j}}$. Set
$\beta_{\hat{i}}=u^{\hat{i}}/u^{\hat{0}}$ and
$H_{\hat{i}\hat{j}}=\delta_{\hat{i}\hat{j}}-\beta_{\hat{i}}\beta_{\hat{j}}$.
The constraints determine the other frame components by
\begin{subequations}
\label{E:IS_nonlinear_reconstruction}
\begin{align}
 u^{\hat{0}}&=\sqrt{1+\sum_{i=1}^{3}(u^{\hat{i}})^2},\\
 \pi^{\hat{3}\hat{3}}&=-\frac{H_{\hat{1}\hat{1}}\pi^{\hat{1}\hat{1}}
       +H_{\hat{2}\hat{2}}\pi^{\hat{2}\hat{2}}
       +2H_{\hat{1}\hat{2}}\pi^{\hat{1}\hat{2}}
       +2H_{\hat{1}\hat{3}}\pi^{\hat{1}\hat{3}}
       +2H_{\hat{2}\hat{3}}\pi^{\hat{2}\hat{3}}}{H_{\hat{3}\hat{3}}},\\
 \pi^{\hat{0}\hat{i}}=\pi^{\hat{i}\hat{0}}&=\beta_{\hat{j}}\pi^{\hat{j}\hat{i}},\\
 \pi^{\hat{0}\hat{0}}&=\beta_{\hat{i}}\beta_{\hat{j}}\pi^{\hat{i}\hat{j}},\\
 \pi^{\hat{i}\hat{j}}&=\pi^{\hat{j}\hat{i}}.
\end{align}
\end{subequations}
Indeed, normalization gives $u^{\hat{0}}$, orthogonality gives the
temporal frame components of the shear tensor, and the trace condition
becomes $H_{\hat{i}\hat{j}}\pi^{\hat{i}\hat{j}}=0$. The denominator is positive:
\begin{align*}
 H_{\hat{3}\hat{3}}&=\frac{1+(u^{\hat{1}})^2+(u^{\hat{2}})^2}
                    {1+(u^{\hat{1}})^2+(u^{\hat{2}})^2+(u^{\hat{3}})^2}>0.
\end{align*}
In particular, the formula for $u^{\hat{0}}$ in
\eqref{E:IS_nonlinear_reconstruction} concerns the frame component;
the coordinate component is $u^0=E_{\hat{a}}{}^0u^{\hat{a}}$.
Thus \eqref{E:IS_nonlinear_reconstruction} parametrizes the physical
states without additional algebraic equations for $Q$.

Substitute \eqref{E:IS_nonlinear_reconstruction} into
\eqref{eq:energy-nodiffusion}, \eqref{eq:number-nodiffusion},
\eqref{eq:bulk-nodiffusion}, the three spatial frame components of
\eqref{eq:momentum-nodiffusion}, and the five frame components of
\eqref{eq:shear-nodiffusion} with indices
$(\hat{a},\hat{b})\in\{(\hat{1},\hat{1}),(\hat{2},\hat{2}),(\hat{1},\hat{2}),
(\hat{1},\hat{3}),(\hat{2},\hat{3})\}$.
With frame derivatives, these eleven equations have the form
\begin{align*}
 \sum_{a=0}^{3}\mathcal{B}^{\hat{a}}(y,Q)e_{\hat{a}}(Q)
 &=\mathcal{R}(y,Q).
\end{align*}
Since $e_{\hat{a}}(Q)=E_{\hat{a}}{}^\mu\partial_{y^\mu}Q$, the equations
in the coordinates $y$ are
\begin{align}
\label{E:IS_reduced_nonlinear_system}
 B^\mu(y,Q)\partial_\mu Q&=\mathcal{R}(y,Q).
\end{align}
Here $\partial_\mu=\partial_{y^\mu}$, and the $11\times11$ matrices
are related by
\begin{align*}
 B^\mu(y,Q)&=\sum_{a=0}^{3}E_{\hat{a}}{}^\mu(y)\mathcal{B}^{\hat{a}}(y,Q).
\end{align*}
The scalar fields and frame components defining $Q$ are unchanged by
this conversion to \eqref{E:IS_reduced_nonlinear_system}. The function $\mathcal{R}$ includes
the relaxation terms and the connection terms of the fixed frame;
the latter depend on the prescribed metric and frame and on $Q$, but
contain no derivatives of $Q$.
All derivatives of dependent field components differentiate the
functions in \eqref{E:IS_nonlinear_reconstruction}.

In order to verify that \eqref{E:IS_reduced_nonlinear_system} implies all the
original equations, write the momentum and shear equations
\eqref{eq:momentum-nodiffusion} and \eqref{eq:shear-nodiffusion} with zero
on their right-hand sides. Denote their complete left-hand expressions,
in the fixed orthonormal frame and after substitution of
\eqref{E:IS_nonlinear_reconstruction}, by $M^{\hat{a}}$ and
$S^{\hat{a}\hat{b}}$, respectively. Thus the full momentum and shear
equations require every component of $M$ and $S$ to vanish.
For every differentiable $Q$, normalization, orthogonality, and the
projections in \eqref{eq:shear-nodiffusion} give
\begin{subequations}
\label{E:IS_equation_component_identities}
\begin{align}
 u_{\hat{a}}M^{\hat{a}}&=0,\\
 S^{\hat{a}\hat{b}}&=S^{\hat{b}\hat{a}},\\
 u_{\hat{a}}S^{\hat{a}\hat{b}}&=0,\\
 \eta_{\hat{a}\hat{b}}S^{\hat{a}\hat{b}}&=0.
\end{align}
\end{subequations}
Equations \eqref{E:IS_equation_component_identities} do not require $Q$ to solve
\eqref{E:IS_reduced_nonlinear_system}. They imply
\begin{align*}
 M^{\hat{0}}&=\beta_{\hat{i}}M^{\hat{i}},\\
 S^{\hat{0}\hat{i}}&=\beta_{\hat{j}}S^{\hat{j}\hat{i}},\\
 S^{\hat{0}\hat{0}}&=\beta_{\hat{i}}\beta_{\hat{j}}S^{\hat{i}\hat{j}},\\
 H_{\hat{i}\hat{j}}S^{\hat{i}\hat{j}}&=0.
\end{align*}
If the five selected spatial frame components of $S$ vanish, the
identity $H_{\hat{i}\hat{j}}S^{\hat{i}\hat{j}}=0$ reduces to
$H_{\hat{3}\hat{3}}S^{\hat{3}\hat{3}}=0$.
Since $H_{\hat{3}\hat{3}}>0$, every component of $S$ vanishes.
Likewise, vanishing of the three spatial frame components of $M$
implies $M^{\hat{0}}=0$.
Thus \eqref{E:IS_reduced_nonlinear_system} is equivalent to the full
physical equations.
The derivative of \eqref{E:IS_nonlinear_reconstruction} maps the eleven
independent derivatives bijectively to the vectors satisfying
\eqref{E:IS_linear_constraints}.  Equations \eqref{E:IS_equation_component_identities}
also identify the eleven independent equations. These are invertible
changes of the independent columns and rows used in
\eqref{E:IS_physical_determinant}, and hence
\begin{align}
\label{E:IS_reduced_characteristic_polynomial}
 \det\big(B^\mu(y,Q)\xi_\mu\big)
 &=c(y,Q)F_{\Umb}(\xi),\qquad c(y,Q)\neq0,
\end{align}
where $F_{\Umb}$ denotes \eqref{E:IS_normalized_physical_polynomial}
at the reconstructed state. In particular,
\eqref{E:IS_reduced_nonlinear_system} has precisely the physical
characteristics.

For $\Umb\in\mathscr{S}$, define $\Gamma_{\Umb}^{*}(x)$ from
\eqref{E:IS_normalized_physical_polynomial} as the connected component of $\{F\neq0\}$ containing $-u^\flat$ and set
\begin{align}
\label{E:IS_dual_cone}
 C_{\Umb}(x)
 &=\{V\in T_x\mathcal{M}\,|\,
            \xi_{\mu}V^{\mu}\geq0
               \text{ for every }\xi\in\Gamma_{\Umb}^{*}(x)\}.
\end{align}
These are the dual cones used to define dependence domains; compare
Definition~B.11 and Remark~B.12 of \cite{Disconzi-2024}.
Equation \eqref{E:IS_cone_inclusion} implies
\begin{align}
\label{E:IS_metric_cone_inclusion}
 C_{\Umb}(x)&\subseteq
       \{V\,|\,g(V,V)\leq0,\ V\text{ future-directed}\}\cup\{0\}.
\end{align}
For this pointwise verification only, use a separate orthonormal frame
with timelike member $u$, and write $V=(V^0,\mathbf{V})$, with
$|\mathbf{V}|=(\sum_{i=1}^{3}(V^i)^2)^{\frac{1}{2}}$.
If $\mathbf{V}\neq0$, every covector
$(1,-r\mathbf{V}/|\mathbf{V}|)$ with $0\leq r<1$ belongs to
$K_g^{*,+}$. Equations \eqref{E:IS_cone_inclusion} and
\eqref{E:IS_dual_cone} give $V^0-r|\mathbf{V}|\geq0$.
Taking $r\rightarrow1$ gives $V^0\geq|\mathbf{V}|$.
For $\mathbf{V}=0$, the covector $(1,0,0,0)$ gives $V^0\geq0$.
This proves \eqref{E:IS_metric_cone_inclusion}.

For $p\in X$, let $R_{\Umb}(p)\subseteq\Sigma$ consist of the points
$q$ joined to $p$ in the common evolution region by a locally Lipschitz
curve $\gamma\colon[0,1]\rightarrow X\cup\Sigma$ satisfying
\begin{align*}
 \gamma(0)&=q,\\
 \gamma(1)&=p,\\
 \dot{\gamma}(s)&\in C_{\Umb}(\gamma(s))
                     \quad\text{for almost every }s\in[0,1].
\end{align*}
The cones are understood to extend to the initial hypersurface using
the initial fields. Equation \eqref{E:IS_metric_cone_inclusion} gives
\begin{align}
\label{E:IS_initial_past_inclusion}
 R_{\Umb}(p)&\subseteq\mathcal{J}_g^-(p)\cap\Sigma.
\end{align}

Define the development determined by the physical cones by
\begin{align}
\label{E:IS_physical_development}
 D_{\Umb}^{+}(\mathcal{O})
 &=\{p\in X\cup\Sigma\,|\,R_{\Umb}(p)\subseteq\mathcal{O}\}.
\end{align}
At $p\in\Sigma$ we set $R_{\Umb}(p)=\{p\}$.
The next theorem gives a sufficient approximation hypothesis under
which equality on \eqref{E:IS_physical_development} follows from
equality of the initial data.

We use $H^r(\mathbb{R}^3)$ for the $L^2$ Sobolev space. For a
vector-valued function, its squared Sobolev norm is the sum of the
squared Sobolev norms of its components.
For a domain $\mathcal{D}\subseteq\mathbb{R}^d$ and $\sigma\geq1$,
$G^\sigma(\mathcal{D})$ denotes the Gevrey space of order $\sigma$:
a smooth function $f$ belongs to this space if for every compact
$K\subset\mathcal{D}$ there are $C_K,L_K>0$ such that
\begin{align*}
 \sup_{z\in K}|\partial^\alpha f(z)|
 &\leq C_K L_K^{|\alpha|}(|\alpha|!)^\sigma
       \quad\text{for every multi-index }\alpha.
\end{align*}
For vector-valued functions the condition is imposed componentwise.
For $\mathcal{D}=[0,T]\times\mathbb{R}^3$, all derivatives extend
continuously to $t=0,T$, and the compact sets $K$ in the defining
bounds may meet these endpoints. Initial data are functions on
$\mathbb{R}^3$, whereas spacetime Gevrey regularity concerns functions
on $[0,T]\times\mathbb{R}^3$. The global coefficient bounds
\eqref{E:IS_Gevrey_coefficient_assumption} are imposed separately.

\begin{theorem}[Conditional domain of dependence for Sobolev solutions]
\label{T:IS_Sobolev_domain_of_dependence}
Let $m>\frac{7}{2}$ and $1<\sigma<\frac{11}{10}$.
Let $\Umb_1,\Umb_2\in\mathscr{S}$, and write $Q_i$ for their variables
\eqref{E:IS_independent_variables} and $q_i=Q_i|_{t=0}$.
Fix a constant reference state $\bar{Q}$.
Suppose $q_1=q_2$ on an open set $\mathcal{O}\subseteq\mathbb{R}^3$.
Assume the following properties.
\begin{enumerate}[label=(\roman*)]
\item There is a class $\mathscr{A}\subseteq\bar{Q}+H^m(\mathbb{R}^3)$
of initial data, containing $q_1,q_2$, for which
\eqref{E:IS_reduced_nonlinear_system} has unique solutions satisfying
\begin{align*}
 Q-\bar{Q}&\in C^0([0,T];H^m(\mathbb{R}^3))\cap C^1([0,T];H^{m-1}(\mathbb{R}^3))
\end{align*}
on a common interval $T=T(M)>0$ whenever
$\norm{q-\bar{Q}}_{H^m(\mathbb{R}^3)}\leq M$.
The given $Q_i$ are these Sobolev solutions for the data $q_i$.
For every pair of data with norms at most $M$, the solutions satisfy
\begin{align*}
 \sup_{0\leq t\leq T}\norm{Q(t)-\bar{Q}}_{H^m(\mathbb{R}^3)}
 &\leq C_0(M)\norm{q-\bar{Q}}_{H^m(\mathbb{R}^3)},
\end{align*}
and
\begin{align}
\label{E:IS_Sobolev_difference_assumption}
 \sup_{0\leq t\leq T}\norm{Q(t)-\widetilde{Q}(t)}_{H^{m-1}(\mathbb{R}^3)}
 &\leq C_1(M)\norm{q-\widetilde{q}}_{H^{m-1}(\mathbb{R}^3)}.
\end{align}
The constants depend on $M$, the fixed equations, and the background,
but not on the individual data.
\item There are $h_k\downarrow0$ and data $q_{i,k}\in\mathscr{A}\cap
G^\sigma(\mathbb{R}^3)$, $i=1,2$, such that
\begin{align*}
 \norm{q_i-\bar{Q}}_{H^m(\mathbb{R}^3)}&\leq M,\\
 \norm{q_{i,k}-\bar{Q}}_{H^m(\mathbb{R}^3)}&\leq M,\\
 q_{i,k}&\longrightarrow q_i\text{ in }H^m(\mathbb{R}^3).
\end{align*}
Define
\begin{align}
\label{E:IS_approximating_data_sets}
 \mathcal{O}_k&=\{x\,|\,\operatorname{dist}(x,\mathcal{O}^c)>h_k\},
\end{align}
with $\mathcal{O}_k=\mathbb{R}^3$ if $\mathcal{O}=\mathbb{R}^3$,
and require
\begin{align}
\label{E:IS_approximating_data_agreement}
 q_{1,k}|_{\mathcal{O}_k}&=q_{2,k}|_{\mathcal{O}_k}.
\end{align}
Their Sobolev solutions $Q_{i,k}$ satisfy
$Q_{i,k}\in G^\sigma([0,T]\times\mathbb{R}^3)$ on the common interval.
Substitution of $Q_{i,k}$ into \eqref{E:IS_nonlinear_reconstruction}
determines all dependent frame components of $u$ and $\pi$; together
with the components already in $Q_{i,k}$, these define the full fields
$\Umb_{i,k}$. The constraints hold by those formulas, and the full IS
equations hold by the equivalence with
\eqref{E:IS_reduced_nonlinear_system}. Require in addition that these
fields satisfy the assumptions and inequalities of
Theorem~\ref{thm:causalitynodiffusion} at every
$y\in[0,T]\times\mathbb{R}^3$, so that $\Umb_{i,k}\in\mathscr{S}$.
\item 
There is an open convex state domain $\mathcal{V}$ such that
$B^\mu(y,Q)$ and $\mathcal{R}(y,Q)$ are smooth jointly in $(y,Q)$ on
$[0,T]\times\mathbb{R}^3\times\mathcal{V}$, up to the time endpoints.
There is one compact set $\mathcal{K}\subset\mathcal{V}$ containing
$Q_i(y)$, $Q_{i,k}(y)$, and
$Q_{2,k}(y)+\theta(Q_{1,k}(y)-Q_{2,k}(y))$ for every
$y\in[0,T]\times\mathbb{R}^3$, $i\in\{1,2\}$, $k\geq1$, and
$\theta\in[0,1]$.
There is a constant $b_*>0$, independent of $i,k,t,x$, such that
\begin{align}
\label{E:IS_uniform_matrix_bounds}
 \norm{(B^0(y,Q))^{-1}}+
       \sum_{j=1}^{3}\norm{(B^0(y,Q))^{-1}B^j(y,Q)}&\leq b_*
\end{align}
for every $y\in[0,T]\times\mathbb{R}^3$, $i\in\{1,2\}$, and
$k\geq1$, evaluated at each of the states $Q=Q_i(y)$ and $Q=Q_{i,k}(y)$.
All matrix norms here are Euclidean operator norms.
\item Set $Q_{\theta,k}=Q_{2,k}+\theta(Q_{1,k}-Q_{2,k})$,
$0\leq\theta\leq1$, and define a matrix $C_k$ by its action on
$Z\in\mathbb{R}^{11}$:
\begin{align}
\label{E:IS_difference_coefficient}
 C_k Z
 &=\int_0^1\!\left\{
 D_Q\mathcal{R}(y,Q_{\theta,k})[Z]
 -D_QB^\mu(y,Q_{\theta,k})[Z]\partial_\mu Q_{2,k}
 \right\}\,d\theta.
\end{align}
Let
\begin{align*}
 A_k^j&=(B^0(y,Q_{1,k}))^{-1}B^j(y,Q_{1,k}),\\
 D_k&=(B^0(y,Q_{1,k}))^{-1}C_k.
\end{align*}
For each $k$ there exist $G_k,L_k>0$ such that, for every spacetime
multi-index $\alpha$,
\begin{align}
\label{E:IS_Gevrey_coefficient_assumption}
 \sup_{[0,T]\times\mathbb{R}^3}
 \left(\sum_{j=1}^{3}\norm{\partial^\alpha A_k^j}
                   +\norm{\partial^\alpha D_k}\right)
 &\leq G_k L_k^{|\alpha|}(|\alpha|!)^\sigma.
\end{align}
The constants $G_k,L_k$ need not be bounded independently of $k$.
\end{enumerate}
Then
\begin{align}
\label{E:IS_Sobolev_physical_uniqueness}
 \Umb_1&=\Umb_2\quad\text{on }D_{\Umb_1}^{+}(\mathcal{O}).
\end{align}
In particular the two solutions agree on
$D_g^+(\mathcal{O})\cap(X\cup\Sigma)$.
\end{theorem}

\begin{proof}
$\mathcal{O}=\mathbb{R}^3$, uniqueness in (i) gives
$Q_1=Q_2$ on $[0,T]\times\mathbb{R}^3$; hence
\eqref{E:IS_nonlinear_reconstruction} gives $\Umb_1=\Umb_2$ there.
The conclusions of the theorem follow in this case.
Henceforth assume $\mathcal{O}\neq\mathbb{R}^3$.

\medskip
\noindent\textbf{Step 1: Estimates for differences.}
Let $W_k=Q_{1,k}-Q_{2,k}$.
Subtracting \eqref{E:IS_reduced_nonlinear_system} and using the
fundamental theorem of calculus on the segments $Q_{\theta,k}$ gives
\begin{align}
\label{E:IS_Gevrey_difference_equation}
 \partial_t W_k+\sum_{j=1}^{3}A_k^j\partial_jW_k-D_kW_k&=0.
\end{align}
In \eqref{E:IS_difference_coefficient}, the derivative of $Q_{2,k}$
is a coefficient multiplying $W_k$; it does not differentiate $W_k$.
By \eqref{E:IS_reduced_characteristic_polynomial} and
Theorem~\ref{thm:causalitynodiffusion}, the principal polynomial of
\eqref{E:IS_Gevrey_difference_equation} is hyperbolic in direction $dt$
and its dual cone is $C_{\Umb_{1,k}}$.

The local uniqueness result stated in Section~1.4, Example~5, of
Rauch~\cite{RauchPreciseFiniteSpeed} applies to a hyperbolic
$N\times N$ first-order system with Gevrey coefficients of order
$1<\sigma<\frac{N}{N-1}$.
Here $N=11$, the coefficient bounds are
\eqref{E:IS_Gevrey_coefficient_assumption}, and $W_k$ is smooth.
Thus the local uniqueness hypothesis in Definition~7 of that reference
holds. Theorem~1 of \cite{RauchPreciseFiniteSpeed} implies that every
point of $\operatorname{supp}W_k$ can be reached from
$\operatorname{supp}(W_k|_{t=0})$ by a curve with tangent in
$C_{\Umb_{1,k}}$.
Since \eqref{E:IS_approximating_data_agreement} gives
$\operatorname{supp}(W_k|_{t=0})\subseteq\mathcal{O}_k^c$, we obtain
\begin{align}
\label{E:IS_Gevrey_localized_equality}
 Q_{1,k}&=Q_{2,k}\quad\text{on }D_{\Umb_{1,k}}^+(\mathcal{O}_k).
\end{align}

The solutions $Q_{i,k}$ belong to the Sobolev class in (i), so
\eqref{E:IS_Sobolev_difference_assumption} gives
\begin{align}
\label{E:IS_Sobolev_approximation_limit}
 \sup_{0\leq t\leq T}\norm{Q_{i,k}(t)-Q_i(t)}_{H^{m-1}(\mathbb{R}^3)}
 &\leq C_1(M)\norm{q_{i,k}-q_i}_{H^{m-1}(\mathbb{R}^3)}
 \longrightarrow0.
\end{align}
Since $m-1>\frac{5}{2}$, Sobolev embedding gives uniform convergence
of the fields and their first spatial derivatives.
In particular, on every compact spacetime set,
$A_k^j\longrightarrow A^j$ uniformly, where
$A^j=(B^0(y,Q_1))^{-1}B^j(y,Q_1)$.

\medskip
\noindent\textbf{Step 2: Comparison of the physical cones.}
It remains to compare the domains in
\eqref{E:IS_Gevrey_localized_equality} with
\eqref{E:IS_physical_development} for $\Umb_1$.
For a matrix $\mathcal{M}$ whose eigenvalues are all real, let
$\lambda_{\max}(\mathcal{M})$ denote its largest eigenvalue.
For a spatial vector $\omega$, let
\begin{align}
\label{E:IS_maximum_physical_speed}
 a(y,\omega)&=\lambda_{\max}\left(
                         \sum_{j=1}^{3}A^j(y)\omega_j\right),
\end{align}
and define $a_k$ with $A_k^j$ in place of $A^j$.
All eigenvalues in \eqref{E:IS_maximum_physical_speed} are real.
No symmetry assumption on the matrices is used.
The largest root in $\tau$ of
$\det(\tau I+\sum_jA^j\xi_j)$ is $a(y,-\xi)$.
The definition of the dual cone therefore gives
\begin{align}
\label{E:IS_velocity_cone_test}
 (1,v)\in C_{\Umb_1}(y)
 &\quad\Longleftrightarrow\quad
 v\cdot\omega\leq a(y,\omega)\quad\text{for every }\omega\in\mathbb{R}^3.
\end{align}
The equivalence \eqref{E:IS_velocity_cone_test} holds with
$\Umb_{1,k},a_k$ as well.
Continuity of the roots of a monic polynomial in its coefficients
gives, for every compact spacetime set $K$,
\begin{align}
\label{E:IS_characteristic_speed_convergence}
 \sup_{y\in K,\ |\omega|=1}|a_k(y,\omega)-a(y,\omega)|
 &\longrightarrow0.
\end{align}
This uses compactness of $K\times S^2$ and permits multiple roots.
Moreover, \eqref{E:IS_uniform_matrix_bounds} gives
$|a_k(y,\omega)|,|a(y,\omega)|\leq b_*|\omega|$.
Taking $\omega=v$ in \eqref{E:IS_velocity_cone_test} proves
$|v|\leq b_*$ for every permitted velocity.
Thus the curves defining $R_{\Umb_1}$ and $R_{\Umb_{1,k}}$
can be parametrized by $t$, with spatial Lipschitz constant at most
$b_*$. To check extension to $t=0$, the factor
$(u^\mu\xi_\mu)^5$ in \eqref{E:IS_physical_determinant} shows that
$\xi(u)>0$ throughout $\Gamma_{\Umb}^{*}$, since it is positive at
$-u^\flat$ and cannot vanish in that connected component.
Hence $u\in C_{\Umb}$, and the continuous, locally Lipschitz vector
field $u/dt(u)$ supplies a continuation at any endpoint with $t>0$.
Its spatial speed is at most $b_*$, so continuation cannot escape to
spatial infinity in finite time. This proves extension to $t=0$.

\medskip
\noindent\textbf{Step 3: Passage to the limit in the domains of dependence.}
Fix $p=(t_p,x_p)\in D_{\Umb_1}^+(\mathcal{O})$ with $t_p>0$.
Suppose, along a subsequence, that
$p\notin D_{\Umb_{1,k}}^+(\mathcal{O}_k)$.
There are curves $\gamma_k:[0,t_p]\rightarrow\mathbb{R}^3$ with
$\gamma_k(t_p)=x_p$,
$(1,\gamma_k^{\prime}(t))\in C_{\Umb_{1,k}}(t,\gamma_k(t))$
almost everywhere, and $\gamma_k(0)\notin\mathcal{O}_k$.
They lie in the closed ball of radius $b_*T$ about $x_p$.
Arzel\`a--Ascoli gives a uniformly convergent subsequence
$\gamma_k\longrightarrow\gamma$ on $[0,t_p]$.
For every $0\leq b<c\leq t_p$ and every $\omega\in\mathbb{R}^3$,
\begin{align}
\label{E:IS_curve_integral_inequality}
 \omega\cdot\big(\gamma_k(c)-\gamma_k(b)\big)
 &\leq\int_b^c a_k(t,\gamma_k(t),\omega)\,dt.
\end{align}
Uniform convergence in \eqref{E:IS_characteristic_speed_convergence}
and continuity of $a$ allow passage to the limit in
\eqref{E:IS_curve_integral_inequality}, yielding
\begin{align}
\label{E:IS_limit_curve_integral_inequality}
 \omega\cdot\big(\gamma(c)-\gamma(b)\big)
 &\leq\int_b^c a(t,\gamma(t),\omega)\,dt.
\end{align}
Differentiating \eqref{E:IS_limit_curve_integral_inequality} for a countable dense set of
$\omega$, and then using continuity in $\omega$, gives
$(1,\gamma^{\prime}(t))\in C_{\Umb_1}(t,\gamma(t))$
almost everywhere. Also $\gamma(t_p)=x_p$.
But \eqref{E:IS_approximating_data_sets} and
$\gamma_k(0)\notin\mathcal{O}_k$ imply
$\operatorname{dist}(\gamma_k(0),\mathcal{O}^c)\leq h_k$.
Hence $\gamma(0)\in\mathcal{O}^c$, contradicting
$R_{\Umb_1}(p)\subseteq\mathcal{O}$.
We have proved $p\in D_{\Umb_{1,k}}^+(\mathcal{O}_k)$ for all
sufficiently large $k$.
Equations \eqref{E:IS_Gevrey_localized_equality} and
\eqref{E:IS_Sobolev_approximation_limit} now give $Q_1(p)=Q_2(p)$.
At $t_p=0$ equality is part of the data assumption.
Reconstruction by \eqref{E:IS_nonlinear_reconstruction} proves
\eqref{E:IS_Sobolev_physical_uniqueness} for all the physical fields.

\medskip
\noindent\textbf{Step 4: The metric domain of dependence.}
Finally, if $p\in D_g^+(\mathcal{O})$, then
$\mathcal{J}_g^-(p)\cap\Sigma\subseteq\mathcal{O}$:
each causal curve from $\Sigma$ to $p$ extends to a past-inextendible
causal curve, whose intersection with the Cauchy hypersurface
$\Sigma$ must lie in $\mathcal{O}$.
Equation \eqref{E:IS_initial_past_inclusion} then gives
$p\in D_{\Umb_1}^+(\mathcal{O})$.
This proves the metric conclusion and concludes the proof of
Theorem~\ref{T:IS_Sobolev_domain_of_dependence}.
\end{proof}

We now discuss the assumptions of Theorem \ref{T:IS_Sobolev_domain_of_dependence}.
The Sobolev estimates in (i) are of the form used in energy-based
local well-posedness arguments.  More precisely, suppose a Sobolev
construction supplies an $H^m$ energy estimate and an existence time
controlled by an $H^m$ bound on the data, with the coefficient bounds
and the constants in the energy estimate uniform on the chosen class
$\mathscr{A}$. It then supplies\footnote{If the constant reference state $\bar{Q}$ is itself a solution, a
homogeneous energy estimate for $Q-\bar{Q}$ gives the stated
multiplicative bound; a reference state that is not a solution would
in general produce an additional forcing term.} the common Sobolev interval in (i).
 The difference
estimate is naturally taken one derivative lower, in a typical quasilinear fashion, as subtraction of
quasilinear first-order equations produces a term of the form
$(A^j(Q)-A^j(\widetilde{Q}))\partial_j\widetilde{Q}$, where
$A^j(Q)=(B^0(y,Q))^{-1}B^j(y,Q)$ with the $y$ dependence suppressed.
On a bounded $H^m$ class, Sobolev product and composition estimates
control this term in $H^{m-1}$ by a constant times
$\norm{Q-\widetilde{Q}}_{H^{m-1}}$.
Thus an energy estimate for the linear difference equation, followed
by Gronwall's inequality, gives
\eqref{E:IS_Sobolev_difference_assumption}.
This is the role of the energy estimates in the construction outlined
in Section~\ref{S:LWP}; the required uniform estimates must be
available for the particular class $\mathscr{A}$ used here.
Sobolev existence alone does not imply the quantitative bounds in (i), but such bounds are the typical type of estimates produced by the known methods of establishing local well-posedness for hyperbolic problems.

The additional issue in (ii) is persistence of Gevrey regularity
throughout the common Sobolev interval. Local existence in a Gevrey
norm alone may give times depending on that norm and on the initial
Gevrey radius, and those times need not have a positive lower bound
as $k\rightarrow\infty$. A persistence theorem instead permits the
radius to decrease with time while remaining positive for each
fixed datum throughout its Sobolev lifespan\footnote{This is the Gevrey analogue of a time-independent interval of existence for analytic data in the usual Schauder method, see Section \ref{App:LWPSchauder}.}. This mechanism is
established, for example, by Kukavica and Vicol for the
three-dimensional incompressible Euler equations in a half-space
\cite[Theorem~2.2]{KukavicaVicolHalfSpaceGevrey}, and by Guo and Titi
for nonlinear wave equations with nonlinearities satisfying their
analyticity assumptions
\cite[Theorem~2.4 and Remark~2.12]{GuoTitiGevreyPersistence}.
In the former result the Sobolev exponent is $r>9/2$ and the Gevrey
order is any $s>1$; in the latter, Remark~2.12 extends the spatial
analyticity result to Gevrey orders $s\geq1$.
These results illustrate how one can obtain a common interval from
Sobolev control even when the Gevrey radii and norms of the individual
approximations are not uniform. 
These results concern different equations
and do not establish persistence for the system considered here. Our goal, however, is not to establish to prove such persistent here but simply to point out that the assumption is line with known properties of hyperbolic problems.

We also note that the approximation assumption does not require a positive Gevrey radius
independent of $k$: the exact equality
\eqref{E:IS_Gevrey_localized_equality} holds separately for each $k$,
and the only estimate used to pass to the Sobolev limit is
\eqref{E:IS_Sobolev_approximation_limit}.
For a nonnegative $\rho\in G_c^\sigma(\mathbb{R}^3)$ with integral
one and support in the unit ball, the choice
$q_{i,k}=\bar{Q}+\rho_{h_k}*(q_i-\bar{Q})$, where
$\rho_h(x)=h^{-3}\rho(x/h)$, preserves agreement on $\mathcal{O}_k$:
the convolution at $x\in\mathcal{O}_k$ only uses data in
$\mathcal{O}$. Moreover, the Fourier multiplier of this convolution
has absolute value at most one and converges pointwise to one, so
dominated convergence gives $q_{i,k}\rightarrow q_i$ in $H^m(\mathbb{R}^3)$ and
$\norm{q_{i,k}-\bar{Q}}_{H^m(\mathbb{R}^3)}\leq\norm{q_i-\bar{Q}}_{H^m(\mathbb{R}^3)}$.
These data properties do not establish membership in $\mathscr{A}$
or the common-time Gevrey solutions required in (ii).
In particular, (ii) requires the approximating solutions to satisfy
\eqref{eq:C1}--\eqref{eq:C4}; convolution is not asserted to preserve
these inequalities at degenerate states. Nevertheless, it is not difficult to see that for large classes of data such, such preservation does hold. Indeed, if the physical inequalities hold strictly with a uniform 
positive lower bound, then convolution with a sufficiently small parameter
preserves them at the level of the data, as follows. Replace every non-strict
inequality in \eqref{eq:C1}--\eqref{eq:C4} by a strict inequality,
and retain the nonvanishing and positivity assumptions of
Theorem~\ref{thm:causalitynodiffusion}. In state coordinates where
these quantities depend continuously on the state, denote the
resulting admissible set by $\mathcal{G}$. This set is open: the
angular parameters range over the compact square $[0,1]^2$, so
strict positivity in every direction has a positive minimum at
each fixed state, and continuity preserves that minimum under a
sufficiently small change of state. Suppose, more strongly, that
the closures of the ranges of $q_1,q_2$ lie in a compact set
$\mathcal{K}_0\subset\mathcal{G}$ and that
\begin{align*}
 \delta&=\operatorname{dist}(\mathcal{K}_0,\mathcal{G}^c)>0.
\end{align*}
For the convolutions above, the mean value formula gives
\begin{align*}
 \norm{q_{i,k}-q_i}_{L^\infty(\mathbb{R}^3)}
 &\leq h_k\norm{\nabla q_i}_{L^\infty(\mathbb{R}^3)}
       \int_{\mathbb{R}^3}|z|\rho(z)\,dz.
\end{align*}
The gradient norm is finite by $m>7/2$ and Sobolev embedding.
Consequently, any $k$ for which the right-hand side is less than
$\delta$ gives $q_{i,k}(x)\in\mathcal{G}$ for every $x$ and $i=1,2$.
We note that convexity of $\mathcal{G}$ is not required for this argument.

To state the strict inequalities quantitatively in terms of
\eqref{eq:C1}--\eqref{eq:C4}, let $F_1,\ldots,F_7$ denote,
respectively, the quantities
\begin{align*}
 \Delta(\mathcal{P}_3;\chi^2,\kappa^2),\quad
 \frac{\mathcal{A}_2}{3},\quad
 1-\frac{\mathcal{A}_2}{3},\quad
 \mathcal{A}_1,\quad \mathcal{A}_0,\quad
 1-\mathcal{A}_2+\mathcal{A}_1-\mathcal{A}_0,\quad
 1-\frac{2\mathcal{A}_2}{3}+\frac{\mathcal{A}_1}{3},
\end{align*}
where each $\mathcal{A}_j$ is evaluated at the state in question and
at $(\chi^2,\kappa^2)$. Uniform strictness means that there is a constant
$c_*>0$ such that
\begin{align*}
 F_\ell(q_i(x);\chi^2,\kappa^2)&\geq c_*
\end{align*}
for every $i=1,2$, $x\in\mathbb{R}^3$, $\ell=1,\ldots,7$, and
$(\chi^2,\kappa^2)\in[0,1]^2$.
On a fixed compact neighborhood of $\mathcal{K}_0$ contained in the
state domain where the coefficients are defined and continuous,
uniform continuity of these seven functions, jointly in the state
and angular variables, gives an $\varepsilon_0>0$ such that
$\norm{q_{i,k}-q_i}_{L^\infty}<\varepsilon_0$ implies
\begin{align*}
 \big|F_\ell(q_{i,k}(x);\chi^2,\kappa^2)
       -F_\ell(q_i(x);\chi^2,\kappa^2)\big|&<\frac{c_*}{2}.
\end{align*}
The estimate for $\norm{q_{i,k}-q_i}_{L^\infty}$ ensures this whenever
$h_k\max_{i=1,2}\norm{\nabla q_i}_{L^\infty}\int|z|\rho(z)\,dz
<\varepsilon_0$, with $\varepsilon_0$ also chosen so that the mollified
states remain in that neighborhood. Consequently,
\begin{align*}
 F_\ell(q_{i,k}(x);\chi^2,\kappa^2)&\geq\frac{c_*}{2}>0
\end{align*}
for every $i,x,\ell,\chi^2,\kappa^2$ in the stated ranges.
Thus the strict inequalities are preserved uniformly by sufficiently
small convolution; the positive lower bound may decrease.

On a compact spatial region, continuity and compactness give a common positive lower bound for the seven quantities $F_\ell$, provided they are strictly positive at every point and in every direction.
On all of $\mathbb{R}^3$, pointwise strictness alone does not supply a uniform positive lower bound.
For data in $\bar{Q}+H^m(\mathbb{R}^3)$, a sufficient additional
condition is $\bar{Q}\in\mathcal{G}$: the function $q_i-\bar{Q}$
is uniformly continuous and square-integrable, so it tends uniformly
to zero at spatial infinity. Thus strict admissibility of
every $q_i(x)$ and of $\bar{Q}$ makes the closure of each range a
compact subset of $\mathcal{G}$. If the admissibility conditions
also depend explicitly on the spacetime point, the positive lower bound and continuity bounds used in this argument must also be uniform in the spacetime point.

We further note that preservation under convolution is an initial-data statement.
To obtain the inequalities for the evolved approximations on the
whole interval, it is sufficient that the given solutions $Q_i$
remain uniformly inside the strict admissible set there: then
\eqref{E:IS_Sobolev_approximation_limit} and Sobolev embedding give
uniform convergence of $Q_{i,k}$ to $Q_i$, which preserves that
positive lower bound for all sufficiently large $k$. This argument presupposes
that the mollified data belong to $\mathscr{A}$ and that their
common-time solutions exist. 
The continuity
$Q_i\in C^0([0,T];\bar{Q}+H^m)$
and the lower bound $F_\ell(q_i(x);\chi^2,\kappa^2)\geq c_*>0$ give $F_\ell(Q_i(t,x);\chi^2,\kappa^2)\geq \frac{c_*}{2}>0$.
These sufficient conditions do not cover states at which a
characteristic inequality is saturated; in particular, strict
positivity of the discriminant excludes repeated roots of the
physical cubic. Our theorem includes those states by keeping
admissibility of the approximating solutions as an explicit
hypothesis.

The compact state and matrix bounds in (iii) are the uniform
coefficient and propagation-speed bounds used in the proof.
Bounds of this kind follow on a fixed compact state set from
uniform bounds on the background and coefficients, provided $B^0$
remains uniformly invertible. On the noncompact spatial domain,
these uniform bounds are additional information.

The coefficient condition \eqref{E:IS_Gevrey_coefficient_assumption}
is also explicit: spatial Gevrey regularity at each time alone is not
its stated hypothesis. To verify (iv) by differentiating the
coefficient formulas, one needs estimates for mixed time and space
derivatives of the solutions, constitutive functions, and background
coefficients, with the factorial growth in
\eqref{E:IS_Gevrey_coefficient_assumption} and bounds uniform in the
spatial variable. Smooth dependence in (iii) alone does not provide
these factorial derivative bounds. The constants in (iv) may,
however, depend on $k$, as expected when the initial Gevrey radii
shrink.

\section{Constraints propagation, proof of Theorem \ref{T:Constraint_propagation}}
\label{S:Constrain_propagation_proof}

We now give a precise formulation of Theorem \ref{T:Constraint_propagation} and then its proof. For this section, $C_b^m$ denotes the space of functions whose coordinate
derivatives of total order at most $m$ are continuous and bounded;
$C_b^\infty=\bigcap_{m\geq0}C_b^m$. Norms of vectors and matrices are
Euclidean norms and their induced operator norms. We use the extension
of \eqref{eq:propsystem} specified by
\begin{align}
\label{E:IS_Constraint_projector}
 \Delta^{\mu\nu}&=g^{\mu\nu}
              -\frac{u^\mu u^\nu}{u_\alpha u^\alpha}.
\end{align}
For a timelike $u$ with $E,\tau_\pi\neq0$, define
\begin{subequations}
\begin{align}
\label{E:IS_Constraint_D}
 \mathcal{D}^{\mu\nu}
 &=\frac{(2\eta+\lambda_{\pi\Pi}\Pi)\Delta^{\mu\nu}
                 +\frac{\tau_{\pi\pi}}{2}\pi^{(\mu\nu)}}{4},\\
\label{E:IS_Constraint_K}
 K^{\mu\nu}&=\frac{2\mathcal{D}^{\mu\nu}}{E\tau_\pi},\\
\label{E:IS_Constraint_H}
 H^{\mu\nu}&=u^\mu u^\nu-K^{\mu\nu}.
\end{align}
\end{subequations}
We write $H(\phi,\psi)=H^{\mu\nu}\phi_\mu\psi_\nu$ and
$K(\phi,\psi)=K^{\mu\nu}\phi_\mu\psi_\nu$.

\begin{theorem}[Propagation of constraints]
\label{T:IS_Propagation_constraints}
Let $\Sigma=\mathbb{T}^3$ or $\mathbb{R}^3$, let $s\geq4$ be an
integer, and let $T>0$. Suppose that $g$ is a Lorentzian metric on a
neighborhood of $[0,T]\times\Sigma$, with
$g,g^{-1}\in C_b^\infty([0,T]\times\Sigma)$ and
$-g^{00}\geq c_g>0$, and that $dt$ is positive on future-directed
timelike vectors. Suppose that the equation of state and the transport
coefficients are smooth on open neighborhoods of the compact closures
of their argument ranges.

Let $\Umb$ be a solution of \eqref{eq:propsystem}, with
\eqref{E:IS_Constraint_projector}, on $[0,T]\times\Sigma$.
Assume that $u$ is future-directed and that constants $c_u,a_*>0$ satisfy
\begin{align*}
 -u_\alpha u^\alpha&\geq c_u,\\
 |E\tau_\pi|&\geq a_*
\end{align*}
on $[0,T]\times\Sigma$. On $\mathbb{T}^3$, set $\overline{\Umb}=0$.
On $\mathbb{R}^3$, let $\overline{\Umb}\in
C_b^\infty([0,T]\times\mathbb{R}^3;\mathbb{R}^{23})$ satisfy
\eqref{Constraints}. Assume
\begin{align}
\label{E:IS_Constraint_solution_regularity}
 \Umb-\overline{\Umb}
 &\in C^0\big([0,T];H^s(\Sigma;\mathbb{R}^{23})\big)
       \cap C^1\big([0,T];H^{s-1}(\Sigma;\mathbb{R}^{23})\big).
\end{align}
Suppose that $\Umb(0,\cdot)$ satisfies \eqref{Constraints} and that
there are $h_0,k_0>0$ such that, for every $x\in\Sigma$,
\begin{subequations}
\label{E:IS_Constraint_initial_bounds}
\begin{align}
\label{E:IS_Constraint_initial_time_bound}
 H(0,x;dt,dt)&\geq h_0,\\
\label{E:IS_Constraint_initial_spatial_bound}
 K(0,x;\phi,\phi)&\geq k_0|\phi|^2
 \quad\text{for every covector $\phi$ with }
                       u^\alpha(0,x)\phi_\alpha=0.
\end{align}
\end{subequations}
Then there exists $T_c\in(0,T]$ such that $\Umb$ satisfies
\eqref{Constraints} on $[0,T_c]\times\Sigma$.

If \eqref{E:IS_Constraint_initial_bounds} holds with $(0,x)$ replaced
by every $(t,x)\in[0,T]\times\Sigma$ and with fixed positive constants
$h_0,k_0$, then $T_c$ can be taken equal to $T$.
The conclusion also holds with $T_c=T$ if $g$ and $\Umb$ are real
analytic on a neighborhood of $[0,T]\times\Sigma$.
\end{theorem}

\begin{proof}
\medskip
\noindent\textbf{Step 1: Symmetry of $\pi^{\mu\nu}$.}
We use\footnote{The proof also applies
when the extension uses $\Delta_*^{\mu\nu}=g^{\mu\nu}+u^\mu u^\nu$.
With $\Phi=u_\alpha u^\alpha+1$ and $\sigma=u_\alpha u^\alpha$, their
difference is $\Delta_*^{\mu\nu}-\Delta^{\mu\nu}
=\frac{\Phi u^\mu u^\nu}{\sigma}$. The coefficients in \eqref{eq:C_def} are
polynomial in $\Delta$. Replacing $\Delta$ by $\Delta_*$ therefore adds
$\Phi$ times functions of the given solution and its first derivatives
to \eqref{eq:phiprop} and \eqref{eq:Phi_Chi}. These terms enter the matrix
$N$ in \eqref{E:IS_Constraint_system}, not its principal matrices.} \eqref{E:IS_Constraint_projector}.
In \eqref{eq:C_def}, every term except
$(\delta_{\pi\pi}-\frac{\tau_{\pi\pi}}{3})\pi^{\mu\beta}g^\alpha{}_{\nu}$
is symmetric in $\mu,\beta$. Antisymmetrizing \eqref{eq:shearprop} gives
\begin{align}
\label{eq:pisymmetry}
 \tau_\pi u^\alpha\nabla_\alpha\pi^{[\mu\beta]}
 &=-\left[1+\left(\delta_{\pi\pi}
              -\frac{\tau_{\pi\pi}}{3}\right)
                     \nabla_\alpha u^\alpha\right]\pi^{[\mu\beta]}.
\end{align}
The Lorentzian Cauchy--Schwarz inequality\footnote{Given two timelike $X,Y$, we can write $Y=cX+W$ with $g(X,W)=0$. Then $g(X,Y)^2-g(X,X)g(Y,Y)=-g(X,X)g(W,W) > 0$ so that $|g(X,Y)| \geq \sqrt{-g(X,X)}\sqrt{-g(Y,Y)}$. Here we apply this inequality with $X=u$ and $Y=dt^\sharp$.} and the time orientation give
$u^0\geq\sqrt{c_gc_u}$. By
\eqref{E:IS_Constraint_solution_regularity} and Sobolev embedding,
$B_u=\sup_{0\leq t\leq T}\|u(t)\|_{C_b^1}<\infty$. Writing
$m_u=\sqrt{c_gc_u}$, direct computation gives gives, for $i,j=1,2,3$,
\begin{align*}
 \left|\frac{u^i}{u^0}\right|&\leq\frac{B_u}{m_u},\\
 \left|\partial_j\left(\frac{u^i}{u^0}\right)\right|
 &\leq\frac{B_u}{m_u}+\frac{B_u^2}{m_u^2}.
\end{align*}
The curves defined by
\begin{align*}
 \frac{d x^i}{d t}&=\frac{u^i(t,x(t))}{u^0(t,x(t))}
\end{align*}
therefore exist back to $t=0$ from every point of $[0,T]\times\Sigma$.
Indeed, the vector field $\frac{u^i}{u^0}$ and its spatial derivatives are bounded. Its integral curves therefore extend backward to $t=0$ since on $\mathbb{R}^3$, the bound on  $\frac{u^i}{u^0}$prevents escape to infinity in finite time, while on $\mathbb T^3$ continuation follows from compactness and the absence of a boundary.
Along each curve, \eqref{eq:pisymmetry}, expressed in
coordinate components, is a homogeneous linear ordinary differential
equation for $\pi^{[\mu\beta]}$. Uniqueness with zero initial value gives
$\pi^{[\mu\beta]}=0$ on $[0,T]\times\Sigma$. In particular,
$\pi^{(\mu\nu)}=\pi^{\mu\nu}$ in \eqref{E:IS_Constraint_D}.

\medskip
\noindent\textbf{Step 2: Equations for normalization, orthogonality, and trace.}
Define
\begin{align*}
 \sigma&=u_\alpha u^\alpha,\\
 \Phi&=\sigma+1,\\
 \Psi^\mu&=u_\alpha\pi^{\alpha\mu},\\
 \chi&=g_{\alpha\beta}\pi^{\alpha\beta}.
\end{align*}
Contracting \eqref{eq:momentumprop} with $u_\mu$, and using
$u_\mu\Delta^{\mu\alpha}=0$ and
$u_\mu\nabla_\alpha\pi^{\mu\alpha}
=\nabla_\alpha\Psi^\alpha-\pi^{\mu\alpha}\nabla_\alpha u_\mu$, gives
\begin{align}
\label{eq:phiprop}
 \frac{E}{2}u^\alpha\nabla_\alpha\Phi+\nabla_\alpha\Psi^\alpha
 &=\Phi\pi^{\mu\alpha}\nabla_\alpha u_\mu.
\end{align}
The contractions of \eqref{eq:C_def} needed for \eqref{eq:shearprop} are
\begin{subequations}
\begin{align}
\label{E:IS_Constraint_orthogonality_contraction}
\begin{split}
 u_\beta\mathcal{C}_\pi^{\mu\beta\alpha}{}_{\nu}
 &=2\mathcal{D}^{\mu\alpha}u_\nu
   -\tau_\pi u^\alpha
              (\sigma\pi^\mu{}_{\nu}+u^\mu\Psi_\nu)\\
 &\quad+\left(\delta_{\pi\pi}-\frac{\tau_{\pi\pi}}{3}\right)
                        \Psi^\mu\delta^\alpha{}_{\nu}
   +\frac{\tau_{\pi\pi}}{4}
          (\Delta^{\alpha\mu}\Psi_\nu+\Psi^\alpha\delta^\mu{}_{\nu}),
\end{split}\\
\label{E:IS_Constraint_trace_contraction}
 g_{\mu\beta}\mathcal{C}_\pi^{\mu\beta\alpha}{}_{\nu}
 &=\left(\delta_{\pi\pi}-\frac{\tau_{\pi\pi}}{3}\right)
                  \chi\delta^\alpha{}_{\nu}
   -\left(2\tau_\pi+\frac{\tau_{\pi\pi}}{2\sigma}\right)
                  u^\alpha\Psi_\nu.
\end{align}
\end{subequations}
For \eqref{E:IS_Constraint_trace_contraction}, the terms proportional to
$\eta+\frac{\lambda_{\pi\Pi}\Pi}{2}$ cancel because
$g_{\mu\beta}\Delta^{\mu\beta}=3$. The two terms proportional to
$\tau_{\pi\pi}$ that contain $\pi^\alpha{}_{\nu}$ cancel; their
uncancelled contribution is
$-\frac{\tau_{\pi\pi}u^\alpha\Psi_\nu}{2\sigma}$.
The product rule, \eqref{E:IS_Constraint_orthogonality_contraction}, 
\eqref{E:IS_Constraint_trace_contraction}, and
$u_\nu\nabla_\alpha u^\nu=\frac{\nabla_\alpha\Phi}{2}$ now give
\begin{subequations}
\label{eq:Phi_Chi}
\begin{align}
\label{E:IS_Constraint_orthogonality_equation}
\begin{split}
 \mathcal{D}^{\mu\alpha}\nabla_\alpha\Phi
       +\tau_\pi u^\alpha\nabla_\alpha\Psi^\mu
 &=-\frac{\tau_{\pi\pi}}{4}
       \left(\Psi^\alpha\nabla_\alpha u^\mu
             +\Delta^{\alpha\mu}\Psi_\nu\nabla_\alpha u^\nu\right)\\
 &\quad-\left[1+\left(\delta_{\pi\pi}
                   -\frac{\tau_{\pi\pi}}{3}\right)
                    \nabla_\alpha u^\alpha\right]\Psi^\mu\\
 &\quad+\tau_\pi u^\alpha
       \left(\Phi\pi^\mu{}_{\nu}\nabla_\alpha u^\nu
                         +u^\mu\Psi_\nu\nabla_\alpha u^\nu\right),
\end{split}\\
\label{E:IS_Constraint_trace_equation}
\begin{split}
 \tau_\pi u^\alpha\nabla_\alpha\chi
 &=-\left[1+\left(\delta_{\pi\pi}
                  -\frac{\tau_{\pi\pi}}{3}\right)
                   \nabla_\alpha u^\alpha\right]\chi\\
 &\quad+\left(2\tau_\pi+\frac{\tau_{\pi\pi}}{2\sigma}\right)
                  u^\alpha\Psi_\nu\nabla_\alpha u^\nu.
\end{split}
\end{align}
\end{subequations}
All coefficients in \eqref{eq:phiprop},
\eqref{E:IS_Constraint_orthogonality_equation}, and
\eqref{E:IS_Constraint_trace_equation} are
evaluated at the given solution. These equations form a homogeneous
linear system for $\mathbf{V}=(\Phi,\Psi^\mu,\chi)^T$.

\medskip
\noindent\textbf{Step 3: The characteristic determinant.}
In coordinates, the system is
\begin{align}
\label{E:IS_Constraint_system}
 M^\alpha\partial_\alpha\mathbf{V}+N\mathbf{V}&=0,
\end{align}
where $N$ includes the connection coefficients and the terms in
\eqref{eq:phiprop} and \eqref{eq:Phi_Chi} with no derivatives of
$\mathbf{V}$. Its principal matrices are
\begin{align}
\label{E:IS_Constraint_matrices}
 M^\alpha&=
 \begin{pmatrix}
  \frac{E}{2}u^\alpha&\delta^\alpha{}_{\nu}&0\\
  \mathcal{D}^{\mu\alpha}&\tau_\pi u^\alpha\delta^\mu{}_{\nu}&0^\mu\\
  0&0_\nu&\tau_\pi u^\alpha
 \end{pmatrix}.
\end{align}
For a covector $\phi$, set $M(\phi)=M^\alpha\phi_\alpha$.
With $K$ and $H$ from \eqref{E:IS_Constraint_K} and
\eqref{E:IS_Constraint_H}, for $z=u^\alpha\phi_\alpha\neq0$, elimination of the five diagonal
entries $\tau_\pi z$ gives
\begin{align}
\label{E:IS_Constraint_determinant}
\begin{split}
 \det M(\phi)
 &=(\tau_\pi z)^5\left(\frac{E}{2}z
             -\frac{\mathcal{D}^{\mu\nu}\phi_\mu\phi_\nu}{\tau_\pi z}\right)\\
 &=\frac{E\tau_\pi^5}{2}z^4H(\phi,\phi).
\end{split}
\end{align}
The expression $\frac{E\tau_\pi^5}{2}z^4H(\phi,\phi)$ is polynomial
in $\phi$, so \eqref{E:IS_Constraint_determinant} also holds when $z=0$.

\medskip
\noindent\textbf{Step 4: Positivity near the initial data.}
We now use the assumed bounds \eqref{E:IS_Constraint_initial_bounds}.
By \eqref{E:IS_Constraint_solution_regularity} and Sobolev embedding,
$u$ and $K$ are continuous in time with values in $L^\infty$.
The projection, in the Euclidean norm on covectors, onto
$\{\phi\,|\,u^\alpha\phi_\alpha=0\}$ has entries
\begin{align*}
 Q_{\alpha\beta}(u)
 &=\delta_{\alpha\beta}
       -\frac{u^\alpha u^\beta}{\sum_{\gamma=0}^3(u^\gamma)^2}.
\end{align*}
Here the indices label the entries of a Euclidean matrix and are not
raised or lowered with $g$. For $G=\sup_{[0,T]\times\Sigma}\|g\|$, the lower bound for
$-u_\alpha u^\alpha$ gives $\sum_{\gamma=0}^3(u^\gamma)^2\geq\frac{c_u}{G}$.
Consequently $Q(u(t))\rightarrow Q(u(0))$ in $L^\infty$ as
$t\rightarrow0$. Set
\begin{align*}
 e_Q(t)&=\sup_{0\leq r\leq t}\|Q(u(r))-Q(u(0))\|_{L^\infty(\Sigma)},\\
 e_K(t)&=\sup_{0\leq r\leq t}\|K(r)-K(0)\|_{L^\infty(\Sigma)},\\
 B_K&=\|K(0)\|_{L^\infty(\Sigma)}.
\end{align*}
For $|\phi|=1$ and $u^\alpha(t,x)\phi_\alpha=0$, put
$\phi_*=Q(u(0,x))\phi$. Then
\begin{align*}
 |\phi_*-\phi|&\leq e_Q(t),\\
 |\phi_*|&\geq1-e_Q(t),\\
 |\phi_*|&\leq1.
\end{align*}
When $e_Q(t)\leq1$, \eqref{E:IS_Constraint_initial_spatial_bound} gives
\begin{align}
\label{E:IS_Constraint_K_continuity_bound}
\begin{split}
 K(t,x;\phi,\phi)
 &\geq K(0,x;\phi_*,\phi_*)-e_K(t)-2B_Ke_Q(t)\\
 &\geq k_0-(2k_0+2B_K)e_Q(t)-e_K(t).
\end{split}
\end{align}
Since $e_Q(t),e_K(t)\rightarrow0$ and $H(t)\rightarrow H(0)$ in
$L^\infty$, there exists $T_c\in(0,T]$ such that
\begin{align*}
 e_Q(T_c)&\leq1,\\
 (2k_0+2B_K)e_Q(T_c)+e_K(T_c)&\leq\frac{k_0}{2},\\
 \sup_{0\leq t\leq T_c}\|H(t)-H(0)\|_{L^\infty(\Sigma)}&\leq\frac{h_0}{2}.
\end{align*}
For $0\leq t\leq T_c$, the initial bound
\eqref{E:IS_Constraint_initial_time_bound} and
\eqref{E:IS_Constraint_K_continuity_bound} imply
\begin{subequations}
\label{E:IS_Constraint_positive_bounds}
\begin{align}
 H(dt,dt)&\geq\frac{h_0}{2},\\
 K(\phi,\phi)&\geq\frac{k_0}{2}|\phi|^2
       \quad\text{when }u^\alpha\phi_\alpha=0.
\end{align}
\end{subequations}
The time $T_c$ depends on $h_0,k_0$, the bounds for $u,K$, and their
moduli of continuity in time in $L^\infty$.

These inequalities do not use the constraints at positive time.
The bounds \eqref{E:IS_Constraint_positive_bounds} show that $H$ has signature $(+---)$: its restriction to the
three-dimensional subspace $u^\alpha\phi_\alpha=0$ is negative
definite, and $H(dt,dt)>0$.

\medskip
\noindent\textbf{Step 5: The eigenvalues and eigenvectors.}
Fix a point in this interval, a covector $\xi$ with $H(\xi,\xi)>0$,
and a covector $\zeta$ independent of $\xi$. The equation
$H(\zeta+\beta\xi,\zeta+\beta\xi)=0$ has roots
\begin{align}
\label{E:IS_Constraint_quadratic_roots}
 \beta_\pm&=\frac{-H(\xi,\zeta)
       \pm\sqrt{H(\xi,\zeta)^2-H(\xi,\xi)H(\zeta,\zeta)}}{H(\xi,\xi)}.
\end{align}
Indeed, putting $\zeta_\perp=\zeta-\frac{H(\xi,\zeta)}{H(\xi,\xi)}\xi$ gives
\begin{align*}
 H(\xi,\zeta_\perp)&=0,\\
 H(\xi,\zeta)^2-H(\xi,\xi)H(\zeta,\zeta)
 &=-H(\xi,\xi)H(\zeta_\perp,\zeta_\perp)>0.
\end{align*}
Here $\zeta_\perp\neq0$, and the orthogonal complement of $\xi$ for
$H$ is negative definite. Thus $\beta_\pm$ are real and distinct.
The other root of \eqref{E:IS_Constraint_determinant} is
\begin{align}
\label{E:IS_Constraint_multiple_root}
 \beta_0&=-\frac{u^\alpha\zeta_\alpha}{u^\alpha\xi_\alpha}.
\end{align}
The denominator is nonzero: otherwise $H(\xi,\xi)=-K(\xi,\xi)<0$.
For $\phi_0=\zeta+\beta_0\xi$, we have
$H(\phi_0,\phi_0)=-K(\phi_0,\phi_0)<0$.
Consequently $\beta_-<\beta_0<\beta_+$, and
\eqref{E:IS_Constraint_determinant} gives algebraic multiplicity four
for $\beta_0$ and one for each of $\beta_\pm$.
It also gives $\det M(\xi)\neq0$.

At $\phi_0$, the equation $M(\phi_0)(a,b^\mu,c)^T=0$ reads
\begin{align*}
 \phi_{0\mu}b^\mu&=0,\\
 a\mathcal{D}^{\mu\alpha}\phi_{0\alpha}&=0.
\end{align*}
Since $K(\phi_0,\phi_0)>0$ and $E\tau_\pi\neq0$, the vector
$\mathcal{D}^{\mu\alpha}\phi_{0\alpha}$ is nonzero. Hence $a=0$,
$b$ belongs to a three-dimensional subspace, and $c$ is free. The
kernel therefore has dimension four. At $\beta_\pm$, put
$\phi_\pm=\zeta+\beta_\pm\xi$ and
$z_\pm=u^\alpha\phi_{\pm\alpha}$. Then $z_\pm\neq0$, and a kernel
vector is
\begin{align*}
 \mathbf{r}_\pm&=
 \begin{pmatrix}
  1\\
  -\dfrac{\mathcal{D}^{\mu\alpha}\phi_{\pm\alpha}}{\tau_\pi z_\pm}\\
  0
 \end{pmatrix}.
\end{align*}
Substitution into \eqref{E:IS_Constraint_matrices} gives
\begin{align*}
 \frac{E}{2}z_\pm
  -\frac{\mathcal{D}^{\mu\alpha}\phi_{\pm\mu}\phi_{\pm\alpha}}
         {\tau_\pi z_\pm}
 &=\frac{E}{2z_\pm}H(\phi_\pm,\phi_\pm)=0,\\
 \mathcal{D}^{\mu\alpha}\phi_{\pm\alpha}
   +\tau_\pi z_\pm\left(-\frac{\mathcal{D}^{\mu\alpha}\phi_{\pm\alpha}}
                                  {\tau_\pi z_\pm}\right)&=0.
\end{align*}
The trace amplitude is zero, so $M(\phi_\pm)\mathbf{r}_\pm=0$.
Thus $-M(\xi)^{-1}M(\zeta)$ has six linearly independent real
eigenvectors. The cubic $\mathcal{P}_3$ does not enter this calculation.

\medskip
\noindent\textbf{Step 6: Uniqueness and continuation.}
For $\xi=dt$ and $\zeta=k_j\,dx^j$, $|k|=1$, let
$A_c=-M(dt)^{-1}M(k_j\,dx^j)$. The three functions in
\eqref{E:IS_Constraint_quadratic_roots} and \eqref{E:IS_Constraint_multiple_root}
are smooth in the coefficients and $k$. Under \eqref{E:IS_Constraint_positive_bounds}, define
\begin{align*}
 P_j&=\prod_{\substack{\ell\in\{-,0,+\}\\\ell\neq j}}
              \frac{A_c-\beta_\ell I_6}{\beta_j-\beta_\ell},
              \qquad j\in\{-,0,+\},\\
 R_c&=\sum_{j\in\{-,0,+\}}P_j^TP_j.
\end{align*}
The diagonalizability proved in Step~5 gives
\begin{align*}
 R_cA_c&=A_c^TR_c,\\
 \frac{1}{3}|h|^2&\leq h^TR_ch
       =\sum_{j\in\{-,0,+\}}|P_jh|^2
       \leq\left(\sum_{j\in\{-,0,+\}}\|P_j\|^2\right)|h|^2,
       \qquad h\in\mathbb{R}^6.
\end{align*}
The lower bound follows from $h=P_-h+P_0h+P_+h$ and
Cauchy--Schwarz. Write $Z_c=(u,E,\tau_\pi,\mathcal{D})$,
and let $\mathscr{K}_c$ be the compact closure of its range, contained
in the set where $E\tau_\pi\neq0$, $H(dt,dt)>0$, and
$K(\phi,\phi)>0$ for nonzero $\phi$ with $u^\alpha\phi_\alpha=0$. The strict order $\beta_-<\beta_0<\beta_+$ gives
\begin{align*}
 \delta_c&=\min_{Z_c\in\mathscr{K}_c,\,|k|=1}
            \min\{\beta_0-\beta_-,\beta_+-\beta_0\}>0,\\
 L_c&=\max_{Z_c\in\mathscr{K}_c,\,|k|=1}\|A_c\|<\infty,\\
 \|P_j\|&\leq\left(\frac{2L_c}{\delta_c}\right)^2.
\end{align*}
The formulas for $\beta_j$ and $P_j$ are smooth on a neighborhood of
$\mathscr{K}_c\times S^2$. Thus, for each integer $m\geq0$,
\begin{align*}
 \max_{|\alpha|+|\gamma|\leq m}
 \sup_{Z_c\in\mathscr{K}_c,\,|k|=1}
 \|\partial_{Z_c}^{\alpha}\partial_k^{\gamma}R_c(Z_c,k)\|
 &\leq C_{m,\mathscr{K}_c}<\infty.
\end{align*}
Derivatives in $k$ are taken on an annulus containing $S^2$, using the
degree-zero extension of $R_c$. The assumption \eqref{E:IS_Constraint_solution_regularity}, the
bounds on $g$ and the constitutive functions, and Sobolev embedding give
$Z_c\in C^0([0,T_c];C_b^1(\Sigma))\cap C^1([0,T_c];C_b^0(\Sigma))$.
The matrix $(M^0)^{-1}N$ belongs to
$C^0([0,T_c];L^\infty)$. The bounds for the derivatives of $R_c$
therefore give the coefficient regularity used in
\cite[Remark 3.2]{Shao:2023psr} on $\mathbb{T}^3$ and
\cite[Sections 0.8 and 5.2]{TaylorPSOandPDE} on $\mathbb{R}^3$.
For the sign in the evolution equation, the identity
$R_c(-A_c)=(-A_c)^TR_c$ follows from $R_cA_c=A_c^TR_c$.

With $\mathbf{V}=(\Phi,\Psi^\mu,\chi)^T$ and $N$ from
\eqref{E:IS_Constraint_system}, these linear estimates give
\begin{align}
\label{E:IS_Constraint_uniqueness_estimate}
 \|\mathbf{V}(t)\|_{L^2(\Sigma)}^2
 &\leq C\|\mathbf{V}(0)\|_{L^2}(\Sigma)^2
       +C\int_0^t\|\mathbf{V}(r)\|_{L^2(\Sigma)}^2\,\mathrm{d}r.
\end{align}
The constant depends on the interval, the Sobolev bound for the given
solution, the coefficient bounds, $h_0,k_0,a_*,c_u$, and
$\|(M^0)^{-1}N\|_{L^\infty(\Sigma)}$.
On $\mathbb{R}^3$, the identities in \eqref{Constraints} for
$\overline{\Umb}$ express every component of $\mathbf{V}$ as a sum
of products containing at least one component of
$\Umb-\overline{\Umb}$. The other factors are bounded, so
$\mathbf{V}\in C^0([0,T_c];L^2(\Sigma))$. Since $\mathbf{V}(0)=0$, Gronwall's inequality applied to
\eqref{E:IS_Constraint_uniqueness_estimate} gives $\mathbf{V}=0$. Together with $\pi^{[\mu\nu]}=0$ from \eqref{eq:pisymmetry}, this proves
\eqref{Constraints} on $[0,T_c]\times\Sigma$.

If \eqref{E:IS_Constraint_initial_bounds} holds throughout
$[0,T]\times\Sigma$, the compactness and eigenvalue calculations in
Steps~5--6 hold on this whole interval, so
\eqref{E:IS_Constraint_uniqueness_estimate} gives
$\mathbf{V}=0$ for $0\leq t\leq T$.
If $g$ and $\Umb$ are real analytic, every component of
$\pi^{[\mu\nu]},\Phi,\Psi^\mu,\chi$ is real analytic and vanishes on
$(0,T_c)\times\Sigma$. The identity theorem for real-analytic functions
extends this vanishing to $(0,T)\times\Sigma$; continuity gives the
endpoint values. This proves the two assertions with $T_c=T$.
\end{proof}

\begin{corollary}
Suppose that all the hypotheses of
Theorem~\ref{T:IS_Propagation_constraints} hold, except that
\eqref{E:IS_Constraint_initial_bounds} is replaced by
\eqref{eq:H5}, imposed at $t=0$ for every $x\in\Sigma$ and
$(\chi^2,\kappa^2)\in[0,1]^2$.
If $\Sigma=\mathbb{R}^3$, assume in addition that there exists
$b_*>0$ such that
\begin{align*}
 y_0(0,x;\chi^2,\kappa^2)&\geq b_*
\end{align*}
for every $x\in\mathbb{R}^3$ and
$(\chi^2,\kappa^2)\in[0,1]^2$.
Then \eqref{E:IS_Constraint_initial_bounds} holds.
Consequently, there exists $T_c\in(0,T]$ such that $\Umb$
satisfies \eqref{Constraints} on $[0,T_c]\times\Sigma$.
If $g$ and $\Umb$ are real analytic on a neighborhood of
$[0,T]\times\Sigma$, then $T_c$ can be taken equal to $T$.
\end{corollary}
\begin{proof}
At $t=0$ the constraints hold. In the frame introduced in
Section~\ref{sec:EOM}, with $e_0=u$, the matrices of $K$ and $H$ are
\begin{align*}
 (K^{AB})&=\operatorname{diag}(0,c_1,c_2,c_3),\\
 (H^{AB})&=\operatorname{diag}(1,-c_1,-c_2,-c_3),\\
 c_a&=\frac{\eta+\frac{1}{2}\lambda_{\pi\Pi}\Pi
                       +\frac{1}{4}\tau_{\pi\pi}\Lambda_a}{E\tau_\pi},
 \qquad a=1,2,3.
\end{align*}
If \eqref{eq:H5} holds, then $0<c_a\leq1$. In particular,
for every timelike covector $\xi$,
\begin{align*}
 H(\xi,\xi)&\geq-g^{\mu\nu}\xi_\mu\xi_\nu>0,\\
 K(\phi,\phi)&>0
 \quad\text{if }\phi\neq0\text{ and }u^\alpha\phi_\alpha=0.
\end{align*}
These strict inequalities give \eqref{E:IS_Constraint_initial_bounds}
by compactness on the torus. On $\mathbb{R}^3$, a bound $y_0\geq b_*>0$
for all initial points and directions gives
$K(\phi,\phi)\geq b_*g^{\mu\nu}\phi_\mu\phi_\nu$ when
$u^\alpha\phi_\alpha=0$. The compact closure of the initial values of
$(g,g^{-1},u)$ consists of nondegenerate Lorentzian metrics and unit
timelike vectors. On this closure, the minimum of
$g^{\mu\nu}\phi_\mu\phi_\nu$ over $|\phi|=1$ and
$u^\alpha\phi_\alpha=0$ is a number $d_g>0$.
Thus $k_0=b_*d_g$ and $h_0=c_g$ give
\eqref{E:IS_Constraint_initial_bounds} on $\mathbb{R}^3$.
\end{proof}

\section{Strong hyperbolicity, proof of Theorem \ref{thm:nodiffusionstronghyperbolicity}}\label{App:hyperbolicity}

We begin stating precise versions of Definition~\ref{D:Strong_hyperbolicity} and
Theorem~\ref{thm:nodiffusionstronghyperbolicity}.

\begin{definition}[Pointwise strong hyperbolicity]
\label{D:IS_Precise_strong_hyperbolicity}
Let $(M,g)$ be a four-dimensional manifold with a continuous Lorentzian
metric, and let $\mathbb{A}^\alpha(x,Z)$ be real, continuous $N\times N$
principal coefficient matrices of a first-order system. For a point $x$,
a state $Z$, and a covector $\phi\in T_x^*M$, write
\begin{align*}
 \mathbb{A}(x,Z;\phi)&=\mathbb{A}^\alpha(x,Z)\phi_\alpha.
\end{align*}
The principal part is \emph{strongly hyperbolic at $(x,Z)$ with respect to}
a timelike covector $\xi\in T_x^*M$ if
\medskip
\begin{enumerate}[label=(H\Roman*),nosep]
\item $\mathbb{A}(x,Z;\xi)$ is invertible;
\item for every spacelike covector $\zeta\in T_x^*M$, the matrix\footnote{The minus sign is simply a convention that makes $\Lambda$ equal the matrix eivenvalues, i.e., $\mathbb{A}(\Lambda\xi+\zeta)\mathbf{r}=0$ is equivalent to $ -\mathbb{A}(\xi)^{-1}\mathbb{A}(\zeta)\mathbf{r}=\Lambda \mathbf{r}$.}
$-\mathbb{A}(x,Z;\xi)^{-1}\mathbb{A}(x,Z;\zeta)$ is diagonalizable over
$\mathbb{R}$.
\end{enumerate}
\medskip
Thus (HII) requires a basis of $\mathbb{R}^N$ consisting of solutions
$\mathbf r$ of
\begin{align*}
 \mathbb{A}(x,Z;\Lambda\xi+\zeta)\mathbf r&=0
\end{align*}
with real $\Lambda$. We say the system is strongly hyperbolic at $(x,Z)$
if such a timelike $\xi$ exists, and along a continuous field $\Umb$ if
this holds at every $(x,\Umb(x))$.
\end{definition}

\begin{theorem}[Precise strong-hyperbolicity statement]
\label{T:IS_Precise_strong_hyperbolicity}
Let $(M,g)$ be a time-oriented four-dimensional manifold with a continuous
Lorentzian metric. Let $\Umb=(u^\nu,\varepsilon,n,\Pi,\pi^{\nu0},
\pi^{\nu1},\pi^{\nu2},\pi^{\nu3})^T$ be a continuous field satisfying
\eqref{Constraints}, with $u$ future-directed. Suppose that the equation
of state $P=P(\varepsilon,n)$ is $C^1$ and that the transport coefficients
are continuous, real-valued functions, each defined on an open neighborhood
of the values of its arguments. Evaluate the principal matrices by
\eqref{eq:fluidprincipalpart} at $(x,\Umb(x))$.

Set $E=\varepsilon+P+\Pi$, let $\Lambda_1,\Lambda_2,\Lambda_3$ be the
eigenvalues of $\pi^\mu{}_{\nu}$ on $u^\perp$, and set $\Lambda_0=0$.
At every $x\in M$, assume $\tau_\Pi,\tau_\pi\neq0$ and
$E+\Lambda_A\neq0$ for $A=0,1,2,3$. With $\mathcal{P}_3$ and its
coefficients as in \eqref{eq:nodiffusiondet} and its discriminant as in
\eqref{supeq:discriminant}, assume the following inequalities for every
$\chi^2,\kappa^2\in[0,1]$:
\bml
\bea
\label{supeq:H1}
0& < &\Delta(\mathcal{P}_3;\chi^2,\kappa^2),\\
\label{supeq:H2}
0 &<& \frac{\mathcal{A}_2(\chi^2,\kappa^2)}{3} \leq 1,\quad 0< \mathcal{A}_1(\chi^2,\kappa^2),\quad 0 < \mathcal{A}_0(\chi^2,\kappa^2)\\
\label{supeq:H3}
0 &\leq & 1-\mathcal{A}_2(\chi^2,\kappa^2) + \mathcal{A}_1(\chi^2,\kappa^2) - \mathcal{A}_0(\chi^2,\kappa^2) ,\\
\label{supeq:H4}
0 &\leq & 1-\frac{2}{3}\mathcal{A}_2(\chi^2,\kappa^2) + \frac{1}{3}\mathcal{A}_1(\chi^2,\kappa^2) ,\\
\label{supeq:H5}
0 &< & \frac{\frac{1}{2} (2\eta + \lambda_{\pi\Pi}\Pi)+\frac{1}{4} \tau_{\pi\pi}\Lambda_{\hat{v}^2}}{\tau_\pi E} \leq 1.
\eea
\eml
In addition, at each $x$ assume 
\eqref{E:IS_Separation_polynomial} for all
$\chi^2,\kappa^2\in[0,1]$, or
\eqref{E:IS_Separation_transport}, with the quantities defined in
\eqref{E:IS_Separation_coefficients}. All these quantities are evaluated
at the same state $\Umb(x)$.

Then, at every $x\in M$, the  principal part \eqref{eq:fluidprincipalpart} is strongly
hyperbolic at $(x,\Umb(x))$ with respect to \emph{every} timelike covector
$\xi\in T_x^*M$, in the sense of
Definition~\ref{D:IS_Precise_strong_hyperbolicity}. The eigenvectors in
(HII) span all of $\mathbb{R}^{23}$. For the system without a baryon current,
the same statement holds for the $22\times22$ principal matrices obtained
by removing the baryon equation and unknown and setting $nP_n=0$ in the
characteristic coefficients and \eqref{E:IS_Separation_coefficients}.
\end{theorem}

The field $\Umb$ is not assumed to solve the equations. This statement
uses no spacetime derivatives of $g$ or $\Umb$. Additional regularity
and uniform bounds needed for the Sobolev estimates are treated in
Section~\ref{S:IS_Linear_estimates_hypotheses}.

\begin{proof}[Proof of Theorem~\ref{T:IS_Precise_strong_hyperbolicity} with baryon current]
Fix $x\in M$ and evaluate all coefficients at $(x,\Umb(x))$.
We suppress these arguments and write
$\mathbb{A}(\phi)=\mathbb{A}^\alpha\phi_\alpha$.

\medskip
\noindent\textbf{Step 1: Bounds on the characteristic roots.}
The condition \eqref{supeq:H5} is identical to $0<y_0\leq1$,

We use the identities
\eqref{supeq:A0root}--\eqref{supeq:A2root}. 
The condition $\Delta(\mathcal{P}_3)>0$ guarantees that these roots
are real and distinct.
Assume by contradiction that there exists some $y_i\leq0$.
Without loss of generality, we say that $i=3$, since the coefficients
are symmetric under exchange of indices.
If $y_3=0$, we already reach a contradiction since $\mathcal{A}_0>0$.
Consider the case where $y_3<0$.
This condition implies $y_1y_2<0$.
Also, since $\mathcal{A}_1>0$, we have
\begin{align*}
 (y_2+y_1)y_3&>-y_1y_2>0.
\end{align*}
Since $y_3<0$, this gives $y_1+y_2<0$.
Furthermore, $\mathcal{A}_2>0$, so
\begin{align*}
 0<\mathcal{A}_2&=y_1+y_2+y_3<y_3<0,
\end{align*}
which is a contradiction. Thus $y_1,y_2,y_3>0$.

For the upper bound, the conditions
\eqref{supeq:H2}--\eqref{supeq:H4} give, for every $t>0$,
\begin{align*}
 \mathcal{P}_3(1+t)
 &=t^3+(3-\mathcal{A}_2)t^2
 +(3-2\mathcal{A}_2+\mathcal{A}_1)t
 +1-\mathcal{A}_2+\mathcal{A}_1-\mathcal{A}_0>0.
\end{align*}
Indeed, all coefficients on the right are nonnegative and $t^3>0$.
There is consequently no root greater than one. We order the cubic
roots and obtain
\begin{align}
\label{E:IS_Strict_cubic_roots}
 0<y_1<y_2<y_3\leq1.
\end{align}
The inequalities hold for every $(\chi^2,\kappa^2)\in[0,1]^2$.

\medskip
\noindent\textbf{Step 2: Separation under either additional hypothesis.}
We use the frame introduced in Section~\ref{sec:EOM}, with $e_0=u$.
The timelike and spacelike covectors in
Definition~\ref{D:IS_Precise_strong_hyperbolicity} remain arbitrary.
For $v>0$, put $k_i=\frac{v_i}{v}$, $k=(k_1,k_2,k_3)^T$, and
\begin{align*}
 \boldsymbol{\pi}&=\operatorname{diag}(\Lambda_1,\Lambda_2,\Lambda_3),\\
 q&=k^T\boldsymbol{\pi}k=\Lambda_{\hat{v}^2},\\
 |k|^2&=1.
\end{align*}
In this step, vector and matrix norms are the Euclidean norm and its
induced operator norm in the spatial part of this frame.
Define a $3\times3$ matrix polynomial by
\begin{align}
\label{E:IS_Separation_matrix}
 \mathcal{F}(y,k)
 &=(EI_3+\boldsymbol{\pi})y-(a+bq)I_3-b\boldsymbol{\pi}
   -c\,kk^T-d\,(\boldsymbol{\pi}k)k^T
   -e\,k(\boldsymbol{\pi}k)^T.
\end{align}
Here $a,b,c,d,e$ are exactly the coefficients in
\eqref{E:IS_Separation_coefficients}.
We verify how \eqref{E:IS_Separation_matrix} occurs in the determinant.
By homogeneity, take $v=1$ in the calculation of
\eqref{E:IS_Four_matrix}. Contracting
\eqref{E:IS_Shear_coefficient} with $\phi_\alpha$ gives, for
$i,j\in\{1,2,3\}$ with no summation on these indices,
\begin{align*}
 \frac{\phi_\alpha\mathcal{C}_{\pi}^{\alpha i}{}_{j}}{\tau_\pi}
 &=\big(a+bq+b\Lambda_i-\Lambda_i z^2\big)\delta_{ij}
   +k_i k_j\left(\frac{a}{3}+d\Lambda_i-\frac{b}{3}\Lambda_j\right),\\
 \frac{\mathcal{C}_{\Pi j}}{\tau_\Pi}
 &=\left(\frac{\zeta+\delta_{\Pi\Pi}\Pi}{\tau_\Pi}
           +\frac{\lambda_{\Pi\pi}}{\tau_\Pi}\Lambda_j\right)k_j.
\end{align*}
Substitution in \eqref{E:IS_Four_matrix} therefore yields
\begin{align*}
 (z\mathcal{X})^i{}_{j}
 &=(E+\Lambda_i)z^2\delta_{ij}
   -(a+bq+b\Lambda_i)\delta_{ij}
   -k_i k_j(c+d\Lambda_i+e\Lambda_j)
   =\mathcal{F}_{ij}(z^2,k),\\
 (z\mathcal{X})^0{}_{0}&=E(z^2-y_0),\\
 (z\mathcal{X})^0{}_{j}&=0.
\end{align*}
Taking the determinant and using
\eqref{E:IS_full_coefficient} and
\eqref{E:IS_full_angular_determinant} gives the polynomial identity
\begin{align}
\label{E:IS_Separation_matrix_determinant}
 \det\mathcal{F}(y,k)
 &=\left(\prod_{i=1}^3(E+\Lambda_i)\right)
       \mathcal{P}_3(y;\chi^2,\kappa^2).
\end{align}
For fixed coefficients and $k$, both sides of
\eqref{E:IS_Separation_matrix_determinant} are polynomials in $y$.
The calculation establishes their equality for every $y>0$
with $y\neq y_0$. Their difference therefore has infinitely many
distinct roots and is the zero polynomial. Hence
\eqref{E:IS_Separation_matrix_determinant} holds for every
$y\in\mathbb{R}$.

Suppose \eqref{E:IS_Separation_transport} holds, and denote its
positive difference by
\begin{align}
\label{E:IS_Separation_transport_difference}
 \gamma_\pi
 &=m_\pi\frac{|a-bE|-|b|L_\pi}{|E|}
      -|c|-(|d|+|e|)L_\pi>0.
\end{align}
Since $|q|\leq L_\pi$, equation
\eqref{E:IS_Extra_squared_speed} gives
\begin{subequations}
\label{E:IS_Separation_norm_estimates}
\begin{align}
 |y_0-b|&\geq\frac{|a-bE|-|b|L_\pi}{|E|},\\
 |\boldsymbol{\pi}h|&\geq m_\pi|h|
 \qquad(h\in\mathbb{R}^3),\\
 \big\|c\,kk^T+d\,(\boldsymbol{\pi}k)k^T
                 +e\,k(\boldsymbol{\pi}k)^T\big\|
 &\leq |c|+(|d|+|e|)L_\pi.
\end{align}
\end{subequations}
The identity
\begin{align}
\label{E:IS_Separation_matrix_at_extra_root}
 \mathcal{F}(y_0,k)
 &=(y_0-b)\boldsymbol{\pi}
        -c\,kk^T-d\,(\boldsymbol{\pi}k)k^T
        -e\,k(\boldsymbol{\pi}k)^T
\end{align}
and \eqref{E:IS_Separation_norm_estimates} give
\begin{align}
\label{E:IS_Separation_matrix_lower_bound}
 |\mathcal{F}(y_0,k)h|&\geq\gamma_\pi|h|
 \qquad h\in\mathbb{R}^3.
\end{align}
In particular, the matrix $\mathcal{F}(y_0,k)$ in
\eqref{E:IS_Separation_matrix_at_extra_root} is invertible. Thus \eqref{E:IS_Separation_matrix_determinant} and
$E+\Lambda_i\neq0$ imply
\eqref{E:IS_Separation_polynomial} for every $\chi^2,\kappa^2$.
Under either alternative, therefore,
\begin{align}
\label{E:IS_Separated_four_roots}
 \mathcal{P}_3(y_0;\chi^2,\kappa^2)
 &=\prod_{j=1}^3(y_0-y_j)\neq0.
\end{align}
Together with \eqref{E:IS_Strict_cubic_roots}, this proves that the four
numbers $y_0,y_1,y_2,y_3$ are positive and pairwise distinct.

\medskip
\noindent\textbf{Step 3: The characteristic eigenvalues.}
Fix an arbitrary timelike covector $\xi$ and an arbitrary spacelike covector
$\zeta$ and write
\begin{align*}
 \Xi_\alpha(\Lambda)&=\Lambda\xi_\alpha+\zeta_\alpha,\\
 z(\Lambda)&=u^\alpha\Xi_\alpha(\Lambda),\\
 v^\mu(\Lambda)&=\Delta^{\mu\alpha}\Xi_\alpha(\Lambda),\\
 v(\Lambda)^2&=g_{\mu\nu}v^\mu(\Lambda)v^\nu(\Lambda).
\end{align*}
To verify (HII), we must show that the solutions of
\begin{align}
\label{E:IS_Characteristic_eigenvalue}
 \det\big(\mathbb{A}^\alpha\Xi_\alpha(\Lambda)\big)&=0
\end{align}
are real and that their (right) eigenvectors form a basis of
$\mathbb{R}^{23}$.

Since $\xi$ and $u$ are timelike, $u^\alpha\xi_\alpha\neq0$.
The root supplied by $z=0$ is
\begin{align}
\label{E:IS_Advective_eigenvalue}
 \Lambda^{(z)}&=-\frac{u^\alpha\zeta_\alpha}{u^\beta\xi_\beta}
 \in\mathbb{R}.
\end{align}
For the other roots we consider the equations
\begin{align}
\label{E:IS_Characteristic_scalar_equations}
 z(\Lambda)^2-v(\Lambda)^2
 y_a\big(\chi(\Lambda)^2,\kappa(\Lambda)^2\big)&=0,
 \qquad a\in\{0,1,2,3\}.
\end{align}
For $v(\Lambda)>0$, the quantities $\chi(\Lambda)^2$ and
$\kappa(\Lambda)^2$ in
\eqref{E:IS_Characteristic_scalar_equations} satisfy
\eqref{eq:normalizedcharacteristics} with
\begin{align*}
 \hat{v}_i(\Lambda)
 &=\frac{\Lambda\xi_\alpha e_i^\alpha
                   +\zeta_\alpha e_i^\alpha}{v(\Lambda)},
 \qquad i=1,2,3.
\end{align*}
The ratios $\hat{v}_i(\Lambda)^2$, and hence the values
$y_a\big(\chi(\Lambda)^2,\kappa(\Lambda)^2\big)$, may vary
with $\Lambda$.

With $S$ and $H_3$ from \eqref{E:IS_quadratic_polynomial} and
\eqref{E:IS_H_angular_expression}, define
\begin{align}
\label{E:IS_Degree_eight_polynomial}
 G_8(\phi)&=S(\phi)H_3(\phi).
\end{align}
By \eqref{E:IS_Degree_eight_polynomial}, $G_8$ is a homogeneous polynomial
of degree eight. For $v>0$ it satisfies
\begin{align}
\label{E:IS_Degree_eight_product}
 G_8(\phi)&=\prod_{a=0}^{3}\big(z^2-v^2y_a(\chi^2,\kappa^2)\big),
\end{align}
and \eqref{E:IS_H_time_axis} gives $G_8(\phi)=z^8$ at $v=0$.
Equations \eqref{E:IS_Degree_eight_product} and
\eqref{E:IS_full_angular_determinant} give the full determinant:
\begin{align}
\label{E:IS_Full_determinant_for_proof}
 \det\mathbb{A}(\phi)&=c_0z^{15}G_8(\phi),
\end{align}
where $c_0\neq0$ is given by \eqref{E:IS_full_coefficient}.
For any timelike $\phi$ with $v>0$ and any $a\in\{0,1,2,3\}$,
\begin{align}
\label{E:IS_Timelike_factor_bound}
 z^2-v^2y_a&\geq z^2-v^2
 =-g^{\alpha\beta}\phi_\alpha\phi_\beta>0.
\end{align}
At $v=0$, one has $G_8(\phi)=z^8>0$ for every nonzero timelike $\phi$.
Thus $G_8(\xi)\neq0$, and
\eqref{E:IS_Full_determinant_for_proof} proves (HI).
Moreover, $\Lambda\mapsto G_8(\Xi(\Lambda))$ has degree exactly eight.

For $v(\Lambda)>0$, denote the left-hand side of
\eqref{E:IS_Characteristic_scalar_equations} by $f_a(\Lambda)$.
If $v(\Lambda)=0$, define $f_a(\Lambda)=z(\Lambda)^2$.
The coefficients of $\mathcal{P}_3$ are polynomial functions of $k$ by
\eqref{E:IS_Separation_matrix_determinant}, since $\prod_{i=1}^{3}(E+\Lambda_i)\neq0$.
Its distinct ordered roots $y_1,y_2,y_3$ are continuous by the
implicit function theorem, and $y_0$ is continuous by
\eqref{E:IS_Extra_squared_speed}.
Since $0<v^2y_a\leq v^2$, the definitions at $v=0$ make every $f_a$
continuous on $\mathbb{R}$.

The covectors $\xi,\zeta$ are linearly independent, so
$\Xi(\Lambda)\neq0$ for every $\Lambda$.
At $\Lambda=\Lambda^{(z)}$, this gives $v>0$ and hence
\begin{align*}
 f_a(\Lambda^{(z)})&=-v(\Lambda^{(z)})^2
 y_a\big(\chi(\Lambda^{(z)})^2,\kappa(\Lambda^{(z)})^2\big)<0.
\end{align*}
On the other hand, the estimate
$f_a(\Lambda)\geq-g^{\alpha\beta}\Xi_\alpha(\Lambda)
\Xi_\beta(\Lambda)$, used in \eqref{E:IS_Timelike_factor_bound}, gives
\begin{align*}
 \liminf_{\Lambda\rightarrow\pm\infty}
 \frac{f_a(\Lambda)}{\Lambda^2}
 &\geq-g^{\alpha\beta}\xi_\alpha\xi_\beta>0.
\end{align*}
The intermediate value theorem therefore supplies a root
$\Lambda_-^{(a)}<\Lambda^{(z)}$ and a root
$\Lambda_+^{(a)}>\Lambda^{(z)}$ of $f_a$ for each $a\in\{0,1,2,3\}$.
At such a root, $v$ cannot vanish: otherwise $f_a=z^2=0$ would imply
$\Xi=0$. If $f_a$ and $f_b$ vanished at the same $\Lambda$, with $a\neq b$,
then $y_a=y_b$  at $(\chi(\Lambda)^2,\kappa(\Lambda)^2)$,
 contrary to
\eqref{E:IS_Strict_cubic_roots} and \eqref{E:IS_Separated_four_roots}.
We have therefore obtained eight distinct real roots of the degree-eight
polynomial $G_8(\Xi(\Lambda))$. These are all of its roots, and each is
simple. None equals $\Lambda^{(z)}$.
The proof only uses $\zeta\notin\operatorname{span}\{\xi\}$, so this
conclusion also holds for any real covector independent of $\xi$.

\medskip
\noindent\textbf{Step 4: Completeness of the eigenvectors.}
Consider the case $z^{(z)}=0$, where
$\Xi_\alpha^{(z)}=\Lambda^{(z)}\xi_\alpha+\zeta_\alpha$.
We also write $(v^{(z)})^\mu=\Delta^{\mu\alpha}\Xi_\alpha^{(z)}$
and $w_\nu^{(z)}=\pi_\nu{}^\alpha\Xi_\alpha^{(z)}$.

At this covector,
\begin{align*}
 G_8(\Xi^{(z)})&=v^8\prod_{a=0}^{3}y_a>0,
\end{align*}
so \eqref{E:IS_Full_determinant_for_proof} gives algebraic multiplicity
exactly fifteen for $\Lambda^{(z)}$.
In the matrix calculations, $\sim$ means equality up to invertible row
and column operations; these preserve rank and nullity.

After moving the energy and baryon rows before the momentum rows,
the principal part may be expressed as
\begin{align}
\label{E:IS_Zero_root_matrix}
\mathbb{A}^\alpha\Xi_\alpha^{(z)} &\sim
\begin{bmatrix}
E\Xi_\nu^{(z)} + w_\nu^{(z)}& 0& 0 & 0 & 0_\nu& 0_\nu& 0_\nu& 0_\nu\\
n\Xi_\nu^{(z)}& 0 & 0& 0 & 0_\nu& 0_\nu& 0_\nu& 0_\nu\\
-u^\mu w_\nu^{(z)}& (v^\mu)^{(z)} P_\varepsilon& (v^\mu)^{(z)} P_n& (v^\mu)^{(z)}& \tensor{g}{^\mu_\nu}\Xi_0^{(z)}& \tensor{g}{^\mu_\nu}\Xi_1^{(z)}& \tensor{g}{^\mu_\nu}\Xi_2^{(z)}& \tensor{g}{^\mu_\nu}\Xi_3^{(z)}\\
(\zeta + \delta_{\Pi\Pi}\Pi)\Xi_\nu^{(z)} + \lambda_{\Pi\pi} w_\nu^{(z)}& 0& 0& 0& 0_\nu& 0_\nu& 0_\nu& 0_\nu\\
(\tensor{\mathcal{C}}{_\pi^{0\mu}_\nu})^{(z)}& 0^{\mu}& 0^{\mu}& 0^{\mu}& \tensor{0}{^\mu_\nu}& \tensor{0}{^\mu_\nu}& \tensor{0}{^\mu_\nu}& \tensor{0}{^\mu_\nu}\\
(\tensor{\mathcal{C}}{_\pi^{1\mu}_\nu})^{(z)}& 0^{\mu}& 0^{\mu}& 0^{\mu}& \tensor{0}{^\mu_\nu}& \tensor{0}{^\mu_\nu}& \tensor{0}{^\mu_\nu}& \tensor{0}{^\mu_\nu}\\
(\tensor{\mathcal{C}}{_\pi^{2\mu}_\nu})^{(z)}& 0^{\mu}& 0^{\mu}& 0^{\mu}& \tensor{0}{^\mu_\nu}& \tensor{0}{^\mu_\nu}& \tensor{0}{^\mu_\nu}& \tensor{0}{^\mu_\nu}\\
(\tensor{\mathcal{C}}{_\pi^{3\mu}_\nu})^{(z)}& 0^{\mu}& 0^{\mu}& 0^{\mu}& \tensor{0}{^\mu_\nu}& \tensor{0}{^\mu_\nu}& \tensor{0}{^\mu_\nu}& \tensor{0}{^\mu_\nu}
\end{bmatrix},
\end{align}
where
\begin{align*}
(\tensor{\mathcal{C}}{_\pi^{\beta\mu}_\nu})^{(z)} &= (2\eta + \lambda_{\pi\Pi}\Pi)\left[\frac{1}{2}((v^\mu)^{(z)}\tensor{g}{^\beta_\nu} + \tensor{g}{^\mu_\nu}(v^\beta)^{(z)}) - \frac{1}{3}\Delta^{\mu\beta}v_\nu^{(z)}\right] \notag\\
&\quad + \left(\delta_{\pi\pi}-\frac{\tau_{\pi\pi}}{3}\right)\pi^{\mu\beta}\Xi_\nu^{(z)} + \frac{\tau_{\pi\pi}}{2}\left((v^{(\beta})^{(z)}\tensor{\pi}{^{\mu)}_\nu} + (w^{(\mu})^{(z)}\tensor{g}{^{\beta)}_\nu}\right)-\frac{\tau_{\pi\pi}}{3}\Delta^{\mu\beta}w_\nu^{(z)}.
\end{align*}

The covector $\Xi^{(z)}$ is nonzero. In the frame used in Step~2,
$\Xi_0^{(z)}=u^\alpha\Xi_\alpha^{(z)}=0$.
After relabeling the spatial basis vectors, we may assume
$\Xi_3^{(z)}\neq0$.
For a field component $X$, let $C_X$ denote its column in
\eqref{E:IS_Zero_root_matrix}.
The columns $C_{\pi^{\mu3}}$, $\mu=0,1,2,3$, have entries only in the
four momentum rows and form $\Xi_3^{(z)}I_4$ there.
They remain unchanged in the operations
\begin{subequations}
\label{E:IS_Column_operations}
\begin{align}
 C_{\pi^{\mu b}}&\longmapsto C_{\pi^{\mu b}}
 -\frac{\Xi_b^{(z)}}{\Xi_3^{(z)}}C_{\pi^{\mu3}},
 \qquad \mu\in\{0,1,2,3\},\quad b\in\{0,1,2\},\\
 C_{\varepsilon}&\longmapsto C_{\varepsilon}
 -\sum_{\mu=0}^{3}\frac{(v^\mu)^{(z)}P_{\varepsilon}}
 {\Xi_3^{(z)}}C_{\pi^{\mu3}},\\
 C_n&\longmapsto C_n
 -\sum_{\mu=0}^{3}\frac{(v^\mu)^{(z)}P_n}
 {\Xi_3^{(z)}}C_{\pi^{\mu3}},\\
 C_{\Pi}&\longmapsto C_{\Pi}
 -\sum_{\mu=0}^{3}\frac{(v^\mu)^{(z)}}
 {\Xi_3^{(z)}}C_{\pi^{\mu3}},\\
 C_{u^\nu}&\longmapsto C_{u^\nu}
 +\sum_{\mu=0}^{3}\frac{u^\mu w_\nu^{(z)}}
 {\Xi_3^{(z)}}C_{\pi^{\mu3}},
 \qquad \nu\in\{0,1,2,3\}.
\end{align}
\end{subequations}
Every operation is invertible because it adds multiples of unchanged
columns to another column. The resulting matrix is

\begin{align}
\label{E:IS_Zero_root_reduced}
\mathbb{A}^\alpha\Xi_\alpha^{(z)}&\sim
\begin{bmatrix}
E\Xi_\nu^{(z)} + w_\nu^{(z)}& 0& 0 & 0 & 0_\nu& 0_\nu& 0_\nu& 0_\nu\\
n\Xi_\nu^{(z)}& 0 & 0& 0 & 0_\nu& 0_\nu& 0_\nu& 0_\nu\\
\tensor{0}{^\mu_\nu}& 0^\mu& 0^\mu& 0^\mu& \tensor{0}{^\mu_\nu}& \tensor{0}{^\mu_\nu}& \tensor{0}{^\mu_\nu}& \tensor{g}{^\mu_\nu}\Xi_3^{(z)}\\
(\zeta + \delta_{\Pi\Pi}\Pi)\Xi_\nu^{(z)} + \lambda_{\Pi\pi} w_\nu^{(z)}& 0& 0& 0& 0_\nu& 0_\nu& 0_\nu& 0_\nu\\
(\tensor{\mathcal{C}}{_\pi^{0\mu}_\nu})^{(z)}& 0^{\mu}& 0^{\mu}& 0^{\mu}& \tensor{0}{^\mu_\nu}& \tensor{0}{^\mu_\nu}& \tensor{0}{^\mu_\nu}& \tensor{0}{^\mu_\nu}\\
(\tensor{\mathcal{C}}{_\pi^{1\mu}_\nu})^{(z)}& 0^{\mu}& 0^{\mu}& 0^{\mu}& \tensor{0}{^\mu_\nu}& \tensor{0}{^\mu_\nu}& \tensor{0}{^\mu_\nu}& \tensor{0}{^\mu_\nu}\\
(\tensor{\mathcal{C}}{_\pi^{2\mu}_\nu})^{(z)}& 0^{\mu}& 0^{\mu}& 0^{\mu}& \tensor{0}{^\mu_\nu}& \tensor{0}{^\mu_\nu}& \tensor{0}{^\mu_\nu}& \tensor{0}{^\mu_\nu}\\
(\tensor{\mathcal{C}}{_\pi^{3\mu}_\nu})^{(z)}& 0^{\mu}& 0^{\mu}& 0^{\mu}& \tensor{0}{^\mu_\nu}& \tensor{0}{^\mu_\nu}& \tensor{0}{^\mu_\nu}& \tensor{0}{^\mu_\nu}
\end{bmatrix},
\end{align}

In \eqref{E:IS_Zero_root_reduced}, the three columns indexed by
$\varepsilon,n,\Pi$ and the twelve columns indexed by
$\pi^{\mu b}$, $\mu\in\{0,1,2,3\}$, $b\in\{0,1,2\}$, are zero.
Thus
\begin{subequations}
\label{E:IS_Advective_kernel_dimension}
\begin{align}
 \operatorname{rank}\mathbb{A}(\Xi^{(z)})&\leq8,\\
 \dim\ker\mathbb{A}(\Xi^{(z)})&\geq15.
\end{align}
\end{subequations}
More explicitly, let $T$ be the invertible column transformation
specified by \eqref{E:IS_Column_operations}. The fifteen coordinate vectors
for the zero columns lie in the kernel of the transformed matrix.
Their images under $T$ are fifteen independent vectors in
$\ker\mathbb{A}(\Xi^{(z)})$.
This kernel is the eigenspace of
$-\mathbb{A}(\xi)^{-1}\mathbb{A}(\zeta)$ at $\Lambda^{(z)}$.
Since its dimension cannot exceed the algebraic multiplicity fifteen,
there are exactly fifteen linearly independent eigenvectors for this root.

By Step~3, the other eight eigenvalues $\Lambda_\pm^{(a)}$,
$a\in\{0,1,2,3\}$, are real, simple, and distinct from $\Lambda^{(z)}$.
Each therefore has a one-dimensional real eigenspace.
Eigenspaces for distinct eigenvalues are linearly independent, giving
$15+8=23$ independent eigenvectors in $\mathbb{R}^{23}$.
This verifies (HII) and concludes the proof for the system with a baryon
current.
The regularity and uniform bounds needed for the Sobolev estimates
are established in Section~\ref{S:IS_Linear_estimates_hypotheses}.

\end{proof}

\begin{proof}[Proof of Theorem~\ref{T:IS_Precise_strong_hyperbolicity} without baryon current] For the theory without a baryon current, remove the baryon equation and
its unknown, and set $nP_n=0$ in \eqref{E:IS_Separation_coefficients} and
in the characteristic coefficients. The determinant has the same factors
$S,H_3$ and a factor $z^{14}$ instead of $z^{15}$. Thus Steps~1--3 apply
with either additional hypothesis, giving eight simple eigenvalues.
At a nonzero covector with $z=0$, the eighteen non-momentum rows involve
only the four velocity columns. Together with the four momentum rows,
their rank is at most eight. The $22\times22$ matrix therefore has nullity
at least fourteen. Since the the eigenvalue determined by $z(\Lambda)=0$, i.e., $\Lambda^{(z)}$,
has algebraic multiplicity
fourteen, its eigenspace has dimension exactly fourteen. Adding the eight
simple eigenspaces gives $14+8=22$ independent eigenvectors. This concludes
the proof of 
Theorem~\ref{T:IS_Precise_strong_hyperbolicity}, and hence of
Theorem~\ref{thm:nodiffusionstronghyperbolicity}.
\end{proof}

\subsection{Dynamic Metric}
Here, we elaborate more on the claim that Theorem~\ref{thm:nodiffusionstronghyperbolicity} is compatible with Einstein's equations. We sketch the argument as follows. Let $K_{\rho\mu\nu}$ be the gradients of the spacetime metric:
\be
K_{\rho\mu\nu} = \partial_\rho g_{\mu\nu}.
\ee
Let $K_{\rho,\frak{A}} \in\{K_{\rho\mu\nu}:\mu\leq \nu\}$ be the $10$ independent components of $K_{\rho\mu\nu}$ for fixed $\rho$. Here, instead of evolving the $10$ independent components of the metric for the mixed-order system in Eq.~\eqref{eq:mixedorderEOM}, we instead let
\be
\textbf{\textrm{U}}_{G} = \begin{bmatrix}
g_{\frak{A}}\\
K_{0,\frak{A}}\\
K_{1,\frak{A}}\\
K_{2,\frak{A}}\\
K_{3,\frak{A}}
\end{bmatrix}\in\mathbb{R}^{50}
\ee
be the $50$ independent degrees of freedom for the gravity portion. Einstein's equations in the Harmonic gauge may then be recast as
\bml
\bea
K_{0\mu\nu} &=& \partial_0 g_{\mu\nu}\\
g^{ij}\partial_0K_{i\mu\nu} &=& g^{ij}\partial_i K_{0\mu\nu},\\
-g^{00}\partial_0K_{0\mu\nu} &=& 2g^{0j}\partial_jK_{0\mu\nu} + g^{ij}\partial_jK_{i\mu\nu} + \mathcal{O}(g,K_\rho),
\eea
\eml
We remark that this system of equations is more general than Einstein's equations, as it includes $40$ extra degrees of freedoms, one must supply the initial data hypersurface $\Sigma$ with
\be
\mathring{K}_{\rho\mu\nu} = \partial_\rho \mathring{g}_{\mu\nu}
\ee
such that the relationship between the degrees of freedom propagate throughout their time of existence. The total solution vector is then
\be
\textbf{\textrm{U}} = \textbf{\textrm{U}}_M \oplus \textbf{\textrm{U}}_G \in \mathbb{R}^{73}.
\ee
We then may cast Einstein's equations along with the matter equations as a quasilinear system of effective order 1:
\be
\mathbb{A}^\alpha\partial_\alpha\textbf{\textrm{U}} + \mathbb{B} = \boldsymbol{0}_{73}
\ee
The total principal part is written in block form as
\be
\label{eq:blocks}
\mathbb{A}^\alpha =
\begin{bmatrix}
\mathbb{A}_M^\alpha & \boldsymbol{0}_{50\times 23}\\
\boldsymbol{0}_{23\times 50} & \mathbb{A}_G^\alpha
\end{bmatrix}.
\ee
The coefficients of the principal part for the gravity sector then takes the form
\bml
\label{eq:gravityPPcoeff}
\bea
\mathbb{A}_G^0 &=&
\begin{bmatrix}
 I_{10}& \boldsymbol{0}_{10\times 10}& \boldsymbol{0}_{10\times 10}& \boldsymbol{0}_{10\times 10}& \boldsymbol{0}_{10\times 10}\\
\boldsymbol{0}_{10\times 10}& -g^{00} I_{10}& \boldsymbol{0}_{10\times 10}& \boldsymbol{0}_{10\times 10}&\boldsymbol{0}_{10\times 10}\\
\boldsymbol{0}_{10\times 10}& \boldsymbol{0}_{10\times 10}& g^{11} I_{10}& g^{12} I_{10}& g^{13} I_{10}\\
\boldsymbol{0}_{10\times 10}& \boldsymbol{0}_{10\times 10}& g^{21} I_{10}& g^{22} I_{10}& g^{23} I_{10}\\
\boldsymbol{0}_{10\times 10}& \boldsymbol{0}_{10\times 10}& g^{31} I_{10}& g^{00} I_{10}& g^{33} I_{10}
\end{bmatrix},\\
\mathbb{A}_G^j &=&
\begin{bmatrix}
\boldsymbol{0}_{10\times 10}& \boldsymbol{0}_{10\times 10}& \boldsymbol{0}_{10\times 10}& \boldsymbol{0}_{10\times 10}& \boldsymbol{0}_{10\times 10}\\
\boldsymbol{0}_{10\times 10}& 2g^{0j} I_{10}& g^{1j} I_{10}& g^{2j} I_{10}& g^{3j} I_{10}\\
\boldsymbol{0}_{10\times 10}& g^{j1} I_{10}& \boldsymbol{0}_{10\times 10}& \boldsymbol{0}_{10\times 10}& \boldsymbol{0}_{10\times 10}\\
\boldsymbol{0}_{10\times 10}& g^{j2} I_{10}& \boldsymbol{0}_{10\times 10}& \boldsymbol{0}_{10\times 10}& \boldsymbol{0}_{10\times 10}\\
\boldsymbol{0}_{10\times 10}& g^{j3} I_{10}& \boldsymbol{0}_{10\times 10}& \boldsymbol{0}_{10\times 10}& \boldsymbol{0}_{10\times 10}\\
\end{bmatrix}.
\eea
\eml
We proved that the matter sector was strongly hyperbolic, and in \cite{Fischer_Marsden_Einstein_FOSH_1972}, it was shown that this first-order casting of Einstein's equations above was symmetric hyperbolic, as long as the metric remains Lorentzian. These two results immediately satisfy (HI), since as in the mixed-order case, we know
\be
\det(\mathbb{A}^\alpha\xi_\alpha) = \det(\mathbb{A}_M^\alpha\xi_\alpha)\det(\mathbb{A}_G^{\alpha}\xi_\alpha),
\ee
and both the matter and gravity sector have non-vanishing characteristic determinant for timelike $\xi^\mu$ by definition. Furthermore, the principal part for Einstein's equations are completely decoupled from the matter equations, meaning that proving the existence of an eigenbasis for each the matter system and gravity system individually gives an eigenbasis for the entire matter + gravity system of equations immediately. We proved the existence of $23$ linearly-independent right eigenvectors in Theorem~\ref{thm:nodiffusionstronghyperbolicity}, and \cite{Fischer_Marsden_Einstein_FOSH_1972} provides the remaining $50$ due to the symmetric hyperbolicity of the system. It is clear by the block diagonal form of Eq.~\eqref{eq:blocks} that these sets are independent of one another. Thus, there exists $73$ linearly independent eigenvectors for the fluid + Einstein system, thus verifying (HII) and proving strong hyperbolicity.

With this calculation at hand, we can then give a precise formulation and proof of corollary \ref{cor:SHeinstein}
similarly to what was done in Section \ref{App:hyperbolicity}. We can also give a precise formulation and proof corollary
\ref{C:LWPcoupled} as done in Section \ref{App:LWPSchauder}. Both formulations are routine extensions of the detailed argument presented in Sections \ref{App:hyperbolicity} and \ref{App:LWPSchauder} and will not be given here for the sake of brevity. The only point that needs clarification in generalizing Theorem \ref{T:IS_Precise_Sobolev} to the coupled system is that we do not include the assumptions on the background metric made in that theorem, since the metric will now be one of the unknowns. In particular, analyticity of the metric is not required.

\section{Local well-posedness in Sobolev spaces, proof of Theorem \ref{T:LWP}}\label{App:LWPSchauder}

In this appendix, we give a precise Sobolev formulation of Theorem
\ref{T:LWP}. We consider a prescribed metric on
$[0,T_g]\times\Sigma$, where $\Sigma=\mathbb{T}^3$ or
$\mathbb{R}^3$. 

We use the vector $\Umb$ from \eqref{eq:FirstOrderQPDE} and restore
explicitly the dependence of the coefficients\footnote{This is done primarily for the sake of generality, but it might be useful in specific applications.} on $y=(t,x)$:
\begin{align}
\label{E:IS_Precise_coordinate_system}
 \mathbb{A}^0(y,\Umb)\partial_t\Umb
 +\sum_{j=1}^{3}\mathbb{A}^j(y,\Umb)\partial_j\Umb
 +\mathbb{B}(y,\Umb)\Umb&=0.
\end{align}
Thus, the connection terms of the prescribed metric are included
in $\mathbb{B}$. For an open set $\mathcal{O}\subset\mathbb{R}^{23}$,
let $\mathcal{C}(y)$ consist of the states in $\mathcal{O}$ that
satisfy \eqref{Constraints} with the metric $g(y)$ and whose
four-velocity is future-directed. In all Sobolev norms, we use
coordinate components, take spatial derivatives, and sum the
squared norms of the components. The spaces on $\mathbb{T}^3$ are
periodic Sobolev spaces.

Appendix~\ref{App:hyperbolicity} proves diagonalizability and separation
of the distinct eigenvalues. Section~\ref{S:IS_Linear_estimates_hypotheses}
verifies the hypotheses of the linear estimates used here.

\begin{theorem}[Local well-posedness in Sobolev spaces]
\label{T:IS_Precise_Sobolev}
Let $\Sigma=\mathbb{T}^3$ or $\mathbb{R}^3$, let $s\geq4$ be an
integer, and let $T_g>0$. Suppose that $g$ is a prescribed
Lorentzian metric on an open neighborhood of
$[0,T_g]\times\Sigma$, that the hypersurfaces $\{t=\mathrm{constant}\}$
are spacelike, and that $dt$ is positive on future-directed
timelike vectors. Fix a real-analytic orthonormal frame,
independent of the fluid solution, as in
Section \ref{S:IS_domain_of_dependence}. Assume that the components of the metric and frame are real
analytic. Write $\mathsf{E}=(E_{\hat{a}}{}^\mu)$ for the frame
matrix. For every spacetime multi-index $\alpha$, assume that
there is $G_\alpha<\infty$ such that
\begin{align}
\label{E:IS_Precise_metric_frame_bounds}
\begin{split}
 \sup_{y\in[0,T_g]\times\Sigma}
 \big(&\norm{\partial_y^\alpha g(y)}
       +\norm{\partial_y^\alpha g^{-1}(y)}
 \\
      &+\norm{\partial_y^\alpha\mathsf{E}(y)}
       +\norm{\partial_y^\alpha\mathsf{E}^{-1}(y)}\big)
 \leq G_\alpha.
\end{split}
\end{align}
Assume also that
\begin{align}
\label{E:IS_Precise_time_covector_bound}
 -g^{00}(y)&\geq c_g>0,
 \qquad y\in[0,T_g]\times\Sigma.
\end{align}

Suppose that the equation of state and the transport coefficients
are real-analytic functions of their arguments on an open set
containing all the states under consideration. Let
$\mathcal{O}\subset\mathbb{R}^{23}$ be an open set on which the
matrices in \eqref{E:IS_Precise_coordinate_system} are defined and
real analytic jointly in $(y,Z)$, and on which
$\mathbb{A}^0(y,Z)$ is invertible. For every compact set
$\mathcal{L}\subset\mathcal{O}$ and every pair of multi-indices
$\alpha,\beta$, assume that there is a finite constant
$C_{\alpha,\beta,\mathcal{L}}$ such that
\begin{align}
\label{E:IS_Precise_coefficient_bounds}
\begin{split}
 \sup_{\substack{y\in[0,T_g]\times\Sigma\\Z\in\mathcal{L}}}
 \bigg(&\sum_{\mu=0}^{3}
  \norm{\partial_y^\alpha\partial_Z^\beta\mathbb{A}^\mu(y,Z)}
  +\norm{\partial_y^\alpha\partial_Z^\beta\mathbb{B}(y,Z)}
 \\
 &+\norm{\partial_y^\alpha\partial_Z^\beta
              (\mathbb{A}^0(y,Z))^{-1}}\bigg)
 \leq C_{\alpha,\beta,\mathcal{L}}.
\end{split}
\end{align}
All matrix norms in this theorem are Euclidean operator norms.

At every $y\in[0,T_g]\times\Sigma$ and every
$Z\in\mathcal{C}(y)$, assume that the hypotheses and inequalities
of Theorem \ref{thm:nodiffusionstronghyperbolicity} hold, and that
$E+\Lambda_A>0$ for $A=0,1,2,3$. 

Fix compact sets $\mathcal{K},\mathcal{K}_1$ with
\begin{align*}
 \mathcal{K}\subset\operatorname{int}\mathcal{K}_1
 \subset\mathcal{K}_1\subset\mathcal{O}.
\end{align*}
When $\Sigma=\mathbb{R}^3$, assume there is $c_*>0$, independent
of $y,Z,\chi^2,\kappa^2$, such that
\begin{align}
\label{E:IS_Schauder_positive_lower_bound}
\begin{split}
 \min\bigg\{&|\tau_\Pi|,|\tau_\pi|,
       \min_{0\leq A\leq3}(E+\Lambda_A),
       \Delta(\mathcal{P}_3;\chi^2,\kappa^2),
       \frac{\mathcal{A}_2}{3},\mathcal{A}_1,\mathcal{A}_0,
 \\
       &\frac{\frac{1}{2}(2\eta+\lambda_{\pi\Pi}\Pi)
                    +\frac{1}{4}\tau_{\pi\pi}\Lambda_{\hat{v}^2}}
                   {\tau_\pi E}\bigg\}\geq c_*
\end{split}
\end{align}
for $y\in[0,T_g]\times\mathbb{R}^3$,
$Z\in\mathcal{K}_1\cap\mathcal{C}(y)$, and
$(\chi^2,\kappa^2)\in[0,1]^2$. All quantities in \eqref{E:IS_Schauder_positive_lower_bound}
are evaluated at $(y,Z,\chi^2,\kappa^2)$.

For $\Sigma=\mathbb{R}^3$, require one of the two alternatives on
$y\in[0,T_g]\times\mathbb{R}^3$ and
$Z\in\mathcal{K}_1\cap\mathcal{C}(y)$: either there exists
$\rho_*>0$, independent of $y,Z,\chi^2,\kappa^2$, such that
\begin{align}
\label{E:IS_Precise_separation_lower_bound}
 \big|\mathcal{P}_3(y_0;\chi^2,\kappa^2)\big|
 &\geq\rho_*
 \quad\text{for all }(\chi^2,\kappa^2)\in[0,1]^2,
\end{align}
or there exist $\gamma_*>0$ and $D_*<\infty$ such that
\begin{align*}
 \gamma_\pi&\geq\gamma_*,\\
 \left|\prod_{i=1}^3(E+\Lambda_i)\right|&\leq D_*,
\end{align*}
where $\gamma_\pi$ is defined in
\eqref{E:IS_Separation_transport_difference}.

When $\Sigma=\mathbb{R}^3$, prescribe a real-analytic reference
solution $\overline{\Umb}$ of
\eqref{E:IS_Precise_coordinate_system} on a neighborhood of
$[0,T_g]\times\mathbb{R}^3$ such that
$\overline{\Umb}(y)\in\mathcal{K}\cap\mathcal{C}(y)$.
For every spacetime multi-index $\alpha$, assume that
\begin{align}
\label{E:IS_Precise_reference_bounds}
 \sup_{y\in[0,T_g]\times\mathbb{R}^3}
 |\partial_y^\alpha\overline{\Umb}(y)|
 &\leq R_\alpha<\infty.
\end{align}
When $\Sigma=\mathbb{T}^3$,
set $\overline{\Umb}=0$.

For $M>0$, prescribe initial data $\mathring{\Umb}$ such that
\begin{subequations}
\label{E:IS_Precise_initial_data}
\begin{align}
 \mathring{\Umb}(x)&\in\mathcal{K}\cap\mathcal{C}(0,x),
 \qquad x\in\Sigma,
 \\
 W_0&=\mathring{\Umb}-\overline{\Umb}(0,\cdot)
       \in H^s(\Sigma;\mathbb{R}^{23}),
 \\
 \norm{W_0}_{H^s(\Sigma)}&\leq M.
\end{align}
\end{subequations}
Then there are $T\in(0,T_g]$ and a unique solution
$\Umb$ of \eqref{E:IS_Precise_coordinate_system} with initial
value $\mathring{\Umb}$, such that
\begin{subequations}
\label{E:IS_Precise_solution}
\begin{align}
 W:=\Umb-\overline{\Umb}
 &\in C^0\big([0,T];H^s(\Sigma;\mathbb{R}^{23})\big)
       \cap C^1\big([0,T];H^{s-1}(\Sigma;\mathbb{R}^{23})\big),
 \\
 \Umb(t,x)&\in\mathcal{K}_1\cap\mathcal{C}(t,x),
 \qquad (t,x)\in[0,T]\times\Sigma.
\end{align}
\end{subequations}
The solution also satisfies \eqref{eq:conservationlaw} and
\eqref{eq:supplementalIS}.

The time $T$ can be chosen independently of the individual initial
datum in \eqref{E:IS_Precise_initial_data}. It may depend on
$s,M,\mathcal{K},\mathcal{K}_1,T_g,c_g$, the fixed constitutive
functions, the constants in
\eqref{E:IS_Precise_metric_frame_bounds}, the constants in
\eqref{E:IS_Precise_coefficient_bounds}, the lower bound in
\eqref{E:IS_Schauder_positive_lower_bound},
the positive lower bound in \eqref{E:IS_Precise_separation_lower_bound},
and, on $\mathbb{R}^3$, the fixed reference solution and
\eqref{E:IS_Precise_reference_bounds}.
There is a finite constant $C_M$, with the same allowed dependence,
such that
\begin{align}
\label{E:IS_Precise_uniform_estimate}
 \sup_{0\leq t\leq T}
 \big(\norm{W(t)}_{H^s(\Sigma)}
       +\norm{\partial_tW(t)}_{H^{s-1}(\Sigma)}\big)
 &\leq C_M.
\end{align}
Uniqueness is among solutions satisfying \eqref{Constraints} and
the regularity in \eqref{E:IS_Precise_solution}, for the same prescribed
metric, whose ranges remain in a compact subset of $\mathcal{O}$.

For two initial data satisfying \eqref{E:IS_Precise_initial_data}
and the same reference solution, their solutions satisfy
\begin{align}
\label{E:IS_Precise_difference_estimate}
 \sup_{0\leq t\leq T}
 \norm{\Umb(t)-\widetilde{\Umb}(t)}_{H^{s-1}(\Sigma)}
 &\leq C_M\norm{\mathring{\Umb}-\mathring{\widetilde{\Umb}}}
                    _{H^{s-1}(\Sigma)},
\end{align}
with $C_M$ enlarged if necessary, but with the same allowed
dependence. 

Moreover, let $W_{0,q}\rightarrow W_0$ in $H^s(\Sigma)$,
with every initial datum satisfying
\eqref{E:IS_Precise_initial_data}. On the common interval $[0,T]$,
\begin{align}
\label{E:IS_Precise_continuous_dependence}
\begin{split}
 \sup_{0\leq t\leq T}
 \big(&\norm{\Umb_q(t)-\Umb(t)}_{H^s(\Sigma)}
 \\
 &+\norm{\partial_t\Umb_q(t)-\partial_t\Umb(t)}
                         _{H^{s-1}(\Sigma)}\big)
 \longrightarrow0.
\end{split}
\end{align}

\end{theorem}

\begin{remark}
On $\mathbb{T}^3$, compactness, the nonvanishing assumptions
of Theorem~\ref{thm:nodiffusionstronghyperbolicity}, and
\eqref{eq:H1}, \eqref{eq:H2}, \eqref{eq:H5}, and
\eqref{E:IS_Separation_polynomial} give
\eqref{E:IS_Schauder_positive_lower_bound} and
\eqref{E:IS_Precise_separation_lower_bound} on
$\{(y,Z)\,|\,y\in[0,T_g]\times\mathbb{T}^3,
Z\in\mathcal{K}_1\cap\mathcal{C}(y)\}$
for every $(\chi^2,\kappa^2)\in[0,1]^2$.
Under \eqref{E:IS_Separation_transport}, the estimate
\eqref{E:IS_Transport_quantitative_separation} gives
$\rho_*=\frac{\gamma_*^3}{D_*}$ on $\mathbb{R}^3$.
These constants enter the dependence of $T$ and $C_M$.

For the linear estimates and continuous-dependence argument we use
\cite[Proposition 3.1, Theorem 2.5, and Sections 4.4 and 4.6]{Shao:2023psr}
on $\mathbb{T}^3$, and the estimates of
\cite[Sections 0.8 and 5.2]{TaylorPSOandPDE} on $\mathbb{R}^3$.
\end{remark}

\subsection{Verification of the hypotheses of the linear estimates}
\label{S:IS_Linear_estimates_hypotheses}

We verify the hypotheses used in
\cite[Proposition 3.1 and Assumptions (A3)--(A4)]{Shao:2023psr} and
\cite[Sections 0.8 and 5.2]{TaylorPSOandPDE}.
The algebraic conclusions are those of
Theorem~\ref{T:IS_Precise_strong_hyperbolicity}; the bounds on the
coefficients and their derivatives are
\eqref{E:IS_Precise_metric_frame_bounds},
\eqref{E:IS_Precise_coefficient_bounds},
\eqref{E:IS_Schauder_positive_lower_bound}, and
\eqref{E:IS_Precise_separation_lower_bound}.

Let $\mathscr{K}$ be a compact set of states and metric components
satisfying \eqref{Constraints} and the hypotheses of
Theorem~\ref{T:IS_Precise_strong_hyperbolicity}, in coordinates
satisfying \eqref{E:IS_Precise_time_covector_bound}.
In \eqref{E:IS_Separation_matrix}, the matrix $\boldsymbol{\pi}$
may be written in any orthonormal basis of $u^\perp$, with $k$ the
components of a unit vector in that basis. To verify this, use $r$
for the polynomial variable and write $\mathcal{F}_{\boldsymbol{\pi}}(r,k)$
for the expression in \eqref{E:IS_Separation_matrix}.
For an orthogonal $3\times3$ matrix $O$, substitution gives
\begin{subequations}
\label{E:IS_Cubic_change_of_basis}
\begin{align}
 (Ok)^T(O\boldsymbol{\pi}O^T)(Ok)&=k^T\boldsymbol{\pi}k,\\
 \mathcal{F}_{O\boldsymbol{\pi}O^T}(r,Ok)
 &=O\mathcal{F}_{\boldsymbol{\pi}}(r,k)O^T.
\end{align}
\end{subequations}
Thus \eqref{E:IS_Cubic_change_of_basis} and
\eqref{E:IS_Separation_matrix_determinant} give
\begin{align}
\label{E:IS_Cubic_coefficients_in_components}
 \mathcal{P}_3(r)
 &=\frac{\det\mathcal{F}_{\boldsymbol{\pi}}(r,k)}
          {\det(EI_3+\boldsymbol{\pi})},
\end{align}
where the arguments other than $r$ have been suppressed.
The denominator is nonzero by the assumptions $E+\Lambda_i\neq0$.
In a local orthonormal frame of $u^\perp$ depending smoothly on $g,u$,
\eqref{E:IS_Cubic_coefficients_in_components} expresses the polynomial
coefficients continuously in the tensor components and $k$.
The identities \eqref{E:IS_Cubic_change_of_basis} show that these
expressions agree on overlaps of such frames.
The quantity $y_0=\frac{a+bk^T\boldsymbol{\pi}k}{E}$ is also continuous
by \eqref{E:IS_Extra_squared_speed}.
Therefore \eqref{E:IS_Separation_polynomial} and compactness give
\begin{align}
\label{E:IS_Uniform_polynomial_separation}
 \big|\mathcal{P}_3(y_0;\chi^2,\kappa^2)\big|
 &\geq\rho_*>0
 \quad\text{on }\mathscr{K}\text{ for all }(\chi^2,\kappa^2)\in[0,1]^2.
\end{align}
For the $\mathbb{R}^3$ application,
\eqref{E:IS_Precise_separation_lower_bound} supplies the constant
$\rho_*$. Under alternative \eqref{E:IS_Separation_transport},
\eqref{E:IS_Separation_matrix_lower_bound} gives the estimate
\begin{align}
\label{E:IS_Transport_quantitative_separation}
 \big|\mathcal{P}_3(y_0;\chi^2,\kappa^2)\big|
 &\geq\frac{\gamma_\pi^3}
               {\left|\prod_{i=1}^3(E+\Lambda_i)\right|}.
\end{align}
Indeed, \eqref{E:IS_Separation_matrix_lower_bound} bounds all three
singular values of $\mathcal{F}(y_0,k)$ from below by $\gamma_\pi$.
Taking the determinant in \eqref{E:IS_Separation_matrix_determinant}
then gives \eqref{E:IS_Transport_quantitative_separation}.
The bounds $\gamma_\pi\geq\gamma_*>0$ and
$\left|\prod_{i=1}^3(E+\Lambda_i)\right|\leq D_*$ in
Theorem~\ref{T:IS_Precise_Sobolev} therefore give
$\rho_*=\frac{\gamma_*^3}{D_*}$.

By \eqref{eq:H1}, \eqref{eq:H2}, and \eqref{eq:H5}, continuity
and compactness give constants $\delta_*,a_*,b_*>0$ such that
\begin{subequations}
\label{E:IS_Cubic_and_extra_root_lower_bounds}
\begin{align}
\label{E:IS_Cubic_discriminant_lower_bound}
 \Delta(\mathcal{P}_3;\chi^2,\kappa^2)&\geq\delta_*,\\
\label{E:IS_Cubic_constant_lower_bound}
 \mathcal{A}_0(\chi^2,\kappa^2)&\geq a_*,\\
\label{E:IS_Extra_root_lower_bound}
 y_0(\chi^2,\kappa^2)&\geq b_*
\end{align}
\end{subequations}
for every state in $\mathscr{K}$ and every
$(\chi^2,\kappa^2)\in[0,1]^2$; the discriminant is defined in
\eqref{eq:discriminant}. On $\mathbb{R}^3$,
\eqref{E:IS_Schauder_positive_lower_bound} supplies
\eqref{E:IS_Cubic_and_extra_root_lower_bounds}, with
$\delta_*=a_*=b_*=c_*$ on the set specified in that hypothesis.
Since $0<y_a\leq1$ for $a\in\{0,1,2,3\}$ by \eqref{eq:H5}
and \eqref{E:IS_Strict_cubic_roots}, we obtain
\begin{subequations}
\begin{align}
\label{E:IS_Cubic_root_separation_bound}
 |y_i-y_j|&\geq\sqrt{\delta_*},\qquad 1\leq i<j\leq3,\\
\label{E:IS_Extra_cubic_root_separation_bound}
 |y_0-y_j|&\geq\rho_*,\qquad 1\leq j\leq3,\\
\label{E:IS_Cubic_root_positive_bound}
 y_j&\geq a_*,\qquad 1\leq j\leq3.
\end{align}
\end{subequations}
For \eqref{E:IS_Cubic_root_separation_bound}, the discriminant
identity and \eqref{E:IS_Cubic_discriminant_lower_bound} give
\begin{align*}
 \delta_*&\leq\prod_{1\leq p<q\leq3}(y_p-y_q)^2
 \leq|y_i-y_j|^2.
\end{align*}
For \eqref{E:IS_Extra_cubic_root_separation_bound},
\eqref{E:IS_Uniform_polynomial_separation} and
\eqref{E:IS_Separated_four_roots} give
\begin{align*}
 \rho_*&\leq\left|\prod_{\ell=1}^{3}(y_0-y_\ell)\right|
 \leq|y_0-y_j|.
\end{align*}
Finally, \eqref{E:IS_Cubic_constant_lower_bound} and
\eqref{supeq:A0root} give
\begin{align*}
 a_*&\leq\mathcal{A}_0=y_1y_2y_3\leq y_j,
\end{align*}
which proves \eqref{E:IS_Cubic_root_positive_bound}.

For smooth fields $g,\Umb$ taking values in $\mathscr{K}$, write
$\mathbb{A}^\mu(y)=\mathbb{A}^\mu(y,\Umb(y))$ and define
\begin{align}
\label{E:IS_LWP_principal_matrix}
 A(y,k)&=\big(\mathbb{A}^0(y)\big)^{-1}
              \sum_{i=1}^{3}\mathbb{A}^i(y)k_i,
 \qquad k\in\mathbb{R}^3,\quad |k|=1.
\end{align}
Here $y=(t,x)$ is the spacetime point and $k_i$ are coordinate
components of a spatial covector. Theorem~\ref{T:IS_Precise_strong_hyperbolicity}
gives nine distinct real eigenvalues of \eqref{E:IS_LWP_principal_matrix},
one of algebraic and geometric multiplicity fifteen and eight simple ones.
Order them as $\lambda_1<\cdots<\lambda_9$ and let $j_0$ denote the
index of the eigenvalue of multiplicity fifteen. By
\eqref{E:IS_Advective_eigenvalue},
\begin{align}
\label{E:IS_LWP_multiple_eigenvalue}
 \lambda_{j_0}(y,k)&=\frac{\sum_{i=1}^{3}u^i(y)k_i}{u^0(y)}.
\end{align}
The sign in \eqref{E:IS_LWP_multiple_eigenvalue} follows from
\eqref{E:IS_LWP_principal_matrix}: the parameter in
\eqref{E:IS_Characteristic_eigenvalue} with $\xi=dt$ and
$\zeta=\sum_{i=1}^{3}k_i\,\mathrm{d}x^i$ is $-\lambda$.
The index $j_0$ is locally constant because the nine eigenvalues do
not meet. The eight eigenvalues $\lambda_j$ with $j\neq j_0$ are
smooth functions of the matrix entries and $k$ by the implicit
function theorem. Define
\begin{subequations}
\label{E:IS_Regular_spectral_matrices}
\begin{align}
\label{E:IS_LWP_spectral_projections}
 P_j&=\prod_{\substack{1\leq\ell\leq9\\\ell\neq j}}
        \frac{A-\lambda_\ell I_{23}}{\lambda_j-\lambda_\ell},
 \qquad j=1,\ldots,9,\\
\label{E:IS_LWP_symmetrizing_matrix}
 R&=\sum_{j=1}^{9}P_j^TP_j.
\end{align}
\end{subequations}
The definitions \eqref{E:IS_Regular_spectral_matrices} and
diagonalizability give $\sum_{j=1}^{9}P_j=I_{23}$ and
$P_jA=\lambda_jP_j$. Hence
\begin{subequations}
\label{E:IS_Regular_spectral_bounds}
\begin{align}
\label{E:IS_LWP_matrix_symmetry}
 RA&=A^TR,\\
\label{E:IS_LWP_matrix_positivity}
 \frac{1}{9}|h|^2
 &\leq h^TRh=\sum_{j=1}^{9}|P_jh|^2
 \leq\left(\sum_{j=1}^{9}\|P_j\|^2\right)|h|^2,
 \qquad h\in\mathbb{R}^{23}.
\end{align}
\end{subequations}
The definition \eqref{E:IS_LWP_symmetrizing_matrix} and
Cauchy--Schwarz give the lower bound in
\eqref{E:IS_LWP_matrix_positivity} from
$h=\sum_{j=1}^{9}P_jh$.
Theorem~\ref{T:IS_Precise_strong_hyperbolicity},
\eqref{E:IS_Precise_time_covector_bound}, and compactness give
constants $L_A,L_0<\infty$ and $\delta_A>0$ such that
\begin{subequations}
\label{E:IS_Uniform_spectral_bounds}
\begin{align}
\label{E:IS_Principal_matrix_bound}
 \|A\|&\leq L_A,\\
\label{E:IS_Time_matrix_inverse_bound}
 \|(\mathbb{A}^0)^{-1}\|&\leq L_0,\\
\label{E:IS_Eigenvalue_difference_bound}
 |\lambda_i-\lambda_j|&\geq\delta_A,\qquad i\neq j,\\
\label{E:IS_Spectral_projection_bound}
 \|P_j\|&\leq\left(\frac{2L_A}{\delta_A}\right)^8,
 \qquad j=1,\ldots,9,
\end{align}
\end{subequations}
for every state in $\mathscr{K}$ and every $|k|=1$.
The constants in \eqref{E:IS_Principal_matrix_bound}--\eqref{E:IS_Eigenvalue_difference_bound}
are obtained by taking extrema on $\mathscr{K}\times S^2$.
Since $|\lambda_j|\leq\|A\|$, substitution of
\eqref{E:IS_Principal_matrix_bound} and
\eqref{E:IS_Eigenvalue_difference_bound} in
\eqref{E:IS_LWP_spectral_projections} gives
\eqref{E:IS_Spectral_projection_bound}.

Let $\mathcal{Z}(y)$ consist of the entries of
$\mathbb{A}^0(y),\ldots,\mathbb{A}^3(y)$ and $u(y)$.
On a time interval $[0,T]$, for $N\geq1$, set
\begin{align}
\label{E:IS_Matrix_entry_derivative_bound}
 M_N&=\max_{1\leq|\alpha|\leq N}
       \sup_{y\in[0,T]\times\Sigma}
       |\partial_y^\alpha\mathcal{Z}(y)|.
\end{align}
The multi-index $\alpha$ in \eqref{E:IS_Matrix_entry_derivative_bound}
includes time and spatial derivatives.
The simple eigenvalues $\lambda_j$, $j\neq j_0$, have smooth local
extensions as functions of the matrix entries by the implicit
function theorem. Since the trace counts algebraic multiplicities,
\begin{align}
\label{E:IS_Multiple_eigenvalue_trace}
 \lambda_{j_0}
 &=\frac{1}{15}\left(\operatorname{tr}A
             -\sum_{\substack{1\leq j\leq9\\j\neq j_0}}\lambda_j\right)
\end{align}
for matrices arising from states in $\mathscr{K}$.
Equation \eqref{E:IS_Multiple_eigenvalue_trace} therefore gives a
smooth local extension of $\lambda_{j_0}$ using only the matrix
entries. The denominators $\lambda_j-\lambda_\ell$ in
\eqref{E:IS_LWP_spectral_projections} stay nonzero in neighborhoods
of the compact set $\mathscr{K}\times S^2$ by
\eqref{E:IS_Eigenvalue_difference_bound}.
For each $N$, the derivatives through order $N$ of these local
extensions of $P_j,R$ have a common finite bound on a finite cover
of $\mathscr{K}\times S^2$.
Each chain-rule term of total order at most $N$ contains at most
$N$ positive-order derivatives of $\mathcal{Z}$, giving
\begin{align}
\label{E:IS_Spectral_matrix_derivative_bound}
\begin{split}
 &\max_{|\alpha|+|\beta|\leq N}
   \sup_{\substack{y\in[0,T]\times\Sigma\\|k|=1}}
 \bigg(\|\partial_y^\alpha\partial_k^\beta R(y,k)\|
 \\
 &\hspace{35mm}
       +\sum_{j=1}^{9}\|\partial_y^\alpha\partial_k^\beta P_j(y,k)\|\bigg)
 \leq C_{N,\mathscr{K}}(1+M_N)^N.
\end{split}
\end{align}
The constant in \eqref{E:IS_Spectral_matrix_derivative_bound}
depends on $N$, $\mathscr{K}$, and the constants in
\eqref{E:IS_Time_matrix_inverse_bound} and
\eqref{E:IS_Eigenvalue_difference_bound}, not on $y,k$ or the
approximation index. The extensions
$P_j(y,k)=P_j(y,\frac{k}{|k|})$ and
$R(y,k)=R(y,\frac{k}{|k|})$ have degree zero in $k$;
the derivatives in $k$ in
\eqref{E:IS_Spectral_matrix_derivative_bound} are evaluated using
these extensions on $\frac{1}{2}<|k|<2$.
For $\Sigma=\mathbb{R}^3$,
\eqref{E:IS_Spectral_matrix_derivative_bound} holds when
\eqref{E:IS_Uniform_spectral_bounds} holds throughout the coefficient
range and $M_N<\infty$ in \eqref{E:IS_Matrix_entry_derivative_bound}.
Equations \eqref{E:IS_LWP_matrix_symmetry},
\eqref{E:IS_LWP_matrix_positivity}, and
\eqref{E:IS_Spectral_matrix_derivative_bound} verify the matrix
properties used in \cite[Section 5.2]{TaylorPSOandPDE} and
\cite[Assumptions (A3)--(A4)]{Shao:2023psr}.
For the fixed time covector $dt$, the simple roots remain real and
distinct in a neighborhood within the states satisfying
\eqref{Constraints}, by the real implicit function theorem.
The eigenvalue \eqref{E:IS_LWP_multiple_eigenvalue} retains
multiplicity fifteen by \eqref{E:IS_Full_determinant_for_proof},
and \eqref{E:IS_Advective_kernel_dimension} gives its eigenspace
dimension there. Thus the diagonalizability and positive lower
bounds in \eqref{E:IS_Uniform_spectral_bounds} persist on a smaller
such neighborhood.

\subsection{Analytic approximation of the initial data}

Let $Q$ be the eleven variables in \eqref{E:IS_independent_variables}.
Write $\Umb=\ISReconstruction(y,Q)$ for the map given by
\eqref{E:IS_nonlinear_reconstruction}, followed by conversion from
the fixed orthonormal frame to coordinate components.

The map $\ISReconstruction(y,Q)$ is real analytic for every real $Q$.
Indeed, the definitions in \eqref{E:IS_nonlinear_reconstruction} give
\begin{align}
\label{E:IS_Reconstruction_denominator_bound}
 H_{\hat{3}\hat{3}}
 &=\frac{1+(u^{\hat{1}})^2+(u^{\hat{2}})^2}
          {1+(u^{\hat{1}})^2+(u^{\hat{2}})^2+(u^{\hat{3}})^2}
 \geq\frac{1}{1+R^2}>0,\qquad |Q|\leq R.
\end{align}
The formulas \eqref{E:IS_nonlinear_reconstruction},
\eqref{E:IS_Reconstruction_denominator_bound}, and
\eqref{E:IS_Precise_metric_frame_bounds} consequently give, for
all multi-indices $\alpha,\beta$ and every $R>0$,
\begin{align}
\label{E:IS_Reconstruction_derivative_bounds}
 \sup_{\substack{y\in[0,T_g]\times\Sigma\\|Q|\leq R}}
 |\partial_y^\alpha\partial_Q^\beta\ISReconstruction(y,Q)|
 &\leq C_{\alpha,\beta,R}<\infty.
\end{align}
Here $C_{\alpha,\beta,R}$ depends on $\alpha,\beta,R$ and the
constants $G_\gamma$ with $|\gamma|\leq|\alpha|$ in
\eqref{E:IS_Precise_metric_frame_bounds}.

Let $Q_0$ be extracted from $\mathring{\Umb}$. On $\mathbb{R}^3$,
let $\overline{Q}_0$ be extracted from
$\overline{\Umb}(0,\cdot)$; on $\mathbb{T}^3$, put
$\overline{Q}_0=0$.

The bounds \eqref{E:IS_Precise_metric_frame_bounds} and
\eqref{E:IS_Precise_initial_data}, applied to the components in
\eqref{E:IS_independent_variables}, give
$Q_0-\overline{Q}_0\in H^s(\Sigma;\mathbb{R}^{11})$.

For $N\in\mathbb{N}$, let $P_N$ be the spatial Fourier projection
onto $|k|\leq N$, using the Fourier series on $\mathbb{T}^3$ and
the Fourier transform on $\mathbb{R}^3$. Define
\begin{subequations}
\label{E:IS_Precise_analytic_data}
\begin{align}
 Q_{0,N}&=\overline{Q}_0+P_N(Q_0-\overline{Q}_0),
 \\*
 \mathring{\Umb}_N(x)&=\ISReconstruction((0,x),Q_{0,N}(x)).
\end{align}
\end{subequations}
The projection of a real function is real. Its Fourier transform
has compact support, or its Fourier series is a finite sum;
therefore $P_N(Q_0-\overline{Q}_0)$ is real analytic. The reference
field and frame are real analytic, so
$\mathring{\Umb}_N$ is real analytic and satisfies
\eqref{Constraints} exactly.

The Fourier definition of $P_N$ and dominated convergence give
\begin{align*}
 \norm{P_N(Q_0-\overline{Q}_0)}_{H^s(\Sigma)}
 &\leq\norm{Q_0-\overline{Q}_0}_{H^s(\Sigma)},
 \\
 \norm{(P_N-I)(Q_0-\overline{Q}_0)}_{H^s(\Sigma)}
 &\longrightarrow0.
\end{align*}
The bound \eqref{E:IS_Reconstruction_derivative_bounds} and the
product and composition estimates of
\cite[Sections 3.1 and 3.5]{TaylorPSOandPDE}, with
$s>\frac{3}{2}$, give a constant $C$ independent of $N$ such that
\begin{align}
\label{E:IS_Precise_analytic_convergence}
 \norm{\mathring{\Umb}_N-\mathring{\Umb}}_{H^s(\Sigma)}
 &\leq C\norm{(P_N-I)(Q_0-\overline{Q}_0)}_{H^s(\Sigma)}
 \longrightarrow0.
\end{align}
Here $C$ depends on $s$, a bound for
$\norm{Q_0-\overline{Q}_0}_{H^s(\Sigma)}$,
the constants in \eqref{E:IS_Precise_metric_frame_bounds},
and, on $\mathbb{R}^3$, those in
\eqref{E:IS_Precise_reference_bounds}, through order $s$. 

Let
\begin{align}
\label{E:IS_Initial_state_distance}
 d_*&=\operatorname{dist}\big(\mathcal{K},
       \mathbb{R}^{23}\backslash\operatorname{int}\mathcal{K}_1\big)>0.
\end{align}
By \eqref{E:IS_Precise_analytic_convergence} and Sobolev embedding,
there exists $N_0\in\mathbb{N}$ such that
\begin{align}
\label{E:IS_Analytic_data_uniform_error}
 \norm{\mathring{\Umb}_N-\mathring{\Umb}}_{L^\infty(\Sigma)}
 &\leq C_sC\norm{(P_N-I)(Q_0-\overline{Q}_0)}_{H^s(\Sigma)}
 <\frac{d_*}{2},\qquad N\geq N_0.
\end{align}
Here $C_s$ is the norm of
$H^s(\Sigma;\mathbb{R}^{23})\subset L^\infty(\Sigma;\mathbb{R}^{23})$.
The integer $N_0$ may depend on the initial data, $d_*$, and $C_sC$.
Since $\mathring{\Umb}(x)\in\mathcal{K}$ by
\eqref{E:IS_Precise_initial_data},
\eqref{E:IS_Initial_state_distance} and
\eqref{E:IS_Analytic_data_uniform_error} imply
$\mathring{\Umb}_N(x)\in\mathcal{K}_1$ for every $x\in\Sigma$
and every $N\geq N_0$. We retain only these indices in
\eqref{E:IS_Precise_analytic_data}.

\subsection{The equation on the noncompact spatial domain}

For $\Sigma=\mathbb{R}^3$, substitute
$\Umb=\overline{\Umb}+W$ into
\eqref{E:IS_Precise_coordinate_system}. The equation for $W$ is
\begin{align}
\label{E:IS_Precise_reference_equation}
 \mathbb{A}^\mu(y,\overline{\Umb}+W)\partial_\mu W
 &=F(y,W),
\end{align}
where the sum over $\mu=0,1,2,3$ is understood and
\begin{subequations}
\label{E:IS_Precise_reference_source}
\begin{align}
\begin{split}
 F(y,W)
 &=-\mathbb{B}(y,\overline{\Umb}+W)(\overline{\Umb}+W)
 \\
 &\quad-\mathbb{A}^\mu(y,\overline{\Umb}+W)
                         \partial_\mu\overline{\Umb}.
\end{split}
\end{align}
\begin{align}
 F(y,0)&=0.
\end{align}
\end{subequations}
Equation \eqref{E:IS_Precise_coordinate_system} for
$\overline{\Umb}$ gives $F(y,0)=0$ in
\eqref{E:IS_Precise_reference_source}. For the composition estimates,
fix a smooth function $b$ of compact support in $\mathcal{O}$ that
is one on a neighborhood of $\mathcal{K}_1$. Multiplying the
coefficient functions by $b(Z)$ and extending by zero outside
$\mathcal{O}$ gives smooth functions on all of $\mathbb{R}^{23}$
with the derivative bounds obtained from
\eqref{E:IS_Precise_coefficient_bounds} by the product rule. Their values
are unchanged for $Z\in\mathcal{K}_1$, including the reference
states in $\mathcal{K}$. This also permits difference estimates
without requiring $\mathcal{O}$ to be convex.

For $W\in H^s(\mathbb{R}^3)$ with $\overline{\Umb}+W$
taking values in $\mathcal{K}_1$, the bounds \eqref{E:IS_Precise_coefficient_bounds} and
\eqref{E:IS_Precise_reference_bounds}, together with
\cite[Sections 3.1 and 3.5]{TaylorPSOandPDE}, give
\begin{align}
\label{E:IS_Reference_source_Sobolev_bound}
 \norm{F((t,\cdot),W)}_{H^s(\mathbb{R}^3)}
 &\leq C_R\norm{W}_{H^s(\mathbb{R}^3)},
 \qquad \norm{W}_{H^s(\mathbb{R}^3)}\leq R.
\end{align}
The expression $F((t,\cdot),W)$ denotes
$x\mapsto F((t,x),W(x))$.
The constant $C_R$ in \eqref{E:IS_Reference_source_Sobolev_bound}
is independent of $t\in[0,T_g]$ and depends on $s,R,\mathcal{K}_1$,
$b$, and the constants through order $s+1$ in
\eqref{E:IS_Precise_coefficient_bounds} and
\eqref{E:IS_Precise_reference_bounds}.

\subsection{Schauder construction and uniform estimates}

We use the analytic approximation in
\eqref{E:IS_Precise_analytic_data}. On $\mathbb{R}^3$, every
Sobolev norm is taken for the difference from the fixed reference
solution. Thus, for the analytic solutions $\Umb_n$ and their
initial values $\mathring{\Umb}_n$, put
\begin{align*}
 W_n&=\Umb_n-\overline{\Umb},
 \\
 W_{0,n}&=\mathring{\Umb}_n-\overline{\Umb}(0,\cdot).
\end{align*}
On the torus, these formulas reduce to $W_n=\Umb_n$ and
$W_{0,n}=\mathring{\Umb}_n$.

We write $n$ for the retained indices $N\geq N_0$ in
\eqref{E:IS_Precise_analytic_data}.
By \eqref{E:IS_Analytic_data_uniform_error}, the data
$\mathring{\Umb}_n$ take values in
$\mathcal{K}_1\cap\mathcal{C}(0,x)$ and satisfy
\eqref{E:IS_Precise_separation_lower_bound} with a constant
$\rho_*$ independent of $n$.
On any interval $[0,T]$ of analytic existence on which
\begin{align}
\label{E:IS_Analytic_time_derivative_bound}
 \sup_n\sup_{0\leq t\leq T}
 \norm{\partial_t\Umb_n(t)}_{L^\infty(\Sigma)}&\leq B,
\end{align}
\eqref{E:IS_Analytic_data_uniform_error} and integration in time give
\begin{align}
\label{E:IS_Analytic_range_displacement}
 |\Umb_n(t,x)-\mathring{\Umb}(x)|
 &\leq Bt+\frac{d_*}{2}.
\end{align}
If $BT\leq\frac{d_*}{4}$, then
\eqref{E:IS_Initial_state_distance} and
\eqref{E:IS_Analytic_range_displacement} give
\begin{align}
\label{E:IS_Analytic_state_range_distance}
 \operatorname{dist}\big(\Umb_n(t,x),
       \mathbb{R}^{23}\backslash\mathcal{K}_1\big)
 &\geq\frac{d_*}{4},
 \qquad (t,x)\in[0,T]\times\Sigma.
\end{align}
The estimate \eqref{E:IS_Schauder_sequence_bound}, applied on the
intervals of analytic existence as in Step~1, gives
\eqref{E:IS_Analytic_time_derivative_bound} by Sobolev embedding
and \eqref{E:IS_Precise_reference_bounds}. More precisely, we may take
\begin{align}
\label{E:IS_Analytic_time_derivative_constant}
 B&=\begin{cases}
 C_{s-1}C_M,&\Sigma=\mathbb{T}^3,\\
 C_{s-1}C_M+R_{(1,0,0,0)},&\Sigma=\mathbb{R}^3,
 \end{cases}
\end{align}
where $C_{s-1}$ is the norm of
$H^{s-1}(\Sigma;\mathbb{R}^{23})\subset L^\infty(\Sigma;\mathbb{R}^{23})$.
In \eqref{E:IS_Analytic_time_derivative_constant},
$R_{(1,0,0,0)}$ is the bound on $\partial_t\overline{\Umb}$ from
\eqref{E:IS_Precise_reference_bounds}; on the torus,
$\overline{\Umb}=0$.
While \eqref{E:IS_Analytic_state_range_distance} holds, the hypotheses
of Theorem~\ref{T:IS_Precise_Sobolev} give
\eqref{E:IS_Precise_separation_lower_bound} along $\Umb_n$.
Under \eqref{E:IS_Separation_transport}, this follows from
\eqref{E:IS_Transport_quantitative_separation}.

\medskip
\noindent\textbf{Step 1: Analytic solutions and the constraints.}
The invertibility of $\mathbb{A}^0$ gives the
Cauchy--Kovalevskaya form of
\eqref{E:IS_Precise_coordinate_system}. The analytic initial data
in \eqref{E:IS_Precise_analytic_data} therefore give analytic
solutions of \eqref{E:IS_Precise_coordinate_system}.

Theorem~\ref{T:IS_Propagation_constraints}, with the initial
bounds \eqref{E:IS_Constraint_initial_bounds} obtained from
\eqref{eq:H5} and \eqref{E:IS_Extra_root_lower_bound}, shows that
these analytic solutions satisfy \eqref{Constraints}.
Theorem~\ref{T:IS_Precise_strong_hyperbolicity} then applies with
the coefficients evaluated at $(y,\Umb_n(y))$.

The common positive time interval is obtained by the
continuation argument in the Schauder method; see
\cite[Chapter 5, Section 3(c)]{John-Book-1982} and
\cite[Chapter VI, Sections 8--10]{Courant_and_Hilbert_book_2}.
In this continuation argument, the estimates
\eqref{E:IS_Schauder_sequence_bound} are obtained on the
intervals of analytic existence. 

Together with \eqref{E:IS_Analytic_state_range_distance}, they
give a common interval whose length depends on
\eqref{E:IS_Schauder_initial_bound},
\eqref{E:IS_Precise_metric_frame_bounds},
\eqref{E:IS_Precise_coefficient_bounds},
\eqref{E:IS_Schauder_positive_lower_bound},
\eqref{E:IS_Precise_separation_lower_bound}, and, on
$\mathbb{R}^3$, \eqref{E:IS_Precise_reference_bounds}, with the
constant dependence stated in Theorem~\ref{T:IS_Precise_Sobolev}.

\medskip
\noindent\textbf{Step 2: Estimates on the common interval.}
On $\mathbb{T}^3$, we use
\cite[Proposition 3.1, equation (3.6), and Remarks 3.2--3.3]{Shao:2023psr}.
The proof of \cite[equation (3.6)]{Shao:2023psr} applies with the right-hand
side of the linear equation in place of the source denoted there by $R(v)$.
On $\mathbb{R}^3$, we use the whole-space linear theory and the energy
estimates for symmetrizable systems in
\cite[Sections 0.8 and 5.2]{TaylorPSOandPDE}. In that case the unknown
is $W_n=\Umb_n-\overline{\Umb}$,
with \eqref{E:IS_Precise_coefficient_bounds},
\eqref{E:IS_Precise_reference_bounds},
\eqref{E:IS_Schauder_positive_lower_bound}, and
\eqref{E:IS_Precise_separation_lower_bound}.

The identities \eqref{E:IS_LWP_matrix_symmetry} and
\eqref{E:IS_LWP_matrix_positivity} and the bound
\eqref{E:IS_Spectral_matrix_derivative_bound} verify the matrix
properties required by
\cite[Proposition 3.1]{Shao:2023psr} and
\cite[Sections 0.8 and 5.2]{TaylorPSOandPDE}, using
Theorem~\ref{T:IS_Precise_strong_hyperbolicity} and the positive
lower bounds in
\eqref{E:IS_Schauder_positive_lower_bound} and
\eqref{E:IS_Precise_separation_lower_bound}.

The estimates are applied after the analytic solutions have been
constructed and their constraints propagated. The coefficients are evaluated along solutions satisfying
\eqref{Constraints}; the linear estimates apply to arbitrary
$23$-component unknowns.

By \eqref{E:IS_Precise_analytic_convergence}, after discarding
finitely many initial approximations,
\begin{align}
\label{E:IS_Schauder_initial_bound}
 \sup_n\norm{W_{0,n}}_{H^s(\Sigma)}&\leq M+1.
\end{align}

The estimates of
\cite[Proposition 3.1, equation (3.6)]{Shao:2023psr} on
$\mathbb{T}^3$ and \cite[Sections 0.8 and 5.2]{TaylorPSOandPDE}
on $\mathbb{R}^3$, with
\eqref{E:IS_Precise_coefficient_bounds},
\eqref{E:IS_Schauder_positive_lower_bound}, \eqref{E:IS_Precise_separation_lower_bound}, and
\eqref{E:IS_Schauder_initial_bound}, give, after shortening the
common interval if necessary,
\begin{align}
\label{E:IS_Schauder_sequence_bound}
 \sup_n\sup_{0\leq t\leq T}
 \big(\norm{W_n(t)}_{H^s(\Sigma)}
       +\norm{\partial_tW_n(t)}_{H^{s-1}(\Sigma)}\big)&\leq C_M.
\end{align}
Here and in the difference estimates, $C_M$ is independent of
the approximation indices and has the dependence specified in
Theorem \ref{T:IS_Precise_Sobolev}.

On $\mathbb{R}^3$, the estimate of $F$ used here is
\eqref{E:IS_Reference_source_Sobolev_bound}; the reference solution
is fixed throughout the approximation and comparison arguments.

Define
\begin{subequations}
\label{E:IS_Schauder_normalized}
\begin{align}
 A^j(y,Z)&=(\mathbb{A}^0(y,Z))^{-1}\mathbb{A}^j(y,Z),
 \\
 f(y,Z)&=-(\mathbb{A}^0(y,Z))^{-1}\mathbb{B}(y,Z)Z.
\end{align}
\end{subequations}
With $A^j,f$ from \eqref{E:IS_Schauder_normalized}, set
\begin{align}
\label{E:IS_Schauder_source}
\begin{split}
 G(y,W)&=f(y,\overline{\Umb}+W)-\partial_t\overline{\Umb}
 \\
 &\quad-\sum_{j=1}^3 A^j(y,\overline{\Umb}+W)
                                     \partial_j\overline{\Umb}.
\end{split}
\end{align}
Then all approximations satisfy
\begin{align}
\label{E:IS_Schauder_normalized_equation}
 \partial_tW_n+\sum_{j=1}^3 A^j(y,\Umb_n)\partial_jW_n
 &=G(y,W_n).
\end{align}
For $D_{n,m}=W_n-W_m=\Umb_n-\Umb_m$, subtraction gives
\begin{align}
\label{E:IS_Schauder_difference_equation}
\begin{split}
 \partial_tD_{n,m}
  +\sum_{j=1}^3 A^j(y,\Umb_n)\partial_jD_{n,m}
 &=G(y,W_n)-G(y,W_m)
 \\
 &\quad-\sum_{j=1}^3
    \big(A^j(y,\Umb_n)-A^j(y,\Umb_m)\big)\partial_jW_m.
\end{split}
\end{align}

The bounds \eqref{E:IS_Schauder_sequence_bound},
\eqref{E:IS_Precise_coefficient_bounds}, and
\eqref{E:IS_Precise_reference_bounds}, together with
\cite[Sections 3.1 and 3.5]{TaylorPSOandPDE}, give
\begin{align}
\label{E:IS_Schauder_difference_source}
\begin{split}
 &\norm{G((t,\cdot),W_n(t))-G((t,\cdot),W_m(t))}_{H^{s-1}(\Sigma)}
 \\
 &\quad+\sum_{j=1}^3
 \norm{\big(A^j((t,\cdot),\Umb_n(t))-A^j((t,\cdot),\Umb_m(t))\big)
                                \partial_jW_m}_{H^{s-1}(\Sigma)}
 \leq C_M\norm{D_{n,m}}_{H^{s-1}(\Sigma)}.
\end{split}
\end{align}
Indeed, $s-1>\frac{3}{2}$, so $H^{s-1}$ is an algebra, and
$\norm{\partial_jW_m}_{H^{s-1}}\leq C_M$.

For these composition estimates the coefficients may be
multiplied by the function $b$ used to obtain
\eqref{E:IS_Reference_source_Sobolev_bound}; this changes neither
\eqref{E:IS_Schauder_normalized_equation} nor
\eqref{E:IS_Schauder_difference_equation}, because the solution
values lie in $\mathcal{K}_1$.

Applying \cite[Proposition 3.1, equation (3.6)]{Shao:2023psr}
on $\mathbb{T}^3$, or \cite[Sections 0.8 and 5.2]{TaylorPSOandPDE}
on $\mathbb{R}^3$, in order $s-1$ to
\eqref{E:IS_Schauder_difference_equation}, with
\eqref{E:IS_Schauder_difference_source} and Gronwall's inequality,
yields
\begin{align}
\label{E:IS_Schauder_Cauchy_estimate}
 \sup_{0\leq t\leq T}\norm{D_{n,m}(t)}_{H^{s-1}(\Sigma)}
 &\leq C_M\norm{W_{0,n}-W_{0,m}}_{H^{s-1}(\Sigma)}.
\end{align}
The strict inequality $s-1>\frac{3}{2}+1$ needed for this
application holds because $s\geq4$.

\medskip
\noindent\textbf{Step 3: Passage to the limit and uniqueness.}
Equations \eqref{E:IS_Precise_analytic_convergence} and
\eqref{E:IS_Schauder_Cauchy_estimate} make $W_n$ a Cauchy
sequence in $C^0([0,T];H^{s-1}(\Sigma))$.
Equation \eqref{E:IS_Schauder_normalized_equation} then gives
convergence of $\partial_tW_n$ in
$C^0([0,T];H^{s-2}(\Sigma))$.
The bound \eqref{E:IS_Schauder_sequence_bound} and weak
compactness imply that the limit $W$ belongs to
$L^\infty([0,T];H^s(\Sigma))$. Since $s-1>\frac{5}{2}$, the
convergence in $H^{s-1}$ is also convergence in $C^1$ in the
spatial variables. Thus the equations, initial values, and
algebraic constraints pass to the limit.  For the time-continuity
argument on $\mathbb{T}^3$, we use
\cite[Section 4.4, Proposition 4.4]{Shao:2023psr};
take $v_n=w_n=W_n$ and $w_{0,n}=W_{0,n}$ in the
system \eqref{E:IS_Schauder_normalized_equation}.
The bound in \cite[equation (4.11)]{Shao:2023psr}
is \eqref{E:IS_Schauder_sequence_bound}. The hypotheses in
\cite[equation (4.12)]{Shao:2023psr} are the convergence in
$C^0([0,T];H^{s-1}(\Sigma))\cap C^1([0,T];H^{s-2}(\Sigma))$ just established
and \eqref{E:IS_Precise_analytic_convergence} at $t=0$.

Equations \eqref{E:IS_Regular_spectral_bounds},
\eqref{E:IS_Spectral_matrix_derivative_bound}, and
\eqref{E:IS_Schauder_sequence_bound} allow the application of
\cite[Section 4.4, Proposition 4.4]{Shao:2023psr}, giving
$W\in C^0([0,T];H^s(\Sigma))$.
 Equation
\eqref{E:IS_Schauder_normalized_equation} gives
$\partial_tW\in C^0([0,T];H^{s-1}(\Sigma))$.
On $\mathbb{R}^3$, the argument of
\cite[Section 4.4, Proposition 4.4]{Shao:2023psr} is applied to the
equation \eqref{E:IS_Precise_reference_equation}, with the whole-space linear estimates
of \cite[Sections 0.8 and 5.2]{TaylorPSOandPDE}
and \eqref{E:IS_Precise_coefficient_bounds},
\eqref{E:IS_Precise_reference_bounds},
\eqref{E:IS_Uniform_spectral_bounds}, and
\eqref{E:IS_Schauder_sequence_bound}.

The calculation
\eqref{E:IS_Schauder_difference_equation}--\eqref{E:IS_Schauder_Cauchy_estimate}
also applies to two solutions in this class. It gives
\eqref{E:IS_Precise_difference_estimate}. With equal initial
data, its right-hand side is zero, which proves uniqueness.
The constants and time interval used in
\eqref{E:IS_Schauder_sequence_bound} are chosen from the common
initial bound \eqref{E:IS_Schauder_initial_bound}
and the fixed sets $\mathcal{K},\mathcal{K}_1$. They therefore give
the uniformity asserted in
\eqref{E:IS_Precise_uniform_estimate} for the stated class of data.

\subsection{Continuous dependence in the top Sobolev norm}

On $\mathbb{T}^3$, we follow the proof of
\cite[Theorem 2.5 and Section 4.6]{Shao:2023psr}, using linear evolution
operators whose coefficients are evaluated along the Sobolev solutions
already obtained by the Schauder construction.
Equations \eqref{E:IS_Regular_spectral_bounds} and
\eqref{E:IS_Spectral_matrix_derivative_bound} verify the matrix
properties used in \cite[Section 4.6]{Shao:2023psr}. On $\mathbb{R}^3$, we use the same
argument with the whole-space evolution estimates of
\cite[Sections 0.8 and 5.2]{TaylorPSOandPDE}, applied to
$W=\Umb-\overline{\Umb}$
with \eqref{E:IS_Precise_reference_bounds},
\eqref{E:IS_Precise_coefficient_bounds},
\eqref{E:IS_Uniform_spectral_bounds}, and
\eqref{E:IS_Precise_uniform_estimate}.

For the solution $W$, define
\begin{subequations}
\begin{align}
\label{E:IS_Top_norm_evolution_operator}
 (L_W(t)v)(x)
 &=-\sum_{j=1}^{3}A^j((t,x),\overline{\Umb}(t,x)+W(t,x))
                         \partial_jv(x),\\
\label{E:IS_Top_norm_commutator_operator}
 K_W(t)&=L_W(t)-\langle D\rangle^sL_W(t)\langle D\rangle^{-s},\\
\label{E:IS_Top_norm_spatial_multiplier}
 \langle D\rangle&=(1-\Delta_x)^{\frac{1}{2}}.
\end{align}
\end{subequations}

The source is $G$ from \eqref{E:IS_Schauder_source}.
On the bounded class in
\eqref{E:IS_Precise_uniform_estimate},
\eqref{E:IS_Precise_coefficient_bounds},
\eqref{E:IS_Precise_reference_bounds},
\cite[Sections 3.1 and 3.5]{TaylorPSOandPDE}, and
\cite[Section 3.6]{TaylorPSOandPDE} give
\begin{align}
\label{E:IS_Schauder_top_norm_coefficient_estimates}
\begin{split}
 &\norm{G((t,\cdot),W(t))-G((t,\cdot),\widetilde{W}(t))}_{H^s(\Sigma)}
 \\
 &\quad+\norm{K_W(t)-K_{\widetilde{W}}(t)}
                              _{\mathcal{L}(L^2(\Sigma),L^2(\Sigma))}
 \leq C_M\norm{W-\widetilde{W}}_{H^s(\Sigma)}.
\end{split}
\end{align}

For the norm of $K_W(t)-K_{\widetilde{W}}(t)$ in
\eqref{E:IS_Schauder_top_norm_coefficient_estimates},
\eqref{E:IS_Top_norm_spatial_multiplier} gives
$\langle D\rangle^{-s}\colon L^2(\Sigma)\rightarrow H^s(\Sigma)$,
and $H^s(\Sigma)\subset W^{1,\infty}(\Sigma)$.
Equation \eqref{E:IS_Top_norm_commutator_operator} gives
\begin{align*}
 K_W(t)-K_{\widetilde{W}}(t)
 &=-[\langle D\rangle^s,L_W(t)-L_{\widetilde{W}}(t)]
                                      \langle D\rangle^{-s}.
\end{align*}
Substitution of \eqref{E:IS_Top_norm_evolution_operator} and the
commutator estimate in \cite[Section 3.6]{TaylorPSOandPDE}, applied
to the differences of its coefficients, give the stated operator
norm bound in \eqref{E:IS_Schauder_top_norm_coefficient_estimates}.

The hypotheses for the linear evolution estimates used in
\cite[Section 4.6]{Shao:2023psr} follow from
\eqref{E:IS_Regular_spectral_bounds},
\eqref{E:IS_Spectral_matrix_derivative_bound},
\eqref{E:IS_Precise_coefficient_bounds}, and
\eqref{E:IS_Precise_uniform_estimate}.
 On $\mathbb{T}^3$
we use \cite[Proposition 3.1 and Remark 3.2]{Shao:2023psr}; on
$\mathbb{R}^3$ we use \cite[Sections 0.8 and 5.2]{TaylorPSOandPDE}.
The reference solution is fixed throughout the $\mathbb{R}^3$ comparison.
The argument of \cite[Section 4.6]{Shao:2023psr},
with \eqref{E:IS_Schauder_top_norm_coefficient_estimates}, gives
\begin{align}
\label{E:IS_Schauder_top_norm_limit}
 W_{0,q}\longrightarrow W_0\text{ in }H^s(\Sigma)
 \quad&\Longrightarrow\quad
 W_q\longrightarrow W\text{ in }C^0([0,T];H^s(\Sigma)).
\end{align}

Finally, subtracting
\eqref{E:IS_Schauder_normalized_equation} for two Sobolev
solutions and using
\eqref{E:IS_Precise_coefficient_bounds},
\eqref{E:IS_Precise_uniform_estimate}, and
\cite[Sections 3.1 and 3.5]{TaylorPSOandPDE} in order $s-1$ gives
\begin{align}
\label{E:IS_Time_derivative_continuous_dependence_bound}
 \sup_{0\leq t\leq T}
 \norm{\partial_tW_q(t)-\partial_tW(t)}_{H^{s-1}(\Sigma)}
 &\leq C_M\sup_{0\leq t\leq T}\norm{W_q(t)-W(t)}_{H^s(\Sigma)}.
\end{align}

Equations \eqref{E:IS_Schauder_top_norm_limit} and
\eqref{E:IS_Time_derivative_continuous_dependence_bound} give
convergence of $\partial_tW_q$ to $\partial_tW$ in
$C^0([0,T];H^{s-1}(\Sigma))$. Since the reference field is the same for all the
solutions, these differences equal the differences of $\Umb_q$
and $\Umb$ and of their time derivatives. This gives
\eqref{E:IS_Precise_continuous_dependence} in both spatial domains.

This concludes the proof of Theorem \ref{T:IS_Precise_Sobolev}.

\subsection{An example satisfying the hypotheses}

We verify that the hypotheses of Theorem~\ref{T:IS_Precise_Sobolev}
hold for a choice of constitutive functions and a nonempty open set
of initial data in the variables \eqref{E:IS_independent_variables}.
The construction applies to both $\Sigma=\mathbb{T}^3$ and
$\Sigma=\mathbb{R}^3$.

\medskip
\noindent\textbf{Step 1: Constitutive functions and a constant state.}
Take $g_{\mu\nu}=\operatorname{diag}(-1,1,1,1)$ and the coordinate
orthonormal frame. Set
\begin{subequations}
\label{E:IS_Nonempty_constitutive_functions}
\begin{align}
 P(\varepsilon,n)&=\frac{\varepsilon}{3},\\
 \tau_\Pi=\tau_\pi&=1,\\
 \eta=\zeta&=\frac{1}{2},\\
 \delta_{\Pi\Pi}&=2,\\
 \delta_{\pi\pi}&=\frac{1}{3},\\
 \tau_{\pi\pi}&=4,\\
 \lambda_{\Pi\pi}=\lambda_{\pi\Pi}&=0.
\end{align}
\end{subequations}
Let $\Umb_*$ be the constant state whose components are
\begin{subequations}
\label{E:IS_Nonempty_state}
\begin{align}
 u_*^\mu&=(1,0,0,0),\\
 (\varepsilon_*,n_*,\Pi_*)&=\left(\frac{9}{8},1,-\frac{1}{2}\right),\\
 (\pi_*^{\mu\nu})&=\operatorname{diag}
       \left(0,-\frac{1}{6},-\frac{1}{12},\frac{1}{4}\right).
\end{align}
\end{subequations}
It satisfies \eqref{Constraints}. At \eqref{E:IS_Nonempty_state},
$E=1$, $P_\varepsilon=\frac{1}{3}$, $P_n=0$, and
\begin{align*}
 (\Lambda_1,\Lambda_2,\Lambda_3)
 &=\left(-\frac{1}{6},-\frac{1}{12},\frac{1}{4}\right),\\
 \min_{0\leq A\leq3}(E+\Lambda_A)&=\frac{5}{6}.
\end{align*}
The coefficients in \eqref{E:IS_Separation_coefficients} satisfy
\begin{align*}
 (a,b,c,d,e)&=\left(\frac{1}{2},1,0,0,0\right),\\
 m_\pi&=\frac{1}{12},\\
 L_\pi&=\frac{1}{4}.
\end{align*}
Consequently, \eqref{E:IS_Separation_transport_difference} gives
\begin{align}
\label{E:IS_Nonempty_transport_difference}
 \gamma_\pi(\Umb_*)&=\frac{1}{12}
       \left(\left|\frac{1}{2}-1\right|-\frac{1}{4}\right)
       =\frac{1}{48}>0.
\end{align}
Thus \eqref{E:IS_Separation_transport} holds at $\Umb_*$.

\medskip
\noindent\textbf{Step 2: The roots for every spatial covector.}
For $k=(k_1,k_2,k_3)^T$ with $|k|=1$, put
\begin{align}
\label{E:IS_Nonempty_q}
 q&=-\frac{1}{6}k_1^2-\frac{1}{12}k_2^2+\frac{1}{4}k_3^2
       \in\left[-\frac{1}{6},\frac{1}{4}\right].
\end{align}
At the state \eqref{E:IS_Nonempty_state},
\eqref{E:IS_Separation_matrix} becomes
\begin{align*}
 \mathcal{F}_{ij}(r,k)
 &=\left((1+\Lambda_i)r-\frac{1}{2}-q-\Lambda_i\right)\delta_{ij},
 \qquad i,j\in\{1,2,3\},
\end{align*}
with no summation on $i,j$; $r$ is the polynomial variable.
Since $\prod_{i=1}^{3}(1+\Lambda_i)=\frac{275}{288}$,
\eqref{E:IS_Separation_matrix_determinant} gives
\begin{align}
\label{E:IS_Nonempty_cubic}
 \mathcal{P}_3(r;k)
 &=\frac{(5r-6q-2)(11r-12q-5)(5r-4q-3)}{275}.
\end{align}
Here $k_i^2=\hat{v}_i^2$ in \eqref{eq:normalizedcharacteristics}.
The roots of \eqref{E:IS_Nonempty_cubic} and the value in
\eqref{E:IS_Extra_squared_speed} are
\begin{subequations}
\label{E:IS_Nonempty_roots}
\begin{align}
 y_1&=\frac{6q+2}{5},\\
 y_2&=\frac{12q+5}{11},\\
 y_3&=\frac{4q+3}{5},\\
 y_0&=\frac{1}{2}+q.
\end{align}
\end{subequations}
Equations \eqref{E:IS_Nonempty_roots} give
\begin{subequations}
\label{E:IS_Nonempty_root_differences}
\begin{align}
 y_j-y_i
 &=\frac{(\Lambda_j-\Lambda_i)(\frac{1}{2}-q)}
          {(1+\Lambda_i)(1+\Lambda_j)},\qquad 1\leq i<j\leq3,\\
 y_i-y_0
 &=\frac{\Lambda_i(\frac{1}{2}-q)}{1+\Lambda_i},\qquad 1\leq i\leq3.
\end{align}
\end{subequations}
By \eqref{E:IS_Nonempty_q}, $\frac{1}{2}-q\geq\frac{1}{4}$. Equations
\eqref{E:IS_Nonempty_roots} and
\eqref{E:IS_Nonempty_root_differences} imply
\begin{subequations}
\label{E:IS_Nonempty_root_bounds}
\begin{align}
 0<y_1<y_2<y_0<y_3&<1,\\
 \frac{1}{5}\leq y_i&\leq\frac{4}{5},\qquad 1\leq i\leq3,\\
 \frac{1}{3}\leq y_0&\leq\frac{3}{4}.
\end{align}
\end{subequations}
Equations \eqref{E:IS_Nonempty_cubic} and
\eqref{E:IS_Nonempty_q} give
\begin{subequations}
\label{E:IS_Nonempty_discriminant_separation}
\begin{align}
\label{E:IS_Nonempty_discriminant}
 \Delta(\mathcal{P}_3;k)
 &=\frac{576(1-2q)^6}{228765625}
       \geq\frac{9}{228765625}>0,\\
\label{E:IS_Nonempty_separation}
 \mathcal{P}_3(y_0;k)
 &=-\frac{(\frac{1}{2}-q)^3}{275}
       \leq-\frac{1}{17600}<0.
\end{align}
\end{subequations}
In particular, \eqref{E:IS_Nonempty_separation} also verifies
\eqref{E:IS_Separation_polynomial}.
The coefficient identities
\eqref{supeq:A0root}--\eqref{supeq:A2root} and
\eqref{E:IS_Nonempty_root_bounds} yield
\begin{subequations}
\label{E:IS_Nonempty_coefficient_bounds}
\begin{align}
 \frac{1}{5}\leq\frac{\mathcal{A}_2}{3}&\leq\frac{4}{5},\\
 \mathcal{A}_1&\geq\frac{3}{25},\\
 \mathcal{A}_0&\geq\frac{1}{125},\\
 1-\mathcal{A}_2+\mathcal{A}_1-\mathcal{A}_0
 &=\prod_{i=1}^{3}(1-y_i)\geq\frac{1}{125},\\
 1-\frac{2}{3}\mathcal{A}_2+\frac{1}{3}\mathcal{A}_1
 &=\frac{1}{3}\sum_{1\leq i<j\leq3}(1-y_i)(1-y_j)
       \geq\frac{1}{25}.
\end{align}
\end{subequations}
Equations \eqref{E:IS_Nonempty_root_bounds},
\eqref{E:IS_Nonempty_discriminant}, and
\eqref{E:IS_Nonempty_coefficient_bounds} verify
\eqref{eq:H1}--\eqref{eq:H5}, with strict upper bounds where an
upper bound occurs. The inequalities hold for every $|k|=1$.

\medskip
\noindent\textbf{Step 3: An open set of states and a reference solution.}
For \eqref{E:IS_Nonempty_constitutive_functions} and the prescribed
constant metric, the matrices in
\eqref{E:IS_Precise_coordinate_system} have no explicit dependence
on $y$. At $\Umb_*$ and $\phi=dt$, the scalar $z$ is one and
$v^\mu=w^\mu=0$. Substitution in \eqref{E:IS_Four_matrix} gives
\begin{align*}
 \mathcal{X}(dt;\Umb_*)&=\operatorname{diag}
       \left(1,\frac{5}{6},\frac{11}{12},\frac{5}{4}\right).
\end{align*}
Eliminating the nineteen scalar and tensor columns in
\eqref{eq:fluidprincipalpart}, whose diagonal entries are one,
therefore gives
\begin{align}
\label{E:IS_Nonempty_time_matrix}
 \det\mathbb{A}^0(\Umb_*)&=\frac{275}{288}>0.
\end{align}
The matrices are real analytic near $\Umb_*$, and
\eqref{E:IS_Nonempty_time_matrix} makes $(\mathbb{A}^0)^{-1}$
real analytic there as well.

The expressions in \eqref{eq:H1}--\eqref{eq:H5},
\eqref{E:IS_Separation_polynomial}, and
\eqref{E:IS_Separation_transport_difference} are continuous in
states satisfying \eqref{Constraints} near $\Umb_*$ and in $k$;
for the polynomial coefficients this follows from
\eqref{E:IS_Cubic_coefficients_in_components}.
The positive constants in
\eqref{E:IS_Nonempty_transport_difference},
\eqref{E:IS_Nonempty_discriminant_separation}, and
\eqref{E:IS_Nonempty_coefficient_bounds}, together with
\eqref{E:IS_Nonempty_root_bounds} and compactness of $S^2$,
provide a number $r_*>0$ with the following properties.
For every $Z$ satisfying \eqref{Constraints} and
$|Z-\Umb_*|\leq3r_*$, the inequalities
\eqref{eq:H1}--\eqref{eq:H5} hold, and
\begin{subequations}
\label{E:IS_Nonempty_neighborhood_bounds}
\begin{align}
 \min\left\{\min_{0\leq A\leq3}(E+\Lambda_A),
       \Delta(\mathcal{P}_3;k),\frac{\mathcal{A}_2}{3},
       \mathcal{A}_1,\mathcal{A}_0,y_0\right\}
 &\geq c_*:=\frac{9}{457531250},\\
 |\mathcal{P}_3(y_0;k)|&\geq\rho_*:=\frac{1}{35200},\\
 \gamma_\pi(Z)&\geq\frac{1}{96}
\end{align}
\end{subequations}
for every $|k|=1$. To obtain one such $r_*$, continuity at each
$(\Umb_*,k)$ preserves the positive lower bounds in
\eqref{E:IS_Nonempty_transport_difference},
\eqref{E:IS_Nonempty_discriminant_separation},
\eqref{E:IS_Nonempty_coefficient_bounds}, and
\eqref{E:IS_Nonempty_root_bounds} with one half of their values,
and preserves positivity
of $1-\frac{\mathcal{A}_2}{3}$ and $1-y_0$.
A finite cover of $S^2$ by these neighborhoods gives one positive
radius valid for every $k$.
Decreasing $r_*$ if necessary, continuity in
\eqref{E:IS_Nonempty_time_matrix} also gives
\begin{align*}
 \det\mathbb{A}^0(Z)&\geq\frac{275}{576}
       \qquad\text{for every }|Z-\Umb_*|\leq3r_*,
\end{align*}
including states not satisfying \eqref{Constraints}; all matrix
entries are real analytic on a neighborhood of this closed ball.
We also take $r_*<\frac{1}{8}$, so that the velocities in this
ball have $u^0>0$.

Using Euclidean balls in $\mathbb{R}^{23}$, set
\begin{subequations}
\label{E:IS_Nonempty_state_sets}
\begin{align}
 \mathcal{O}&=B_{3r_*}(\Umb_*),\\
 \mathcal{K}&=\overline{B_{r_*}(\Umb_*)},\\
 \mathcal{K}_1&=\overline{B_{2r_*}(\Umb_*)}.
\end{align}
\end{subequations}
They satisfy the inclusions in
Theorem~\ref{T:IS_Precise_Sobolev}. On $\mathcal{O}$ in
\eqref{E:IS_Nonempty_state_sets}, the coefficient matrices and
$(\mathbb{A}^0)^{-1}$ are real analytic. Their derivatives with
respect to $Z$ have finite maxima on every compact
$\mathcal{L}\subset\mathcal{O}$, and all positive-order derivatives
with respect to $y$ vanish. This verifies
\eqref{E:IS_Precise_coefficient_bounds}.
The metric and frame satisfy
\eqref{E:IS_Precise_metric_frame_bounds} with $G_0=4$ and
$G_\alpha=0$ for $|\alpha|\geq1$, and
\eqref{E:IS_Precise_time_covector_bound} with $c_g=1$.
Since $\tau_\Pi=\tau_\pi=1$, equations
\eqref{E:IS_Nonempty_neighborhood_bounds} supply
\eqref{E:IS_Schauder_positive_lower_bound} and
\eqref{E:IS_Precise_separation_lower_bound} on
$\mathcal{K}_1\cap\mathcal{C}(y)$, with the displayed constants
$c_*,\rho_*$, independently of $y$.

For $\Sigma=\mathbb{R}^3$, the required reference solution is
\begin{subequations}
\label{E:IS_Nonempty_reference}
\begin{align}
 \overline{u}^\mu(t,x)&=(1,0,0,0),\\
 (\overline{\varepsilon},\overline{n})(t,x)
 &=\left(\frac{9}{8},1\right),\\
 \overline{\Pi}(t,x)&=-\frac{1}{2}e^{-t},\\
 (\overline{\pi}^{\mu\nu}(t,x))
 &=e^{-t}\operatorname{diag}
       \left(0,-\frac{1}{6},-\frac{1}{12},\frac{1}{4}\right).
\end{align}
\end{subequations}
These fields are real analytic and satisfy \eqref{Constraints}.
Their only nonzero derivatives are time derivatives of
$\overline{\Pi}$ and $\overline{\pi}^{\mu\nu}$. In particular,
\begin{align*}
 \partial_\alpha\overline{u}^\mu&=0,\\
 \Delta^{\mu\alpha}\partial_\alpha(\overline{P}+\overline{\Pi})&=0,\\
 \partial_\alpha\overline{\pi}^{\mu\alpha}&=0,\\
 \partial_t\overline{\Pi}+\overline{\Pi}&=0,\\
 \partial_t\overline{\pi}^{\mu\nu}+\overline{\pi}^{\mu\nu}&=0.
\end{align*}
Together with the constant values of
$\overline{\varepsilon},\overline{n}$, substitution of
\eqref{E:IS_Nonempty_reference} and
\eqref{E:IS_Nonempty_constitutive_functions} into
\eqref{eq:propsystem} therefore verifies that
\eqref{E:IS_Nonempty_reference} solves
\eqref{E:IS_Precise_coordinate_system}.
Moreover,
\begin{align*}
 |\overline{\Umb}(t,x)-\Umb_*|
 &=\frac{5}{6\sqrt{2}}(1-e^{-t})
       \leq\frac{5}{6\sqrt{2}}t,\qquad t\geq0.
\end{align*}
With
\begin{align}
\label{E:IS_Nonempty_reference_time}
 T_g&=\min\left\{1,\frac{3\sqrt{2}}{5}r_*\right\}>0,
\end{align}
we obtain $|\overline{\Umb}(t,x)-\Umb_*|\leq\frac{r_*}{2}$ for
$0\leq t\leq T_g$. Thus \eqref{E:IS_Nonempty_reference} takes
values in $\mathcal{K}\cap\mathcal{C}(y)$ on this interval.
The bounds
\begin{align*}
 \sup_{0\leq t\leq T_g}|\overline{\Umb}(t,x)|
 &\leq\frac{\sqrt{2081}}{24},\\
 \sup_{0\leq t\leq T_g}|\partial_t^m\overline{\Umb}(t,x)|
 &\leq\frac{5}{6\sqrt{2}},\qquad m\geq1,
\end{align*}
and the vanishing of all spatial derivatives verify
\eqref{E:IS_Precise_reference_bounds} for every spacetime
multi-index. On $\mathbb{T}^3$, we use the same $T_g$ in
\eqref{E:IS_Nonempty_reference_time} and retain the theorem's
convention $\overline{\Umb}=0$.

\medskip
\noindent\textbf{Step 4: Sobolev initial data.}
In the order prescribed by \eqref{E:IS_independent_variables}, put
\begin{align*}
 Q_*&=\left(\frac{9}{8},1,-\frac{1}{2},
       -\frac{1}{6},-\frac{1}{12},0,0,0,0,0,0\right).
\end{align*}
For the constant frame, the map defined by
\eqref{E:IS_nonlinear_reconstruction} has no dependence on $y$;
write it as $\ISReconstruction(Q)$. It satisfies
$\ISReconstruction(Q_*)=\Umb_*$.
For $h\in H^s(\Sigma;\mathbb{R}^{11})$, define initial data by
\begin{align}
\label{E:IS_Nonempty_initial_data}
 \mathring{\Umb}_h(x)&=\ISReconstruction(Q_*+h(x)).
\end{align}
Equation \eqref{E:IS_nonlinear_reconstruction} ensures that these
data satisfy \eqref{Constraints} and have future-directed
velocity. By \eqref{E:IS_Reconstruction_derivative_bounds} and
the product and composition estimates in
\cite[Sections 3.1 and 3.5]{TaylorPSOandPDE}, there is a constant
$C_s\geq1$, depending on $s,\Sigma,Q_*$, such that
\begin{align}
\label{E:IS_Nonempty_data_estimate}
 \norm{\mathring{\Umb}_h-\Umb_*}_{H^s(\Sigma)}
 &\leq C_s\norm{h}_{H^s(\Sigma)},
 \qquad \norm{h}_{H^s(\Sigma)}\leq1.
\end{align}
For $\mathbb{R}^3$, the composition is applied to the function
$h\mapsto\ISReconstruction(Q_*+h)-\ISReconstruction(Q_*)$,
which vanishes at $h=0$.
Let $C_{\mathrm{emb}}$ be the norm of the embedding
$H^s(\Sigma;\mathbb{R}^{23})\subset L^\infty(\Sigma;\mathbb{R}^{23})$,
and set
\begin{align*}
 \epsilon_s&=\min\left\{1,\frac{1}{2C_s},
                   \frac{r_*}{2C_{\mathrm{emb}}C_s}\right\}>0.
\end{align*}
For every $\norm{h}_{H^s(\Sigma)}<\epsilon_s$,
\eqref{E:IS_Nonempty_data_estimate} implies
\begin{align*}
 \norm{\mathring{\Umb}_h-\Umb_*}_{L^\infty(\Sigma)}
 &<\frac{r_*}{2},\\
 \norm{\mathring{\Umb}_h-\Umb_*}_{H^s(\Sigma)}&<\frac{1}{2}.
\end{align*}
Hence $\mathring{\Umb}_h(x)\in\mathcal{K}\cap\mathcal{C}(0,x)$
for every $x\in\Sigma$ by \eqref{E:IS_Nonempty_state_sets}.
On $\mathbb{R}^3$, equation \eqref{E:IS_Nonempty_reference} gives
$\overline{\Umb}(0,x)=\Umb_*$, so
\eqref{E:IS_Precise_initial_data} holds with $M=1$.
On $\mathbb{T}^3$, the constant $\Umb_*$ belongs to $H^s$, and
\begin{align*}
 \norm{\mathring{\Umb}_h}_{H^s(\mathbb{T}^3)}
 &<\norm{\Umb_*}_{H^s(\mathbb{T}^3)}+\frac{1}{2}.
\end{align*}
Thus \eqref{E:IS_Precise_initial_data} holds there with
$M=\norm{\Umb_*}_{H^s(\mathbb{T}^3)}+1$.
The ball $\{h\in H^s(\Sigma;\mathbb{R}^{11})\,|\,
\norm{h}_{H^s(\Sigma)}<\epsilon_s\}$ is nonempty and open.
Its distinct elements give distinct data in
\eqref{E:IS_Nonempty_initial_data}, because their components in
\eqref{E:IS_independent_variables} are exactly $Q_*+h$.
This verifies every hypothesis of
Theorem~\ref{T:IS_Precise_Sobolev} for the stated constitutive
functions and a nonempty open set of data in
$Q_*+H^s(\Sigma;\mathbb{R}^{11})$, on both spatial domains.

\section{Causality of Generalized Maxwell-Cattaneo}\label{App:MC}

Here, we list more general forms of the extrema calculated in the generalized Maxwell-Cattaneo limit where $\lambda_{\Pi\pi} = \delta_{\pi\pi} = \tau_{\pi\pi} = 0$, these conditions include, but do not impose the physical assumptions made in the main paper for the interested reader. For convenience, we once again provide the determinant in this case:
\be
\label{eq:MCdetApp}
\det(\mathbb{A}^\alpha\phi_\alpha) = \tau^{17}v^{17}\hat{z}^{15}\left(E\hat{z}^2 - \frac{\eta_{\textrm{eff}}}{\tau_\pi}\right)\mathcal{P}_3(\hat{z}^2,\hat{v}_1^2,\hat{v}_2^2,\hat{v}_3^2).
\ee
as well as the corresponding (angle-dependent) coefficients:
\bml
\label{eq:MCcoeffApp}
\bea
\mathcal{A}_0(\chi^2,\kappa^2) &=& \dfrac{\eta_{\textrm{eff}}^2}{\tau_\pi^2}\frac{\dfrac{\zeta_{\textrm{eff}}}{\tau_\Pi} + \dfrac{\eta_{\textrm{eff}}}{\tau_\pi} + P_\varepsilon \left(E + \Lambda_{\hat{v}^2}\right)}{ (E + \Lambda_1) (E + \Lambda_2) (E + \Lambda_3)},\\
\mathcal{A}_1(\chi^2,\kappa^2) &=& \frac{\eta_{\textrm{eff}}}{\tau_\pi}\left[\dfrac{3\dfrac{\eta_{\textrm{eff}}^2}{\tau_\pi^2} + 2\dfrac{\zeta_{\textrm{eff}}}{\tau_\Pi}\left(E - \dfrac{1}{2}\Lambda_{\hat{v}^2}\right)}{(E + \Lambda_1)(E + \Lambda_2)(E +\Lambda_3)} + P_\varepsilon\sum_{k = 1}^3\frac{1 - \hat{v}_k^2}{E + \Lambda_k}\right],\\
\mathcal{A}_2(\chi^2,\kappa^2) &=& P_\varepsilon + \dfrac{\zeta_{\textrm{eff}}}{\tau_\Pi}\sum_{k = 1}^3\frac{\hat{v}_k^2}{E + \Lambda_k} + \dfrac{\eta_{\textrm{eff}}}{\tau_\pi}\sum_{k = 1}^3\frac{1}{E + \Lambda_k}.
\eea
\eml
Recall the definitions
\bml
\bea
\frac{\zeta_{\textrm{eff}}}{\tau_\Pi} &\equiv & nP_n + \frac{\zeta}{\tau_\Pi} + \delta_{\Pi\Pi}\Pi + \frac{1}{3}\frac{\eta_{\textrm{eff}}}{\tau_\pi},\\\frac{\eta_{\textrm{eff}}}{\tau_\pi} &\equiv & \eta + \frac{1}{2}\lambda_{\pi\Pi}\Pi,\\ c_{\textrm{eff},a}^2 &\equiv & P_\varepsilon + \frac{1}{E +\Lambda_a}\frac{\zeta_{\textrm{eff}}}{\tau_\Pi}.
\eea
\eml
We will rearrange indices and impose:
\bml
\bea
\Lambda_1 &\leq& \Lambda_2\leq \Lambda_3,\\
-E < \Lambda_1 &\leq &0 \leq \Lambda_3.
\eea
\eml
We stress that the causality conditions in Theorem~\ref{thm:causalitynodiffusion} can be applied as usual due to the identical structure of the characteristic determinant. However, these original nonlinear causality conditions depend on $\chi^2,\kappa^2\in[0,1]$ which can be thought of as cosine functions in local spherical coordinates. These conditions must be evaluated for all angles in the unit interval, meaning that one could calculate the minima or maxima relative to $\chi^2$ and $\kappa^2$, and then impose the conditions solely on these extrema (instead of the entire interval). Thus, constraining \emph{only the extrema} is necessary and sufficient for the original causality conditions, yet significantly less computationally intensive. In terms of equations, we reiterate the main paper's example by remarking that the following statements are equivalent (where $(\chi^2,\kappa^2) = (\widetilde{\chi}^2,\widetilde{\kappa}^2)\in[0,1]^2$ is a critical point providing the global minimum in the domain $[0,1]^2$):
\be
\forall \chi^2,\kappa^2\in [0,1];\quad \mathcal{A}_0(\chi^2,\kappa^2)\geq 0\quad\Leftrightarrow\quad \min_{\chi^2,\kappa^2\in[0,1]}\mathcal{A}_0(\chi^2,\kappa^2) \equiv \mathcal{A}_0(\widetilde{\chi}^2,\widetilde{\kappa}^2)\geq 0.
\ee
The coefficients $\mathcal{A}_1$, $\mathcal{A}_2$ and superpositions thereof may be dealt with exactly in this manner.\\

\noindent\emph{Step I: Removing angles.} We shall assume that $\eta_{\textrm{eff}}/\tau_\pi\neq 0$ along with $\Lambda_1\neq \Lambda_3$, as these conditions are always true for non-zero shear-stress (also, plugging in $\eta_{\textrm{eff}}/\tau_\pi = 0$ decouples the cubic into a monomial in $y = \hat{z}^2$ times $y^2$). The existence of real roots is already proven for only bulk viscosity in \cite{Bemfica:2019cop}. Thus, the conditions $\mathcal{A}_0(\chi^2,\kappa^2)\geq 0$ and $0\leq \mathcal{A}_2(\chi^2,\kappa^2)\leq 3$ are equivalent to 
\bml
\bea
0\leq \min_{\chi^2,\kappa^2\in[0,1]}\mathcal{A}_0(\chi^2,\kappa^2) &=&
\begin{dcases}
\frac{\eta_{\textrm{eff}}^2}{\tau_\pi^2}\dfrac{c_{\textrm{eff},1}^2 + \dfrac{1}{E + \Lambda_1}\dfrac{\eta_{\textrm{eff}}}{\tau_\pi}}{(E + \Lambda_2)(E + \Lambda_3)}, &P_\varepsilon\geq 0,\\
\frac{\eta_{\textrm{eff}}^2}{\tau_\pi^2}\dfrac{c_{\textrm{eff},3}^2 + \dfrac{1}{E + \Lambda_3}\dfrac{\eta_{\textrm{eff}}}{\tau_\pi}}{(E + \Lambda_1)(E + \Lambda_2)}, &P_\varepsilon < 0.
\end{dcases}\\
0\leq \min_{\chi^2,\kappa^2\in[0,1]}\mathcal{A}_2(\chi^2,\kappa^2) &=&
\begin{dcases}
c_{\textrm{eff},3}^2 + \frac{\eta_{\textrm{eff}}}{\tau_\pi}\sum_{k = 1}^3\frac{1}{E +\Lambda_k}, &\frac{\zeta_{\textrm{eff}}}{\tau_\Pi}\geq 0,\\
c_{\textrm{eff},1}^2 + \frac{\eta_{\textrm{eff}}}{\tau_\pi}\sum_{k = 1}^3\frac{1}{E +\Lambda_k}, &\frac{\zeta_{\textrm{eff}}}{\tau_\Pi} < 0.
\end{dcases}\\
3\geq \max_{\chi^2,\kappa^2\in[0,1]}\mathcal{A}_2(\chi^2,\kappa^2) &=& 
\begin{dcases}
c_{\textrm{eff},1}^2 + \frac{\eta_{\textrm{eff}}}{\tau_\pi}\sum_{k = 1}^3\frac{1}{E +\Lambda_k}, &\frac{\zeta_{\textrm{eff}}}{\tau_\Pi}\geq 0,\\
c_{\textrm{eff},3}^2 + \frac{\eta_{\textrm{eff}}}{\tau_\pi}\sum_{k = 1}^3\frac{1}{E +\Lambda_k}, &\frac{\zeta_{\textrm{eff}}}{\tau_\Pi} < 0.
\end{dcases}
\eea
\eml
All these calculations are performed on Mathematica, using the \texttt{Maximize[]} and \texttt{Minimize[]} functions that evaluate global extrema provided a set of conditions \cite{CordeiroGit2026}. As we cannot guarantee the relative sizes of all of the dynamics/transport coefficients, the maxima and minima often have a large amount of subcases. Thus, for the remaining conditions, we evaluate the expression at multiple values of $\chi^2$ and $\kappa^2$ such that it is guaranteed that the true minimum exists at one of the points. For instance, one may show that $\mathcal{A}_1(\chi^2,\kappa^2)\geq 0$ for all $\chi^2,\kappa^2\in[0,1]$ is equivalent to satisfying the following condition for all $a = 1,2,3$:
\be
\frac{\eta}{\tau_\pi}\frac{2\left(E - \dfrac{\Lambda_a}{2}\right)(E + \Lambda_a)c_{\textrm{eff},a}^2+ 3\dfrac{\eta_{\textrm{eff}}}{\tau_\pi}E}{(E + \Lambda_1)(E + \Lambda_2)(E + \Lambda_3)} \geq 0
\ee
Notice that it suffices to show for all $(a,b,c) = (1,2,3),(1,3,2),(2,3,1)$ that $1 - \mathcal{A}_2 + \mathcal{A}_1 -\mathcal{A}_0 \geq 0$ for all $\chi^2,\kappa^2\in[0,1]$ is equivalent to
\be
\left(1 - \dfrac{1}{E + \Lambda_a}\dfrac{\eta_{\textrm{eff}}}{\tau_\pi}\right)\left(1 - \dfrac{1}{E + \Lambda_b}\dfrac{\eta_{\textrm{eff}}}{\tau_\pi}\right)\left(1-c_{\textrm{eff,c}}^2 - \dfrac{1}{E + \Lambda_c}\dfrac{\eta_{\textrm{eff}}}{\tau_\pi}\right)\geq 0
\ee
as well as for all $a = 1,2,3$ that
\begin{align}
1 &-\frac{2}{3}\left(c_{\textrm{eff},a}^2 + \frac{\eta_{\textrm{eff}}}{\tau_\pi}\sum_{k = 1}^3\frac{1}{E + \Lambda_k}\right) + \dfrac{\eta_{\textrm{eff}}}{\tau_\pi}\left[\frac{\left(E - \dfrac{\Lambda_a}{3}\right) (E + \Lambda_a)c_{\textrm{eff},a}^2 + \dfrac{8}{9}\dfrac{\eta_{\textrm{eff}}}{\tau_\pi}E}{(E + \Lambda_1)(E + \Lambda_2)(E + \Lambda_3)}\right]\geq 0
\end{align}
is equivalent to $1 - \frac{2}{3}\mathcal{A}_2 + \frac{1}{3}\mathcal{A}_1 \geq 0$ for all $\chi^2,\kappa^2\in[0,1]$. Using the above inequalities removes the angular dependence from Theorem~1, except for the discriminant condition $\Delta\geq 0$.\\

\noindent\emph{Step II: Removing the discriminant condition.} We thank L.~Gavassino for suggesting this step. Here, we show that there exists a regime where the discriminant condition in Eq.~\eqref{eq:C1} may be removed entirely for the generalized Maxwell-Cattaneo equations. Define the points
\be
\hat{z}_a^2 = \frac{\eta_{\textrm{eff}}}{\tau_\pi}\frac{1}{E + \Lambda_a}.
\ee
By our choice of indices, we remark that $0\leq \hat{z}_3^2 \leq \hat{z}_2^2 \leq \hat{z}_1^2$. Evaluating the cubic polynomial at these three points in Eq.~\eqref{eq:MCdetApp}, one finds that the cubic portion of the characteristic determinant evaluated at these points is equal to 
\be
\mathcal{P}_3(\hat{z}^2 = \hat{z}_a^2) = (-1)^a\frac{\eta_{\textrm{eff}}^2}{\tau_\pi^2}\hat{v}_a^2c_{\textrm{eff},a}^2\prod_{b\neq a}\frac{|\Lambda_b - \Lambda_a|}{(E + \Lambda_a)(E + \Lambda_b)}
\ee
Notice that the overall sign of the polynomial depends on two terms. $(-1)^a$ and effective modulations to the hydrodynamic speed of sound, $c_{\textrm{eff},a}^2$. For the rest of the proof, we shall assume that
\be
\forall a = 1,2,3;\qquad c_{\textrm{eff},a}^2 \geq 0.
\ee
Notice that this assumption is not that stringent. In fact, it is always true if $P_\varepsilon \geq 0$ and $\zeta_{\textrm{eff}}/\tau_\Pi \geq 0$ - two conditions that are guaranteed by \emph{linear} causality \cite{Hiscock_Lindblom_stability_1983,Olson:1989ey}. Suppose that for all $a,b= 1,2,3$, $\Lambda_a\neq \Lambda_b$ for any $b\neq a$. This assumption guarantees that $\hat{z}_3^2 < \hat{z}_2^2 < \hat{z}_1^2$. Since $\mathcal{P}_3(y)$ is a cubic in $y \in\mathbb{R}$, it follows that $\mathcal{P}_3\in C^\infty (\mathbb{R})$. Thus, we are guaranteed $3$ real roots since the function crosses the $y = 0$ axis three times, recognizing that for arbitrarily large/small $y$,
\be
\lim\limits_{y\rightarrow \pm\infty}\mathcal{P}_3(y) = \pm \infty,
\ee
and then applying the intermediate value theorem three times using $\mathcal{P}_3(\hat{z}_1^2) \leq 0$, $\mathcal{P}_3(\hat{z}_2^2) \geq 0$, and $\mathcal{P}_3(\hat{z}^2 = \hat{z}_3^2) \leq 0$. It remains to show that this holds for the cases where $\Lambda_1 = \Lambda_2$ ($\hat{z}_1^2 = \hat{z}_2^2$) and $\Lambda_2 = \Lambda_3$ ($\hat{z}_2^2 = \hat{z}_3^2$). Remember that we are assuming that not all three can be equal, as this case would be the zero shear-stress case. If $\Lambda_1 = \Lambda_2$, then $\hat{z}_1^2 = \hat{z}_2^2$ is a root of $\mathcal{P}_3$. One can write
\be
\mathcal{P}_3(y)\bigg|_{\Lambda_2 = \Lambda_1} = \left(y - \hat{z}_1^2\right)\mathcal{Q}_{13}(y),
\ee
where $\mathcal{Q}_{13}$ is a quadratic polynomial in $y$ such that
\begin{align}
\mathcal{Q}_{13}(y) = y^2& -\left[P_\varepsilon + \left(\frac{1 - \hat{v}_3^2}{E + \Lambda_1} + \frac{\hat{v}_3^2}{E + \Lambda_3}\right)\frac{\zeta_{\textrm{eff}}}{\tau_\Pi} + \left(\frac{1 }{E + \Lambda_1} + \frac{1}{E + \Lambda_3}\right)\frac{\eta_{\textrm{eff}}}{\tau_\Pi}\right]y\notag\\
& + \left[\left(\frac{1 - \hat{v}_3^2}{E + \Lambda_3} + \frac{\hat{v}_3^2}{E + \Lambda_1}\right)P_\varepsilon + \frac{1}{(E + \Lambda_1)(E + \Lambda_3)}\left(\frac{\zeta_{\textrm{eff}}}{\tau_\Pi} + \frac{\eta_{\textrm{eff}}}{\tau_\Pi}\right)\right]\frac{\eta_{\textrm{eff}}}{\tau_\Pi}.
\end{align}
In a similar fashion, notice that
\bml
\bea
\mathcal{Q}_{13}(\hat{z}_1^2) &=& -(1 - \hat{v}_3^2)\frac{\eta_{\textrm{eff}}}{\tau_\Pi}\frac{\Lambda_3 - \Lambda_1}{(E + \Lambda_1)(E + \Lambda_3)}c_{\textrm{eff},1}^2,\\
\mathcal{Q}_{13}(\hat{z}_3^2) &=& \hat{v}_3^2\frac{\eta_{\textrm{eff}}}{\tau_\Pi}\frac{\Lambda_3 - \Lambda_1}{(E + \Lambda_1)(E + \Lambda_3)}c_{\textrm{eff},3}^2.
\eea
\eml
Per our assumptions, notice that $\hat{z}_1^2 > \hat{z}_3^2$, while $\mathcal{Q}_{13}(\hat{z}_1^2)\leq 0$ and $\mathcal{Q}_{13}(\hat{z}_3^2)\geq 0$. Thus, it follows that there exists one \emph{real} root in between $\hat{z}_3^2$ and $\hat{z}_1^2$ by the intermediate value theorem. Since roots appear in conjugate pairs for polynomials with real coefficients, it follows that there then must be two real roots in $\mathcal{Q}_{13}$, thus proving that there must be three real roots in $\mathcal{P}_3$.

Now, suppose instead that $\Lambda_2 = \Lambda_3$ (and not $\Lambda_1$). Then $\hat{z}_2^2 = \hat{z}_3^2$ is a root, such that $\mathcal{P}_3(\hat{z}_3^2) = 0$. However, recall that $\hat{z}_3^2 < \hat{z}_1^2$, and $\mathcal{P}_3(\hat{z}_1^2) \leq 0$, whereas there exists $y' > \hat{z}_1^2$ such that $\mathcal{P}_3(y') > 0$, as $\mathcal{P}_3$ is a cubic with a positive leading order coefficient, and is thus an increasing, continuous function for large $y$. Therefore, the intermediate value theorem guarantees that either $\hat{z}_1^2$ is a root, or that there exists some root $y_0 > \hat{z}_1^2$ such that $\mathcal{P}_3$ has at least two real roots. Again, $\mathcal{P}_3$ is a cubic with real coefficients, meaning that there must then be three real roots. Therefore, in any case, we have three real roots under the assumption that $c_{\textrm{eff},a}^2\geq 0$ for any $a = 1,2,3$, meaning that we can discard the discriminant condition in Eq.~\eqref{eq:C1}, as it is superfluous.

\end{document}